\documentclass[11pt]{article}

\usepackage{amsmath,amssymb,amsfonts,amsthm,mathtools}
\usepackage{thmtools}
\makeatletter
\IfFormatAtLeastTF{2024-11-01}{
  \renewcommand*\@addtoreset[2]{%
    \bgroup
      \edef\aliasctr@@truelist{\aliasctr@follow{#2}}%
      \let\@elt\relax
      \expandafter\@cons\aliasctr@@truelist{{#1}}%
    \egroup
    \expandafter\xdef\csname theH#1\endcsname{%
      \expandafter\noexpand\csname theH#2\endcsname.%
      \noexpand\the\noexpand\value{#1}}%
  }
}{}
\makeatother
\usepackage{thm-restate}

\usepackage[utf8]{inputenc}
\usepackage[T1]{fontenc}
\usepackage[tt=false, type1=true]{libertine}
\usepackage[varqu]{zi4}
\usepackage[libertine]{newtxmath}
\usepackage[margin=1in]{geometry}
\usepackage{microtype,parskip}
\usepackage{placeins}
\usepackage{float}

\usepackage{graphicx,xcolor}

\usepackage{booktabs,array,multirow,tabularx}

\usepackage{xcolor}
\usepackage[ruled,vlined]{algorithm2e}
\SetKwInput{KwRequire}{Requires}

\SetAlgoNoLine
\DontPrintSemicolon
\LinesNumbered

\SetKwComment{Comment}{\textcolor{blue}{// }}{}
\newcommand{\cmt}[1]{\Comment*[r]{\textcolor{blue}{#1}}}

\usepackage{essgraph}
\usepackage{subcaption}
\newcommand{\EssDagEdgeLineWidth}{1.2pt}

\tikzset{
  dag clean edge/.style={
    draw=black,
    line width=\EssDagEdgeLineWidth,
    ->,
    shorten >=1.5pt,
    shorten <=1.5pt
  },
  dag clean curved edge/.style={
    dag clean edge,
    preaction={
      draw=white,
      line width=3.4pt,
      shorten >=1.5pt,
      shorten <=1.5pt
    }
  }
}

\definecolor{ForestGreen}{rgb}{0.1333,0.5451,0.1333}
\definecolor{DarkRed}{rgb}{0.65,0,0}
\usepackage{hyperref}
\hypersetup{
    colorlinks=true,
    linkcolor=DarkRed,
    filecolor=magenta,      
    urlcolor=cyan,
    citecolor=ForestGreen,
    linktocpage=true,
    pagebackref=true,
    bookmarks=true,
    bookmarksopen=true,
    bookmarksnumbered=true
}
\usepackage[nameinlink]{cleveref}
\usepackage{csquotes}
\usepackage[style=alphabetic,natbib=true,maxalphanames=3,minalphanames=3,,maxbibnames=99]{biblatex}
\usepackage{tikz}
\usetikzlibrary{arrows.meta,positioning,calc,decorations.pathreplacing,matrix,patterns,decorations.pathmorphing, decorations.markings,external, fit}
\IfFileExists{cache.tex}{
    \input{cache.tex}
}{}

\usepackage{pdfcomment}

\usepackage[showdeletions]{color-edits}
\addauthor[Mahdi]{m}{cyan}
\addauthor[MTH]{mth}{blue}
\addauthor[Kavi]{kavi}{magenta}
\addauthor[Soheil]{soh}{orange}

\usepackage{xspace,nicefrac} 
\theoremstyle{definition}
\theoremstyle{plain}

\newtheorem{theorem}{Theorem}
\newtheorem{corollary}{Corollary}[section]
\newtheorem{lemma}[corollary]{Lemma}

\newtheorem{proposition}[corollary]{Proposition}
\newtheorem{definition}[corollary]{Definition}

\newtheorem{conjecture}[corollary]{Conjecture}

\newtheorem{claim}{Claim}[lemma]

\AddToHook{cmd/appendix/before}{\crefalias{section}{appendix} \crefalias{subsection}{appendix}}

\newcommand{\lovasz}{Lov\'asz\xspace}
\newcommand{\gyori}{Gy\H{o}ri\xspace}
\newcommand{\gyorilovasz}{\gyori--\lovasz}
\newcommand{\essentialassignmentcond}{Flow-Essential Assignment Condition\xspace}
\newcommand{\essentialassignmentcondition}{Flow-Essential Assignment Condition\xspace}

\newcommand{\fractionalesscond}{Flow-Essential Split-Assignment Condition\xspace}
\newcommand{\FractionalEssAssign}{flow-essential split-assignment\xspace}
\xspaceaddexceptions{’}

\newcommand{\caC}{\mathcal{C}}
\newcommand{\caR}{\mathcal{R}}
\newcommand{\caO}{\mathcal{O}}

\newcommand{\term}{t}
\newcommand{\cp}{c}
\newcommand{\conn}{\kappa}
\newcommand{\connvg}[2]{\conn_{#2}(#1)}
\newcommand{\Reach}{\mathfrak{R}}

\newcommand{\assgn}{\phi}
\newcommand{\assgninv}{\phi^{-1}}
\newcommand{\assgnalt}{\varphi}

\newcommand{\assgnw}{\psi}
\newcommand{\assgnwinv}{\psi^{-1}}
\newcommand{\assgncompact}{\sigma}
\newcommand{\crt}{\xi}
\DeclareMathOperator{\PT}{PT}

\newcommand{\Phifrac}{\bar\Phi}

\newcommand{\cut}{C}
\newcommand{\lcut}{L_C}
\newcommand{\scut}{S_C}
\newcommand{\rcut}{R_C}
\newcommand{\lcutc}[1]{L_{#1}}
\newcommand{\scutc}[1]{S_{#1}}
\newcommand{\rcutc}[1]{R_{#1}}

\newcommand{\demand}{%
  \ifmmode d\else demand\fi
}

\newcommand{\weight}{w}
\newcommand{\conf}{\tau}

\title{Breaking the Exponential Barrier: The First Polynomial-Time Algorithm for the \gyori--\lovasz Theorem~\footnote{The main result was obtained in late April 2026, after almost three years of sustained effort by the authors. Since then, the authors have focused on developing the full exposition, including detailed proofs, intuition, and the key barriers overcome. All ideas, technical results, proofs, and scientific contributions are entirely the authors' own; AI assistance was limited to proofreading, language refinement, and minor expository improvements.}}

\author{ 
Mohammad T. Hajiaghayi\thanks{ACM Fellow \& IOI'97 Medalist} 
\and Mahdi JafariRaviz\thanks{IOI'19 Medalist} 
\and Alireza Kaviani\thanks{IOI'21 \& IOI'22 Medalist} 
\and Soheil Mohammadkhani \thanks{IOI'23 Medalist\\ Email: \texttt{\{hajiaghayi,mahdijafariraviz,akaviani05,mohammadkhani.soheil\}@gmail.com}} }

\date{\today}

\begin{document}
\date{}
\maketitle
\begin{abstract}
We give the first polynomial-time algorithm, after half a century, for the celebrated \gyori--\lovasz theorem, which resolved a conjecture of Frank (1975). The theorem, one of the simplest existential theorems to explain, states that every $k$-connected graph can be partitioned into $k$ disjoint connected subgraphs of arbitrary prescribed positive sizes. This is a fundamental structural result with broad applications, such as flexible allocation of connected subnetworks of prescribed sizes in sufficiently connected cloud infrastructures.

While \lovasz (1977) gave a highly non-constructive proof for a stronger directed version using algebraic topology, \gyori's original constructive proof (1976) requires exponential time. 
Despite more than 50 years of effort, no polynomial-time algorithm was known even for $k>4$. Determining the computational complexity of the \gyori--\lovasz theorem---whether it admits even a sub-exponential-time algorithm or is computationally hard (in particular, PLS-complete or PPAD)---has remained one of the central open problems in algorithmic graph theory.

In this paper, we finally resolve this long-standing problem by a fundamentally new proof of the existential theorem via introducing the novel concept of \emph{flow-essential assignment}, which genuinely marries matching and cut structures and yields the first polynomial-time constructive algorithm for the \gyori--\lovasz theorem.
In fact, we obtain a polynomial-time algorithm for \lovasz's stronger directed version, whose proof was non-constructive even for DAGs; for DAGs, we further obtain a near-linear-time algorithm. We also develop polynomial-time algorithms for weighted generalizations where the seminal work of Chen, Kleinberg, \lovasz, Rajaraman, Sundaram, and Vetta (JACM'07) on confluent flows established only existential non-constructive results.

\end{abstract}
\newpage
\tableofcontents
\newpage
\section{Introduction}

Perhaps one of the best-known and simplest-to-explain existential theorems in algorithmic graph theory states that every $k$-connected graph---one that remains connected after the removal of any $k-1$ vertices---can be partitioned into $k$ disjoint connected subgraphs of arbitrary prescribed positive sizes. This is a fundamental problem with broad applications, including, to name just a few, flexible allocation of connected subnetworks of prescribed sizes in cloud infrastructures and clustering tasks such as partitioning images into connected regions of predicted sizes; see, e.g., \citep{casel2023efficient,niklanovits2025connected} for further applications to road networks, robotics, and image processing. Although existential proofs of this theorem have been known for decades---yielding exponential-time algorithms---obtaining a polynomial-time algorithm, even for $k>4$, remained open despite persistent efforts for more than 50 years.
The situation is particularly intriguing because other important problems such as Nash Equilibrium and Weighted Local Max-Cut also admit existential proofs implying exponential-time algorithms, yet were later shown to be PPAD-complete~\citep{papadimitriou1994complexity,daskalakis2009complexity,chen2009settling} and PLS-complete~\citep{johnson1988easy,schaffer1991simple}, respectively, naturally leading to the belief that this problem might likewise be computationally intractable. In this paper, we finally resolve this major open problem by giving the first polynomial-time algorithm via a fundamentally new proof of the existential theorem.

Indeed, the existential theorem above is more general and was first conjectured and partially resolved by \citet{frank1975combinatorial}: given a $k$-connected graph $G$, distinct vertices $v_1,\ldots,v_k$, and positive integers $n_1,\ldots,n_k$ satisfying $\sum_i n_i = |V(G)|$, there exist pairwise disjoint connected subgraphs $G_1,\ldots,G_k$ such that each $G_i$ contains $v_i$ and has exactly $n_i$ vertices.
Here, both the roots and the sizes of the parts are prescribed in advance. While \gyori~\citep{gyori1976division} gave the first constructive proof in 1976 with an exponential-time algorithm, \lovasz~\citep{lovasz1977homology} independently established a stronger directed version in 1977 via a highly non-constructive algebraic topological argument: Let $G$ be a digraph, let $S=\{v_1,\ldots,v_k\}\subseteq V(G)$, and suppose $G$ is $k$-connected to $S$, i.e., each vertex $v \not\in S$
can be joined to $S$ by $k$ paths which are vertex-disjoint except for the common vertex $v$. For any positive integers $n_1,\ldots,n_k$ satisfying $\sum_{i=1}^k n_i = |V(G)|$, the graph $G$ contains pairwise vertex-disjoint arborescences $A_1,\ldots,A_k$ such that each $A_i$ is rooted at $v_i$ and satisfies $|V(A_i)| = n_i$, where an arborescence is a directed tree in which every vertex reaches the root.

\gyori's proof is constructive but inherently exponential; a clear modern exposition was later given by \citet{hoyer2016gyori}.
The argument maintains a partition into connected parts and repeatedly attempts to enlarge a deficient part by moving vertices while preserving connectivity. When a direct move is impossible, the proof recursively follows the components that would be disconnected, creating a chain of dependent choices known as a \emph{cascade}. By carefully rerouting vertices along such cascades, the proof eventually increases the deficient part while maintaining connectivity throughout. However, cascades can be arbitrarily long and their number grows exponentially, so any direct implementation of \gyori's method yields an exponential-time algorithm. 

\lovasz's proof is a landmark application of algebraic-topological methods to discrete graph theory.  At its core, \lovasz's proof transforms a discrete graph partitioning problem into a continuous topological one. Rather than directly searching over combinatorial structures, the proof considers a topological space $K$ representing feasible spanning configurations, together with a continuous map assigning to each configuration a vector encoding the sizes of its components. The target is to show that this map attains a prescribed vector corresponding to the desired partition. Drawing on homological connectivity arguments---functioning as a high-dimensional analogue of the Intermediate Value Theorem---\lovasz proves that $K$
is sufficiently connected to force the map to pass through the target point. This establishes the existence of the required partition or arborescence structure in an elegant but fundamentally non-constructive way, yielding no efficient algorithm for finding it. \lovasz's proof parallels the existence proof for Nash equilibria: both rely on non-constructive topological arguments rather than efficient algorithms. Just as Nash used Brouwer's fixed-point theorem to establish existence, \lovasz used algebraic topology to guarantee the required partition. 
Since such principles often underlie computationally hard classes such as PPAD and PLS, it was natural to suspect that the celebrated \gyori--\lovasz theorem might likewise be computationally intractable---a possibility formalized by \citet{chandran2018spanning}, who placed the problem in PLS and asked whether it is PLS-complete\footnote{Strictly speaking, our polynomial-time algorithms do not rule this out!}, and by \citet{meunier2022existence}, who asked whether it lies in PPAD.

The problem has attracted significant effort for small values of $k$ and special graph classes.
Polynomial-time algorithms were previously known only in restricted settings: \citet{suzuki1990linear} solved the case $k=2$, \citet{suzuki1990algorithm} and \citet{wada1993efficient} handled $k=3$, and \citet{hoyer2019independent} resolved the general case $k=4$ following earlier planar results of \citet{nakano1997linear}. For $k=2$, \citet{an2022diameter} also recently gave a direct polynomial-time search based on \lovasz's proof.
Subsequent work obtained polynomial-time algorithms only for restricted graph classes or relaxed variants: chordal graphs and near-prescribed sizes \citep{casel2023efficient}, bounded-treewidth graphs for fixed $k$ \citep{casel2026connected}, constant-factor size violations for uniform capacities \citep{borndorfer2021connected}, and stronger $\Omega(k\log^2 n)$ connectivity \citep{niklanovits2025connected}. For the exact theorem on general graphs, the best known construction remained exponential: \citet{chandran2018spanning} gave a $4^n$-time local search, improving to $2^{\caO((n/k)\log k)}$ under $(2k-2)$-connectivity. Moreover, $k$-connectivity is essential: without it, the problem is NP-complete even when all parts have equal size and no roots are prescribed \citep{dyer1985complexity}. Thus, no polynomial-time algorithm was known for the general \gyori--\lovasz theorem, or even for any fixed $k\geq 5$.

In this paper, we give the first polynomial-time algorithm for the \gyori--\lovasz theorem for arbitrary $k$. Our approach develops a fundamentally new proof of the existential theorem and introduces the novel concept of \emph{flow-essential assignment}, which genuinely marries matching and cut structures. The algorithm constructs the partition incrementally through edge contractions: at each step, we identify an edge $(v_i,u)$ whose contraction permanently assigns $u$ to the part rooted at $v_i$ while preserving a carefully designed global invariant.
The central challenge is designing an invariant that is simultaneously strong enough to guarantee the existence of a valid contractible edge and structured enough to be checked efficiently. Natural candidates, such as preserving $k$-connectivity itself, fail dramatically: we show examples where no valid contraction exists despite the graph remaining highly connected. Our flow-essential assignment framework overcomes this barrier by maintaining the key invariant throughout, yielding a fully polynomial-time construction.
The same framework extends to the directed theorem of \lovasz, providing the first direct polynomial-time construction for a result previously known only through non-constructive topological arguments. It also extends to directed acyclic graphs, where we get a near-linear-time algorithm, and to the weighted setting.
More broadly, our approach reveals that the \gyori--\lovasz theorem is not purely about connectivity: prescribed part sizes impose a matching structure, and it is precisely the genuine marriage of matching and cuts that breaks the exponential barrier.

Confluent flows---where all flow leaving a node follows a single outgoing edge---arise naturally in networking and Internet routing. A seminal result of Chen, Kleinberg, \lovasz, Rajaraman, Sundaram, and Vetta~\citep{chen2007almost} gave a polynomial-time algorithm achieving congestion $1+\ln k$ whenever a splittable flow of congestion $1$ exists, together with a nearly matching lower bound of $H_k$, the $k$-th harmonic number. They further proved that this gap collapses in $k$-connected graphs with $k$ sinks, showing the existence of confluent flows with congestion at most $C+d_{\max}$, where $C$ is the optimal splittable congestion and $d_{\max}$ is the maximum demand. This weighted generalization of the \gyori--\lovasz theorem relied on the topological methods of \lovasz~\citep{lovasz1977homology} and was non-constructive. When we generalize our algorithms to the weighted setting, it provides the missing constructive step by giving the first polynomial-time algorithm for the corresponding existential theorem. The algorithms remain polynomial even with exponentially large weights (demands) and nonuniform capacities.

\subsection{Our Results}
\label{subsec:our-results}
Matchings and cuts are among the most studied objects in combinatorial optimization. The two classical proofs of the \gyorilovasz theorem---the constructive proof of \gyori~\citep{gyori1976division} and the topological proof of \lovasz~\citep{lovasz1977homology}---treat the theorem as a pure connectivity problem: they reason about cuts, vertex-disjoint paths, and connected parts. The algorithms for small $k$~\citep{suzuki1990linear,suzuki1990algorithm,wada1993efficient,hoyer2019independent} do the same. As we show in this paper, matching plays an equally central role. The reason is the capacity constraint: each root must receive a prescribed number of vertices. This counting requirement is a bipartite matching constraint between vertices and roots. We marry the two concepts and introduce the \essentialassignmentcondition. It asks for an assignment with two properties: every vertex goes to a root whose removal decreases the connectivity of that vertex to the root set, and every root receives exactly its prescribed number of vertices. We show how to maintain this condition until the algorithm reaches the partition. To the best of our knowledge, no earlier work considered such a combination.

We work in the directed setting of \lovasz. Let $G=(V,E)$ be a directed graph with a set $T=\{t_1,\dots,t_k\}$ of $k$ terminals. We use \emph{terminal} throughout the paper for what prior work and the discussion above call a root. The graph is $k$-$T$-connected if every non-terminal vertex has $k$ vertex-disjoint directed paths to $k$ distinct terminals. This hypothesis is exactly the one of \lovasz, who calls it $k$-connectivity to $T$. \gyori instead assumes standard $k$-connectivity; this case is included, since a $k$-connected graph---with each undirected edge replaced by two opposite directed edges---is $k$-$T$-connected for every set of $k$ terminals. Given nonnegative integers $c_1,\dots,c_k$ with $\sum_{i=1}^k c_i=|V\setminus T|$, the goal is to partition $V$ into parts $V_1,\dots,V_k$ such that $t_i\in V_i$, $|V_i|=c_i+1$ (that is, $c_i$ non-terminal vertices and $t_i$), and every vertex of $V_i$ has a directed path to $t_i$ inside $V_i$. We refer to \Cref{sec:preliminaries} for the formal definitions.

We call this the unweighted setting. In \Cref{sec:weighted}, we study the generalization introduced by \citet{chen2007almost} for confluent flows, in which every non-terminal carries an integer weight and $c_i$ bounds the total weight of the part of $t_i$. We call it the weighted setting, and unit weights bring us back to the unweighted one. We still prove the unweighted setting on its own, because it is the setting of the \gyorilovasz theorem itself and it presents the main ideas of this paper in their clearest form. Our main result is the following theorem, which we prove in \Cref{sec:algorithm}.

\begin{restatable}{theorem}{thmktconn}
\label{thm:k-t-conn}
Let $G = (V, E)$ be a directed graph and $T = \{t_1,\dots,t_k\}$ be a set of distinct terminals. Assume that $G$ is $k$-$T$-connected and let $c_1, \dots, c_k$ be nonnegative integers such that $\sum_{i=1}^k c_i = |V \setminus T|$. Then there exists a polynomial-time algorithm that partitions the vertices of the graph into $V_1, \dots, V_k$, such that $t_i \in V_i$, $|V_i| = c_i + 1$, and the induced subgraph $G[V_i]$ is connected to $t_i$.
\end{restatable}

Every part of the final partition must stay connected to its terminal. A natural approach is to build the parts by contraction. Call a non-terminal vertex with an edge to a terminal a pre-terminal. The contraction of an edge $(p,t)$ from a pre-terminal $p$ into a terminal $t$ permanently adds $p$ to the part of $t$. To repeat this step until no vertex remains, we must guarantee that after each contraction, the resulting instance can still be completed to a valid partition. We therefore seek a condition that can be maintained throughout the algorithm and guarantees that, at every step, some pre-terminal can be contracted while preserving the condition.

\begin{figure}[htbp]
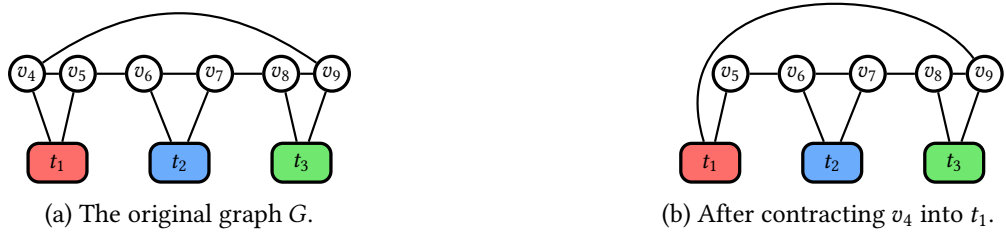

\centering

\begin{subfigure}{0.48\textwidth}
\centering
\begin{essgraph}[scale=0.82]

  \essterminal{t1}{(-2.40,-1.75)}{$t_1$}
  \essterminal{t2}{( 0.00,-1.75)}{$t_2$}
  \essterminal{t3}{( 2.40,-1.75)}{$t_3$}

  \essvertex{v4}{(-3.00, 0.00)}{$v_4$}
  \essvertex{v5}{(-2.00, 0.00)}{$v_5$}
  \essvertex{v6}{(-0.70, 0.00)}{$v_6$}
  \essvertex{v7}{( 0.70, 0.00)}{$v_7$}
  \essvertex{v8}{( 2.00, 0.00)}{$v_8$}
  \essvertex{v9}{( 3.00, 0.00)}{$v_9$}

  \begin{scope}[on background layer]
    \essedge{t1}{v4}
    \essedge{t1}{v5}

    \essedge{t2}{v6}
    \essedge{t2}{v7}

    \essedge{t3}{v8}
    \essedge{t3}{v9}

    \essedge{v4}{v5}
    \essedge{v5}{v6}
    \essedge{v6}{v7}
    \essedge{v7}{v8}
    \essedge{v8}{v9}

    \draw[line width=\EssEdgeLineWidth]
      (v4) to[out=35,in=145,looseness=1.08] (v9);
  \end{scope}

\end{essgraph}
\caption{The original graph $G$.}
\label{fig:original-graph}
\end{subfigure}
\hfill
\begin{subfigure}{0.48\textwidth}
\centering
\begin{essgraph}[scale=0.82]

  \essterminal{t1}{(-2.40,-1.75)}{$t_1$}
  \essterminal{t2}{( 0.00,-1.75)}{$t_2$}
  \essterminal{t3}{( 2.40,-1.75)}{$t_3$}

  \essvertex{v5}{(-2.00, 0.00)}{$v_5$}
  \essvertex{v6}{(-0.70, 0.00)}{$v_6$}
  \essvertex{v7}{( 0.70, 0.00)}{$v_7$}
  \essvertex{v8}{( 2.00, 0.00)}{$v_8$}
  \essvertex{v9}{( 3.00, 0.00)}{$v_9$}

  \begin{scope}[on background layer]
    \essedge{t1}{v5}

    \essedge{t2}{v6}
    \essedge{t2}{v7}

    \essedge{t3}{v8}
    \essedge{t3}{v9}

    \essedge{v5}{v6}
    \essedge{v6}{v7}
    \essedge{v7}{v8}
    \essedge{v8}{v9}

    \draw[line width=\EssEdgeLineWidth]
      (t1) .. controls (-3.50,1.85) and (1.55,2.00) .. (v9);
  \end{scope}

\end{essgraph}
\caption{After contracting $v_4$ into $t_1$.}
\label{fig:contract-v4}
\end{subfigure}

\caption{A graph that is $3$-connected to the set of terminals
$T=\{t_1,t_2,t_3\}$, but contracting any pre-terminal into an adjacent
terminal destroys this property. For example, after contracting $v_4$ into
$t_1$, the vertex $v_5$ has only two neighbors, namely $t_1$ and $v_6$,
and therefore cannot have three vertex-disjoint paths to $T$. By symmetry,
the same obstruction occurs for every such contraction. Replacing every
undirected edge by two oppositely directed edges yields the same
counterexample in the directed setting.}
\label{fig:k-conn-counter-example}
\end{figure}

The most natural candidate is $k$-connectivity to the terminal set $T$, but we cannot always maintain this property under contraction. \Cref{fig:k-conn-counter-example} gives a graph that is $3$-connected to the set of terminals $T=\{t_1,t_2,t_3\}$, yet contracting any pre-terminal destroys this property. \citet{gyori1981partition} proved the converse of the \gyori--\lovasz theorem: if a graph admits the desired partition for every choice of $k$ terminals and positive part sizes that sum to $|V|$, then it is $k$-connected. Together, this converse and the contraction counterexample suggest that a useful invariant must reason explicitly about the terminals and most importantly, their capacities. In \Cref{app:relaxed-conditions}, we introduce compact connectivity and local connectivity, two relaxations that incorporate the capacities but still fail to provide a maintainable, efficiently checkable invariant: $k$-$T$-connectivity implies compact connectivity, and compact connectivity implies local connectivity. We show that the \gyorilovasz conclusion still holds under both, which generalizes the theorem itself. Similar to $k$-$T$-connectivity, compact connectivity is also impossible to maintain: there are instances in which no contraction preserves it. Despite this, we show that the framework developed in this paper yields a polynomial-time algorithm under the compact connectivity condition, providing one example of the broader applicability of our algorithm beyond $k$-$T$-connected graphs (\Cref{cor:compact-connectivity-partition}). For local connectivity, the weakest of the three, we conjecture that some contraction always preserves it (\Cref{conj:local-connectivity-contraction}). However, we know of no efficient way to check this condition when $k$ is part of the input, so even this route seems unlikely to give a polynomial-time algorithm. The \essentialassignmentcondition, which we describe next, overcomes this barrier.

\paragraph{The \essentialassignmentcondition.}
We introduce the \essentialassignmentcondition, a polynomial-time checkable generalization of $k$-$T$-connectivity that makes contraction possible. Its scope is much wider: every $k$-$T$-connected graph satisfies it, but so does every instance consisting of $k$ disjoint arborescences of the prescribed sizes, even though such an instance may be far from $k$-$T$-connected. Thus, the condition captures a large family of graphs that already admit the desired partition. For a non-terminal vertex $v$, let $\connvg{v}{G}$ denote the maximum number of vertex-disjoint paths from $v$ to distinct terminals. We call this quantity the terminal connectivity of $v$. We build the condition from the terminals whose removal decreases this quantity.

\begin{restatable}[Flow-Essential Terminals]{definition}{defessentialterminal}
	\label{def:essential-terminal}
	We call a terminal $t \in T$ \emph{flow-essential} for a non-terminal vertex $v \in V \setminus T$, or simply \emph{essential} for $v$, if removing $t$ from $G$ reduces the terminal connectivity of $v$ by exactly $1$, i.e., $\connvg{v}{G \setminus \{t\}} = \connvg{v}{G} - 1$. Equivalently, $t$ must be the destination of a path in every maximum-cardinality family of vertex-disjoint paths from $v$ to $T$.
\end{restatable}

Essential terminals admit two complementary views, and both are central to our proofs. The first view, using paths and illustrated in \Cref{fig:v-essential-terminals}, is immediate from the definition: a terminal $t$ is not essential for $v$ precisely if there is a family of $\connvg{v}{G}$ vertex-disjoint paths from $v$ to distinct terminals that avoids $t$. Equivalently, $t$ is essential if every such maximum family contains a path ending at $t$. The second view, using cuts and illustrated in \Cref{fig:tightest-v-T-cut}, follows from Menger's theorem. The terminal connectivity of $v$ equals the size of a minimum cut that separates $v$ from $T$, and the separator of such a cut can contain terminals. The intersection of all these minimum cuts is again a minimum cut, the tightest one. A terminal is essential for $v$ exactly when it lies in the separator of the tightest minimum cut (\Cref{lem:essential-terminal-tightest-cut}). In words, the essential terminals of $v$ are the terminals inside the bottleneck between $v$ and the terminal set.

\begin{figure}[htbp]
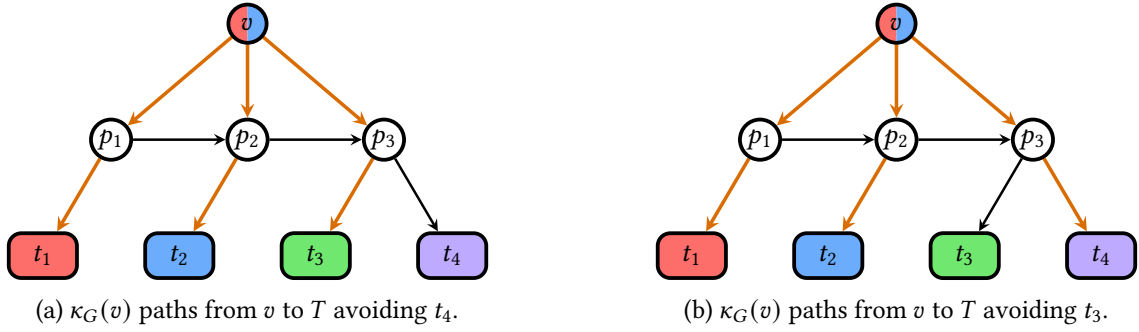

    \centering

    \begin{subfigure}[b]{0.48\textwidth}
        \centering
        \begin{essgraph}[scale=0.95, >=stealth]
            \essterminal{t1}{(-3,0)}{$t_1$}
            \essterminal{t2}{(-1,0)}{$t_2$}
            \essterminal{t3}{( 1,0)}{$t_3$}
            \essterminal{t4}{( 3,0)}{$t_4$}

            \essvertex{p1}{(-2,1.7)}{$p_1$}
            \essvertex{p2}{( 0,1.7)}{$p_2$}
            \essvertex{p3}{( 2,1.7)}{$p_3$}

            \essvertex[colors={t1/0.5,t2/0.5}]{v}{(0,3.4)}{$v$}

            \essedge[extra={->,draw=orange!85!black,line width=1.4pt}]{v}{p1}
            \essedge[extra={->,draw=orange!85!black,line width=1.4pt}]{p1}{t1}

            \essedge[extra={->,draw=orange!85!black,line width=1.4pt}]{v}{p2}
            \essedge[extra={->,draw=orange!85!black,line width=1.4pt}]{p2}{t2}

            \essedge[extra={->,draw=orange!85!black,line width=1.4pt}]{v}{p3}
            \essedge[extra={->,draw=orange!85!black,line width=1.4pt}]{p3}{t3}

            \essedge[extra={->}]{p1}{p2}
            \essedge[extra={->}]{p2}{p3}
            \essedge[extra={->}]{p3}{t4}
        \end{essgraph}
        \caption{$\connvg{v}{G}$ paths from $v$ to $T$ avoiding $t_4$.}
        \label{fig:v-avoid-t4}
    \end{subfigure}
    \hfill
    \begin{subfigure}[b]{0.48\textwidth}
        \centering
        \begin{essgraph}[scale=0.95, >=stealth]
            \essterminal{t1}{(-3,0)}{$t_1$}
            \essterminal{t2}{(-1,0)}{$t_2$}
            \essterminal{t3}{( 1,0)}{$t_3$}
            \essterminal{t4}{( 3,0)}{$t_4$}

            \essvertex{p1}{(-2,1.7)}{$p_1$}
            \essvertex{p2}{( 0,1.7)}{$p_2$}
            \essvertex{p3}{( 2,1.7)}{$p_3$}

            \essvertex[colors={t1/0.5,t2/0.5}]{v}{(0,3.4)}{$v$}

            \essedge[extra={->,draw=orange!85!black,line width=1.4pt}]{v}{p1}
            \essedge[extra={->,draw=orange!85!black,line width=1.4pt}]{p1}{t1}

            \essedge[extra={->,draw=orange!85!black,line width=1.4pt}]{v}{p2}
            \essedge[extra={->,draw=orange!85!black,line width=1.4pt}]{p2}{t2}

            \essedge[extra={->,draw=orange!85!black,line width=1.4pt}]{v}{p3}
            \essedge[extra={->,draw=orange!85!black,line width=1.4pt}]{p3}{t4}

            \essedge[extra={->}]{p1}{p2}
            \essedge[extra={->}]{p2}{p3}
            \essedge[extra={->}]{p3}{t3}
        \end{essgraph}
        \caption{$\connvg{v}{G}$ paths from $v$ to $T$ avoiding $t_3$.}
        \label{fig:v-avoid-t3}
    \end{subfigure}

    \caption{
    The terminals $t_3$ and $t_4$ are not essential for $v$: there are families of three vertex-disjoint paths from $v$ to $T$ that avoid $t_3$ and, respectively, avoid $t_4$. In contrast, $t_1$ and $t_2$ are essential for $v$, since every family of three vertex-disjoint paths from $v$ to $T$ must use both of them.
    }
    \label{fig:v-essential-terminals}
\end{figure}

\begin{figure}[htbp]
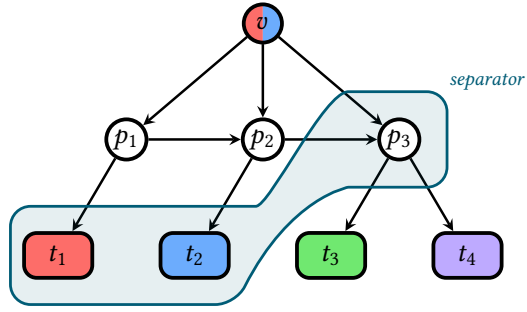

    \centering
    \begin{essgraph}[scale=0.95, >=stealth]
        \essterminal{t1}{(-3,0)}{$t_1$}
        \essterminal{t2}{(-1,0)}{$t_2$}
        \essterminal{t3}{( 1,0)}{$t_3$}
        \essterminal{t4}{( 3,0)}{$t_4$}

        \essvertex{p1}{(-2,1.7)}{$p_1$}
        \essvertex{p2}{( 0,1.7)}{$p_2$}
        \essvertex{p3}{( 2,1.7)}{$p_3$}
        \essvertex[colors={t1/0.5,t2/0.5}]{v}{(0,3.4)}{$v$}

        \essedge[extra={->}]{v}{p1}
        \essedge[extra={->}]{p1}{t1}
        \essedge[extra={->}]{v}{p2}
        \essedge[extra={->}]{p2}{t2}
        \essedge[extra={->}]{v}{p3}
        \essedge[extra={->}]{p3}{t3}
        \essedge[extra={->}]{p3}{t4}
        \essedge[extra={->}]{p1}{p2}
        \essedge[extra={->}]{p2}{p3}

        \begin{scope}[on background layer]
            \path[
                draw=\EssCutOneColor,
                fill=\EssCutOneColor,
                fill opacity=0.10,
                line width=1.10pt,
                rounded corners=8pt
            ]
                (-3.70,0.72) -- (-0.40,0.72)
                .. controls (0.10,0.72) and (0.90,2.40) .. (1.55,2.40)
                -- (2.70,2.40) -- (2.70,1.00) -- (1.55,1.00)
                .. controls (1.25,1.00) and (0.40,0.85) .. (-0.40,-0.72)
                -- (-3.70,-0.72) -- cycle;
        \end{scope}
        \node[
            anchor=south west,
            font=\scriptsize\itshape,
            text=\EssCutOneColor,
            fill=white,
            inner sep=0.6pt
        ] at (2.70,2.40) {separator};
    \end{essgraph}
    \caption{The figure shows the tightest minimum cut separating $v$ from $T$. Its separator contains $t_1$, $t_2$, and $p_3$. In particular, $t_1$ and $t_2$ belong to the
    separator of the tightest minimum cut and are therefore essential for $v$.}
    \label{fig:tightest-v-T-cut}
\end{figure}

The idea of \essentialassignmentcondition is to assign each non-terminal vertex to a terminal that is essential for it, while respecting the prescribed capacities: terminal $t_i$ must receive $c_i$ vertices. We define an assignment as a function $\assgn: V\setminus T \to T$ that maps every non-terminal vertex to one terminal. For a terminal $t$, we write $\assgninv(t)$ for the set of vertices mapped to $t$ and $\cp_t$ for its capacity. Throughout the paper, we sometimes write $\cp_i$ in place of $\cp_{t_i}$.

\begin{restatable}[\essentialassignmentcondition]{definition}{defessentialassignment}
	\label{def:essential-assignment-condition}
	We say that $G$ satisfies the \essentialassignmentcondition if there exists an assignment $\assgn$ such that:
	\begin{enumerate}
		\item for every vertex $v \in V \setminus T$, the terminal $\assgn(v)$ is essential for $v$; and
		\item for every terminal $t \in T$, we have $|\assgninv(t)| = \cp_t$.
	\end{enumerate}
\end{restatable}

We refer to \Cref{fig:essential-and-assignment} in \Cref{sec:essential-assignment-conditions} for a visualization of the \essentialassignmentcondition.

The two properties of the condition tie the two sides of the problem together. The first property is the cut side: the assignment sends each vertex into its bottleneck. The second property is the counting side: the assignment saturates every capacity exactly. Together, the two properties ask for a bipartite matching between the vertices and their essential terminals. Every vertex is matched once, and every terminal $t$ is matched exactly $\cp_t$ times. The two extremes show the range of the condition. In a $k$-$T$-connected graph, a family of $k$ vertex-disjoint paths from a vertex must end at all $k$ terminals. Hence every terminal is essential for every vertex, and any allocation that respects the capacities satisfies the condition. In a disjoint union of $k$ arborescences, each vertex reaches only the root of its own tree, so this root is its only essential terminal. The assignment of every vertex to its root fills each capacity exactly.

We can check the condition in polynomial time: for each vertex, one maximum-flow computation finds its essential terminals (\Cref{prop:essential-terminals-single-flow}), and then a single bipartite maximum-flow computation finds an assignment that satisfies the condition. We call such an assignment a witness. More importantly, we show how to maintain the condition in polynomial time until the algorithm builds the full partition. This gives the following strengthening of \Cref{thm:k-t-conn}. Since a $k$-$T$-connected graph always satisfies the condition, \Cref{thm:k-t-conn} follows.

\begin{restatable}{theorem}{thmessassign}
\label{thm:ess-assign-cond}
Let $G = (V, E)$ be a directed graph and $T = \{t_1,\dots,t_k\}$ be a set of distinct terminals. Assume that $G$ satisfies the \essentialassignmentcondition and let $c_1, \dots, c_k$ be nonnegative integers such that $\sum_{i=1}^k c_i = |V \setminus T|$. Then there exists a polynomial-time algorithm that partitions the vertices of the graph into $V_1, \dots, V_k$, such that $t_i \in V_i$, $|V_i| = c_i + 1$, and the induced subgraph $G[V_i]$ is connected to $t_i$.
\end{restatable}

\paragraph{Maintaining the condition.}
\providecommand{\OverviewScale}{1}

\providecommand{\OverviewXScale}{0.90}

\providecommand{\OverviewPanelWidth}{0.47\textwidth}

\captionsetup{font=small,labelfont=bf}
\captionsetup[subfigure]{font=footnotesize,labelfont=bf,skip=2pt}

\newcommand{\overviewterminals}{%
  \essterminal{t1}{(-2.60,0.00)}{$t_1$}%
  \essterminal{t2}{( 0.00,0.00)}{$t_2$}%
  \essterminal{t3}{( 2.60,0.00)}{$t_3$}%
}

\newcommand{\overviewedges}{%
  \essedge[extra={->}]{v8}{v4}%
  \essedge[extra={->}]{v8}{v5}%
  \essedge[extra={->}]{v5}{v4}%
  \essedge[extra={->}]{v4}{t1}%
  \essedge[extra={->}]{v4}{t2}%
  \essedge[extra={->}]{v5}{t2}%
  \essedge[extra={->}]{v9}{v6}%
  \essedge[extra={->}]{v9}{v7}%
  \essedge[extra={->}]{v6}{v7}%
  \essedge[extra={->}]{v6}{t2}%
  \essedge[extra={->}]{v7}{t2}%
  \essedge[extra={->}]{v7}{t3}%
}

\begin{figure*}[!tb]
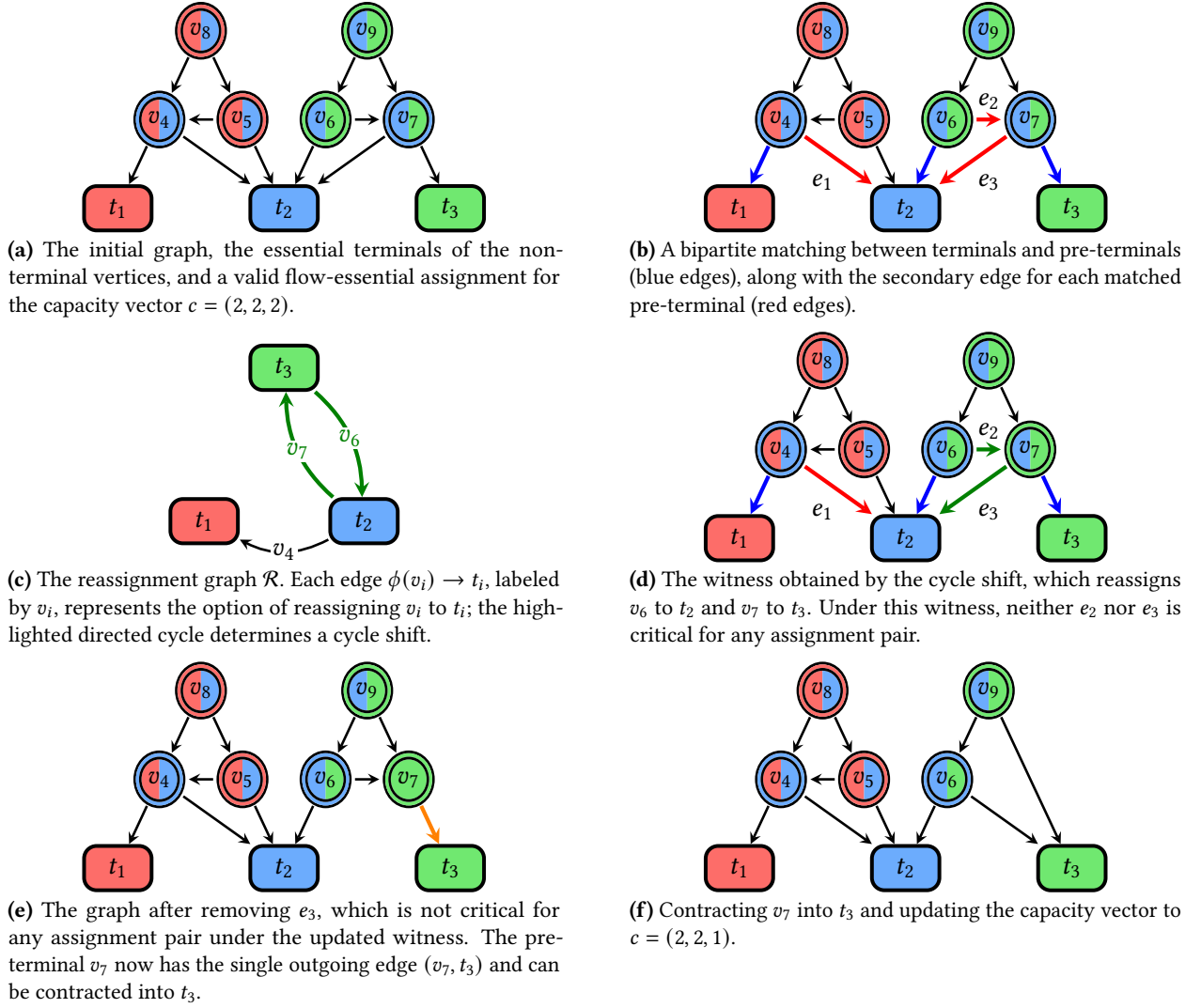

  \centering

  \begin{subfigure}[t]{\OverviewPanelWidth}
    \centering
    \begin{essgraph}[
      scale=\OverviewScale,
      xscale=\OverviewXScale,
      >=stealth,
      line cap=round,
      line join=round
    ]
      \overviewterminals

      \essvertex[colors={t1/0.5,t2/0.5},assigned=t2]{v4}{(-1.95,1.25)}{$v_4$}
      \essvertex[colors={t1/0.5,t2/0.5},assigned=t1]{v5}{(-0.65,1.25)}{$v_5$}
      \essvertex[colors={t2/0.5,t3/0.5},assigned=t3]{v6}{( 0.65,1.25)}{$v_6$}
      \essvertex[colors={t2/0.5,t3/0.5},assigned=t2]{v7}{( 1.95,1.25)}{$v_7$}
      \essvertex[colors={t1/0.5,t2/0.5},assigned=t1]{v8}{(-1.30,2.50)}{$v_8$}
      \essvertex[colors={t2/0.5,t3/0.5},assigned=t3]{v9}{( 1.30,2.50)}{$v_9$}

      \overviewedges
    \end{essgraph}
    \caption{The initial graph, the essential terminals of the non-terminal vertices, and a valid flow-essential assignment for the capacity vector $\cp = (2, 2, 2)$.}
    \label{fig:overview-assignment}
  \end{subfigure}
  \hfill
  \begin{subfigure}[t]{\OverviewPanelWidth}
    \centering
    \begin{essgraph}[
      scale=\OverviewScale,
      xscale=\OverviewXScale,
      >=stealth,
      line cap=round,
      line join=round
    ]
      \overviewterminals

      \essvertex[colors={t1/0.5,t2/0.5},assigned=t2]{v4}{(-1.95,1.25)}{$v_4$}
      \essvertex[colors={t1/0.5,t2/0.5},assigned=t1]{v5}{(-0.65,1.25)}{$v_5$}
      \essvertex[colors={t2/0.5,t3/0.5},assigned=t3]{v6}{( 0.65,1.25)}{$v_6$}
      \essvertex[colors={t2/0.5,t3/0.5},assigned=t2]{v7}{( 1.95,1.25)}{$v_7$}
      \essvertex[colors={t1/0.5,t2/0.5},assigned=t1]{v8}{(-1.30,2.50)}{$v_8$}
      \essvertex[colors={t2/0.5,t3/0.5},assigned=t3]{v9}{( 1.30,2.50)}{$v_9$}

      \essedge[extra={->}]{v8}{v4}
      \essedge[extra={->}]{v8}{v5}
      \essedge[extra={->}]{v5}{v4}
      \essedge[extra={->,blue,ultra thick}]{v4}{t1}
      \essedge[extra={->,red,ultra thick}]{v4}{t2}
      \node[below left=4pt,inner sep=1pt] at ($(v4)!0.5!(t2)$) {$e_1$};
      \essedge[extra={->}]{v5}{t2}
      \essedge[extra={->}]{v9}{v6}
      \essedge[extra={->}]{v9}{v7}
      \essedge[extra={->,red,ultra thick}]{v6}{v7}
      \node[above=3pt,inner sep=1pt] at ($(v6)!0.5!(v7)$) {$e_2$};
      \essedge[extra={->,blue,ultra thick}]{v6}{t2}
      \essedge[extra={->,red,ultra thick}]{v7}{t2}
      \node[below right=4pt,inner sep=1pt] at ($(v7)!0.5!(t2)$) {$e_3$};
      \essedge[extra={->,blue,ultra thick}]{v7}{t3}
    \end{essgraph}
\caption{A bipartite matching between terminals and pre-terminals (blue edges), along with the secondary edge for each matched pre-terminal (red edges).}
    \label{fig:overview-matching}
  \end{subfigure}

  \vspace{0.35em}

  \begin{subfigure}[t]{\OverviewPanelWidth}
    \centering
    \begin{essgraph}[
      scale=\OverviewScale,
      xscale=\OverviewXScale,
      >=stealth,
      line cap=round,
      line join=round
    ]
      \essterminal{t1}{(0.00,0.00)}{$t_1$}
      \essterminal{t2}{(2.5,0.00)}{$t_2$}
      \essterminal{t3}{(1.25,2.16)}{$t_3$}

      \begin{scope}[on background layer]
        \draw[->,line width=\EssEdgeLineWidth]
          (t2) to[bend left=25]
          node[pos=0.5,fill=white,inner sep=1pt] {$v_4$} (t1);
        \draw[->,green!50!black,ultra thick]
          (t3) to[bend left=25]
          node[pos=0.5,fill=white,inner sep=1pt] {$v_6$} (t2);
        \draw[->,green!50!black,ultra thick]
          (t2) to[bend left=25]
          node[pos=0.5,fill=white,inner sep=1pt] {$v_7$} (t3);
      \end{scope}
    \end{essgraph}
\caption{The reassignment graph $\caR$. Each edge $\assgn(v_i)\to t_i$, labeled by $v_i$, represents the option of reassigning $v_i$ to $t_i$; the highlighted directed cycle determines a cycle shift.}
    \label{fig:overview-relation}
  \end{subfigure}
  \hfill
  \begin{subfigure}[t]{\OverviewPanelWidth}
    \centering
    \begin{essgraph}[
      scale=\OverviewScale,
      xscale=\OverviewXScale,
      >=stealth,
      line cap=round,
      line join=round
    ]
      \overviewterminals

      \essvertex[colors={t1/0.5,t2/0.5},assigned=t2]{v4}{(-1.95,1.25)}{$v_4$}
      \essvertex[colors={t1/0.5,t2/0.5},assigned=t1]{v5}{(-0.65,1.25)}{$v_5$}
      \essvertex[colors={t2/0.5,t3/0.5},assigned=t2]{v6}{( 0.65,1.25)}{$v_6$}
      \essvertex[colors={t2/0.5,t3/0.5},assigned=t3]{v7}{( 1.95,1.25)}{$v_7$}
      \essvertex[colors={t1/0.5,t2/0.5},assigned=t1]{v8}{(-1.30,2.50)}{$v_8$}
      \essvertex[colors={t2/0.5,t3/0.5},assigned=t3]{v9}{( 1.30,2.50)}{$v_9$}

      \essedge[extra={->}]{v8}{v4}
      \essedge[extra={->}]{v8}{v5}
      \essedge[extra={->}]{v5}{v4}
      \essedge[extra={->,blue,ultra thick}]{v4}{t1}
      \essedge[extra={->,red,ultra thick}]{v4}{t2}
      \node[below left=4pt,inner sep=1pt] at ($(v4)!0.5!(t2)$) {$e_1$};
      \essedge[extra={->}]{v5}{t2}
      \essedge[extra={->}]{v9}{v6}
      \essedge[extra={->}]{v9}{v7}
      \essedge[extra={->,green!50!black,ultra thick}]{v6}{v7}
      \node[above=3pt,inner sep=1pt] at ($(v6)!0.5!(v7)$) {$e_2$};
      \essedge[extra={->,blue,ultra thick}]{v6}{t2}
      \essedge[extra={->,green!50!black,ultra thick}]{v7}{t2}
      \node[below right=4pt,inner sep=1pt] at ($(v7)!0.5!(t2)$) {$e_3$};
      \essedge[extra={->,blue,ultra thick}]{v7}{t3}
    \end{essgraph}
\caption{The witness obtained by the cycle shift, which reassigns $v_6$ to $t_2$ and $v_7$ to $t_3$. Under this witness, neither $e_2$ nor $e_3$ is critical for any assignment pair.}
    \label{fig:overview-shifting}
  \end{subfigure}

  \vspace{0.35em}

  \begin{subfigure}[t]{\OverviewPanelWidth}
    \centering
    \begin{essgraph}[
      scale=\OverviewScale,
      xscale=\OverviewXScale,
      >=stealth,
      line cap=round,
      line join=round
    ]
      \overviewterminals

      \essvertex[colors={t1/0.5,t2/0.5},assigned=t2]{v4}{(-1.95,1.25)}{$v_4$}
      \essvertex[colors={t1/0.5,t2/0.5},assigned=t1]{v5}{(-0.65,1.25)}{$v_5$}
      \essvertex[colors={t2/0.5,t3/0.5},assigned=t2]{v6}{( 0.65,1.25)}{$v_6$}
      \essvertex[colors={t3/1},assigned=t3]{v7}{( 1.95,1.25)}{$v_7$}
      \essvertex[colors={t1/0.5,t2/0.5},assigned=t1]{v8}{(-1.30,2.50)}{$v_8$}
      \essvertex[colors={t2/0.5,t3/0.5},assigned=t3]{v9}{( 1.30,2.50)}{$v_9$}

      \essedge[extra={->}]{v8}{v4}
      \essedge[extra={->}]{v8}{v5}
      \essedge[extra={->}]{v5}{v4}
      \essedge[extra={->}]{v4}{t1}
      \essedge[extra={->}]{v4}{t2}
      \essedge[extra={->}]{v5}{t2}
      \essedge[extra={->}]{v9}{v6}
      \essedge[extra={->}]{v9}{v7}
      \essedge[extra={->}]{v6}{v7}
      \essedge[extra={->}]{v6}{t2}
      \essedge[extra={->,orange,ultra thick}]{v7}{t3}
    \end{essgraph}
\caption{The graph after removing $e_3$, which is not critical for any assignment pair under the updated witness. The pre-terminal $v_7$ now has the single outgoing edge $(v_7,t_3)$ and can be contracted into $t_3$.}
    \label{fig:overview-removal}
  \end{subfigure}
  \hfill
  \begin{subfigure}[t]{\OverviewPanelWidth}
    \centering
    \begin{essgraph}[
      scale=\OverviewScale,
      xscale=\OverviewXScale,
      >=stealth,
      line cap=round,
      line join=round
    ]
      \overviewterminals

      \essvertex[colors={t1/0.5,t2/0.5},assigned=t2]{v4}{(-1.95,1.25)}{$v_4$}
      \essvertex[colors={t1/0.5,t2/0.5},assigned=t1]{v5}{(-0.65,1.25)}{$v_5$}
      \essvertex[colors={t2/0.5,t3/0.5},assigned=t2]{v6}{( 0.65,1.25)}{$v_6$}
      \essvertex[colors={t1/0.5,t2/0.5},assigned=t1]{v8}{(-1.30,2.50)}{$v_8$}
      \essvertex[colors={t2/0.5,t3/0.5},assigned=t3]{v9}{( 1.30,2.50)}{$v_9$}

      \essedge[extra={->}]{v8}{v4}
      \essedge[extra={->}]{v8}{v5}
      \essedge[extra={->}]{v5}{v4}
      \essedge[extra={->}]{v4}{t1}
      \essedge[extra={->}]{v4}{t2}
      \essedge[extra={->}]{v5}{t2}
      \essedge[extra={->}]{v9}{v6}
      \essedge[extra={->}]{v9}{t3}
      \essedge[extra={->}]{v6}{t3}
      \essedge[extra={->}]{v6}{t2}
    \end{essgraph}
\caption{Contracting $v_7$ into $t_3$ and updating the capacity vector to $\cp = (2, 2, 1)$.}
    \label{fig:overview-contraction}
  \end{subfigure}

  \caption{The sequence of transformations in the algorithm until it performs a contraction step. A detailed explanation of this example appears on \cpageref{par:complete-example-explanation}.}
  \vspace{0.35em}
  \begin{minipage}{\textwidth}
    \textbf{Visual conventions.}
    Rectangular vertices represent terminals, distinguished by color. Colored sectors within a non-terminal vertex indicate its essential terminals, while a colored outer ring, when present, indicates the terminal to which the vertex is assigned. These conventions are used throughout the paper.
  \end{minipage}
  \label{fig:overview-full}
\end{figure*}

The algorithm (\Cref{alg:gl-partition}) keeps a witness $\assgn$ and repeatedly applies one of three operations: (i) it removes a terminal whose capacity has dropped to zero, (ii) it contracts a pre-terminal with a single outgoing edge into the terminal that this edge enters and decreases the capacity of that terminal, or (iii) it removes one edge, chosen so that the graph still has a witness after the removal. Only operation (ii) builds the partition: each contraction commits one pre-terminal to a part. Operation (i) removes a terminal once its part is complete. The real work is in operation (iii), and it is our answer to the contraction barrier: we do not search for a contraction that the graph tolerates. Instead, we remove edges until some pre-terminal is down to one outgoing edge and the contraction is forced. Each operation removes a vertex or an edge, so after at most $|V|+|E|$ steps the graph is empty. To assemble the partition, we undo the contractions. Every contracted vertex entered its part through one edge to a vertex already in that part. These edges form $k$ disjoint arborescences of the prescribed sizes, one rooted at each terminal, so every part is connected to its terminal.

Operations (i) and (ii) preserve the condition for simple reasons. If the capacity of a terminal is zero, the witness assigns no vertex to it, and the removal of one terminal never destroys the essentiality of another (\Cref{lem:essential-survives-terminal-removal,lem:remove-zero-capacity-terminal}). If a pre-terminal has a single outgoing edge, then the terminal of that edge is its only essential terminal, so the witness already assigns the pre-terminal there. The contraction of the edge translates every family of disjoint paths across it, and all other essentialities survive (\Cref{lem:essential-survives-degree-one-leaf-contraction,lem:contract-degree-one-potential-leaf}).

The main difficulty is to show that operation (iii) is available whenever operations (i) and (ii) are not. In that case, every capacity is positive and every pre-terminal has at least two outgoing edges.
Our first step is to prove that the bipartite graph between terminals and pre-terminals has a matching that saturates the terminal side, where terminal $t$ and pre-terminal $p$ are adjacent if they have an edge in $G$. Suppose otherwise. By Hall's theorem, there is an inclusion-minimal Hall-deficient set $S$ of terminals; writing $\PT(G,S)$ for its neighboring pre-terminals, we have $|\PT(G,S)|<|S|$. Using inclusion-minimality and a rerouting argument, we show that no terminal in $S$ is essential for any vertex outside $\PT(G,S)$ (\Cref{lem:terminal-leaf-hall-set-not-essential-outside}).

Positive capacities now rule out this deficiency. The condition's witness must assign at least one vertex to every terminal in $S$; these vertices are distinct, and each lies in $\PT(G,S)$ because its assigned terminal is essential for it. Hence $|\PT(G,S)|\ge |S|$, contradicting $|\PT(G,S)|<|S|$. The desired matching therefore exists (\Cref{lem:terminal-leaf-matching}). The initial graph and witness in our running example are shown in \Cref{fig:overview-assignment}, and its matching edges are shown in blue in \Cref{fig:overview-matching}.

Recall that operation (ii) does not apply, so every pre-terminal has at least two outgoing edges.
For each terminal $t_i$, let $p_i$ denote its matched pre-terminal. Choose an outgoing edge of $p_i$ distinct from the matching edge $(p_i,t_i)$, call it the \emph{secondary edge} of $p_i$, and write it as $e_i=(p_i,q_i)$. These secondary edges are shown in red in \Cref{fig:overview-matching}. Our next goal is to remove one of $e_1,e_2,\dots,e_k$ while preserving the \essentialassignmentcondition.

Let $\assgn$ be the current witness to the \essentialassignmentcondition. Call an edge $e$ critical for a pair $(v,t)$ if $t$ is essential for $v$ in the current graph but not after the removal of $e$. If some $e_i$ is not critical for any assignment pair $(v,\assgn(v))$, then the removal of $e_i$ keeps the current witness valid, and we are done. Otherwise, for every $e_i$ choose a $v_i$ such that $e_i$ is critical for the assignment pair $(v_i,\assgn(v_i))$. The key observation is that $t_i$ is also essential for $v_i$ (\Cref{lem:critical-implies-assignable}). We sketch the proof argument next. Since $e_i$ is critical for assigning $v_i$ to $\assgn(v_i)$, deleting it must lower the terminal connectivity of $v_i$. If $t_i$ were not essential for $v_i$, there would be a maximum family of disjoint paths from $v_i$ to terminals that avoids $t_i$. This family must use $e_i$, since otherwise it would survive the deletion of $e_i$. Rerouting the path that uses $e_i$ at $p_i$ along the matching edge $(p_i,t_i)$ would give a family of the same size in $G\setminus e_i$, a contradiction. This gives $v_i$ an alternative: it can switch its assigned terminal from $\assgn(v_i)$ to $t_i$.

A second key observation is that $t_i$ is a strictly \enquote{safer} assignment option for $v_i$. Precisely, for any terminal $t$, let $\crt_{v_i}(t)$ denote the number of edges among $e_1, \dots, e_k$ that are critical for assigning $v_i$ to $t$. We show that $\crt_{v_i}(t_i) < \crt_{v_i}(\assgn(v_i))$ in two steps. This means that reassigning $v_i$ to $t_i$ strictly decreases the number of critical edges for $v_i$. First, $e_i$ is critical for assigning $v_i$ to $\assgn(v_i)$ by construction, but not for assigning it to $t_i$ (\Cref{lem:own-edge-not-critical}). Second, for every $j\neq i$, if $e_j$ is critical for assigning $v_i$ to $t_i$, then it was already critical for assigning $v_i$ to $\assgn(v_i)$ (\Cref{lem:critical-transfer}). The two steps together imply that $\crt_{v_i}(t_i) < \crt_{v_i}(\assgn(v_i))$.

We sketch proof idea for the second step. Suppose for a contradiction that some $e_j$ is critical for assigning $v_i$ to $t_i$, but not for assigning $v_i$ to $\assgn(v_i)$. Let $G_i=G\setminus\{e_i\}$, let $G_j=G\setminus\{e_j\}$, and let $H=G\setminus\{e_i,e_j\}$. Let $\kappa=\connvg{v_i}{G}$. We show that the connectivity of $v_i$ is $\kappa-1$ in all three graphs $G_i$, $G_j$, and $H$ (\Cref{claim:all-conns-equal}). Using the notion of the tightest minimum cut separating $v_i$ from $T$ (see \Cref{def:tightest-min-cut}), we consider $C_i$ and $C_j$, the tightest minimum cuts separating $v_i$ from $T$ in $G_i$ and $G_j$, respectively. We show that $C_i$ and $C_j$ are both minimum cuts in $H$. Their intersection and union are therefore also minimum cuts in $H$ (\Cref{cor:mincut-closure}). We identify the location of the endpoints of $e_i$ and $e_j$ in the intersection and union cuts, and show that both edges cross the intersection but neither crosses the union. Hence the union remains a valid cut when we add both edges back to recover $G$, although its size is $\kappa-1$, a contradiction.

All that remains is to turn these safer alternatives into progress. We do this by repeatedly reassigning vertices to their safer terminals, until one of the edges $e_1, \dots, e_k$ is no longer critical for any assignment pair. The reassignment of a single vertex would violate the exact capacity constraints of the witness, so we must coordinate several reassignments at once. We place all these alternatives in a directed reassignment graph $\caR$ on the terminals. For each $i$, we place an edge from $\assgn(v_i)$ to $t_i$ which represents the option of reassigning $v_i$ to $t_i$. The reassignment graph for our running example is shown in \Cref{fig:overview-relation}. Every terminal has exactly one incoming edge in $\caR$; therefore, $\caR$ contains a directed cycle. In one iteration of the algorithm, we simultaneously perform the reassignments represented by this cycle. Each terminal on the cycle loses one assigned vertex and gains one, while every reassigned vertex remains assigned to an essential terminal. Thus, the resulting cycle shift preserves both the capacity and essentiality requirements. The updated witness is shown in \Cref{fig:overview-shifting}.

For any assignment $\assgnalt$, define the potential function $\Phi(\assgnalt) = \sum_{v \in V \setminus T} \crt_v(\assgnalt(v))$; thus, $\Phi(\assgnalt)$ is the total number of critical edges counted across all assignment pairs $(v,\assgnalt(v))$. Since $\crt_v(\assgnalt(v)) \le k$, we observe that $\Phi(\assgnalt) \le k|V \setminus T|$. If $\assgn$ is the witness before the cycle shift operation above and $\assgn'$ is the witness after, then $\Phi(\assgn') < \Phi(\assgn)$, which means that the cycle shift strictly decreases the total number of critical edges over all assignment pairs. Since $\Phi(\cdot)$ is polynomially bounded, after polynomially many iterations of the cycle shift operation, we eventually reach a witness for which one of the edges $e_1, \dots, e_k$ is not critical for any assignment pair and can be removed. The resulting edge removal is shown in \Cref{fig:overview-removal}, and the subsequent contraction is shown in \Cref{fig:overview-contraction}.

\paragraph{Explanation of the example.}\label{par:complete-example-explanation}
\Cref{fig:overview-full} illustrates the primary sequence of algorithmic steps on an example graph. We begin in \Cref{fig:overview-assignment} with a graph (which is not $3$-$T$-connected) and terminal capacities $\cp = (2, 2, 2)$, where every non-terminal vertex has terminal connectivity $2$ to $T = \{t_1, t_2, t_3\}$. First, we compute the essential terminals for every non-terminal vertex. For example, $t_3$ is essential for $v_6$ because removing $t_3$ reduces the connectivity of $v_6$ to $T$ from $2$ to $1$. In the figure, the inner shading of each non-terminal vertex denotes its essential terminals, while the colored outer ring indicates its assigned terminal (e.g., $v_6$ is assigned to $t_3$). The displayed assignment satisfies the \essentialassignmentcondition.

Because no terminal has capacity zero and no pre-terminal has out-degree one, the direct reduction steps cannot be applied immediately. Therefore, as shown in \Cref{fig:overview-matching}, we find the matching $M = \{(v_4, t_1), (v_6, t_2), (v_7, t_3)\}$ from pre-terminals to terminals, shown by the blue edges. For each matched pre-terminal, we choose a secondary edge ($e_1, e_2$, and $e_3$, shown in red).

Under this initial assignment, each secondary edge is critical for at least one assignment pair. For instance, $e_2$ is critical for assigning $v_6$ to $\assgn(v_6)=t_3$: removing $e_2$ drops the connectivity of $v_6$ to $1$ and eliminates all paths from $v_6$ to $t_3$, leaving $t_2$ as its sole essential terminal. Thus, none of $e_1,e_2,e_3$ can be removed while preserving the current witness.

To make a secondary edge removable, we update the witness by a cycle shift. In \Cref{fig:overview-relation}, we construct the reassignment graph. For each secondary edge $e_i$, we choose a vertex $v_i$ such that $e_i$ is critical for assigning $v_i$ to $\assgn(v_i)$. Here, the chosen vertices are $v_4$ for $e_1$, $v_6$ for $e_2$, and $v_7$ for $e_3$, assigned to $t_2$, $t_3$, and $t_2$, respectively. Adding the edge $\assgn(v_i)\to t_i$ for each $i$ produces a directed cycle between $t_2$ and $t_3$, highlighted in green.

We shift the assignment along this directed cycle by reassigning $v_6$ from $t_3$ to $t_2$ and $v_7$ from $t_2$ to $t_3$. The resulting assignment is shown in \Cref{fig:overview-shifting}. Both new assigned terminals are essential for their respective vertices, and the cycle shift preserves the number of vertices assigned to each terminal. Hence the resulting assignment remains a valid witness. Under this witness, neither $e_2$ nor $e_3$ is critical for any assignment pair. As shown in \Cref{fig:overview-removal}, we remove $e_3$, leaving the pre-terminal $v_7$ with the single outgoing edge $(v_7,t_3)$, highlighted in orange. In \Cref{fig:overview-contraction}, operation (ii) contracts $v_7$ into $t_3$ and updates the capacity vector to $\cp = (2, 2, 1)$, after which the algorithm continues on the reduced graph.

\paragraph{The weighted version.}
We next allow each non-terminal vertex $v$ to have a positive integer weight $\weight_v$. The capacity $\cp_t$ now bounds weight instead of size: the total weight of the part of $t$ must be at most $\cp_t+w_{\max}-1$, where $w_{\max}$ is the maximum vertex weight. This additive violation is unavoidable: take $k$ terminals of capacity one and one non-terminal $v$ of weight $k$, with an edge from $v$ to every terminal. The instance is $k$-$T$-connected, but every partition places $v$ in one part, whose weight exceeds its capacity by $k-1=w_{\max}-1$. This weighted generalization is the existence theorem behind the confluent flow result of \citet{chen2007almost}: their proof for $k$-connected graphs relies on the topological argument of \lovasz, and they left an algorithm for it as an open problem. For undirected graphs, \citet{chandran2018spanning} gave a constructive proof by local search, but their algorithm takes exponential time. Our main result in this setting is the weighted analogue of \Cref{thm:k-t-conn}.
\begin{restatable}{theorem}{thmweightedktconn}
\label{thm:weighted-k-t-conn}
Let $G$ be a directed graph with terminal set $T$. Each non-terminal vertex $v$ has a positive integer weight $\weight_v$, and each terminal $t$ has a nonnegative integer capacity $\cp_t$, with $\sum_{v\in V\setminus T}\weight_v\le\sum_{t\in T}\cp_t$. If $G$ is $k$-$T$-connected, then there is a polynomial-time algorithm that partitions $V(G)$ into parts $\langle V_t\rangle_{t\in T}$ such that $t\in V_t$, $G[V_t]$ is connected to $t$, and
\begin{equation*}
        \sum_{v\in V_t\setminus T}\weight_v \le \cp_t+w_{\max}-1
\end{equation*}
for every terminal $t\in T$.
\end{restatable}

Besides adding the weights, \Cref{thm:weighted-k-t-conn} relaxes the capacity constraint of \Cref{thm:k-t-conn}, since the capacities can now sum to more than the total weight. The unweighted theorem remains a special case, since we can give every vertex weight one and keep the capacities summing to the number of non-terminals; the bound then reads $|V_t\setminus T|\le\cp_t$, and the equality of the two sums makes every part receive exactly $\cp_t$ non-terminals.

To get a polynomial-time algorithm, we follow the ideas of the unweighted case. The first difficulty is in the \essentialassignmentcondition itself: even in a $k$-$T$-connected graph, a flow-essential assignment of the non-terminals to the terminals that fills every capacity exactly can fail to exist. We therefore relax the condition in two ways and obtain the \fractionalesscond. First, each vertex now distributes its weight over its essential terminals in nonnegative integral units: a \FractionalEssAssign is a function $\assgnw$ that sends $\assgnw(v,t)$ units of the weight of $v$ to the terminal $t$, with $\sum_{t\in T}\assgnw(v,t)=\weight_v$. Second, the exact capacity equalities become capacity upper bounds. The essentiality rule stays the same.

\begin{restatable}[\fractionalesscond]{definition}{deffractionalesscond}
	\label{def:fractional-essential-assignment-condition}
	We say that $G$ satisfies the \fractionalesscond if there exists a \FractionalEssAssign $\assgnw$ such that:
	\begin{enumerate}
		\item for every vertex $v \in V \setminus T$ and terminal $t \in T$ with $\assgnw(v,t) > 0$, the terminal $t$ is essential for $v$; and
		\item for every terminal $t \in T$, we have $\sum_{v \in V \setminus T} \assgnw(v,t) \le \cp_t$.
	\end{enumerate}
\end{restatable}

We again call such a $\assgnw$ a witness. Consider the bipartite graph with the non-terminals on one side and the terminals on the other, in which a non-terminal $v$ and a terminal $t$ are adjacent if $t$ is essential for $v$. Add a source that supplies $\weight_v$ units to each $v$, and a sink that takes at most $\cp_t$ units from each $t$. The witnesses are exactly the integral flows that saturate the source. One maximum-flow computation therefore checks the condition, and, once we place costs on the pairs $(v,t)$, one minimum-cost-flow computation finds a cheapest witness (\Cref{prop:best-assign-is-poly}). As in the unweighted case, we maintain the condition until the algorithm builds the full partition. This gives the following strengthening of \Cref{thm:weighted-k-t-conn}. In a $k$-$T$-connected graph, every terminal is essential for every vertex, so any split of the weights that respects the capacities is a witness. Such a split exists because the capacities sum to at least the total weight, and \Cref{thm:weighted-k-t-conn} follows.

\begin{restatable}{theorem}{thmweighted}
\label{gyori-weighted-poly-theorem}
Let $G$ be a directed graph with terminal set $T$. Each non-terminal vertex $v$ has a positive integer weight $\weight_v$, and each terminal $t$ has a nonnegative integer capacity $\cp_t$. If $G$ satisfies the \fractionalesscond, then there is a polynomial-time algorithm that partitions $V(G)$ into parts $\langle V_t\rangle_{t\in T}$ such that $t\in V_t$, $G[V_t]$ is connected to $t$, and
\begin{equation*}
        \sum_{v\in V_t\setminus T}\weight_v \le \cp_t+w_{\max}-1
\end{equation*}
for every terminal $t\in T$.
\end{restatable}

The algorithm (\Cref{alg:gl-weighted}) keeps the shape of the unweighted one, and porting it raises the second and third difficulties. Operations (i) and (ii) stay as before, except that a contraction decreases the capacity of the terminal by the weight of the contracted pre-terminal. For operation (iii), we again need the matching of every terminal $t_i$ to its own pre-terminal $p_i$. The failure of this matching is the second difficulty, and we treat it below.

When the matching exists, we choose a secondary outgoing edge $e_i$ of each $p_i$ and seek an edge that is critical for no pair $(v,t)$ with $\assgnw(v,t)>0$. Recall how the unweighted algorithm finds such an edge. If every secondary edge is critical for an assignment pair, the reassignment graph $\caR$ contains a directed cycle, and shifting the assignments along this cycle produces another witness with strictly smaller potential. Since the unweighted potential is polynomially bounded, repeating this cycle shift polynomially many times eventually yields a witness for which some secondary edge is non-critical.

The same reassignment argument applies to split-assignments. Define
\begin{equation*}
        \Phifrac(\assgnw)=\sum_{v\in V\setminus T}\sum_{t\in T}\assgnw(v,t)\crt_v(t),
\end{equation*}
where $\crt_v(t)$ is the number of secondary edges that are critical for $(v,t)$. Equivalently, each unit that $v$ assigns to an essential terminal $t$ incurs cost $\crt_v(t)$. If every secondary edge is critical for some pair receiving positive weight, then the reassignment graph again contains a directed cycle; shifting one unit of weight along this cycle preserves the split-assignment constraints and strictly decreases $\Phifrac$.

This leads to the third difficulty. Unlike its unweighted counterpart, $\Phifrac$ need not be polynomially bounded in the input size: the integral weights are given in binary and may be exponentially large. Thus, performing one-unit cycle shifts as a local search may require exponentially many iterations. Instead, we bypass the local search and compute a minimum-potential witness directly by a single minimum-cost-flow computation, using $\crt_v(t)$ as the cost of sending one unit from $v$ to $t$ (\Cref{prop:best-assign-is-poly}). This witness guarantees a non-critical secondary edge: otherwise, the cycle shift above would produce a witness of smaller potential, contradicting minimality (\Cref{lem:weighted-matching-removable-edge}). We remove this edge to make progress.

We now return to the second difficulty, when the bipartite graph between terminals and pre-terminals has no matching that saturates the terminal side. Let $S\subseteq T$ be an inclusion-minimal Hall-deficient set of terminals, and let $\PT(G,S)$ denote its neighboring pre-terminals. Minimality gives $|\PT(G,S)|=|S|-1$ (\Cref{lem:minimal-hall-set}). As in the unweighted argument, no terminal in $S$ is essential for any non-terminal outside $\PT(G,S)$ (\Cref{lem:terminal-leaf-hall-set-not-essential-outside}). Thus, only vertices in $\PT(G,S)$ can send positive weight to terminals in $S$. Unlike in the unweighted case, this does not lead to a contradiction. Instead, we match the pre-terminals in $\PT(G,S)$ to distinct terminals in $S$. Each matched pair becomes a final part, and the one unmatched terminal becomes a part by itself. We remove $S\cup \PT(G,S)$ and continue on the remaining graph. We call this operation rounding (\Cref{alg:gl-round-and-remove}).

After the removal, the total weight of the remaining non-terminals is still at most the total capacity of the remaining terminals, because the witness already sent all of this weight to those terminals. 

The real challenge is to prove that the split-assignment restricted to the remaining instance is still a valid witness. Fix a remaining vertex $v$ and a remaining terminal $t$ that receives positive weight from $v$. Since the original split-assignment is a valid witness, $t$ is essential for $v$ in the original instance. Yet deleting $S\cup \PT(G,S)$ may lower the terminal connectivity of $v$, so it is not obvious that $t$ remains essential afterwards. We establish this in two steps. First, temporarily replace $S$ in the terminal set by $\PT(G,S)$ and delete $S$ from the graph, obtaining terminal set $(T\setminus S)\cup \PT(G,S)$. This operation preserves the terminal connectivity of $v$, and the structure of the tightest minimum cut then implies that $t$ remains essential. Second, remove the temporary terminals in $\PT(G,S)$ one at a time. At each removal, the previously established fact that deleting a terminal other than $t$ cannot destroy the essentiality of $t$ applies (\Cref{lem:essential-survives-terminal-removal}). Thus $t$ remains essential after rounding (\Cref{lem:essentiality-survives-rounding}). Consequently, every essentiality required by the restricted split-assignment survives the removal, so the restriction is still a valid witness (\Cref{lem:weighted-round-and-remove}).

Rounding is the only step in which a part can exceed its capacity. Every other part grows by contractions, and a contraction never takes more weight than the current capacity of its terminal. A part that we complete by rounding receives at most one pre-terminal, of weight at most $w_{\max}$, while its capacity is at least one, so it exceeds its capacity by at most $w_{\max}-1$. Each step of the algorithm removes at least one vertex or one edge, so the algorithm stops after at most $|V|+|E|$ steps. Each step performs a polynomial number of maximum-flow, minimum-cost-flow, and bipartite matching computations, so the whole algorithm runs in polynomial time. We provide a more careful analysis of the running time in \Cref{sec:weighted}.

\paragraph{Near-linear time on DAGs.}
Even for directed acyclic graphs, no polynomial-time algorithm computing a \gyorilovasz partition was known. Acyclicity is not an obvious simplification: for confluent flows, the acyclic case is already the general case \citep{chen2007almost}. \Cref{thm:weighted-k-t-conn} already applies to a DAG, and we now improve the algorithm on this special case. A DAG is $k$-$T$-connected if and only if every non-terminal vertex has out-degree at least $k$ (\Cref{lem:k-conn-dag}). Take any terminal $t$ and let $p$ be the first pre-terminal of $t$ in a topological order. Contracting $p$ into $t$ preserves the out-degree condition (\Cref{lem:contract-exists-dag}), so contractions alone build the partition. A heap-based implementation of this rule runs in near-linear time (\Cref{sec:dag}), even in the weighted setting.

\begin{restatable}{theorem}{thmdag}
\label{dag-algo-theorem}
Let $D=(V,E)$ be a directed acyclic graph and let $T=\{t_1,\dots,t_k\}$ be a set of distinct terminals. Each non-terminal vertex $v$ has a positive integer weight $\weight_v$, and each terminal $t$ has a nonnegative integer capacity $\cp_t$ such that
    $$\sum_{v \in V \setminus T} w_v \le \sum_{t \in T} c_t.$$
    If $D$ is $k$-$T$-connected, then there is an $\caO(m\log n)$-time algorithm that partitions $V$ into $V_1,\dots,V_k$  such that $t_i\in V_i$, $D[V_i]$ is connected to $t_i$, and
\begin{equation*}
        \sum_{v\in V_i\setminus T}\weight_v \le \cp_{t_i}+w_{\max}-1
\end{equation*}
for every $i\in [k]$.
\end{restatable}

\section{Preliminaries}
\label{sec:preliminaries}

Let $G = (V, E)$ be a directed graph on $n = |V|$ vertices and $m = |E|$ edges with a designated set of $k$ terminals $T = \{t_1, t_2, \dots, t_k\} \subseteq V$. Each terminal $t \in T$ has a capacity $c_t \in \mathbb{N}$. For notational clarity, we occasionally use $c_i$ as shorthand for $c_{t_i}$. We assume without loss of generality that the graph has no self-loops or parallel edges, and that terminals have no outgoing edges. In particular, whenever a contraction creates duplicate directed edges, we keep only one copy, and if the contraction creates an outgoing edge from a terminal, we drop the edge. For any vertex $v \in V$, let $N^+(v)$ denote its set of out-neighbors and $d^+(v) = |N^+(v)|$ its out-degree. For $V' \subseteq V$, let $G[V']$ denote the subgraph induced by $V'$. For $v \in V'$, we say that $G[V']$ is connected to $v$ if every $u \in V'$ has a directed path to $v$ in $G[V']$.

In the weighted setting of \Cref{sec:weighted}, each non-terminal vertex has a positive integer weight $w_v \in \mathbb{N}^+$, and the total weight of the non-terminal vertices does not exceed the total capacity of the terminals, i.e.,
$\sum_{v \in V \setminus T} w_v \le \sum_{t \in T} c_t$. We denote the maximum weight of a non-terminal vertex by $w_{\max} = \max_{v \in V \setminus T} w_v$.

\begin{definition}[Pre-terminals and contraction]
        \label{def:potential-leaf-contraction}

        A vertex $p \in V \setminus T$ is a {\em pre-terminal} if it has an outgoing edge $(p,t)$ to some terminal $t \in T$. For a subset $S \subseteq T$, we let $\PT(G,S)$ denote the set of pre-terminals with an outgoing edge to a terminal in $S$. If $(p,t)$ exists, then {\em contracting} the pre-terminal $p$ into the terminal $t$ means removing $p$, redirecting every incoming edge of $p$ to $t$, deleting all outgoing edges of $p$, and removing any duplicate directed edges created by this redirection.
\end{definition}

We next introduce the notion of terminal connectivity used throughout this work. Unlike standard vertex connectivity, this notion measures the connectivity of a vertex with respect to the terminal set $T$.

\begin{definition}[Terminal connectivity]
        \label{def:connectivity}

        For any non-terminal vertex $v \in V \setminus T$, the terminal connectivity of $v$ to $T$, denoted by $\connvg{v}{G}$, is the maximum cardinality of a family of paths that all start at $v$, end at distinct terminals in $T$, and are pairwise vertex-disjoint except for their common starting vertex $v$. Throughout the paper we call such a family vertex-disjoint. We say that $G$ is $k$-$T$-connected if $\connvg{v}{G} = k = |T|$ for every $v \in V \setminus T$.
\end{definition}

Note that if a graph is $k$-vertex-connected in the standard sense, it is immediately $k$-$T$-connected for any terminal set $T$ of size $k$. We next define the notion of a vertex cut separating a vertex from the terminal set.

\begin{definition}[Cut and separator]
        \label{def:cut}

        For a non-terminal vertex $v \in V \setminus T$, a cut separating $v$ from $T$ is a partition
        $C = (\lcutc{C}, \scutc{C}, \rcutc{C})$ of $V$ such that:
        \begin{itemize}
                \item $v \in \rcutc{C}$,
                \item $T \subseteq \lcutc{C} \cup \scutc{C}$ (equivalently, $T \cap \rcutc{C} = \emptyset$), and
                \item there is no directed edge from $\rcutc{C}$ to $\lcutc{C}$.
        \end{itemize}

        The sets $\lcutc{C}$ and $\rcutc{C}$ are called the left and right sides of the cut, respectively, while $\scutc{C}$ is called the separator. The size of the cut, denoted by $|C|$, is the cardinality of its separator, i.e., $|C| = |\scutc{C}|$.
\end{definition}

The separator $\scutc{C}$ structurally blocks all directed paths from $v$ to $T$, as any path from $v \in \rcutc{C}$ to $T \subseteq \lcutc{C} \cup \scutc{C}$ must intersect $\scutc{C}$ before it can enter $\lcutc{C}$. We now state the directed vertex-version of Menger's Theorem.

\begin{lemma}[Menger's theorem~\citep{menger1927allgemeinen}]
	\label{lem:mengers-theorem}
	For every non-terminal vertex $v \in V \setminus T$, the maximum number of vertex-disjoint paths from $v$ to distinct terminals equals the minimum size of a cut separating $v$ from $T$. That is, $\connvg{v}{G} = \min \{ |\scutc{C}| : C = (\lcutc{C}, \scutc{C}, \rcutc{C}) \text{ separates } v \text{ from } T \}$.
\end{lemma}

\subsection{Union, Intersection, and the Tightest Min-Cut}
\label{subsec:submodularity}

We next introduce two operations on cuts, namely union and intersection, and use them to expose the structure of minimum cuts separating a vertex from $T$. These operations imply a closure property for minimum cuts and, in turn, lead to the definition of the tightest minimum cut.

\begin{definition}[Union and intersection of cuts]
        \label{def:cut-operations}

        Let $C_1 = (\lcutc{C_1}, \scutc{C_1}, \rcutc{C_1})$ and
        $C_2 = (\lcutc{C_2}, \scutc{C_2}, \rcutc{C_2})$ be two cuts separating
        a vertex $v$ from $T$. We define their union $C_1 \cup C_2$ and intersection
        $C_1 \cap C_2$ by setting $\lcutc{C_1 \cup C_2} = \lcutc{C_1} \cup \lcutc{C_2}$,
        $\rcutc{C_1 \cup C_2} = \rcutc{C_1} \cap \rcutc{C_2}$,
        $\lcutc{C_1 \cap C_2} = \lcutc{C_1} \cap \lcutc{C_2}$, and
        $\rcutc{C_1 \cap C_2} = \rcutc{C_1} \cup \rcutc{C_2}$; in each case, the separator contains all remaining vertices. The same definitions are summarized in \Cref{tab:cut-intersection-union}.
\end{definition}

\begin{table}[h]
	\centering
	\begin{subtable}[t]{0.48\textwidth}
		\centering
		\renewcommand{\arraystretch}{1.4}
		\begin{tabular}{c|c|c|c}
			& $\lcutc{C_2}$ & $\scutc{C_2}$ & $\rcutc{C_2}$ \\ \hline
			$\lcutc{C_1}$ & $L$ & $L$ & $L$ \\ \hline
			$\scutc{C_1}$ & $L$ & $S$ & $S$ \\ \hline
			$\rcutc{C_1}$ & $L$ & $S$ & $R$
		\end{tabular}
		\caption{Vertex membership in $C_1 \cup C_2$.}
		\label{tab:cut-union}
	\end{subtable}
	\hfill
	\begin{subtable}[t]{0.48\textwidth}
		\centering
		\renewcommand{\arraystretch}{1.4}
		\begin{tabular}{c|c|c|c}
			& $\lcutc{C_2}$ & $\scutc{C_2}$ & $\rcutc{C_2}$ \\ \hline
			$\lcutc{C_1}$ & $L$ & $S$ & $R$ \\ \hline
			$\scutc{C_1}$ & $S$ & $S$ & $R$ \\ \hline
			$\rcutc{C_1}$ & $R$ & $R$ & $R$
		\end{tabular}
		\caption{Vertex membership in $C_1 \cap C_2$.}
		\label{tab:cut-intersection}
	\end{subtable}
	 \caption{
        Definition of the union and intersection of two cuts.
        For a vertex $x$, its row indicates its membership in $C_1$ and its
        column indicates its membership in $C_2$; the corresponding entry
        specifies whether $x$ belongs to the left side ($L$), separator ($S$),
        or right side ($R$) of the resulting cut.
        }
	\label{tab:cut-intersection-union}
\end{table}

These two operations are illustrated on a directed graph in
\Cref{fig:cut-operation-example}.

\definecolor{CutLeftColor}{RGB}{52,126,167}
\definecolor{CutSeparatorColor}{RGB}{202,126,32}
\definecolor{CutRightColor}{RGB}{115,88,151}

\newcommand{\CutOperationEdges}{%
  \essedge[extra={->}]{v}{d}%
  \essedge[extra={->}]{v}{e}%
  \essedge[extra={->}]{v}{f}%
  \essedge[extra={->}]{d}{p}%
  \essedge[extra={->}]{d}{q}%
  \essedge[extra={->}]{e}{q}%
  \essedge[extra={->}]{f}{q}%
  \essedge[extra={->}]{f}{r}%
  \essedge[extra={->}]{p}{t1}%
  \essedge[extra={->}]{q}{t2}%
  \essedge[extra={->}]{r}{t3}%
}

\newcommand{\CutPartLabel}[3]{%
  \node[
    font=\scriptsize\bfseries,
    text=#3!72!black,
    inner sep=0pt
  ] at #1 {$#2$};%
}

\newcommand{\CutDescendingRegions}{%
  \begin{scope}[on background layer]
    \begin{scope}[transparency group, opacity=0.11]
      \begin{pgfinterruptboundingbox}
        \clip
          (-5,-5) -- (5,-5) -- (5,-2.54) -- (-5,3.72) -- cycle;
      \end{pgfinterruptboundingbox}
      \draw[draw=CutLeftColor, line width=48pt,
        line cap=round, line join=round]
        (-2.4,1.50) -- (-2.4,0) -- (0,0);
      \draw[draw=CutLeftColor, line width=48pt, line cap=round]
        (-2.4,1.50)
        .. controls (-2.4,0.65) and (-1.00,0) .. (0,0);
    \end{scope}
    \begin{scope}[transparency group, opacity=0.09]
      \begin{pgfinterruptboundingbox}
        \clip
          (-5,5.54) -- (5,-0.72) -- (5,8) -- (-5,8) -- cycle;
      \end{pgfinterruptboundingbox}
      \draw[draw=CutRightColor, line width=48pt,
        line cap=round, line join=round]
        (0,4.50) -- (2.4,3.00) -- (2.4,1.50);
      \draw[draw=CutRightColor, line width=48pt, line cap=round]
        (0,4.50) -- (0,3.00)
        .. controls (1.10,3.00) and (2.4,2.35) .. (2.4,1.50);
    \end{scope}
    \draw[
      draw=CutSeparatorColor,
      draw opacity=0.18,
      preaction={draw=white, draw opacity=1, line width=36pt},
      line width=26pt,
      line cap=round,
      line join=round
    ]
      (-2.4,3.00) -- (0,1.50) -- (2.4,0);
  \end{scope}%
}

\newcommand{\CutAscendingRegions}{%
  \begin{scope}[on background layer]
    \begin{scope}[transparency group, opacity=0.11]
      \begin{pgfinterruptboundingbox}
        \clip
          (-5,-5) -- (5,-5) -- (5,3.72) -- (-5,-2.54) -- cycle;
      \end{pgfinterruptboundingbox}
      \draw[draw=CutLeftColor, line width=48pt,
        line cap=round, line join=round]
        (2.4,1.50) -- (2.4,0) -- (0,0);
      \draw[draw=CutLeftColor, line width=48pt, line cap=round]
        (2.4,1.50)
        .. controls (2.4,0.65) and (1.00,0) .. (0,0);
    \end{scope}
    \begin{scope}[transparency group, opacity=0.09]
      \begin{pgfinterruptboundingbox}
        \clip
          (-5,-0.72) -- (5,5.54) -- (5,8) -- (-5,8) -- cycle;
      \end{pgfinterruptboundingbox}
      \draw[draw=CutRightColor, line width=48pt,
        line cap=round, line join=round]
        (0,4.50) -- (-2.4,3.00) -- (-2.4,1.50);
      \draw[draw=CutRightColor, line width=48pt, line cap=round]
        (0,4.50) -- (0,3.00)
        .. controls (-1.10,3.00) and (-2.4,2.35) .. (-2.4,1.50);
    \end{scope}
    \draw[
      draw=CutSeparatorColor,
      draw opacity=0.18,
      preaction={draw=white, draw opacity=1, line width=36pt},
      line width=26pt,
      line cap=round,
      line join=round
    ]
      (-2.4,0) -- (0,1.50) -- (2.4,3.00);
  \end{scope}%
}

\newcommand{\CutUpperRegions}{%
  \begin{scope}[on background layer]
    \begin{scope}[transparency group, opacity=0.11]
      \begin{pgfinterruptboundingbox}
        \clip
          (-5,-5) -- (5,-5) -- (5,3.72)
          -- (0,0.59) -- (-5,3.72) -- cycle;
      \end{pgfinterruptboundingbox}
      \draw[draw=CutLeftColor, line width=48pt,
        line cap=round, line join=round]
        (-2.4,1.50) -- (-2.4,0) -- (0,0) -- (2.4,0) -- (2.4,1.50);
      \draw[draw=CutLeftColor, line width=48pt, line cap=round]
        (-2.4,1.50)
        .. controls (-2.4,0.65) and (-1.00,0) .. (0,0)
        .. controls (1.00,0) and (2.4,0.65) .. (2.4,1.50);
    \end{scope}
    \begin{scope}[transparency group, opacity=0.09]
      \draw[draw=CutRightColor, line width=48pt, line cap=round]
        (0,4.50) -- (0,3.00);
    \end{scope}
    \draw[
      draw=CutSeparatorColor,
      draw opacity=0.18,
      preaction={draw=white, draw opacity=1, line width=36pt},
      line width=26pt,
      line cap=round,
      line join=round
    ]
      (-2.4,3.00) -- (0,1.50) -- (2.4,3.00);
  \end{scope}%
}

\newcommand{\CutLowerRegions}{%
  \begin{scope}[on background layer]
    \fill[fill=CutLeftColor, fill opacity=0.11]
      (0,0) ellipse (0.95 and 0.72);
    \begin{scope}[transparency group, opacity=0.09]
      \begin{pgfinterruptboundingbox}
        \clip
          (-5,-0.72) -- (0,2.41) -- (5,-0.72)
          -- (5,8) -- (-5,8) -- cycle;
      \end{pgfinterruptboundingbox}
      \draw[draw=CutRightColor, line width=48pt,
        line cap=round, line join=round]
        (-2.4,1.50) -- (-2.4,3.00) -- (0,4.50)
        -- (2.4,3.00) -- (2.4,1.50);
      \draw[draw=CutRightColor, line width=48pt, line cap=round]
        (-2.4,1.50)
        .. controls (-2.4,2.20) and (-0.80,3.00) .. (0,3.00)
        .. controls (0.80,3.00) and (2.4,2.20) .. (2.4,1.50);
    \end{scope}
    \draw[
      draw=CutSeparatorColor,
      draw opacity=0.18,
      preaction={draw=white, draw opacity=1, line width=36pt},
      line width=26pt,
      line cap=round,
      line join=round
    ]
      (-2.4,0) -- (0,1.50) -- (2.4,0);
  \end{scope}%
}

\begin{figure}[!htb]
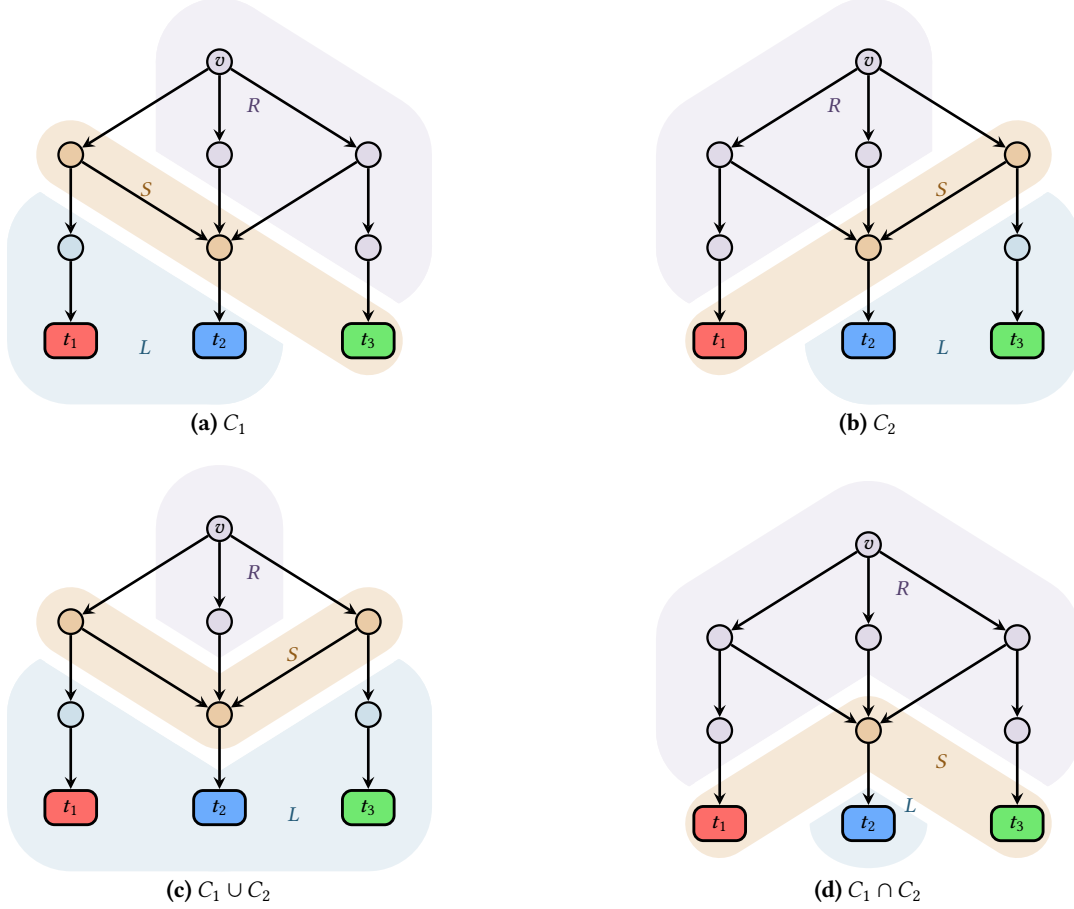

  \centering

  \begin{subfigure}[t]{0.48\textwidth}
    \centering

    \tikzsetnextfilename{cut-operation-c1}

    \begin{essgraph}[
        scale=0.82,
        >=stealth,
        node scale=0.58,
        term scale=0.72,
        node font={\scriptsize},
        line cap=round,
        line join=round
      ]

      \CutDescendingRegions

      \essterminal{t1}{(-2.4,0)}{$t_1$}
      \essterminal{t2}{( 0.0,0)}{$t_2$}
      \essterminal{t3}{( 2.4,0)}{$t_3$}

      \essvertex[extra={fill=CutRightColor!22}]{v}{(0,4.50)}{$v$}

      \essvertex[extra={fill=CutSeparatorColor!40}]{d}{(-2.4,3.00)}{}

      \essvertex[extra={fill=CutRightColor!22}]{e}{(0.0,3.00)}{}

      \essvertex[extra={fill=CutRightColor!22}]{f}{(2.4,3.00)}{}

      \essvertex[extra={fill=CutLeftColor!24}]{p}{(-2.4,1.50)}{}

      \essvertex[extra={fill=CutSeparatorColor!40}]{q}{(0.0,1.50)}{}

      \essvertex[extra={fill=CutRightColor!22}]{r}{(2.4,1.50)}{}

      \CutOperationEdges

      \CutPartLabel{(-1.20,-0.10)}{L}{CutLeftColor}
      \CutPartLabel{(-1.18, 2.45)}{S}{CutSeparatorColor}
      \CutPartLabel{( 0.55, 3.80)}{R}{CutRightColor}

    \end{essgraph}

    \caption{$C_1$}
    \label{fig:cut-operation-c1}
  \end{subfigure}
  \hfill
  \begin{subfigure}[t]{0.48\textwidth}
    \centering

    \tikzsetnextfilename{cut-operation-c2}

    \begin{essgraph}[
        scale=0.82,
        >=stealth,
        node scale=0.58,
        term scale=0.72,
        node font={\scriptsize},
        line cap=round,
        line join=round
      ]

      \CutAscendingRegions

      \essterminal{t1}{(-2.4,0)}{$t_1$}
      \essterminal{t2}{( 0.0,0)}{$t_2$}
      \essterminal{t3}{( 2.4,0)}{$t_3$}

      \essvertex[extra={fill=CutRightColor!22}]{v}{(0,4.50)}{$v$}

      \essvertex[extra={fill=CutRightColor!22}]{d}{(-2.4,3.00)}{}

      \essvertex[extra={fill=CutRightColor!22}]{e}{(0.0,3.00)}{}

      \essvertex[extra={fill=CutSeparatorColor!40}]{f}{(2.4,3.00)}{}

      \essvertex[extra={fill=CutRightColor!22}]{p}{(-2.4,1.50)}{}

      \essvertex[extra={fill=CutSeparatorColor!40}]{q}{(0.0,1.50)}{}

      \essvertex[extra={fill=CutLeftColor!24}]{r}{(2.4,1.50)}{}

      \CutOperationEdges

      \CutPartLabel{( 1.20,-0.10)}{L}{CutLeftColor}
      \CutPartLabel{( 1.18, 2.45)}{S}{CutSeparatorColor}
      \CutPartLabel{(-0.55, 3.80)}{R}{CutRightColor}

    \end{essgraph}

    \caption{$C_2$}
    \label{fig:cut-operation-c2}
  \end{subfigure}

  \par\vspace{0.8em}

  \begin{subfigure}[t]{0.48\textwidth}
    \centering

    \tikzsetnextfilename{cut-operation-union}

    \begin{essgraph}[
        scale=0.82,
        >=stealth,
        node scale=0.58,
        term scale=0.72,
        node font={\scriptsize},
        line cap=round,
        line join=round
      ]

      \CutUpperRegions

      \essterminal{t1}{(-2.4,0)}{$t_1$}
      \essterminal{t2}{( 0.0,0)}{$t_2$}
      \essterminal{t3}{( 2.4,0)}{$t_3$}

      \essvertex[extra={fill=CutRightColor!22}]{v}{(0,4.50)}{$v$}

      \essvertex[extra={fill=CutSeparatorColor!40}]{d}{(-2.4,3.00)}{}

      \essvertex[extra={fill=CutRightColor!22}]{e}{(0.0,3.00)}{}

      \essvertex[extra={fill=CutSeparatorColor!40}]{f}{(2.4,3.00)}{}

      \essvertex[extra={fill=CutLeftColor!24}]{p}{(-2.4,1.50)}{}

      \essvertex[extra={fill=CutSeparatorColor!40}]{q}{(0.0,1.50)}{}

      \essvertex[extra={fill=CutLeftColor!24}]{r}{(2.4,1.50)}{}

      \CutOperationEdges

      \CutPartLabel{(1.20,-0.10)}{L}{CutLeftColor}
      \CutPartLabel{(1.18, 2.47)}{S}{CutSeparatorColor}
      \CutPartLabel{(0.55, 3.80)}{R}{CutRightColor}

    \end{essgraph}

    \caption{$C_1\cup C_2$}
    \label{fig:cut-operation-union}
  \end{subfigure}
  \hfill
  \begin{subfigure}[t]{0.48\textwidth}
    \centering

    \tikzsetnextfilename{cut-operation-intersection}

    \begin{essgraph}[
        scale=0.82,
        >=stealth,
        node scale=0.58,
        term scale=0.72,
        node font={\scriptsize},
        line cap=round,
        line join=round
      ]

      \CutLowerRegions

      \essterminal{t1}{(-2.4,0)}{$t_1$}
      \essterminal{t2}{( 0.0,0)}{$t_2$}
      \essterminal{t3}{( 2.4,0)}{$t_3$}

      \essvertex[extra={fill=CutRightColor!22}]{v}{(0,4.50)}{$v$}

      \essvertex[extra={fill=CutRightColor!22}]{d}{(-2.4,3.00)}{}

      \essvertex[extra={fill=CutRightColor!22}]{e}{(0.0,3.00)}{}

      \essvertex[extra={fill=CutRightColor!22}]{f}{(2.4,3.00)}{}

      \essvertex[extra={fill=CutRightColor!22}]{p}{(-2.4,1.50)}{}

      \essvertex[extra={fill=CutSeparatorColor!40}]{q}{(0.0,1.50)}{}

      \essvertex[extra={fill=CutRightColor!22}]{r}{(2.4,1.50)}{}

      \CutOperationEdges

      \CutPartLabel{(0.68,0.30)}{L}{CutLeftColor}
      \CutPartLabel{(1.18,1.00)}{S}{CutSeparatorColor}
      \CutPartLabel{(0.55,3.80)}{R}{CutRightColor}

    \end{essgraph}

    \caption{$C_1\cap C_2$}
    \label{fig:cut-operation-intersection}
  \end{subfigure}

  \caption{
  Union and intersection of two cuts (that are also minimum) separating
  $v$ from $T$ in the same directed graph. The shading and the labels
  $L$, $S$, and $R$ indicate the three parts of each cut.
  }
  \label{fig:cut-operation-example}

\end{figure}

The next lemma shows that these operations are well behaved: they preserve the property of being a cut separating $v$ from $T$, and their separator sizes satisfy an exact modular identity.

\begin{lemma}
	\label{lem:union-intersection-cut}
	If $C_1$ and $C_2$ are cuts separating $v$ from $T$, then both $C_1 \cup C_2$ and $C_1 \cap C_2$ are valid cuts separating $v$ from $T$. Furthermore, their sizes satisfy the modular equation:
	$$|C_1| + |C_2| = |C_1 \cup C_2| + |C_1 \cap C_2|.$$
\end{lemma}

	\begin{proof}
	We first show that $C_1 \cup C_2$ is a valid cut separating $v$ from $T$. By definition, since $v \in \rcutc{C_1}$ and $v \in \rcutc{C_2}$, we have $v \in \rcutc{C_1} \cap \rcutc{C_2} = \rcutc{C_1 \cup C_2}$. Now, suppose for contradiction that there exists a directed edge $(x, y)$ such that $x \in \rcutc{C_1 \cup C_2}$ and $y \in \lcutc{C_1 \cup C_2}$. From \Cref{tab:cut-union}, $x \in \rcutc{C_1 \cup C_2}$ implies $x \in \rcutc{C_1} \cap \rcutc{C_2}$. Meanwhile, $y \in \lcutc{C_1 \cup C_2}$ implies that $y \in \lcutc{C_1} \cup \lcutc{C_2}$. By symmetry, suppose $y \in \lcutc{C_1}$. This means $x \in \rcutc{C_1}$ and $y \in \lcutc{C_1}$, which contradicts the assumption that $C_1$ is a valid cut with no edges leaving $\rcutc{C_1}$ and entering $\lcutc{C_1}$.
	
	Next, we establish that $C_1 \cap C_2$ is a valid cut separating $v$. Since $v \in \rcutc{C_1}$ and $v \in \rcutc{C_2}$, \Cref{tab:cut-intersection} yields $v \in \rcutc{C_1 \cap C_2}$. Assume for contradiction that there exists a directed edge $(x, y)$ with $x \in \rcutc{C_1 \cap C_2}$ and $y \in \lcutc{C_1 \cap C_2}$. The definition of $\rcutc{C_1 \cap C_2}$ implies that $x \in \rcutc{C_1} \cup \rcutc{C_2}$. By symmetry, suppose $x \in \rcutc{C_1}$. Meanwhile, $y \in \lcutc{C_1 \cap C_2}$ implies $y \in \lcutc{C_1} \cap \lcutc{C_2}$, so $y \in \lcutc{C_1}$. Thus, $(x, y)$ is an edge from $\rcutc{C_1}$ to $\lcutc{C_1}$, contradicting that $C_1$ is a valid cut.
	
	Finally, to verify the size identity $|\scutc{C_1}| + |\scutc{C_2}| = |\scutc{C_1 \cup C_2}| + |\scutc{C_1 \cap C_2}|$, we count the contribution of each vertex $x \in V$ to both sides of the equation based on its membership in $\scutc{C_1}$ and $\scutc{C_2}$:
	\begin{itemize}
		\item If $x \in \scutc{C_1}$ and $x \in \scutc{C_2}$, then $x \in \scutc{C_1 \cap C_2}$ and $x \in \scutc{C_1 \cup C_2}$, contributing $2$ to both sides of the identity.
		\item If $x$ belongs to exactly one of $\scutc{C_1}$ or $\scutc{C_2}$, then a case-by-case check of \Cref{tab:cut-intersection-union} shows that $x$ belongs to exactly one of $\scutc{C_1 \cap C_2}$ or $\scutc{C_1 \cup C_2}$, contributing $1$ to both sides.
		\item If $x \notin \scutc{C_1}$ and $x \notin \scutc{C_2}$, then $x \notin \scutc{C_1 \cap C_2}$ and $x \notin \scutc{C_1 \cup C_2}$, contributing $0$ to both sides.
	\end{itemize}
	Summing these contributions over all vertices yields the desired equation.
\end{proof}

An immediate consequence of \Cref{lem:union-intersection-cut} is that minimum cuts are closed under intersection and union.

\begin{corollary}
	\label{cor:mincut-closure}
	Let $C_1$ and $C_2$ be minimum cuts separating a vertex $v$ from $T$, so $|C_1| = |C_2| = \connvg{v}{G}$. Then both $C_1 \cap C_2$ and $C_1 \cup C_2$ are minimum cuts separating $v$ from $T$, each of size $\connvg{v}{G}$.
\end{corollary}

This closure under intersection yields a canonical minimum cut separating $v$ from $T$, namely the one with the smallest left side.

\begin{definition}[Tightest minimum cut]
	\label{def:tightest-min-cut}
	For any non-terminal vertex $v \in V \setminus T$, the tightest minimum cut separating $v$ from $T$ is the cut obtained by intersecting all minimum cuts separating $v$ from $T$. By \Cref{cor:mincut-closure}, this intersection is again a valid minimum cut of size $\connvg{v}{G}$.
\end{definition}

The tightest minimum cut induces useful containment relations with every other minimum cut.

\begin{lemma}
	\label{obs:tightest-min-cut}

	Let $C$ be the tightest minimum cut separating $v$ from $T$, and let $C'$ be any other minimum cut separating $v$ from $T$. Then $\lcutc{C} \subseteq \lcutc{C'}$ and $\rcutc{C'} \subseteq \rcutc{C}$. Consequently:

	\begin{enumerate}

		\item If a vertex $u$ satisfies $u \in \lcutc{C}$, then
		$u \in \lcutc{C'}$ for every minimum cut $C'$ separating $v$ from $T$.

		\item If a vertex $u$ satisfies $u \in \scutc{C'}$ for some minimum cut $C'$ separating $v$ from $T$, then
		$u \in \scutc{C} \cup \rcutc{C}$.

		\item If a terminal $t$ satisfies $t \in \scutc{C'}$ for some minimum cut $C'$ separating $v$ from $T$, then
		$t \in \scutc{C}$.

	\end{enumerate}
\end{lemma}

\begin{proof}
	Since $C$ is the intersection of all minimum cuts separating $v$ from $T$, we have $C \cap C' = C$. By \Cref{def:cut-operations}, this gives $\lcutc{C} = \lcutc{C} \cap \lcutc{C'}$ and $\rcutc{C} = \rcutc{C} \cup \rcutc{C'}$, and hence $\lcutc{C} \subseteq \lcutc{C'}$ and $\rcutc{C'} \subseteq \rcutc{C}$.

	The first conclusion follows from the former containment. For the second, a vertex in $\scutc{C'}$ cannot lie in $\lcutc{C}$, since every vertex of $\lcutc{C}$ lies in $\lcutc{C'}$; thus it lies in $\scutc{C} \cup \rcutc{C}$. Finally, if a terminal $t \in \scutc{C'}$, then $t \notin \lcutc{C}$ by the same argument and $t \notin \rcutc{C}$ because $C$ separates $v$ from $T$, so $t \in \scutc{C}$.
\end{proof}

\section{Flow-Essential Assignment}
\label{sec:essential-assignment-conditions}

In this section, we introduce the core connectivity condition that makes the contraction step possible. This condition is weaker than $k$-$T$-connectivity, but still strong enough to guarantee the existence of the desired partition. Its guiding principle is that a vertex should be assigned only to terminals that are forced by its connectivity to $T$. We restate the essentiality definition below.

\defessentialterminal*

We next relate the tightest minimum cut (\Cref{def:tightest-min-cut}) and essential terminals. We show that a terminal $t$ is essential for a vertex $v$ if and only if $t$ lies in the separator of the tightest minimum cut separating $v$ from $T$. This cut view is used later when we delete terminals or compare critical edges.

\begin{lemma}
	\label{lem:essential-terminal-tightest-cut}
	Let $G=(V,E)$ be a directed graph with terminal set $T$, let $v \in V \setminus T$, and let $C = (\lcutc{C}, \scutc{C}, \rcutc{C})$ be the tightest minimum cut separating $v$ from $T$ in $G$. A terminal $t \in T$ is essential for $v$ in $G$ if and only if $t \in \scutc{C}$.
\end{lemma}

\begin{proof}
We first prove that membership in $\scutc{C}$ implies essentiality. Assume that $t \in \scutc{C}$, and consider the partition $C' = (\lcutc{C}, \scutc{C} \setminus \{t\}, \rcutc{C})$ in the graph $G \setminus \{t\}$. Then $C'$ is a valid cut in $G \setminus \{t\}$ separating $v$ from the remaining terminals and its size is $\connvg{v}{G} - 1$. By \Cref{lem:mengers-theorem}, this gives $\connvg{v}{G \setminus \{t\}} \le \connvg{v}{G} - 1$. A maximum path family in $G$ contains at most one path through $t$, so deleting that path, if it exists, gives $\connvg{v}{G \setminus \{t\}} \ge \connvg{v}{G} - 1$. Hence $\connvg{v}{G \setminus \{t\}} = \connvg{v}{G} - 1$, so $t$ is essential for $v$.

For the converse, assume that $t$ is essential for $v$. By the definition of essentiality $\connvg{v}{G \setminus \{t\}} = \connvg{v}{G} - 1$. By \Cref{lem:mengers-theorem}, there exists a cut $C' = (\lcutc{C'}, \scutc{C'}, \rcutc{C'})$ in $G \setminus \{t\}$ separating $v$ from $T \setminus \{t\}$ such that $|\scutc{C'}| = \connvg{v}{G} - 1$. Add $t$ to its separator, defining $C'' = (\lcutc{C'}, \scutc{C'} \cup \{t\}, \rcutc{C'})$. Then $C''$ is a valid minimum cut in $G$ separating $v$ from $T$ and contains $t$ in its separator. By item~(3) of \Cref{obs:tightest-min-cut}, every terminal in the separator of any minimum cut lies in the separator of the tightest minimum cut. Hence $t \in \scutc{C}$, concluding the proof.
\end{proof}

Combined with \Cref{obs:tightest-min-cut}, the above lemma shows that if a terminal lies in the separator of some minimum cut $\cut$ separating $v$ from $T$, then it also lies in the separator of the tightest minimum cut, and hence is essential.

Using \Cref{def:essential-terminal}, one can find the essential terminals of a vertex $v$ by computing $\connvg{v}{G}$ and $\connvg{v}{G \setminus \{t\}}$ for each terminal $t$. However, this requires $k+1$ maximum-flow computations. Using the above lemma, we can find all essential terminals of a vertex $v$ by finding the tightest minimum cut which requires only one maximum-flow computation, as stated in the following proposition. Computing the tightest minimum cut is a standard algorithm in the literature and we include a proof for completeness.

\begin{proposition}[Computing essential terminals for one vertex]
\label{prop:essential-terminals-single-flow}
For a fixed vertex $v \in V \setminus T$, the set of terminals essential for $v$ can be computed using one maximum-flow computation in a network with $\caO(n)$ vertices and $\caO(n+m)$ edges and linear-time postprocessing. In particular, this takes $\caO((n+m)^{1+o(1)})$ time using the maximum-flow algorithm of \citet{brand2023deterministic}.
\end{proposition}

\begin{proof}
Let $\kappa=\connvg{v}{G}$ and set $K=k+1$, so $K>\kappa$. We use the standard vertex-splitting network. For every vertex $x\in V$, create two copies $x_{\mathrm{in}}$ and $x_{\mathrm{out}}$. Add the split arc $(x_{\mathrm{in}},x_{\mathrm{out}})$ with capacity $1$ when $x\neq v$ and with capacity $K$ when $x=v$. Add a source $s_{\mathrm{in}}$, a sink $s_{\mathrm{out}}$, and the arc $(s_{\mathrm{in}},v_{\mathrm{in}})$ of capacity $K$. For every edge $(x,y)\in E$, add the arc $(x_{\mathrm{out}},y_{\mathrm{in}})$ of capacity $K$. Finally, for every terminal $t\in T$, add the arc $(t_{\mathrm{out}},s_{\mathrm{out}})$ of capacity $K$. Call the resulting network $H$.

The split arcs of capacity $1$ ensure that each vertex other than $v$ is used by at most one unit of flow. Thus, a family of $r$ vertex-disjoint paths from $v$ to distinct terminals gives an integral $s_{\mathrm{in}}$-$s_{\mathrm{out}}$ flow of value $r$. Conversely, decompose an integral flow of value $r$ into paths and cycles, and discard the cycles. The unit capacities of the split arcs ensure that the resulting $r$ paths are vertex-disjoint outside $v$ and end at distinct terminals. Hence the maximum flow value in $H$ is $\kappa$.

We use the residual graph of a maximum flow to find the tightest minimum cut. Let $f$ be an integral maximum flow of value $\kappa$. For each arc $a=(x,y)$ of $H$, let $c_a$ be its capacity. The residual graph of $f$ contains the forward arc $(x,y)$ if $f(a)<c_a$ and the reverse arc $(y,x)$ if $f(a)>0$. Let $\Reach$ be the set of network vertices from which $s_{\mathrm{out}}$ can be reached in this residual graph. Since $f$ is maximum, the residual graph has no augmenting path from $s_{\mathrm{in}}$ to $s_{\mathrm{out}}$. It follows that $s_{\mathrm{in}}\notin\Reach$.

To turn $\Reach$ into a partition of $V$, we observe that if $x_{\mathrm{in}}\in\Reach$ then $x_{\mathrm{out}}\in\Reach$. This is because if the split arc of $x$ carries positive flow, its reverse residual arc lets $x_{\mathrm{out}}$ reach $x_{\mathrm{in}}$, and from there it can reach $s_{\mathrm{out}}$. If the split arc carries no flow, flow conservation at $x_{\mathrm{in}}$ shows that every arc entering $x_{\mathrm{in}}$ carries no flow. The only residual arc leaving $x_{\mathrm{in}}$ is then the forward split arc, so $x_{\mathrm{out}}$ again reaches $s_{\mathrm{out}}$.

We now define a partition $C=(\lcutc{C},\scutc{C},\rcutc{C})$ of $V$:
\begin{align*}
\lcutc{C}&=\{x\in V\mid x_{\mathrm{in}},x_{\mathrm{out}}\in\Reach\},\\
\scutc{C}&=\{x\in V\mid x_{\mathrm{in}}\notin\Reach,\ x_{\mathrm{out}}\in\Reach\},\\
\rcutc{C}&=\{x\in V\mid x_{\mathrm{in}},x_{\mathrm{out}}\notin\Reach\}.
\end{align*}
We verify that this partition is a cut separating $v$ from $T$. For each terminal $t$, the arc $(t_{\mathrm{out}},s_{\mathrm{out}})$ has capacity $K$ and carries at most $\kappa$ units of flow. This arc is therefore present in the residual graph, so $t_{\mathrm{out}}\in\Reach$ and $t\in\lcutc{C}\cup\scutc{C}$. The arcs $(s_{\mathrm{in}},v_{\mathrm{in}})$ and $(v_{\mathrm{in}},v_{\mathrm{out}})$ also have capacity $K$ and carry at most $\kappa$ units of flow, so both are present in the residual graph. If either copy of $v$ belonged to $\Reach$, these residual arcs would give a path from $s_{\mathrm{in}}$ to $s_{\mathrm{out}}$. Since no such path exists, $v\in\rcutc{C}$. Now suppose that an edge $(x,y)\in E$ goes from $\rcutc{C}$ to $\lcutc{C}$. The corresponding arc $(x_{\mathrm{out}},y_{\mathrm{in}})$ has capacity $K$ and carries at most $\kappa$ units of flow, so it is present in the residual graph. Since $y_{\mathrm{in}}\in\Reach$, this would imply $x_{\mathrm{out}}\in\Reach$, contrary to $x\in\rcutc{C}$. Thus, no edge goes from $\rcutc{C}$ to $\lcutc{C}$, and $C$ separates $v$ from $T$.

We next prove that $C$ is a minimum cut. Consider the flow network cut $(V(H)\setminus\Reach,\Reach)$. If an arc entering $\Reach$ were not saturated, its forward residual arc would place its tail in $\Reach$. Hence every arc entering $\Reach$ is saturated. If an arc leaving $\Reach$ carried positive flow, its reverse residual arc would place its head in $\Reach$. Hence every arc leaving $\Reach$ carries no flow. The net flow entering $\Reach$ is $\kappa$ because the source lies outside $\Reach$ and the sink lies in $\Reach$. Since every entering arc is saturated and every leaving arc carries no flow, this net flow equals the capacity of the network cut. Its capacity is therefore $\kappa$. Because $\kappa<K$, every arc entering $\Reach$ has capacity $1$. These arcs are exactly the split arcs of the vertices in $\scutc{C}$, so $|\scutc{C}|=\kappa$. Therefore, $C$ is a minimum cut separating $v$ from $T$.

It remains to show that $C$ is the tightest minimum cut. Let $C'=(\lcutc{C'},\scutc{C'},\rcutc{C'})$ be any minimum cut separating $v$ from $T$. We build a flow network cut $(X,Y)$ from $C'$ as follows. Put both copies of each vertex in $\rcutc{C'}$ in $X$, put $x_{\mathrm{in}}$ in $X$ and $x_{\mathrm{out}}$ in $Y$ for each $x\in\scutc{C'}$, and put both copies of each vertex in $\lcutc{C'}$ in $Y$. Put $s_{\mathrm{in}}$ in $X$ and $s_{\mathrm{out}}$ in $Y$. The arc out of $s_{\mathrm{in}}$ stays in $X$ because $v\in\rcutc{C'}$, and every arc into $s_{\mathrm{out}}$ starts in $Y$ because every terminal lies in $\lcutc{C'}\cup\scutc{C'}$. An arc corresponding to an edge of $G$ could go from $X$ to $Y$ only if that edge went from $\rcutc{C'}$ to $\lcutc{C'}$, which the definition of a cut forbids. Thus, the only network arcs from $X$ to $Y$ are the split arcs of the vertices in $\scutc{C'}$. The capacity of $(X,Y)$ is therefore $|\scutc{C'}|=\kappa$, so $(X,Y)$ is a minimum network cut.

We now compare the fixed maximum flow $f$ with the minimum network cut $(X,Y)$. The net flow from $X$ to $Y$ equals $\kappa$, and the total capacity of the arcs from $X$ to $Y$ is also $\kappa$. Hence every arc from $X$ to $Y$ is saturated and every arc from $Y$ to $X$ carries no flow. There is therefore no residual arc from $X$ to $Y$. Since $s_{\mathrm{out}}\in Y$, no vertex in $X$ can reach $s_{\mathrm{out}}$, and thus $\Reach\subseteq Y$. If $x\in\lcutc{C}$, both copies of $x$ belong to $\Reach$ and hence to $Y$, so $x\in\lcutc{C'}$. Thus, $\lcutc{C}\subseteq\lcutc{C'}$. If $x\in\rcutc{C'}$, both copies of $x$ belong to $X$ and hence lie outside $\Reach$, so $x\in\rcutc{C}$. Thus, $\rcutc{C'}\subseteq\rcutc{C}$.

These containments hold for every minimum cut $C'$, and $C$ is itself a minimum cut. Thus, $\lcutc{C}$ is the intersection of the left sides of all minimum cuts, while $\rcutc{C}$ is the union of their right sides. By \Cref{def:tightest-min-cut}, $C$ is the tightest minimum cut.

By \Cref{lem:essential-terminal-tightest-cut}, the terminals essential for $v$ are exactly the terminals in $\scutc{C}$. We obtain this set from the maximum flow by one reachability search in the reverse residual graph. The network has $\caO(n)$ vertices and $\caO(n+m)$ arcs, and the search takes linear time. The maximum-flow computation takes $\caO((n+m)^{1+o(1)})$ time by using the algorithm of \citet{brand2023deterministic}.
\end{proof}

The next lemma shows that deleting an edge can affect the essentiality of a terminal for a vertex only if it reduces that vertex's terminal connectivity.
\begin{lemma} \label{lem:essential-connectivity}
Let $G=(V,E)$ have terminal set $T$, let $v\in V\setminus T$, and let $\conn=\connvg{v}{G}$. For every edge $e\in E$, the graph $G'=G\setminus\{e\}$ satisfies $\connvg{v}{G'}\ge\conn-1$. Moreover, if $\connvg{v}{G'}=\conn$, then every terminal essential for $v$ in $G$ remains essential for $v$ in $G'$.
\end{lemma}
\begin{proof}
We first prove the connectivity bound. A family of $\conn$ vertex-disjoint paths from $v$ to distinct terminals is also edge-disjoint. Deleting $e$ destroys at most one of these paths, leaving a family of at least $\conn-1$ paths in $G'$. Hence $\connvg{v}{G'}\ge\conn-1$.

Now assume that $\connvg{v}{G'}=\conn$, and let $t$ be essential for $v$ in $G$. Suppose for contradiction that $t$ is not essential for $v$ in $G'$. Then $G'$ contains a family of $\conn$ vertex-disjoint paths from $v$ to distinct terminals that avoids $t$. This family is also present in $G$ and avoids $t$, contradicting that $t$ is essential for $v$ in $G$. Therefore $t$ is essential for $v$ in $G'$.
\end{proof}

\subsection{\essentialassignmentcondition}
\label{subsec:essential-assignment-condition}

In this section, we work in the unweighted setting and define assignments and the \essentialassignmentcondition using the essentiality defined in \Cref{def:essential-terminal}.

\begin{definition}[Assignment]
	\label{def:assignment}
	An assignment is a function $\assgn: V \setminus T \to T$. For each vertex $v \in V \setminus T$, the terminal $\assgn(v)$ is the terminal to which $v$ is assigned. For each terminal $t \in T$, let
	$$\assgninv(t) = \{v \in V \setminus T : \assgn(v) = t\}.$$
\end{definition}

We now restate the definition of \essentialassignmentcondition. The \essentialassignmentcondition is a capacitated assignment of non-terminals to terminals: each non-terminal is assigned to a terminal that is essential for it, and each terminal $t$ receives exactly $\cp_t$ vertices. This is the point where we marry the matching and cut concepts: the capacitated assignment captures the matching side, while essentiality captures the cut side. We call an assignment witnessing this condition a \emph{flow-essential assignment}; it is the object maintained by the unweighted contraction algorithm (\Cref{alg:gl-partition}).

\defessentialassignment*

\begin{figure}[htbp]
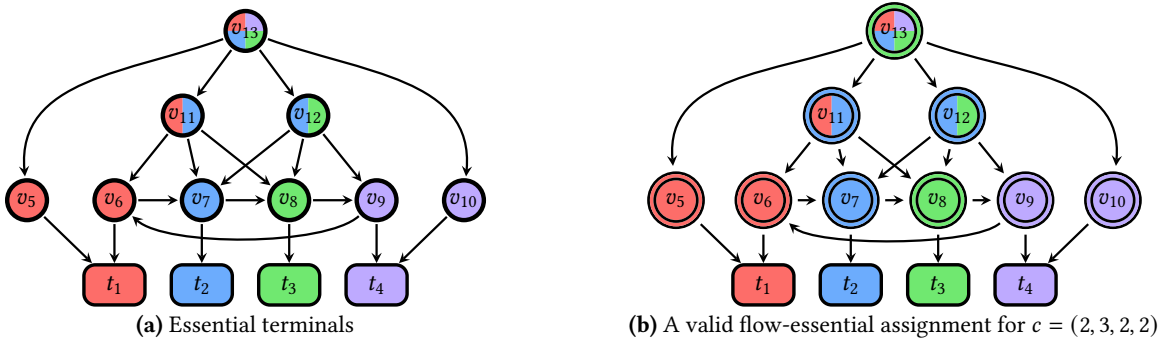

  \centering

  \begin{subfigure}[t]{0.48\textwidth}
    \centering
    \begin{essgraph}[scale=0.85, node scale=1.3, >=stealth]

      \essterminal{t1}{(-2.40, -0.45)}{$t_1$}
      \essterminal{t2}{(-0.80, -0.45)}{$t_2$}
      \essterminal{t3}{( 0.80, -0.45)}{$t_3$}
      \essterminal{t4}{( 2.40, -0.45)}{$t_4$}

      \essvertex[colors={t1/0.25, t2/0.25, t3/0.25, t4/0.25}]{v13}{(0.00, 4.20)}{$v_{13}$}

      \essvertex[colors={t1/0.5, t2/0.5}]{v11}{(-1.15, 2.65)}{$v_{11}$}
      \essvertex[colors={t2/0.5, t3/0.5}]{v12}{( 1.15, 2.65)}{$v_{12}$}

      \essvertex[colors={t1/1.0}]{v5}{(-4.00, 1.10)}{$v_5$}
      \essvertex[colors={t1/1.0}]{v6}{(-2.40, 1.10)}{$v_6$}
      \essvertex[colors={t2/1.0}]{v7}{(-0.80, 1.10)}{$v_7$}
      \essvertex[colors={t3/1.0}]{v8}{( 0.80, 1.10)}{$v_8$}
      \essvertex[colors={t4/1.0}]{v9}{( 2.40, 1.10)}{$v_9$}
      \essvertex[colors={t4/1.0}]{v10}{( 4.00, 1.10)}{$v_{10}$}

      \begin{scope}[on background layer]

        \draw[->, line width=\EssEdgeLineWidth]
          (v13) to[out=200, in=95, looseness=1.05] (v5);

        \essedge[extra={->}]{v13}{v11}
        \essedge[extra={->}]{v13}{v12}

        \draw[->, line width=\EssEdgeLineWidth]
          (v13) to[out=-20, in=85, looseness=1.05] (v10);

        \essedge[extra={->}]{v11}{v6}
        \essedge[extra={->}]{v11}{v7}
        \essedge[extra={->}]{v11}{v8}

        \essedge[extra={->}]{v12}{v7}
        \essedge[extra={->}]{v12}{v8}
        \essedge[extra={->}]{v12}{v9}

        \essedge[extra={->}]{v6}{v7}
        \essedge[extra={->}]{v7}{v8}
        \essedge[extra={->}]{v8}{v9}

        \draw[->, line width=\EssEdgeLineWidth]
          (v9) .. controls (1.45, 0.25) and (-1.45, 0.25) .. (v6);

        \essedge[extra={->}]{v5}{t1}
        \essedge[extra={->}]{v6}{t1}
        \essedge[extra={->}]{v7}{t2}
        \essedge[extra={->}]{v8}{t3}
        \essedge[extra={->}]{v9}{t4}
        \essedge[extra={->}]{v10}{t4}

      \end{scope}

    \end{essgraph}
    \caption{Essential terminals}
    \label{fig:essential-and-assignment-a}
  \end{subfigure}
  \hfill
  \begin{subfigure}[t]{0.48\textwidth}
    \centering
    \begin{essgraph}[scale=0.85, node scale=1.3, >=stealth]

      \essterminal{t1}{(-2.40, -0.45)}{$t_1$}
      \essterminal{t2}{(-0.80, -0.45)}{$t_2$}
      \essterminal{t3}{( 0.80, -0.45)}{$t_3$}
      \essterminal{t4}{( 2.40, -0.45)}{$t_4$}

      \essvertex[colors={t1/0.25, t2/0.25, t3/0.25, t4/0.25}, assigned=t3]{v13}{(0.00, 4.20)}{$v_{13}$}

      \essvertex[colors={t1/0.5, t2/0.5}, assigned=t2]{v11}{(-1.15, 2.65)}{$v_{11}$}
      \essvertex[colors={t2/0.5, t3/0.5}, assigned=t2]{v12}{( 1.15, 2.65)}{$v_{12}$}

      \essvertex[colors={t1/1.0}, assigned=t1]{v5}{(-4.00, 1.10)}{$v_5$}
      \essvertex[colors={t1/1.0}, assigned=t1]{v6}{(-2.40, 1.10)}{$v_6$}
      \essvertex[colors={t2/1.0}, assigned=t2]{v7}{(-0.80, 1.10)}{$v_7$}
      \essvertex[colors={t3/1.0}, assigned=t3]{v8}{( 0.80, 1.10)}{$v_8$}
      \essvertex[colors={t4/1.0}, assigned=t4]{v9}{( 2.40, 1.10)}{$v_9$}
      \essvertex[colors={t4/1.0}, assigned=t4]{v10}{( 4.00, 1.10)}{$v_{10}$}

      \begin{scope}[on background layer]

        \draw[->, line width=\EssEdgeLineWidth]
          (v13) to[out=200, in=95, looseness=1.05] (v5);

        \essedge[extra={->}]{v13}{v11}
        \essedge[extra={->}]{v13}{v12}

        \draw[->, line width=\EssEdgeLineWidth]
          (v13) to[out=-20, in=85, looseness=1.05] (v10);

        \essedge[extra={->}]{v11}{v6}
        \essedge[extra={->}]{v11}{v7}
        \essedge[extra={->}]{v11}{v8}

        \essedge[extra={->}]{v12}{v7}
        \essedge[extra={->}]{v12}{v8}
        \essedge[extra={->}]{v12}{v9}

        \essedge[extra={->}]{v6}{v7}
        \essedge[extra={->}]{v7}{v8}
        \essedge[extra={->}]{v8}{v9}

        \draw[->, line width=\EssEdgeLineWidth]
          (v9) .. controls (1.45, 0.25) and (-1.45, 0.25) .. (v6);

        \essedge[extra={->}]{v5}{t1}
        \essedge[extra={->}]{v6}{t1}
        \essedge[extra={->}]{v7}{t2}
        \essedge[extra={->}]{v8}{t3}
        \essedge[extra={->}]{v9}{t4}
        \essedge[extra={->}]{v10}{t4}

      \end{scope}

    \end{essgraph}
    \caption{A valid flow-essential assignment for $c = (2,3,2,2)$}
    \label{fig:essential-and-assignment-b}
  \end{subfigure}
  
  \caption{Essential terminals and a valid flow-essential assignment. In \subref{fig:essential-and-assignment-a}, the shaded sectors of each non-terminal indicate its essential terminals. In \subref{fig:essential-and-assignment-b}, the colored outer ring indicates the terminal to which each non-terminal is assigned. The displayed assignment is valid for $c=(2,3,2,2)$.}
  \label{fig:essential-and-assignment}
\end{figure}

We illustrate the condition with an example in \Cref{fig:essential-and-assignment}. The figure shows that the essential terminals of $v_{11}$ are $t_1$ and $t_2$. Its terminal connectivity is $3$: there are maximum path families ending at $\{t_1,t_2,t_3\}$ and at $\{t_1,t_2,t_4\}$. Thus $t_3$ and $t_4$ can each be avoided by a maximum family, whereas $t_1$ and $t_2$ cannot. \Cref{fig:essential-and-assignment-b} shows a valid flow-essential assignment for $c = (2,3,2,2)$.

\subsection{\fractionalesscond}
\label{subsec:fractional-essential-assignment-condition}

In this section, we work in the weighted setting and generalize the \essentialassignmentcondition to allow a vertex to split its weight among several terminals.
A \FractionalEssAssign $\assgnw$ records this split: $\assgnw(v,t)$ is the amount of weight that $v$ sends to $t$, and each vertex $v$ sends exactly its full weight $\weight_v$. The term \enquote{split-assignment} means that the weight of one vertex may be split among several terminals. Later, the rounding step and the contractions in \Cref{sec:weighted} turn these split-assignments into the final parts.

\begin{definition}[Flow-Essential Split-Assignment]
	\label{def:fractional-essential-assignment}
	A \FractionalEssAssign is a function $\assgnw: (V \setminus T) \times T \to \mathbb{N}_{\ge 0}$ such that, for every vertex $v \in V \setminus T$,
	$$\sum_{t \in T} \assgnw(v,t) = w_v.$$
	For a vertex $v \in V \setminus T$, let
	$$\assgnw(v) = \{t \in T : \assgnw(v,t) > 0\},$$
	and for a terminal $t \in T$, let
	$$\assgnwinv(t) = \{v \in V \setminus T : \assgnw(v,t) > 0\}.$$
\end{definition}

Not every split is allowed. As in the unweighted setting, a vertex may send positive weight only to terminals that are essential for it; the definition of essentiality does not change. Also, each terminal $t$ may receive at most its capacity $\cp_t$. The \fractionalesscond combines these two requirements.

\deffractionalesscond*

Consider a cost function $\crt$ on vertex-terminal pairs, where $\crt_v(t)$ is the cost of sending one unit of weight from a non-terminal $v$ to a terminal $t$. In \Cref{sec:weighted}, we need to choose, among all flow-essential split-assignments satisfying the \fractionalesscond, one of minimum cost. Sending one unit of weight from a vertex $v$ to a terminal $t$ has cost $\crt_v(t)$, so sending $\assgnw(v,t)$ units costs $\assgnw(v,t)\crt_v(t)$. The next proposition shows how to find such a minimum-cost split-assignment, or determine that none exists.

\begin{algorithm}[!ht]
        \caption{\textsc{MinCostSplitAssignment}$(G,T,\cp,\weight,\crt)$}
        \label{alg:min-cost-assignment}
        \KwIn{Weighted instance $(G=(V,E),T,\cp,\weight)$ and cost function $\crt$}
        \KwRequire{$\crt_v(t)\le |V|$ for every non-terminal $v$ and every terminal $t$ essential for $v$}
        \KwOut{A minimum-cost flow-essential split-assignment $\assgnw$ satisfying the \fractionalesscond, or \textsc{Infeasible} if none exists}

        Compute all essential pairs $(v,t)$ \;
        $W\gets \sum_{u\in V\setminus T}\weight_u$ \cmt{Total supply}
        Create a directed network $H$ with vertices $\{s_{\mathrm{in}},s_{\mathrm{out}}\}\cup (V\setminus T)\cup T$ \;
        \ForEach{$v\in V\setminus T$}{
                Add arc $(s_{\mathrm{in}},v)$ to $H$ with capacity $\weight_v$ and cost $0$ \cmt{Supply of $v$}
        }
        \ForEach{essential pair $(v,t)$ with $v\in V\setminus T$ and $t\in T$}{
                Add arc $(v,t)$ to $H$ with capacity $\weight_v$ and cost $\crt_v(t)$ \cmt{Only essential pairs}
        }
        \ForEach{$t\in T$}{
                Add arc $(t,s_{\mathrm{out}})$ to $H$ with capacity $\cp_t$ and cost $0$ \cmt{Terminal capacity}
        }
        Compute a minimum-cost flow $f$ of value $W$ in $H$, if one exists \;
        \If{no such flow exists}{
                \Return \textsc{Infeasible} \cmt{No feasible split-assignment}
        }
        \Return $\assgnw$, where $\assgnw(v,t)=f(v,t)$ if $(v,t)$ is an arc of $H$, and $\assgnw(v,t)=0$ otherwise \;
\end{algorithm}

\begin{proposition}[Computing a minimum-cost flow-essential split-assignment]
	\label{prop:best-assign-is-poly}
	Let $(G,T,\cp,\weight)$ be a weighted instance, and let $\crt$ be a cost function such that $\crt_v(t) \le |V|$ for every non-terminal vertex $v \in V \setminus T$ and every terminal $t$ essential for $v$. Then one can either compute a \FractionalEssAssign $\assgnw$ satisfying the \fractionalesscond and minimizing
	$$\sum_{v \in V \setminus T}\sum_{t \in T} \assgnw(v,t)\crt_v(t),$$
	or determine that no such \FractionalEssAssign exists. After computing, for each $v \in V \setminus T$, its set of essential terminals by applying \Cref{prop:essential-terminals-single-flow}, this reduces to one minimum-cost flow computation in a bipartite network with $\caO(nk)$ edges. In particular, using the algorithm of \citet{brand2023deterministic}, the total running time is
	$$\caO\bigl(n(n+m)^{1+o(1)} + (nk)^{1+o(1)}\log(nw_{\max})\log n\bigr).$$
	This is polynomial in the input size even if the weights are exponentially large.
\end{proposition}

\begin{proof}
	Let $W = \sum_{v \in V \setminus T} \weight_v$. First compute, for each $v \in V \setminus T$, its set of essential terminals by applying \Cref{prop:essential-terminals-single-flow}; this takes $\caO(n(n+m)^{1+o(1)})$ time in total. We then build the bipartite flow network shown in \Cref{fig:min-cost-fractional-assignment}: each non-terminal supplies its weight, an intermediate arc exists exactly for an essential-terminal pair, and a terminal-to-sink arc enforces its capacity. Formally, create a source $s_{\mathrm{in}}$ and a sink $s_{\mathrm{out}}$. For each non-terminal vertex $v \in V \setminus T$, add an arc $(s_{\mathrm{in}},v)$ of capacity $\weight_v$ and cost $0$. For each terminal $t \in T$, add an arc $(t,s_{\mathrm{out}})$ of capacity $\cp_t$ and cost $0$. Finally, for every non-terminal vertex $v$ and every terminal $t$ essential for $v$, add an arc $(v,t)$ of capacity $\weight_v$ and cost $\crt_v(t)$.

	A feasible integral $s_{\mathrm{in}}$-$s_{\mathrm{out}}$ flow of value $W$ in $H$ is exactly a \FractionalEssAssign satisfying the \fractionalesscond: the arc $(s_{\mathrm{in}},v)$ forces vertex $v$ to send exactly $\weight_v$ units, the arcs $(t,s_{\mathrm{out}})$ enforce the upper bounds $\cp_t$, and the available arcs ensure that $v$ sends positive flow only to its essential terminals. Conversely, every \FractionalEssAssign satisfying the \fractionalesscond defines such a flow by sending $\assgnw(v,t)$ units on $(v,t)$. Under this correspondence, the flow cost is exactly
	$$\sum_{v \in V \setminus T}\sum_{t \in T} \assgnw(v,t)\crt_v(t).$$
	Therefore a minimum-cost flow of value $W$ yields the desired minimum-cost \FractionalEssAssign, and if no such flow exists then no witness to the \fractionalesscond exists.

	The network $H$ has $|V \setminus T| + |T| + 2 = \caO(n)$ vertices and at most
	$$|V \setminus T| + |T| + |V \setminus T|\cdot |T| = \caO(nk)$$
	arcs. Every capacity is at most $W \le n w_{\max}$, and every cost is at most $|V|$. Hence the algorithm of \citet{brand2023deterministic} computes the required minimum-cost flow in $\caO((nk)^{1+o(1)}\log(nw_{\max})\log n)$ time, which together with the preprocessing above gives the stated total running time. Thus the proposition supplies a polynomial-time feasibility test and, when a witness exists, an optimal one for any bounded cost function $\crt$.
\end{proof}

\begin{figure}[htbp]
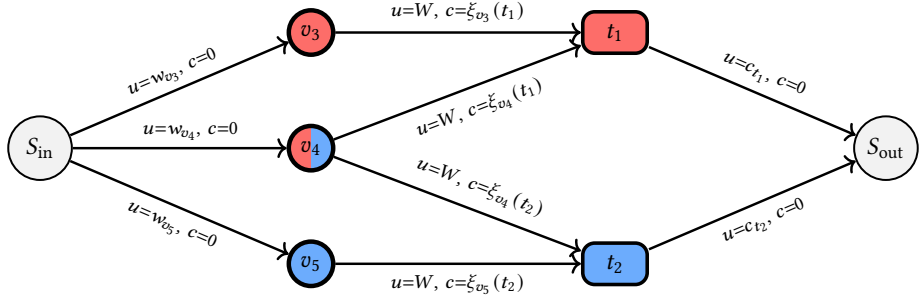

  \centering

  \begin{essgraph}[scale=0.935,node scale=1.2,term scale=0.96,
                   node font={\small}]
    \tikzset{
      flowlabel/.style={
        fill=white,
        inner sep=0.7pt,
        font=\scriptsize
      }
    }

    \essvertex[extra={thick, fill=gray!10, minimum size=9mm}]
      {Sin}{(-6.40, 0.00)}{$S_{\mathrm{in}}$}
    
    \essterminal{t1}{(2.30, 1.75)}{$t_1$}
    \essterminal{t2}{(2.30,-1.75)}{$t_2$}

    \essvertex[colors={t1/1}]{v3}{(-2.30, 1.75)}{$v_3$}
    \essvertex[colors={t1/0.5, t2/0.5}]{v4}{(-2.30, 0.00)}{$v_4$}
    \essvertex[colors={t2/1}]{v5}{(-2.30,-1.75)}{$v_5$}

    \essvertex[extra={thick, fill=gray!10, minimum size=9mm}]
      {Sout}{(6.40, 0.00)}{$S_{\mathrm{out}}$}

    \begin{scope}[on background layer]
      \draw[->, line width=\EssEdgeLineWidth]
        (Sin) -- (v3)
        node[midway, flowlabel, sloped, above=2pt]
        {$u{=}\weight_{v_3},\, c{=}0$};

      \draw[->, line width=\EssEdgeLineWidth]
        (Sin) -- (v4)
        node[pos=0.55, flowlabel, above=2pt]
        {$u{=}\weight_{v_4},\, c{=}0$};

      \draw[->, line width=\EssEdgeLineWidth]
        (Sin) -- (v5)
        node[midway, flowlabel, sloped, below=3pt]
        {$u{=}\weight_{v_5},\, c{=}0$};

      \draw[->, line width=\EssEdgeLineWidth]
        (v3) -- (t1)
        node[midway, flowlabel, above=3pt]
        {$u{=}W,\, c{=}\crt_{v_3}(t_1)$};

      \draw[->, line width=\EssEdgeLineWidth]
        (v4) -- (t1)
        node[pos=0.56, flowlabel, sloped, below=4pt]
        {$u{=}W,\, c{=}\crt_{v_4}(t_1)$};

      \draw[->, line width=\EssEdgeLineWidth]
        (v4) -- (t2)
        node[pos=0.56, flowlabel, sloped, above=4pt]
        {$u{=}W,\, c{=}\crt_{v_4}(t_2)$};

      \draw[->, line width=\EssEdgeLineWidth]
        (v5) -- (t2)
        node[midway, flowlabel, below=3pt]
        {$u{=}W,\, c{=}\crt_{v_5}(t_2)$};

      \draw[->, line width=\EssEdgeLineWidth]
        (t1) -- (Sout)
        node[pos=0.52, flowlabel, sloped, above=3pt]
        {$u{=}\cp_{t_1},\, c{=}0$};

      \draw[->, line width=\EssEdgeLineWidth]
        (t2) -- (Sout)
        node[pos=0.52, flowlabel, sloped, below=3pt]
        {$u{=}\cp_{t_2},\, c{=}0$};
    \end{scope}
  \end{essgraph}

 \caption{The bipartite minimum-cost-flow network for a flow-essential split-assignment. In this example, we assumed that in the original graph, $t_1$ is essential for $v_3$ and $v_4$, and $t_2$ is essential for $v_4$ and $v_5$.}
 \label{fig:min-cost-fractional-assignment}

\end{figure}

\section{The Polynomial-Time Algorithm}
\label{sec:algorithm}
In this section, we will provide a polynomial-time algorithm for the \gyorilovasz problem. The following theorem, restated from the introduction, formalizes our result.

\thmktconn*

To provide the algorithm, we will relax the $k$-$T$-connectedness to the \essentialassignmentcondition. The following theorem is the main result of our algorithm.

\thmessassign*

Having the above theorem, we can directly show \Cref{thm:k-t-conn}.
\begin{proof}[Proof of \Cref{thm:k-t-conn}]
We will show that if $G$ is $k$-$T$-connected, then it satisfies the \essentialassignmentcondition. To see this, note that for every vertex $v \in V \setminus T$,
$\connvg{v}{G} = k = |T|$. Thus every terminal is essential for every vertex since removing any terminal decreases $\connvg{v}{G}$ by 1. Any assignment of vertices to the terminals that respects the sizes shows the \essentialassignmentcondition.
\end{proof}

In the rest of this section, we will focus on proving \Cref{thm:ess-assign-cond}. We will first present the algorithm and then prove its correctness.

\subsection{\textsc{GLPartition} Algorithm}

We present \Cref{alg:gl-partition} which is a recursive subroutine named \textsc{GLPartition}. The algorithm takes as input the graph $G$, the set of terminals $T$ and their corresponding capacities $c=(c_1,\dots,c_k)$. It further takes an assignment $\assgn$ that satisfies the \essentialassignmentcondition. This assignment is a witness to the condition being satisfied, and it will be updated during the execution of the algorithm to make sure that the \essentialassignmentcondition remains satisfied.

\begin{algorithm}[!ht]
	\caption{\textsc{GLPartition}$(G,T,c,\assgn)$}
	\label{alg:gl-partition}
	
	\KwIn{Graph $G=(V,E)$, terminals $T=\{t_1,\dots,t_k\}$, capacities $c=(c_1,\dots,c_k)$, assignment $\assgn: V\setminus T \to T$ 
	}
	\KwOut{Partition $\langle V_1,\dots,V_k\rangle$}
	
	\If{$G=\emptyset$}{
		\Return $\emptyset$ \cmt{No vertices remain}
	}
	
	\If{$\exists i \in [k]: c_i=0$}{ \label{alg:gl-partition:line:rm-empty-t}
		$G' \gets G \setminus \{t_i\}$; $T' \gets T\setminus\{t_i\}$ \cmt{Remove saturated terminal}
		$c' \gets (c_1,\dots,c_{i-1},c_{i+1},\dots,c_k)$ \cmt{Update capacities}
		$\langle V_1,\dots,V_{i-1}, V_{i+1}, \dots V_k\rangle \gets \textsc{GLPartition}(G',T',c',\assgn)$ \cmt{Recursive call}
		$V_i \gets \{t_i\}$ \cmt{Add removed terminal}
		\Return $\langle V_1,\dots,V_k\rangle$
	}
	\label{alg:gl-partition:line:rm-empty-t-end}
	
	\If{$\exists v\in \PT(G,T)$ and $i\in[k]$ s.t. $d^+(v)=1$ and $(v,t_i)\in E$}{ \label{alg:gl-partition:line:contract-leaf}
		Choose such a vertex $v$ and index $i$ \cmt{Pre-terminal of out-degree one}
		Contract $(v,t_i)$ in $G$ to obtain $G'$ \cmt{Contract the pre-terminal}
		$c_i \gets c_i - 1$ \cmt{Update capacities}
		$\langle V_1,\dots,V_k\rangle \gets \textsc{GLPartition}(G',T,c,\assgn\setminus\{v\})$ \cmt{Recursive call}
		$V_i \gets V_i \cup \{v\}$ \cmt{Add contracted vertex}
		\Return $\langle V_1,\dots,V_k\rangle$
	}
	\label{alg:gl-partition:line:contract-leaf-end}
	
	\Else{ \label{alg:gl-partition:line:rm-edge}
		$(\assgn,e_{\mathrm{nc}})\gets \textsc{ShiftAssignment}(G,T,c,\assgn)$ \cmt{Update the assignment}
		$G' \gets G \setminus e_{\mathrm{nc}}$ \cmt{Remove one non-critical edge}
		\Return $\textsc{GLPartition}(G',T,c,\assgn)$ \cmt{Recursive call}
	}
	\label{alg:gl-partition:line:rm-edge-end}
	
\end{algorithm}

To make progress, the algorithm performs three main operations. Operation (i) removes a terminal with zero capacity (Lines~\ref{alg:gl-partition:line:rm-empty-t}--\ref{alg:gl-partition:line:rm-empty-t-end}). Operation (ii) contracts a pre-terminal of out-degree one (Lines~\ref{alg:gl-partition:line:contract-leaf}--\ref{alg:gl-partition:line:contract-leaf-end}). Operation (iii) removes an edge that maintains the \essentialassignmentcondition (Lines~\ref{alg:gl-partition:line:rm-edge}--\ref{alg:gl-partition:line:rm-edge-end}).

If any of the operations are applicable, the subroutine performs one of them and recurses. We will show that if operations (i) and (ii) cannot be performed, then operation (iii) can be performed. To show this, we introduce another subroutine \textsc{ShiftAssignment}.

An edge $e$ of $G$ is critical for assigning a non-terminal $v$ to a terminal $t$ if $t$ is essential for $v$ in $G$, but not in $G\setminus e$. Otherwise, $e$ is non-critical for assigning $v$ to $t$. For an assignment $\assgn$, we say $e$ is critical for $\assgn$ if it is critical for assigning some vertex $v$ to $\assgn(v)$, and non-critical for $\assgn$ otherwise.

The \textsc{ShiftAssignment} subroutine updates the assignment until it finds a non-critical edge $e_{\mathrm{nc}}$ whose removal preserves the \essentialassignmentcondition. It first finds a matching $M=\{(p_1,t_1),\dots,(p_k,t_k)\}$ from pre-terminals to terminals, which we show must exist in \Cref{lem:terminal-leaf-matching}. Since operations (i) and (ii) cannot be performed, every pre-terminal has at least two outgoing edges, so the subroutine chooses, for each pre-terminal $p_i$, one outgoing edge other than the matching edge $(p_i,t_i)$, called the \emph{secondary edge} of $p_i$ and denoted by $e_i=(p_i,q_i)$. If some $e_i$ is not critical for assigning any vertex $v$ to its current terminal $\assgn(v)$, then we can remove it and keep the same assignment (Lines~\ref{alg:shift-assignment:line:non-critical-edge}--\ref{alg:shift-assignment:line:non-critical-edge-end}). Otherwise, for each $i$, the subroutine picks a vertex $v_i$ such that $e_i$ is critical for assigning $v_i$ to $\assgn(v_i)$. In \Cref{lem:critical-implies-assignable}, we show that $t_i$ is also essential for $v_i$, so $v_i$ can switch its assigned terminal from $\assgn(v_i)$ to $t_i$. We therefore construct the reassignment graph $\caR$ on the terminals: it contains the directed edge $\assgn(v_i)\to t_i$ labeled by $v_i$. A directed cycle in this graph allows all labeled vertices on the cycle to switch terminals simultaneously, without changing the number of vertices assigned to any terminal.

After adding all the edges to the reassignment graph $\caR$, each vertex has in-degree exactly one, hence $\caR$ contains a directed cycle. Let $\caC=(t_{i_1},t_{i_2},\dots,t_{i_r},t_{i_1})$ be such a cycle, where indices are taken modulo $r$. For each $j\in[r]$, the edge $t_{i_{j-1}}\to t_{i_j}$ is labeled by $v_{i_j}$, and therefore $\assgn(v_{i_j})=t_{i_{j-1}}$. Shifting the assignment along $\caC$ means constructing an assignment $\assgn'$ defined by
\begin{equation*}
\assgn'(v)=
\begin{cases}
t_{i_j}, & \text{if }v=v_{i_j}\text{ for some }j\in[r],\\
\assgn(v), & \text{otherwise}.
\end{cases}
\end{equation*}
Thus, every vertex labeling an edge of $\caC$ is reassigned from the tail of that edge to its head.

We will show that after such shift, the new assignment $\assgn'$ is still a witness for the \essentialassignmentcondition, and further, for any vertex $v_{i_j}$ in the cycle, the number of critical edges among $e_1,\dots,e_k$ for the new assignment is strictly smaller than before. Thus, after polynomially many shifts, we will find a non-critical edge that can be removed while maintaining the \essentialassignmentcondition.

\begin{algorithm}[!ht]
	\caption{\textsc{ShiftAssignment}$(G,T,c,\assgn)$}
	\label{alg:shift-assignment}
	
	\KwIn{Graph $G=(V,E)$, terminals $T=\{t_1,\dots,t_k\}$, capacities $c=(c_1,\dots,c_k)$, current assignment $\assgn: V\setminus T \to T$}
	\KwRequire{$c_i > 0$ for all $i\in[k]$, and $d^+(v)\ge2$ for all $v\in \PT(G,T)$}
	\KwOut{Updated assignment $\assgn$ and a non-critical edge $e_{\mathrm{nc}}$}
	
	Select a matching $M=\{(p_1,t_1), (p_2,t_2), \dots, (p_k,t_k)\}$ from pre-terminals to terminals \cmt{$p_i \in \PT(G,T)$}
	Each $p_i$ has at least one other outgoing edge; denote it by $(p_i,q_i) \in E$ \cmt{Secondary edge} \label{alg:shift-assignment:line:matching}
	
	\While{\textbf{true}}{ \label{alg:shift-assignment:line:while}
		
		\If{$\exists i$ such that $(p_i,q_i)$ is not critical for the assignment of any vertex $v$ to $\assgn(v)$}{ \label{alg:shift-assignment:line:non-critical-edge}
			$e_{\mathrm{nc}} \gets (p_i,q_i)$ \cmt{Non-critical edge found}
			\Return $(\assgn, e_{\mathrm{nc}})$
		}
		\label{alg:shift-assignment:line:non-critical-edge-end}
		
		Construct the reassignment graph $\caR$ on the terminal set $T$ \label{alg:shift-assignment:line:r-construct}
		
		\ForEach{$i \in \{1,\dots,k\}$}{
			Let $v_i$ be the vertex for which $(p_i,q_i)$ is critical for assigning $v_i$ to $\assgn(v_i)$ \\
			Add a directed edge $\assgn(v_i) \to t_i$ in the reassignment graph $\caR$, labeled by $v_i$
		}
		
		\cmt{Each vertex in the reassignment graph $\caR$ has in-degree exactly one, hence $\caR$ contains a cycle}
		\label{alg:shift-assignment:line:r-construct-end}
		
		Let $\caC$ be any directed cycle in the reassignment graph $\caR$
		
		Construct $\assgn'$ by shifting assignments along cycle $\caC$ \cmt{Reassign terminals according to cycle}
		$\assgn \gets \assgn'$
	}
\end{algorithm}

\subsection{Correctness of the Algorithm}

In this section, we will prove that \Cref{alg:gl-partition} correctly finds the desired partition formalized in the following theorem.

\begin{theorem} \label{thm:gl-partition-correctness}
Let $(G,T,c,\assgn)$ satisfy the \essentialassignmentcondition. Then \textsc{GLPartition}$(G,T,c,\assgn)$ terminates in polynomial time and returns a partition $\langle V_1,\dots,V_k\rangle$ such that, for every $i\in[k]$, $t_i\in V_i$, $|V_i|=c_i+1$, and $G[V_i]$ is connected to $t_i$.
\end{theorem}

Having the above theorem, \Cref{thm:ess-assign-cond} follows directly.

\begin{proof}[Proof of \Cref{thm:ess-assign-cond}]
We will first find a witnessing assignment $\assgn$. Construct a bipartite graph between $V(G)\setminus T$ and $T$, joining each vertex to its essential terminals; these edges can be determined in polynomial time by \Cref{prop:essential-terminals-single-flow}. A maximum-flow computation with capacity one at each non-terminal vertex and capacity $\cp_{t_i}$ at each terminal $t_i$ finds the desired assignment $\assgn$, whose existence is guaranteed by the \essentialassignmentcondition.

By \Cref{thm:gl-partition-correctness}, \textsc{GLPartition}$(G,T,c,\assgn)$ terminates in polynomial time and returns a partition $\langle V_1,\dots,V_k\rangle$ such that, for every $i\in[k]$, $t_i\in V_i$, $|V_i|=c_i+1$, and $G[V_i]$ is connected to $t_i$.
\end{proof}

We begin with a simple structural lemma about deleting a terminal. We then use it to show that if a terminal has capacity zero, then we can remove it from the graph and solve the problem recursively on the remaining graph. This is used in Lines~\ref{alg:gl-partition:line:rm-empty-t}--\ref{alg:gl-partition:line:rm-empty-t-end} of \Cref{alg:gl-partition}.

\begin{lemma} \label{lem:essential-survives-terminal-removal}
Let $G=(V,E)$ be a directed graph with terminal set $T$. Let $t,t'\in T$ be distinct terminals, and let $v\in V\setminus T$. If $t'$ is essential for $v$ in $G$, then $t'$ is essential for $v$ in $G\setminus\{t\}$. In other words, removing $t$ does not affect the essentiality of any other terminal $t'$ for $v$.
\end{lemma}
\begin{proof}
	Let $G'=G\setminus\{t\}$, write $\kappa=\connvg{v}{G}$, and let $C$ be the tightest minimum cut separating $v$ from $T$ in $G$. Since $t'$ is essential for $v$ in $G$, \Cref{lem:essential-terminal-tightest-cut} gives $t' \in \scutc{C}$.

	If $t$ is not essential for $v$, then $t\in\lcutc{C}$ by \Cref{lem:essential-terminal-tightest-cut}. Further, $\connvg{v}{G'}=\kappa$ by definition of essentiality. Thus $\scutc{C}$ is a minimum cardinality separator for $v$ from $T \setminus \{t\}$ in $G'$, and since $t' \in \scutc{C}$, it is essential for $v$ in $G'$.

	If $t$ is essential for $v$, then $t \in \scutc{C}$ by \Cref{lem:essential-terminal-tightest-cut}. Then, by definition $\connvg{v}{G'}=\kappa-1$. Let $C'=(\lcutc{C}, \scutc{C}\setminus\{t\}, \rcutc{C})$. Then $C'$ is a minimum cut separating $v$ from $T \setminus \{t\}$ in $G'$, and since $t' \in \scutc{C'}$, it is essential for $v$ in $G'$.
\end{proof}

\begin{lemma}[Removing a terminal with zero capacity] \label{lem:remove-zero-capacity-terminal}
Let $(G, T, c, \assgn)$ satisfy the \essentialassignmentcondition. Assume that there exists a terminal $t_i$ such that $c_i=0$. Let $G'$ be the graph obtained by removing $t_i$ from $G$. Let $c' = (c_1, \dots, c_{i-1}, c_{i+1}, \dots, c_k)$ be the vector of capacities for the remaining terminals. Then $(G', T\setminus\{t_i\},\; c',\; \assgn)$ also satisfies the \essentialassignmentcondition. Let $\langle V_1,\dots,V_{i-1}, V_{i+1}, \dots V_k\rangle$ be a partition for $(G', T\setminus\{t_i\},\; c',\; \assgn)$. Then $\langle V_1,\dots,V_{i-1}, V_i=\{t_i\}, V_{i+1}, \dots V_k\rangle$ is a valid partition for the original instance.
\end{lemma}
\begin{proof}
	Let $T' = T\setminus\{t_i\}$.
	Since $c_i=0$, no vertex is assigned to $t_i$. Thus, $\assgn$ is also an assignment from $V(G')\setminus T'$ to $T'$, and it satisfies the capacity constraints given by $c'$. We next show that $\assgn$ is valid, meaning that $\assgn(v)$ is essential for all $v$ in $G'$.
	For every $v\in V(G')\setminus T'$, the terminal $\assgn(v)$ is distinct from $t_i$, since no vertex is assigned to $t_i$. Because $\assgn(v)$ is essential for $v$ in $G$, \Cref{lem:essential-survives-terminal-removal} implies that $\assgn(v)$ remains essential for $v$ in $G'$.

	The parts $V_j$ for $j\neq i$ remain valid in $G$, and $V_i=\{t_i\}$ is a valid part since $c_i=0$.
\end{proof}

We continue with the reduction for pre-terminals of out-degree one. We first isolate the graph-theoretic effect of the contraction, independent of any assignment. We then use it to show that if there exists a pre-terminal of out-degree one, then we can contract it and solve the problem recursively on the contracted graph. This is used in Lines~\ref{alg:gl-partition:line:contract-leaf}--\ref{alg:gl-partition:line:contract-leaf-end} of \Cref{alg:gl-partition}.

\begin{lemma} \label{lem:essential-survives-degree-one-leaf-contraction}
Let $G=(V,E)$ be a directed graph with terminal set $T$. Let $p\in\PT(G,T)$ be a pre-terminal in $G$ and assume that $d^+(p)=1$. Let $t$ be its terminal, i.e., $(p,t)\in E$. Let $G'$ be obtained by contracting $(p,t)$ into $t$. Then for every vertex $v\in V(G')\setminus T$ and every terminal $t'\in T$, if $t'$ is essential for $v$ in $G$, then $t'$ is essential for $v$ in $G'$.
\end{lemma}
\begin{proof}
\usetikzlibrary{fit}

\begin{figure}[htbp]
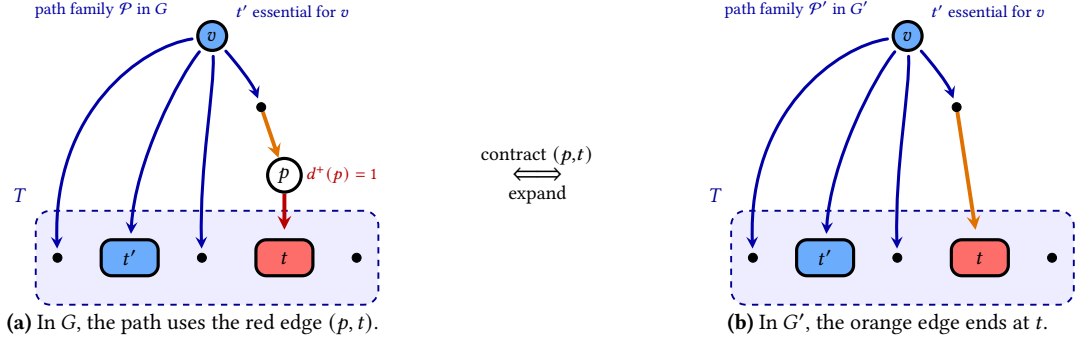

\centering
\begingroup

\tikzset{
  contraction path/.style={
    ->,
    draw=blue!65!black,
    line width=1.45pt,
    line cap=round,
    shorten <=1.5pt,
    shorten >=2pt
  },
  contracted segment/.style={
    ->,
    draw=red!72!black,
    line width=1.85pt,
    line cap=round,
    shorten <=1.5pt,
    shorten >=2pt
  },
  anonymous terminal/.style={
    circle,
    fill=black,
    draw=black,
    inner sep=1.5pt
  },
  internal path dot/.style={
    circle,
    fill=black,
    draw=black,
    inner sep=1.5pt
  }
}

\resizebox{\textwidth}{!}{%
\begin{minipage}{17.8cm}
\centering
\begingroup

\newcommand{\ContractionNodes}[1]{%
  \essterminal{t}{(1.40,-0.50)}{$t$}
  \essterminal{tprime}{(-1.60,-0.50)}{$t'$}

  \node[anonymous terminal] (ta) at (-3.00,-0.50) {};
  \node[anonymous terminal] (tb) at (-0.20,-0.50) {};
  \node[anonymous terminal] (tc) at ( 2.80,-0.50) {};

  \node[internal path dot] (join) at (0.95,2.42) {};

  \essvertex[colors={tprime/1}]{v}{(0,3.80)}{$v$}

  #1

  \begin{scope}[on background layer]
    \node[
      fit=(ta)(tprime)(tb)(t)(tc),
      fill=blue!7,
      draw=blue!55!black,
      dashed,
      rounded corners=7pt,
      line width=0.9pt,
      inner xsep=8pt,
      inner ysep=13pt,
      label={
        [font=\small\bfseries,text=blue!55!black]
        north west:$T$
      }
    ] (terminalset) {};
  \end{scope}

  \node[
    font=\scriptsize,
    text=blue!60!black,
    anchor=south west
  ] at (0.30,4.00) {$t'$ essential for $v$};
}

\newcommand{\UnchangedContractionPaths}{%
  \begin{scope}[on background layer]
    \draw[contraction path]
      (v)
      .. controls (-1.35,3.60) and (-3.10,2.35) ..
      (ta);

    \draw[contraction path]
      (v)
      .. controls (-0.65,3.00) and (-1.50,1.35) ..
      (tprime);

    \draw[contraction path]
      (v)
      .. controls (0.05,2.75) and (-0.25,1.30) ..
      (tb);

    \draw[contraction path]
      (v)
      .. controls (0.45,3.30) and (0.78,2.88) ..
      (join);
  \end{scope}
}

\begin{subfigure}[b]{0.44\textwidth}
\centering
\begin{essgraph}[scale=0.86,>=stealth,node font={\small}]
  \ContractionNodes{\essvertex{p}{(1.40,1.10)}{$p$}}
  \UnchangedContractionPaths

  \essedge[extra={->,draw=orange!88!black,line width=1.85pt}]{join}{p}
  \essedge[extra={->,draw=red!72!black,line width=1.85pt,shorten >=4pt}]{p}{t}

  \node[
    font=\scriptsize,
    text=red!72!black,
    fill=white,
    inner sep=1pt,
    anchor=west
  ] at (1.78,1.10) {$d^+(p)=1$};
  \node[font=\scriptsize,text=blue!65!black,anchor=south]
    at (-2.15,4.00) {path family $\mathcal P$ in $G$};
\end{essgraph}
\caption{In $G$, the path uses the red edge $(p,t)$.}
\label{fig:leaf-contraction-before}
\end{subfigure}
\hfill
\raisebox{2.25cm}{%
  $\underset{\text{\scriptsize expand}}%
    {\overset{\text{\scriptsize contract }(p,t)}{\Longleftrightarrow}}$
}
\hfill
\begin{subfigure}[b]{0.44\textwidth}
\centering
\begin{essgraph}[scale=0.86,>=stealth,node font={\small}]
  \ContractionNodes{}
  \UnchangedContractionPaths

  \essedge[extra={->,draw=orange!88!black,line width=1.85pt,shorten >=4pt}]{join}{t}

  \node[font=\scriptsize,text=blue!65!black,anchor=south]
    at (-2.15,4.00) {path family $\mathcal P'$ in $G'$};
\end{essgraph}
\caption{In $G'$, the orange edge ends at $t$.}
\label{fig:leaf-contraction-after}
\end{subfigure}

\endgroup
\end{minipage}
}%

\endgroup
\caption{Illustration of
  \Cref{lem:essential-survives-degree-one-leaf-contraction}. Since $d^+(p)=1$,
  contracting $p$ into $t$ gives a reversible correspondence between families
  of vertex-disjoint paths from $v$ to distinct terminals in $G$ and $G'$ that
  preserves their terminal endpoints. A path through $p$ must end with $(p,t)$;
  contraction redirects its preceding edge to $t$, and the reverse operation
  expands the redirected edge through $p$. Hence every terminal essential for
  $v$ in $G$, such as $t'$, remains essential for $v$ in $G'$.}
\label{fig:essentiality-leaf-contraction}
\end{figure}

Fix $v\in V(G')\setminus T$. Consider a family of $\connvg{v}{G}$ vertex-disjoint paths from $v$ to distinct terminals in $G$. If one of these paths contains $p$, then because $d^+(p)=1$, that path must end with the edge $(p,t)$, as in \Cref{fig:leaf-contraction-before}. Contracting $(p,t)$ turns this family into a family of vertex-disjoint paths in $G'$ with the same terminal endpoints, as in \Cref{fig:leaf-contraction-after}. Conversely, given such a family in $G'$, if a path uses an edge into $t$ created by the contraction, we expand that edge into the corresponding two-edge path through $p$. Since only one path can end at $t$, this again yields a family of vertex-disjoint paths with the same terminal endpoints in $G$.

Thus the same sets of terminal endpoints can be realized by vertex-disjoint path families from $v$ in $G$ and in $G'$. In particular, every terminal that is essential for $v$ in $G$ remains essential for $v$ in $G'$.
\end{proof}
\begin{lemma}[Contracting a pre-terminal of out-degree one] \label{lem:contract-degree-one-potential-leaf}
Let $(G, T, c, \assgn)$ satisfy the \essentialassignmentcondition. Assume that there exists a pre-terminal $p\in \PT(G,T)$ such that $d^+(p)=1$, and let $t_i$ be its corresponding terminal. Let $G'$ be the graph obtained by contracting the edge $(p,t_i)$ into $t_i$. Let $\assgn' = \assgn\setminus\{p\}$ be the assignment obtained by removing $p$ from the domain of $\assgn$. Then $(G', T, c', \assgn')$ also satisfies the \essentialassignmentcondition, where $c'$ is obtained from $c$ by decreasing $c_i$ by one. Let $\langle V'_1,\dots,V'_k\rangle$ be a partition for $(G', T, c', \assgn')$. Then $\langle V'_1,\dots,V'_{i-1}, V_i'\cup\{p\}, V'_{i+1}, \dots V'_k\rangle$ is a valid partition for the original instance.
\end{lemma}
\begin{proof}
Since $p$ is a pre-terminal of out-degree one, its only outgoing edge is $(p,t_i)$. Thus, $t_i$ is the only essential terminal for $p$, so $\assgn(p)=t_i$. Removing $p$ from the assignment decreases the number of vertices assigned to $t_i$ by one and does not change the number assigned to any other terminal. Therefore, $\assgn'$ satisfies the capacity constraints given by $c'$.

For every $v\in V(G')\setminus T$, the assigned terminal $\assgn'(v)=\assgn(v)$ is essential for $v$ in $G$. Hence \Cref{lem:essential-survives-degree-one-leaf-contraction} implies that $\assgn'(v)$ remains essential for $v$ in $G'$. Therefore, $(G',T,c',\assgn')$ satisfies the \essentialassignmentcondition.

Finally, consider the partition $\langle V'_1,\dots,V'_k\rangle$. For $j\neq i$, the part $V'_j$ remains valid in $G$. Each path to $t_i$ in $G'[V'_i]$ gives a path in $G[V'_i\cup\{p\}]$ by replacing any edge created by the contraction with its two-edge path through $p$. Also, $p$ reaches $t_i$ through $(p,t_i)$. Hence, $V'_i\cup\{p\}$ is connected to $t_i$. Adding $p$ gives the $i$-th part its required size, so the resulting partition is valid for the original instance.
\end{proof}

We finally arrive at the most involved result. We will show that if there is no terminal with zero capacity and no pre-terminal of out-degree one, then we can find an edge that can be removed while maintaining the \essentialassignmentcondition. This is used in Lines~\ref{alg:gl-partition:line:rm-edge}--\ref{alg:gl-partition:line:rm-edge-end} of \Cref{alg:gl-partition}.

\begin{lemma}[Existence of a non-critical edge] \label{lem:non-critical-edge-exists}
Let $(G,T,c,\assgn)$ satisfy the \essentialassignmentcondition.
Assume that $c_i>0$ for every $i\in[k]$, and that
$d^+(v)\ge2$ for every pre-terminal $v\in\PT(G,T)$.
Then \textsc{ShiftAssignment} terminates in polynomial time and returns
an assignment $\assgn'$ and an edge $e_{\mathrm{nc}}$ such that $(G\setminus e_{\mathrm{nc}},T,c,\assgn')$ satisfies the \essentialassignmentcondition.
\end{lemma}

Having the above results, \Cref{thm:gl-partition-correctness} follows directly.
\begin{proof}[Proof of \Cref{thm:gl-partition-correctness}]
We prove the claim by induction on $|V(G)|+|E(G)|$. If $G=\emptyset$, the claim is immediate.

If some terminal has capacity zero, the recursive call is valid and its output extends to the original instance by \Cref{lem:remove-zero-capacity-terminal}. If there is a pre-terminal of out-degree one, the same follows from \Cref{lem:contract-degree-one-potential-leaf}. Otherwise, the assumptions of \Cref{lem:non-critical-edge-exists} hold. Thus \textsc{ShiftAssignment} returns an assignment $\assgn'$ and an edge $e_{\mathrm{nc}}$ such that the recursive call on $(G\setminus e_{\mathrm{nc}},T,c,\assgn')$ is valid and its returned partition is also valid in $G$.

Each recursive call decreases $|V(G)|+|E(G)|$, so the induction applies and there are at most $|V(G)|+|E(G)|$ calls. Since each step is polynomial time by \Cref{lem:non-critical-edge-exists}, the algorithm runs in polynomial time.
\end{proof}

To prove \Cref{lem:non-critical-edge-exists}, we first show in \Cref{lem:terminal-leaf-matching} that the matching $M$ required in Line~\ref{alg:shift-assignment:line:matching} exists. We show this by contradiction and considering an inclusion-minimal Hall-deficient set of terminals $S$, i.e. a set of terminals $S$ such that $|S|>|\PT(G,S)|$ and no proper subset of $S$ has this property. We show that no terminal $t_i$ in $S$ is essential for any vertex outside of $\PT(G,S)$, which together with $c_i > 0$ contradicts the \essentialassignmentcondition. This shows that the matching $M$ exists.

We will then show that the while loop in \Cref{alg:shift-assignment} (Line~\ref{alg:shift-assignment:line:while}) terminates in polynomial time with a non-critical edge. The loop only terminates when a non-critical edge is found in Line~\ref{alg:shift-assignment:line:non-critical-edge}. Otherwise, it will keep shifting the assignment along cycles in the reassignment graph $\caR$. We will show that the total number of tuples $(v, p_i, q_i)$ where the edge $(p_i, q_i)$ is critical for assigning $v$ to $\assgn(v)$ decreases after each shift in the assignments. Since this number is polynomially bounded, the algorithm terminates in polynomial time.

We begin the analysis of \Cref{alg:shift-assignment} with showing the existence of a matching from pre-terminals to terminals in the next few lemmas. We first show the following standard result: for an inclusion-minimal Hall-deficient set of terminals $S$, every proper subset of $S$ satisfies Hall's condition. This structure will be used to show that no terminal in $S$ is essential for any vertex outside of $\PT(G,S)$.

\begin{lemma}[Structure of a minimal Hall-deficient set] \label{lem:minimal-hall-set}\label{lem:weighted-minimal-hall-set}
Let $G=(V,E)$ be a directed graph with terminal set $T$. Let $S\subseteq T$ be inclusion-minimal subject to $|S|>|\PT(G,S)|$. Then, for every $t\in S$ and every $S'\subseteq S\setminus\{t\}$, we have $|\PT(G,S')|\ge |S'|$. Consequently, there is a matching from $S\setminus\{t\}$ to $\PT(G,S)$ that saturates both sides. Moreover, $|\PT(G,S)|=|S|-1$.
\end{lemma}
\begin{proof}
Consider the bipartite graph $H$ with parts $T$ and $\PT(G,T)$, where $t\in T$ is adjacent to $p\in\PT(G,T)$ if $(p,t)\in E$. Fix $t\in S$. Every subset $S'\subseteq S\setminus\{t\}$ is a proper subset of $S$, so minimality gives $|\PT(G,S')|\ge |S'|$. Thus, $S\setminus\{t\}$ satisfies Hall's condition. Also, $|S|-1\le |\PT(G,S\setminus\{t\})|\le |\PT(G,S)|<|S|$. Hence $|\PT(G,S)|=|S|-1$ and $\PT(G,S\setminus\{t\})=\PT(G,S)$. Hall's theorem gives a matching from $S\setminus\{t\}$ to $\PT(G,S)$ that saturates $S\setminus\{t\}$; since both sides have size $|S|-1$, it saturates both sides.
\end{proof}

We next show that the above lemma implies that no terminal in $S$ is essential for any vertex outside of $\PT(G,S)$. Since there are not enough vertices in $\PT(G,S)$ to satisfy the total capacity of terminals in $S$, which is at least $|S|$ because $c_i > 0$ for every terminal $t_i\in S$, this will contradict the \essentialassignmentcondition and show that the matching from pre-terminals to terminals exists.

\begin{lemma}[Hall-deficient terminals are not essential outside their neighborhood] \label{lem:terminal-leaf-hall-set-not-essential-outside}
Let $G=(V,E)$ be a directed graph with terminal set $T$. Let $S\subseteq T$ be inclusion-minimal subject to $|S|>|\PT(G,S)|$. Then no terminal $t\in S$ is essential for any non-terminal vertex $v\in (V\setminus T)\setminus \PT(G,S)$. Equivalently, terminals in $S$ can be essential only for vertices in $\PT(G,S)$.
\end{lemma}
\begin{proof}
\tikzexternaldisable
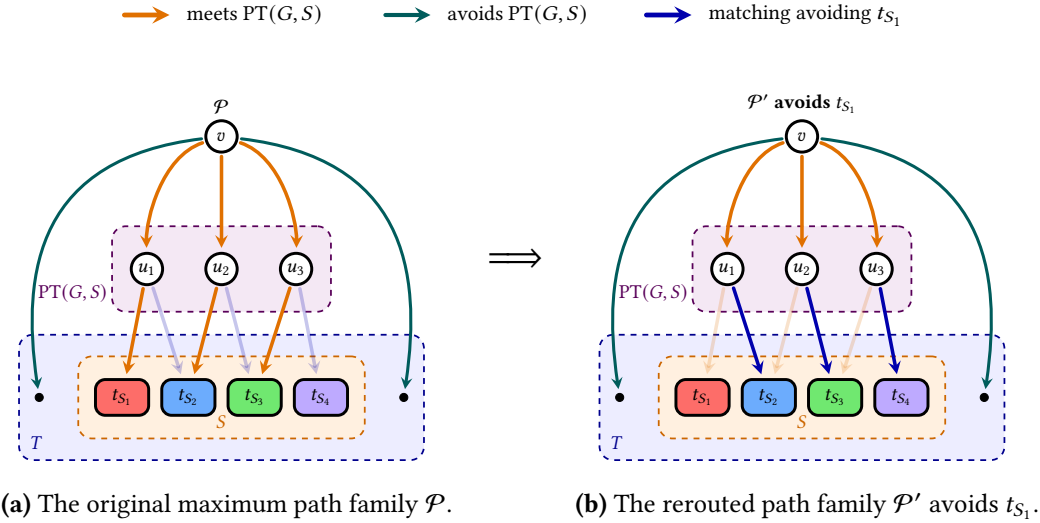
\begin{figure}[htbp]
\centering
\begingroup
\newcommand{\panelcaption}[1]{%
  \par\vspace{6pt}\begin{minipage}{6.85cm}\centering\small #1\end{minipage}%
}

\tikzset{
  >=stealth,
  hall path through/.style={
    ->,
    draw=orange!88!black,
    line width=1.8pt,
    line cap=round,
    shorten <=1.5pt,
    shorten >=2pt
  },
  hall path other/.style={
    ->,
    draw=teal!72!black,
    line width=1.55pt,
    line cap=round,
    shorten <=1.5pt,
    shorten >=2pt
  },
  hall matching edge/.style={
    ->,
    draw=blue!68!black,
    line width=1.8pt,
    line cap=round,
    shorten <=1.5pt,
    shorten >=2pt
  },
  hall faded old edge/.style={
    hall path through,
    opacity=0.20
  },
  hall faded matching edge/.style={
    hall matching edge,
    opacity=0.22
  },
  hall terminal dot/.style={
    circle,
    fill=black,
    draw=black,
    inner sep=1.45pt
  },
  hall S region/.style={
    fill=orange!12,
    draw=orange!80!black,
    dashed,
    rounded corners=7pt,
    line width=0.9pt,
    inner xsep=8pt,
    inner ysep=11pt
  },
  hall PT region/.style={
    fill=violet!9,
    draw=violet!70!black,
    dashed,
    rounded corners=7pt,
    line width=0.9pt,
    inner xsep=9pt,
    inner ysep=13pt
  },
  hall T region/.style={
    fill=blue!7,
    draw=blue!55!black,
    dashed,
    rounded corners=7pt,
    line width=0.9pt,
    inner xsep=8pt,
    inner ysep=22pt
  }
}

\newcommand{\hallpanel}[2]{%
\begin{essgraph}[scale=0.74,>=stealth,node font={\small}]
  \essterminal{tsone}{(-2.40,-0.65)}{$t_{S_1}$}
  \essterminal{tstwo}{(-0.80,-0.65)}{$t_{S_2}$}
  \essterminal{tsthree}{(0.80,-0.65)}{$t_{S_3}$}
  \essterminal{tsfour}{(2.40,-0.65)}{$t_{S_4}$}

  \node[hall terminal dot] (tleft) at (-4.40,-0.65) {};
  \node[hall terminal dot] (tright) at (4.40,-0.65) {};

  \essvertex{uone}{(-1.80,2.45)}{$u_1$}
  \essvertex{utwo}{(0,2.45)}{$u_2$}
  \essvertex{uthree}{(1.80,2.45)}{$u_3$}

  \essvertex{v}{(0,5.65)}{$v$}

  \begin{scope}[on background layer]
    \node[hall T region,
      fit=(tleft)(tsone)(tstwo)(tsthree)(tsfour)(tright)] (tregion) {};
    \node[hall S region,
      fit=(tsone)(tstwo)(tsthree)(tsfour)] (sregion) {};

    \node[hall PT region,fit=(uone)(utwo)(uthree)] (ptregion) {};

    \draw[hall path through]
      (v) .. controls (-1.45,5.16) and (-1.82,3.64) .. (uone);
    \draw[hall path through] (v) -- (utwo);
    \draw[hall path through]
      (v) .. controls (1.45,5.16) and (1.82,3.64) .. (uthree);

    \draw[hall path other]
      (v) .. controls (-4.40,5.12) and (-4.85,2.45) .. (tleft);
    \draw[hall path other]
      (v) .. controls (4.40,5.12) and (4.85,2.45) .. (tright);

    #1
  \end{scope}

  \node[font=\small\bfseries,text=blue!55!black,anchor=south west]
    at ($(tregion.south west)+(0.12,0.10)$) {$T$};
  \node[font=\small\bfseries,text=orange!80!black,anchor=south]
    at ($(sregion.south)+(0,0.10)$) {$S$};
  \node[font=\small\bfseries,text=violet!70!black,anchor=east]
    at ($(ptregion.west)+(0.10,-0.58)$) {$\PT(G,S)$};
  \node[font=\small\bfseries,anchor=south] at (0,6.02) {#2};
\end{essgraph}%
}

\resizebox{\textwidth}{!}{%
\begin{minipage}{16.4cm}
\centering

\begin{tikzpicture}[baseline=-0.5ex,>=stealth]
  \draw[hall path through] (0,0) -- (0.72,0)
    node[right,font=\scriptsize,text=black] {meets $\PT(G,S)$};
  \draw[hall path other] (3.45,0) -- (4.17,0)
    node[right,font=\scriptsize,text=black] {avoids $\PT(G,S)$};
  \draw[hall matching edge] (6.90,0) -- (7.62,0)
    node[right,font=\scriptsize,text=black]
      {matching avoiding $t_{S_1}$};
\end{tikzpicture}

\vspace{0.65cm}

\begin{minipage}[t]{7.10cm}
\centering
\hallpanel{%
  \draw[hall path through] (uone) -- (tsone);
  \draw[hall path through] (utwo) -- (tstwo);
  \draw[hall path through] (uthree) -- (tsthree);

  \draw[hall faded matching edge] (uone) -- (tstwo);
  \draw[hall faded matching edge] (utwo) -- (tsthree);
  \draw[hall faded matching edge] (uthree) -- (tsfour);
}{$\mathcal P$}
\phantomsubcaption\label{fig:bipartite-matching-exists-original}
\panelcaption{\textbf{(\thesubfigure)} The original maximum path family
  $\mathcal P$.}
\end{minipage}
\hspace{-0.28cm}
\raisebox{2.55cm}{\Large$\Longrightarrow$}
\hspace{-0.28cm}
\begin{minipage}[t]{7.10cm}
\centering
\hallpanel{%
  \draw[hall faded old edge] (uone) -- (tsone);
  \draw[hall faded old edge] (utwo) -- (tstwo);
  \draw[hall faded old edge] (uthree) -- (tsthree);

  \draw[hall matching edge] (uone) -- (tstwo);
  \draw[hall matching edge] (utwo) -- (tsthree);
  \draw[hall matching edge] (uthree) -- (tsfour);
}{$\mathcal P'$ avoids $t_{S_1}$}
\phantomsubcaption\label{fig:bipartite-matching-exists-rerouted}
\panelcaption{\textbf{(\thesubfigure)} The rerouted path family
  $\mathcal P'$ avoids $t_{S_1}$.}
\end{minipage}

\end{minipage}
}%

\caption{The rerouting used in
  \Cref{lem:terminal-leaf-hall-set-not-essential-outside}. In
  \subref{fig:bipartite-matching-exists-original}, the orange paths of the
  maximum family $\mathcal P$ end in $S$ through the distinct vertices
  $u_1,u_2,u_3\in\PT(G,S)$; the blue edges of a matching avoiding
  $t_{S_1}$ are shown faded. In
  \subref{fig:bipartite-matching-exists-rerouted}, we replace each original final edge with the matching edge incident to the same
  $u_j$. The resulting family $\mathcal P'$ has the same size, has distinct
  terminal endpoints, and avoids $t_{S_1}$.}
\label{fig:bipartite-matching-exists}

\endgroup
\end{figure}
\tikzexternalenable

\Cref{fig:bipartite-matching-exists} illustrates the argument. Fix a non-terminal vertex $v\notin\PT(G,S)$ and a terminal $t\in S$. By \Cref{lem:minimal-hall-set}, there is a matching $M$ between $\PT(G,S)$ and $S\setminus\{t\}$ that saturates both sides. Let $\mathcal P$ be a maximum family of vertex-disjoint paths from $v$ to distinct terminals, shown in \Cref{fig:bipartite-matching-exists-original}. Let $\mathcal P_1,\dots,\mathcal P_m$ be the paths in $\mathcal P$ that end in $S$, and for each $j\in[m]$ let $u_j$ be the last non-terminal vertex of $\mathcal P_j$. Since $\mathcal P_j$ ends with an edge from $u_j$ to a terminal in $S$, we have $u_j\in\PT(G,S)$. The vertices $u_1,\dots,u_m$ are distinct because the paths are vertex-disjoint, and none of them equals $v$ since $v\notin\PT(G,S)$. For each $j\in[m]$, replace the last edge of $\mathcal P_j$ by the matching edge of $M$ incident to $u_j$, as in \Cref{fig:bipartite-matching-exists-rerouted}. Such an edge exists because $M$ saturates $\PT(G,S)$, and distinct vertices $u_j$ are matched to distinct terminals in $S\setminus\{t\}$. The resulting paths remain vertex-disjoint, have distinct terminal endpoints, and avoid $t$. Thus, $t$ is not essential for $v$.
\end{proof}

\begin{lemma} \label{lem:terminal-leaf-matching}
Let $(G, T, c, \assgn)$ satisfy the \essentialassignmentcondition. Assume that $c_i>0$ for all $i\in[k]$. Then there exists a matching from pre-terminals to terminals, that is, there exist distinct vertices $p_1,\dots,p_k \in \PT(G,T)$ such that $(p_i,t_i)\in E$ for all $i\in[k]$.
\end{lemma}
\begin{proof}
Consider the bipartite graph $H$ with parts $T$ and $\PT(G,T)$, where $t\in T$ is adjacent to $p\in\PT(G,T)$ if $(p,t)\in E$. Suppose, for a contradiction, that $H$ has no matching saturating $T$. By Hall's theorem, there is a nonempty set $S\subseteq T$ with $|\PT(G,S)|<|S|$; choose such a set inclusion-minimal. By \Cref{lem:terminal-leaf-hall-set-not-essential-outside}, no terminal in $S$ is essential for any non-terminal vertex outside $\PT(G,S)$.

Since $c_i>0$ for every $i$, each terminal $t\in S$ has a vertex $v_t$ assigned to it. The vertices $v_t$ are distinct, and each $v_t$ belongs to $\PT(G,S)$, because $\assgn(v_t)=t$ is essential for $v_t$. Thus, every vertex assigned to a terminal in $S$ lies in $\PT(G,S)$, which implies $|\PT(G,S)|\ge |S|$, a contradiction. Therefore, $H$ has a matching saturating $T$, which gives the required vertices $p_1,\dots,p_k$.
\end{proof}

Having the matching $\{(p_1,t_1),\dots,(p_k,t_k)\}$, \Cref{alg:shift-assignment} considers a secondary edge for each $p_i$, denoted by $e_i=(p_i,q_i)$. \Cref{fig:potential-leaf-matching} illustrates the matching edges $(p_i,t_i)$ together with the secondary edges $e_i=(p_i,q_i)$. We will show a slightly stronger variant of \Cref{lem:non-critical-edge-exists}: we can update the assignment $\assgn$ so that one of the edges among $e_1,\dots,e_k$ becomes non-critical for $\assgn$. At each iteration of the while loop (Line~\ref{alg:shift-assignment:line:while}) in \Cref{alg:shift-assignment}, the subroutine first checks if any of these edges is non-critical for the current $\assgn$. If so, the edge is already found and $\assgn$ is a witness to the \essentialassignmentcondition even after removing that non-critical edge. Otherwise, each edge $e_i$ must have a corresponding vertex $v_i$ such that $e_i$ is critical for assigning $v_i$ to $\assgn(v_i)$. We will show in \Cref{lem:critical-implies-assignable} that $t_i$ is also essential for $v_i$, so $v_i$ can be reassigned from $\assgn(v_i)$ to $t_i$ while keeping every vertex on an essential terminal. Reassigning a single vertex this way would break the capacities, so instead we find a directed cycle in the reassignment graph and shift the assignments of the vertices along that cycle, which keeps the number of vertices at each terminal unchanged. 

\begin{figure}[htbp]
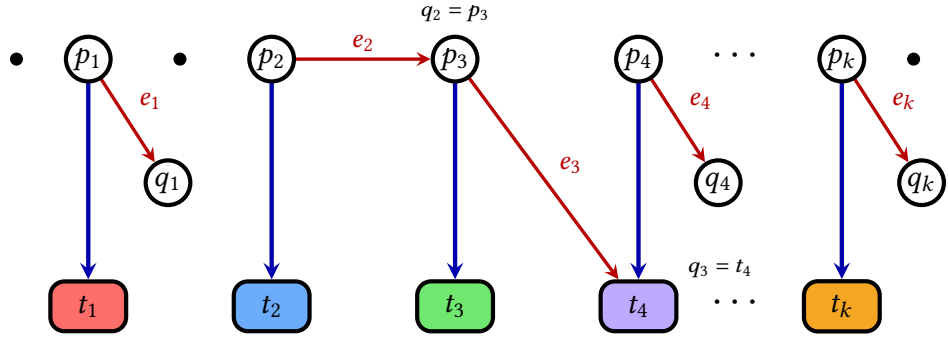

\centering
\begingroup
\tikzset{
  anonymous terminal/.style={
    circle,
    fill=black,
    draw=black,
    inner sep=1.5pt
  },
  matching edge/.style={
    ->,
    draw=blue!68!black,
    line width=1.55pt,
    line cap=round
  },
  ordinary edge/.style={
    ->,
    draw=red!72!black,
    line width=1.20pt,
    line cap=round
  },
  edge label/.style={
    font=\small,
    text=red!72!black,
    fill=white,
    inner sep=1.2pt
  }
}

\resizebox{0.75\textwidth}{!}{%
\begin{essgraph}[scale=1.0,>=stealth,line cap=round,line join=round]
  \essterminal{t1}{(-4.80,0)}{$t_1$}
  \essterminal{t2}{(-2.50,0)}{$t_2$}
  \essterminal{t3}{(-0.20,0)}{$t_3$}
  \essterminal{t4}{( 2.10,0)}{$t_4$}
  \node[font=\Large] at (3.35,0) {$\cdots$};
  \essterminal{tk}{( 4.65,0)}{$t_k$}

  \essvertex{p1}{(-4.80,3.10)}{$p_1$}
  \essvertex{p2}{(-2.50,3.10)}{$p_2$}
  \essvertex{p3}{(-0.20,3.10)}{$p_3$}
  \essvertex{p4}{( 2.10,3.10)}{$p_4$}
  \node[font=\Large] at (3.35,3.10) {$\cdots$};
  \essvertex{pk}{( 4.65,3.10)}{$p_k$}

  \node[anonymous terminal] (a1) at (-5.70,3.10) {};
  \node[anonymous terminal] (a2) at (-3.65,3.10) {};
  \node[anonymous terminal] (a3) at ( 5.55,3.10) {};

  \essvertex{q1}{(-3.80,1.55)}{$q_1$}
  \essvertex{q4}{( 3.10,1.55)}{$q_4$}
  \essvertex{qk}{( 5.65,1.55)}{$q_k$}

  \essedge[extra={matching edge}]{p1}{t1}
  \essedge[extra={matching edge}]{p2}{t2}
  \essedge[extra={matching edge}]{p3}{t3}
  \essedge[extra={matching edge}]{p4}{t4}
  \essedge[extra={matching edge}]{pk}{tk}

  \essedge[extra={ordinary edge}]{p1}{q1}
  \essedge[extra={ordinary edge}]{p2}{p3}
  \essedge[extra={ordinary edge}]{p3}{t4}
  \essedge[extra={ordinary edge}]{p4}{q4}
  \essedge[extra={ordinary edge}]{pk}{qk}

  \node[edge label] at ($(p1)!0.50!(q1)+(0.28,0.25)$) {$e_1$};
  \node[edge label] at ($(p2)!0.50!(p3)+(0,0.20)$) {$e_2$};
  \node[edge label] at ($(p3)!0.50!(t4)+(0.30,0.22)$) {$e_3$};
  \node[edge label] at ($(p4)!0.50!(q4)+(0.28,0.25)$) {$e_4$};
  \node[edge label] at ($(pk)!0.50!(qk)+(0.28,0.25)$) {$e_k$};

  \node[
    font=\scriptsize,
    anchor=south,
    fill=white,
    inner sep=1pt
  ] at ($(p3)+(0,0.43)$) {$q_2=p_3$};
  \node[
    font=\scriptsize,
    anchor=south west,
    fill=white,
    inner sep=1pt
  ] at ($(t4)+(0.58,0.32)$) {$q_3=t_4$};
\end{essgraph}%
}

\caption{
The blue edges depict the matching whose existence is guaranteed by
\Cref{lem:terminal-leaf-matching}. For each terminal $t_i$, we denote by
$p_i$ the pre-terminal matched to $t_i$. Note that there may be additional
pre-terminals besides $p_1,\ldots,p_k$. The red edges are the secondary
edges of the $p_i$'s, and we denote the other endpoint of the secondary edge
of $p_i$ by $q_i$. The secondary neighbor $q_i$ of $p_i$ may be an arbitrary
vertex, including another pre-terminal $p_j$ or even a terminal $t_j$.
}
\label{fig:potential-leaf-matching}
\endgroup
\end{figure}

We intentionally state the next cut lemmas without the assignment $\assgn$. The same facts are used again for flow-essential split-assignments in \Cref{sec:weighted}, where the terminal being tested may be any terminal receiving positive weight. As described above, the first lemma shows that if the edge $(p_i,q_i)$ is critical for assigning a vertex $v$ to a terminal $t$, then the terminal $t_i$ of the matching edge $(p_i,t_i)$ is essential for $v$.

\begin{lemma}[Criticality implies assignability] \label{lem:critical-implies-assignable}
Let $G=(V,E)$ be a directed graph with terminal set $T$. Let $p_i\in\PT(G,T)$ be a pre-terminal and $t_i\in T$ be one of its corresponding terminals, i.e. $(p_i,t_i)\in E$, and let $e_i=(p_i,q_i)$ be an edge distinct from $(p_i,t_i)$. If $e_i$ is critical for assigning a vertex $v\in V\setminus T$ to a terminal $t\in T$, then $t_i$ is essential for $v$. 
\end{lemma}
\begin{proof}

\begin{figure}[htbp]
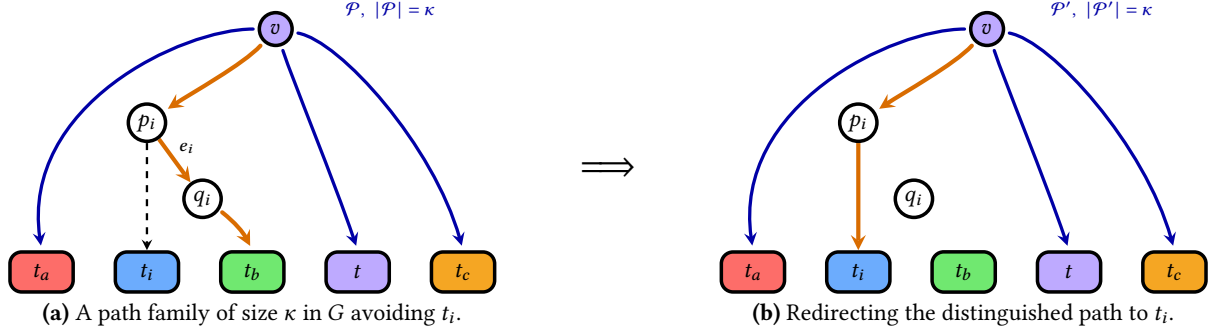

\centering
\begingroup
\tikzset{
path family/.style={
->,
draw=blue!65!black,
line width=1.4pt,
line cap=round,
shorten <=1.5pt,
shorten >=2pt
},
distinguished path/.style={
->,
draw=orange!85!black,
line width=1.8pt,
line cap=round,
shorten <=1.5pt,
shorten >=2pt
},
figure note/.style={
font=\scriptsize,
fill=white,
inner sep=1.4pt
}
}

\newcommand{\CriticalReroutingNodes}{%
  \essterminal{ta}{(-4.0,0)}{$t_a$}
  \essterminal{ti}{(-2.2,0)}{$t_i$}
  \essterminal{tb}{(-0.4,0)}{$t_b$}
  \essterminal{tt}{( 1.4,0)}{$t$}
  \essterminal{tc}{( 3.2,0)}{$t_c$}

  \essvertex{pi}{(-2.2,2.55)}{$p_i$}
  \essvertex{qi}{(-1.25,1.25)}{$q_i$}
  \essvertex[colors={tt/1}]{v}{(0,4.15)}{$v$}
}

\newcommand{\CommonUnchangedPaths}{%
  \begin{scope}[on background layer]
    \draw[path family]
      (v) .. controls (-1.60,4.15) and (-4.35,2.85) .. (ta.north);

    \draw[path family]
      (v) .. controls (0.30,3.25) and (1.00,1.55) .. (tt.north);

    \draw[path family]
      (v) .. controls (1.10,3.95) and (2.70,2.15) .. (tc.north);
  \end{scope}
}

\begin{subfigure}[b]{0.43\textwidth}
\centering
\begin{essgraph}[scale=0.88, >=stealth, node font={\small}]
  \CriticalReroutingNodes

  \essedge[extra={->,dashed}]{pi}{ti}
  \essedge[extra={->,draw=orange!85!black,line width=1.8pt}]{pi}{qi}

  \CommonUnchangedPaths

\begin{scope}[on background layer]
  \draw[distinguished path]
    (v) .. controls (-0.55,3.55) and (-1.45,3.05) .. (pi);
\draw[distinguished path]
  (qi) to[bend left=12] (tb.north);
\end{scope}

  \node[figure note] at (-1.52,2.12) {$e_i$};
  \node[figure note,text=blue!65!black] at (1.95,4.48)
    {$\mathcal{P}$, $\ |\mathcal{P}|=\kappa$};
\end{essgraph}
\caption{A path family of size $\conn$ in $G$ avoiding $t_i$.}
\label{fig:critical-edge-rerouting-a}
\end{subfigure}
\hfill
\raisebox{2.0cm}{\Large$\Longrightarrow$}
\hfill
\begin{subfigure}[b]{0.43\textwidth}
\centering
\begin{essgraph}[scale=0.88, >=stealth, node font={\small}]
  \CriticalReroutingNodes

  \essedge[extra={->,draw=orange!85!black,line width=1.8pt}]{pi}{ti}

  \CommonUnchangedPaths

  \begin{scope}[on background layer]
    \draw[distinguished path]
      (v) .. controls (-0.55,3.55) and (-1.45,3.05) .. (pi);
  \end{scope}

  \node[figure note,text=blue!65!black] at (1.95,4.48)
    {$\mathcal{P}'$, $\ |\mathcal{P}'|=\kappa$};
\end{essgraph}
\caption{Redirecting the distinguished path to $t_i$.}
\label{fig:critical-edge-rerouting-b}
\end{subfigure}

\caption{
Illustration of the rerouting argument in \Cref{lem:critical-implies-assignable}. Since $e_i$ is critical for assigning $v$ to $t$, \Cref{lem:essential-connectivity} implies that deleting $e_i$ decreases $\connvg{v}{G}$ from $\kappa$ to $\kappa-1$; hence every path family of size $\kappa$ must use $e_i$. Now assume, for contradiction, that $t_i$ is not essential for $v$. Then there exists a $\kappa$-path family $\mathcal{P}$ avoiding $t_i$, as in \subref{fig:critical-edge-rerouting-a}, and one of its paths uses $e_i=(p_i,q_i)$. Because $t_i$ is unused, we may keep that path up to $p_i$ and reroute its suffix through the matching edge $(p_i,t_i)$, obtaining the path family $\mathcal{P}'$ shown in \subref{fig:critical-edge-rerouting-b}. This avoids $e_i$ entirely, so it lies in $G' = G \setminus e_i$, contradicting $\connvg{v}{G'}=\kappa-1$. Therefore $t_i$ is essential for $v$.
}
\label{fig:critical-edge-rerouting}
\endgroup
\end{figure}

Let $G'=G\setminus e_i$ and let $\kappa=\connvg{v}{G}$. Since $e_i$ is critical for assigning $v$ to $t$, the terminal $t$ is essential for $v$ in $G$ but not in $G'$. If the connectivity of $v$ were unchanged after deleting $e_i$, then \Cref{lem:essential-connectivity} would imply that $t$ remained essential for $v$ in $G'$, a contradiction. Thus, deleting $e_i$ decreases the connectivity of $v$. By the first part of \Cref{lem:essential-connectivity}, the connectivity decreases by exactly one, so $\connvg{v}{G'}=\kappa-1$.

Suppose, for a contradiction, that $t_i$ is not essential for $v$ in $G$. Then there is a family $\mathcal{P}$ of $\kappa$ vertex-disjoint paths from $v$ to distinct terminals that avoids $t_i$. Since $G'$ has connectivity only $\kappa-1$, the family $\mathcal{P}$ must contain $e_i=(p_i,q_i)$. We refer to \Cref{fig:critical-edge-rerouting} for a visual representation. Consider the path in $\mathcal{P}$ that uses $e_i$ (\Cref{fig:critical-edge-rerouting-a}). Up to the visit to $p_i$, keep this path unchanged, and then replace its remaining suffix by the edge $(p_i,t_i)$ (\Cref{fig:critical-edge-rerouting-b}). This is valid because $t_i$ is not used by any path in $\mathcal{P}$. The resulting family still consists of $\kappa$ vertex-disjoint paths to distinct terminals, but it does not use $e_i$ and therefore lies entirely in $G'$. Hence $\connvg{v}{G'}\ge\kappa$, contradicting the earlier conclusion that $\connvg{v}{G'}=\kappa-1$. Therefore, $t_i$ must be essential for $v$ in $G$.
\end{proof}

The above lemmas give us the ability to shift the assignments along the cycle in the reassignment graph $\caR$. Now all that remains is to prove that this shift operation strictly decreases a nonnegative potential function. We define this potential function next. 

\begin{definition}[Criticality cost] \label{def:critical}
Consider a graph $G$ with terminal set $T=\{t_1,\dots,t_k\}$, a matching $M=\{(p_i,t_i):i\in[k]\}$ from pre-terminals to terminals, and, for every $i\in[k]$, an edge $e_i=(p_i,q_i)$ distinct from $(p_i,t_i)$. For a vertex $v$ and a terminal $t$, define $\crt_v(t)$ as the number of edges among $e_1,\dots,e_k$ that are critical for assigning $v$ to $t$.
\end{definition}

\begin{definition}[Assignment potential] \label{def:assignment-potential}
Fix a graph $G$, terminal set $T$, and criticality cost $\crt$ as in \Cref{def:critical}. For an assignment $\assgn$, define its potential by
\begin{equation*}
\Phi(\assgn)=\sum_{v\in V\setminus T}\crt_v(\assgn(v)).
\end{equation*}
\end{definition}

We arrive at the key lemma of this section, which shows the shift operation in \Cref{alg:shift-assignment} strictly decreases the potential $\Phi$. In other words, the shift decreases the total number of pairs $(v,e_i)$ such that deleting $e_i$ would make the terminal assigned to $v$ no longer essential for $v$.

\begin{lemma}[Cycle shift decreases the potential] \label{lem:cycle-shift-decreases-potential}
Let $(G, T, c, \assgn)$ satisfy the \essentialassignmentcondition. Let $M=\{(p_1,t_1),\dots,(p_k,t_k)\}$ be a matching from pre-terminals to terminals. For every $i\in[k]$, let $e_i=(p_i,q_i)$ be an edge distinct from $(p_i,t_i)$, and assume that $e_i$ is critical for assigning a vertex $v_i$ to $\assgn(v_i)$. Let $\caR$ be the reassignment graph on $T$ containing, for every $i\in[k]$, the edge $\assgn(v_i)\to t_i$ labeled by $v_i$. Then $\caR$ contains a directed cycle $\caC$. Let $\assgn'$ be the assignment obtained by shifting $\assgn$ along $\caC$. Then $(G,T,c,\assgn')$ satisfies the \essentialassignmentcondition and $\Phi(\assgn')<\Phi(\assgn)$.
\end{lemma}

Having the above lemma, we can prove the earlier \Cref{lem:non-critical-edge-exists}.

\begin{proof}[Proof of \Cref{lem:non-critical-edge-exists}]
By \Cref{lem:terminal-leaf-matching}, there exists a matching $M=\{(p_1,t_1),\dots,(p_k,t_k)\}$ from pre-terminals to terminals. Since every pre-terminal has out-degree at least two, for every $i\in[k]$, we can choose an edge $e_i=(p_i,q_i)$ distinct from $(p_i,t_i)$.

Consider an iteration of \textsc{ShiftAssignment}. If some edge $e_i$ is not critical for assigning any vertex $v$ to $\assgn(v)$, then removing $e_i$ preserves the essentiality of $\assgn(v)$ for every vertex $v$. Since removing an edge does not affect the capacity constraints, $(G\setminus e_i,T,c,\assgn)$ satisfies the \essentialassignmentcondition.

Otherwise, for every $i\in[k]$, there exists a vertex $v_i$ such that $e_i$ is critical for assigning $v_i$ to $\assgn(v_i)$. The subroutine constructs the reassignment graph and finds a directed cycle in Lines~\ref{alg:shift-assignment:line:r-construct}--\ref{alg:shift-assignment:line:r-construct-end}. By \Cref{lem:cycle-shift-decreases-potential}, shifting the assignment along this cycle preserves the \essentialassignmentcondition and strictly decreases $\Phi(\assgn)$. Since $0\le\Phi(\assgn)\le k|V\setminus T|$, the subroutine performs at most $k|V\setminus T|$ such shifts. Therefore, \textsc{ShiftAssignment} eventually finds an edge that is not critical for the assignment of any vertex, and this edge can be removed while maintaining the \essentialassignmentcondition.
\end{proof}

\subsection[Proof of Lemma \ref*{lem:cycle-shift-decreases-potential}]{Proof of \Cref{lem:cycle-shift-decreases-potential}}

It remains to prove \Cref{lem:cycle-shift-decreases-potential}. Consider a vertex $v_i$ whose assignment is shifted along the cycle, from $\assgn(v_i)$ to $t_i=\assgn'(v_i)$. We will show that $\crt_{v_i}(t_i)<\crt_{v_i}(\assgn(v_i))$. First, \Cref{lem:critical-implies-assignable} shows that $t_i$ is essential for $v_i$, so the shift preserves the essentiality of the new assignment. Next, \Cref{lem:own-edge-not-critical} shows that $e_i$ is not critical for assigning $v_i$ to $t_i$, although $e_i$ is critical for assigning $v_i$ to $\assgn(v_i)$ by construction. Finally, \Cref{lem:critical-transfer} shows that every other edge $e_j$ that is critical for assigning $v_i$ to $t_i$ is also critical for assigning $v_i$ to $\assgn(v_i)$. Thus, the latter set of critical edges strictly contains the former, which yields the desired decrease in potential. \Cref{table:cycle-shift-decreases-potential-roadmap} summarizes the roadmap of these facts.

\begin{table}[htbp]
    \centering
    \renewcommand{\arraystretch}{1.2}
	\begin{tabular}{@{}
        >{\centering\arraybackslash}p{0.64\linewidth}
        @{\quad}
        >{\centering\arraybackslash}p{0.24\linewidth}
        @{}}        \hline
        \textbf{Fact} & \textbf{Established by} \\
        \hline
        $t_i=\assgn'(v_i)$ is essential for $v_i$. & \Cref{lem:critical-implies-assignable} \\
        $e_i$ is not critical for assigning $v_i$ to $t_i$. & \Cref{lem:own-edge-not-critical} \\
        Every other edge critical for assigning $v_i$ to $t_i$ is also critical for assigning $v_i$ to $\assgn(v_i)$. & \Cref{lem:critical-transfer} \\
        \hline
    \end{tabular}
    \caption{The three facts establishing $\crt_{v_i}(t_i)<\crt_{v_i}(\assgn(v_i))$ for a vertex $v_i$ shifted along the cycle.}
	\label{table:cycle-shift-decreases-potential-roadmap}
\end{table}

\begin{lemma} \label{lem:own-edge-not-critical}
Let $G=(V,E)$ be a directed graph with terminal set $T$. Let $p_i\in\PT(G,T)$ be a pre-terminal and $t_i\in T$ be a terminal corresponding to $p_i$, i.e., $(p_i,t_i)\in E$. Let $e_i=(p_i,q_i)$ be an edge distinct from $(p_i,t_i)$. Then $e_i$ is not critical for assigning any vertex $v\in V\setminus T$ to $t_i$.
\end{lemma}
\begin{proof}
\begin{figure}[htbp]
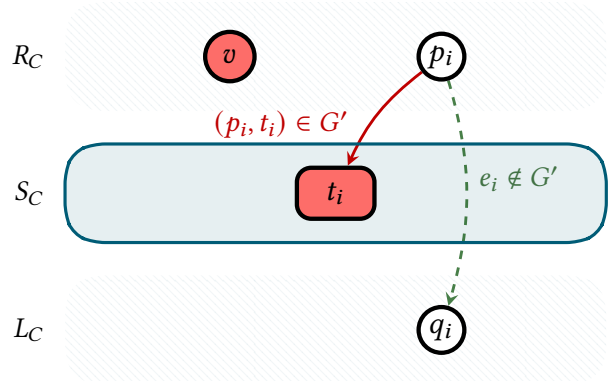

    \centering
    \renewcommand{\EssCutPadding}{1pt}
    \begin{essgraph}[scale=1.08,>=stealth]
        \essvertex{qi}{(1.20,-1.55)}{$q_i$}
        \essterminal{ti}{(0,0)}{$t_i$}
        \essvertex[colors={ti/1}]{v}{(-1.20,1.55)}{$v$}
        \essvertex{pi}{(1.20,1.55)}{$p_i$}

        \essvertex[extra={opacity=0}]{l1}{(-2.75,-1.85)}{}
        \essvertex[extra={opacity=0}]{l2}{(-1.35,-1.85)}{}
        \essvertex[extra={opacity=0}]{l3}{(0,-1.85)}{}
        \essvertex[extra={opacity=0}]{l4}{(2.75,-1.85)}{}
        \essvertex[extra={opacity=0}]{l5}{(-2.75,-1.25)}{}
        \essvertex[extra={opacity=0}]{l6}{(2.75,-1.25)}{}

        \essvertex[extra={opacity=0}]{s1}{(-2.75,-0.24)}{}
        \essvertex[extra={opacity=0}]{s2}{(2.75,-0.24)}{}
        \essvertex[extra={opacity=0}]{s3}{(-2.75,0.24)}{}
        \essvertex[extra={opacity=0}]{s4}{(2.75,0.24)}{}

        \essvertex[extra={opacity=0}]{r1}{(-2.75,1.25)}{}
        \essvertex[extra={opacity=0}]{r2}{(2.75,1.25)}{}
        \essvertex[extra={opacity=0}]{r3}{(-2.75,1.85)}{}
        \essvertex[extra={opacity=0}]{r4}{(0,1.85)}{}
        \essvertex[extra={opacity=0}]{r5}{(2.75,1.85)}{}

        \esscut[
            name=C,
            mode=full,
            left={qi,l1,l2,l3,l4,l5,l6},
            separator={ti,s1,s2,s3,s4},
            right={v,pi,r1,r2,r3,r4,r5}
        ]

        \begin{scope}[on background layer]
            \draw[->,dashed,draw=green!30!black,opacity=0.72,
                  line width=\EssEdgeLineWidth]
                (pi) to[bend left=16]
                node[pos=0.43,right=3pt,inner sep=1pt,
                     font=\small,text=green!25!black]
                    {$e_i\notin G'$}
                (qi);

            \draw[->,draw=red!75!black,line width=\EssEdgeLineWidth]
                (pi) to[bend right=14] (ti);
        \end{scope}

        \node[inner sep=1pt,font=\small,text=red!75!black]
            at (-0.62,0.78) {$(p_i,t_i)\in G'$};

        \node[anchor=east,font=\small\itshape] at (-3.18,-1.55) {$\lcut$};
        \node[anchor=east,font=\small\itshape] at (-3.18,0) {$\scut$};
        \node[anchor=east,font=\small\itshape] at (-3.18,1.55) {$\rcut$};

    \end{essgraph}
    \caption{Illustration for \Cref{lem:own-edge-not-critical}. The assumption
    that $e_i$ is critical for assigning $v$ to $t_i$ forces
    $p_i\in\rcut$ and $q_i\in\lcut$. The edge $(p_i,t_i)$ then forces
    $t_i\in\scut$. By \Cref{obs:tightest-min-cut}, $t_i$ also belongs to the
    separator of the tightest minimum cut in $G'$. Hence
    \Cref{lem:essential-terminal-tightest-cut} implies that $t_i$ is essential
    for $v$ in $G'$, a contradiction.}
    \label{fig:own-edge-not-critical}
\end{figure}

Suppose for contradiction that $e_i$ is critical for assigning some non-terminal vertex $v \in V \setminus T$ to $t_i$. Let $G' = G \setminus e_i$ and let $\conn = \connvg{v}{G}$. Because $e_i$ is critical for assigning $v$ to $t_i$, the terminal $t_i$ is essential for $v$ in $G$ but not in $G'$. Thus, \Cref{lem:essential-connectivity} gives $\connvg{v}{G'} = \conn - 1$.

Now, let $\cut = (\lcut, \scut, \rcut)$ be a minimum cut separating $v$ from $T$ in $G'$, where $v \in \rcut$. The size of this cut is exactly $\connvg{v}{G'} = \conn - 1$. If we were to add the removed edge $e_i = (p_i, q_i)$ back to $G'$, this cut must become invalid, as the connectivity of $v$ in $G$ is $\conn$. For the edge $e_i$ to invalidate the cut, it must cross from $\rcut$ to $\lcut$. Therefore, we must have $p_i \in \rcut$ and $q_i \in \lcut$, as shown in \Cref{fig:own-edge-not-critical}.

Since $(p_i,t_i)$ is an edge of $G'$ and $\cut$ is valid in $G'$, we have $t_i \notin \lcut$. As $t_i$ is a terminal, this gives $t_i \in \scut$. By \Cref{obs:tightest-min-cut}, $t_i$ also belongs to the separator of the tightest minimum cut in $G'$. Hence \Cref{lem:essential-terminal-tightest-cut} implies that $t_i$ is essential for $v$ in $G'$, a contradiction.
\end{proof}

It remains to show that reassigning $v$ to $t_i$ does not introduce any new critical edge among $e_1,\dots,e_k$.

\begin{lemma}[Transfer of criticality] \label{lem:critical-transfer}
	Let $G=(V,E)$ be a directed graph with terminal set $T=\{t_1,\dots,t_k\}$. Let $M=\{(p_i,t_i):i\in[k]\}$ be a matching from pre-terminals to terminals, and for every $i\in[k]$, let $e_i=(p_i,q_i)$ be an edge distinct from $(p_i,t_i)$. Let $v\in V\setminus T$, let $t\in T$, and let $i\in[k]$. Assume that $e_i$ is critical for assigning $v$ to $t$. Then every edge $e_j$ that is critical for assigning $v$ to $t_i$ is also critical for assigning $v$ to $t$.
\end{lemma}

\begin{proof}
\begin{figure}[htbp]
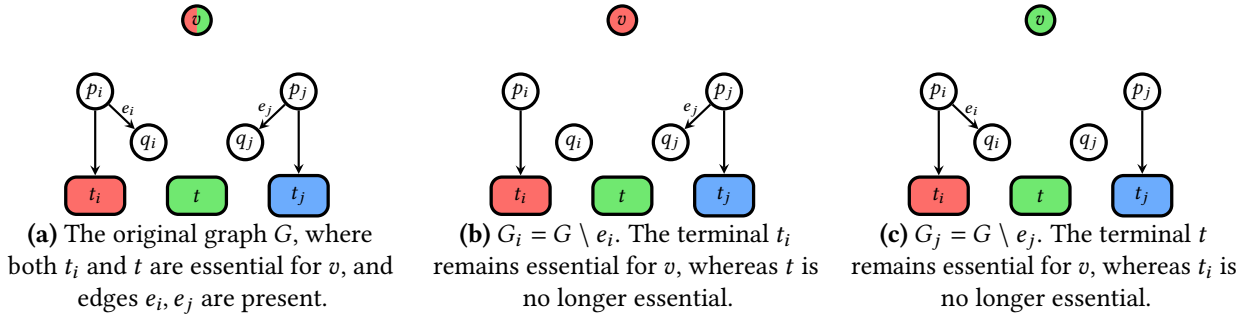

\centering
\begingroup
\newcommand{\panelcaption}[1]{%
  \par\vspace{1pt}\begin{minipage}{5.1cm}\centering\small #1\end{minipage}%
}
\resizebox{\textwidth}{!}{%
\begin{minipage}{16.4cm}
\centering

\begin{minipage}[t]{5.25cm}
\centering
\begin{essgraph}[scale=0.82,>=stealth,node font={\small}]
  \essterminal{ti}{(-2.0,0)}{$t_i$}
  \essterminal{tj}{( 2.0,0)}{$t_j$}
  \essterminal{tt}{( 0.0,0)}{$t$}

  \essvertex{qi}{(-0.95,1.05)}{$q_i$}
  \essvertex{qj}{( 0.95,1.05)}{$q_j$}

  \essvertex{pi}{(-2.0,2.05)}{$p_i$}
  \essvertex{pj}{( 2.0,2.05)}{$p_j$}

  \essvertex[colors={ti/0.5,tt/0.5}]{v}{(0,3.45)}{$v$}

  \essedge[extra={->}]{pi}{ti}
  \essedge[extra={->}]{pj}{tj}
  \essedge[extra={->}]{pi}{qi}
  \essedge[extra={->}]{pj}{qj}

  \node[font=\scriptsize] at (-1.35,1.68) {$e_i$};
  \node[font=\scriptsize] at ( 1.35,1.68) {$e_j$};
\end{essgraph}
\phantomsubcaption\label{fig:transfer-criticality-G}
\panelcaption{\textbf{(a)} The original graph $G$, where both $t_i$ and $t$ are essential for $v$, and edges $e_i, e_j$ are present.}
\end{minipage}
\hfill
\begin{minipage}[t]{5.25cm}
\centering
\begin{essgraph}[scale=0.82,>=stealth,node font={\small}]
  \essterminal{ti}{(-2.0,0)}{$t_i$}
  \essterminal{tj}{( 2.0,0)}{$t_j$}
  \essterminal{tt}{( 0.0,0)}{$t$}

  \essvertex{qi}{(-0.95,1.05)}{$q_i$}
  \essvertex{qj}{( 0.95,1.05)}{$q_j$}

  \essvertex{pi}{(-2.0,2.05)}{$p_i$}
  \essvertex{pj}{( 2.0,2.05)}{$p_j$}

  \essvertex[colors={ti/1}]{v}{(0,3.45)}{$v$}

  \essedge[extra={->}]{pi}{ti}
  \essedge[extra={->}]{pj}{tj}
  \essedge[extra={->}]{pj}{qj}

  \node[font=\scriptsize] at (1.35,1.68) {$e_j$};
\end{essgraph}
\phantomsubcaption\label{fig:transfer-criticality-Gi}
\panelcaption{\textbf{(b)} $G_i = G \setminus e_i$. The terminal $t_i$ remains essential for $v$, whereas $t$ is no longer essential.}
\end{minipage}%
\hfill
\begin{minipage}[t]{5.25cm}
\centering
\begin{essgraph}[scale=0.82,>=stealth,node font={\small}]
  \essterminal{ti}{(-2.0,0)}{$t_i$}
  \essterminal{tj}{( 2.0,0)}{$t_j$}
  \essterminal{tt}{( 0.0,0)}{$t$}

  \essvertex{qi}{(-0.95,1.05)}{$q_i$}
  \essvertex{qj}{( 0.95,1.05)}{$q_j$}

  \essvertex{pi}{(-2.0,2.05)}{$p_i$}
  \essvertex{pj}{( 2.0,2.05)}{$p_j$}

  \essvertex[colors={tt/1}]{v}{(0,3.45)}{$v$}

  \essedge[extra={->}]{pi}{ti}
  \essedge[extra={->}]{pj}{tj}
  \essedge[extra={->}]{pi}{qi}

  \node[font=\scriptsize] at (-1.35,1.68) {$e_i$};
\end{essgraph}
\phantomsubcaption\label{fig:transfer-criticality-Gj}
\panelcaption{\textbf{(c)} $G_j = G \setminus e_j$. The terminal $t$ remains essential for $v$, whereas $t_i$ is no longer essential.}
\end{minipage}

\end{minipage}
}%
\endgroup
\caption{
  Overview of graph configurations $G$, $G_i$, and $G_j$ used in the proof of \Cref{lem:critical-transfer}, illustrating how edge removals affect the essentiality of terminals $t$ and $t_i$ for vertex $v$.
}
\label{fig:transfer-criticality-configuration}
\end{figure}
We prove this lemma by contradiction. Suppose that $e_j$ is not critical for assigning $v$ to $t$, but it is critical for assigning $v$ to $t_i$. By \Cref{lem:critical-implies-assignable}, $t_i$ is essential for $v$. Additionally, by \Cref{lem:own-edge-not-critical}, we know that $e_i$ is not critical for assigning $v$ to $t_i$. Since $e_j$ is critical for that assignment by our initial assumption, it follows that $e_i \neq e_j$.

Let $\kappa=\connvg{v}{G}$, $G_i=G\setminus e_i$, $G_j=G\setminus e_j$, and $H=G\setminus\{e_i,e_j\}$. Based on our definitions and the assumptions made for the sake of contradiction, we can establish the following four facts regarding the connectivity and essentiality of terminals in $G_i$ and $G_j$:

\begin{itemize}
   \item \textit{$\connvg{v}{G_i} = \kappa - 1$ and $t$ is not essential for $v$ in $G_i$:} Due to the criticality of $e_i$ for the assignment of $v$ to $t$, removing $e_i$ implies $t$ is no longer essential for $v$. By \Cref{lem:essential-connectivity}, because the set of essential terminals cannot shrink unless the connectivity drops, the connectivity of $v$ must strictly reduce by one (\Cref{fig:transfer-criticality-Gi}).
    
    \item \textit{$\connvg{v}{G_j} = \kappa - 1$ and $t_i$ is not essential for $v$ in $G_j$:} By our assumption for the sake of contradiction, $e_j$ is critical for assigning $v$ to $t_i$. This means that removing $e_j$ causes $t_i$ to no longer be essential for $v$. Again, by \Cref{lem:essential-connectivity}, losing an essential terminal means the connectivity of $v$ must strictly reduce by one (\Cref{fig:transfer-criticality-Gj}).
    \item \textit{$t_i$ is essential for $v$ in $G_i$:} By \Cref{lem:own-edge-not-critical}, we know that $e_i$ is not critical for the assignment of $v$ to $t_i$ in $G$. Therefore, removing $e_i$ to form $G_i$ ensures $t_i$ remains essential.
    
    \item \textit{$t$ is essential for $v$ in $G_j$:} This follows directly from our initial assumption (for contradiction) that $e_j$ is not critical for assigning $v$ to $t$. Thus, removing $e_j$ to form $G_j$ keeps $t$ essential.
\end{itemize}

\begin{claim} \label{claim:all-conns-equal}
The connectivity of $v$ is $\kappa - 1$ in each of $H$, $G_i$, and $G_j$:
\begin{equation*}
\connvg{v}{H} = \connvg{v}{G_i} = \connvg{v}{G_j} = \kappa - 1
\end{equation*}
\end{claim}
\begin{proof}
\begin{figure}[!htbp]
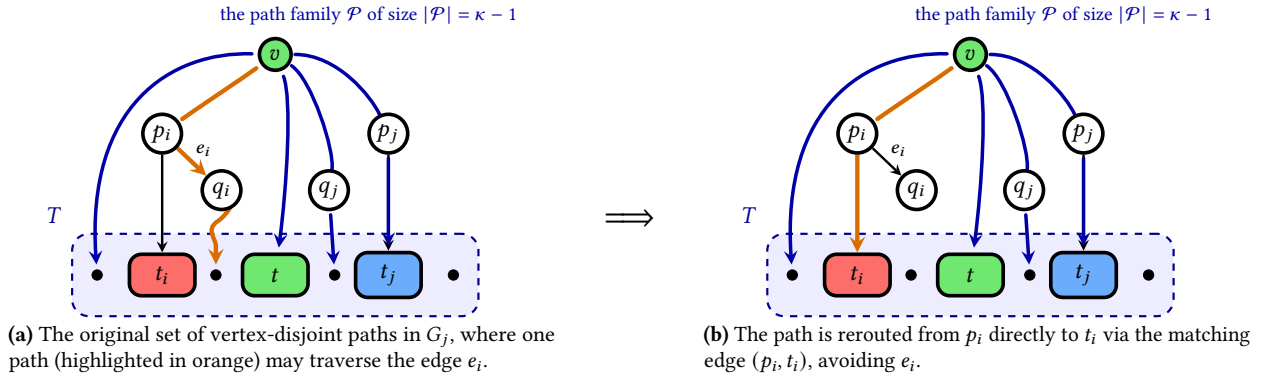

\centering
\begingroup
\tikzset{
  linkage path/.style={
    ->,
    draw=blue!65!black,
    line width=1.45pt,
    line cap=round,
    shorten <=1.5pt,
    shorten >=2pt
  },
  distinguished path/.style={
    linkage path,
    draw=orange!85!black,
    line width=1.8pt
  },
  distinguished segment/.style={
    draw=orange!85!black,
    line width=1.8pt,
    line cap=round,
    shorten <=1.5pt,
    shorten >=1.5pt
  },
  anonymous terminal/.style={
    circle,
    fill=black,
    draw=black,
    inner sep=1.5pt
  },
  path endpoint/.style={
    circle,
    fill=white,
    draw=black,
    line width=0.8pt,
    inner sep=1.35pt
  }
}
\resizebox{\textwidth}{!}{%
\begin{minipage}{18cm}
\centering

\begingroup

\newcommand{\CommonPanelNodes}{%

  \essterminal{ti}{(-2.0,-0.45)}{$t_i$}
  \essterminal{tj}{( 2.0,-0.45)}{$t_j$}
  \essterminal{tt}{( 0.0,-0.45)}{$t$}

  \node[anonymous terminal] (ta) at (-3.15,-0.45) {};
  \node[anonymous terminal] (tb) at (-1.05,-0.45) {};
  \node[anonymous terminal] (tc) at ( 1.05,-0.45) {};
  \node[anonymous terminal] (td) at ( 3.15,-0.45) {};

  \essvertex{qi}{(-0.95,1.05)}{$q_i$}
  \essvertex{qj}{( 0.95,1.05)}{$q_j$}

  \essvertex{pi}{(-2.0,2.05)}{$p_i$}
  \essvertex{pj}{( 2.0,2.05)}{$p_j$}

  \essvertex[colors={tt/1}]{v}{(0,3.45)}{$v$}

  \begin{scope}[on background layer]
    \node[
      fit=(ta)(ti)(tb)(tt)(tc)(tj)(td),
      fill=blue!7,
      draw=blue!55!black,
      dashed,
      rounded corners=7pt,
      line width=0.9pt,
      inner xsep=8pt,
      inner ysep=7pt,
      label={
        [font=\small\bfseries,text=blue!55!black]
        north west:$T$
      }
    ] (terminalset) {};
  \end{scope}

  \node[
    font=\scriptsize,
    fill=white,
    inner sep=1pt
  ] at (-1.27,1.76) {$e_i$};

  \node[
    font=\scriptsize,
    text=blue!65!black,
    anchor=south
  ] at (1.65,3.82)
    {the path family $\mathcal{P}$ of size $\vert{}\mathcal{P}\vert{} = \kappa - 1$};
}

\newcommand{\CommonGraphEdges}{%
  \essedge[extra={->}]{pj}{tj}
}

\newcommand{\CommonDisjointPaths}{%
  \begin{scope}[on background layer]

    \draw[linkage path]
      (v)
      .. controls (-1.85,3.50) and (-3.35,2.55) ..
      (ta);

    \draw[linkage path]
      (v)
      .. controls (0.30,2.65) and (0.20,1.10) ..
      (tt);

    \draw[linkage path]
      (v)
      .. controls (0.65,3.00) and (1.05,1.75) ..
      (qj)
      -- (tc);

    \draw[linkage path]
      (v)
      .. controls (1.10,3.45) and (1.70,2.70) ..
      (pj)
      -- (tj);

  \end{scope}
}

\begin{subfigure}[b]{0.44\textwidth}
\centering

\begin{essgraph}[
  scale=0.82,
  >=stealth,
  node font={\small}
]

  \CommonPanelNodes
  \CommonGraphEdges
  \essedge[extra={->}]{pi}{ti}
  \essedge[extra={->,draw=orange!85!black,line width=1.8pt}]{pi}{qi}
  \CommonDisjointPaths

  \begin{scope}[on background layer]
    \draw[distinguished segment]
      (v)
      .. controls (-0.55,3.05) and (-1.45,2.45) ..
      (pi);

    \draw[distinguished path]
      (qi)
      .. controls (-0.75,0.88) and (-0.80,0.65) ..
      (-1.02,0.50)
      .. controls (-1.24,0.35) and (-1.05,0.20) ..
      (-1.05,0.06)
      --
      (tb);
  \end{scope}

\end{essgraph}

\caption{The original set of vertex-disjoint paths in $G_j$, where one path (highlighted in orange) may traverse the edge $e_i$.}
\label{fig:path-before}
\end{subfigure}
\hfill
\raisebox{2.15cm}{\Large$\Longrightarrow$}
\hfill
\begin{subfigure}[b]{0.44\textwidth}
\centering

\begin{essgraph}[
  scale=0.82,
  >=stealth,
  node font={\small}
]

  \CommonPanelNodes
  \CommonGraphEdges
  \essedge[extra={->,draw=orange!85!black,line width=1.8pt}]{pi}{ti}
  \essedge[extra={->}]{pi}{qi}
  \CommonDisjointPaths

  \begin{scope}[on background layer]
    \draw[distinguished segment]
      (v)
      .. controls (-0.55,3.05) and (-1.45,2.45) ..
      (pi);
  \end{scope}

\end{essgraph}

\caption{The path is rerouted from $p_i$ directly to $t_i$ via the matching edge $(p_i, t_i)$, avoiding $e_i$.}
\label{fig:path-after}
\end{subfigure}

\endgroup
\end{minipage}
}%
\endgroup
\caption{
Illustration for \Cref{claim:all-conns-equal}: Given a set of $\kappa - 1$ vertex-disjoint paths in $G_j$ avoiding $t_i$ as a terminal (since $t_i$ is not essential in $G_j$), we can reroute any path using $e_i$ to avoid it. Consequently, removing $e_i$ preserves the $\kappa - 1$ disjoint paths in $H$, establishing that $\connvg{v}{H} \geq \kappa - 1$.
}
\label{fig:terminal-path-rerouting}
\end{figure}

Thus far, we have established that $\connvg{v}{G_i} = \connvg{v}{G_j} = \kappa - 1$. We now show that the connectivity remains the same in $H$. Consider the graph $G_j$. We know that $\connvg{v}{G_j} = \kappa - 1$, and that $t_i$ is not essential for $v$ in $G_j$. Therefore, there exists a set $\mathcal{P}$ of $\kappa - 1$ vertex-disjoint paths starting from $v$ and ending at distinct terminals in $T$, completely avoiding the terminal $t_i$. 

Notice that $H = G_j \setminus e_i$. If no path in $\mathcal{P}$ contains $e_i$, the claim holds immediately, as the same path family $\mathcal{P}$ remains valid in $H$, yielding $\kappa - 1$ paths from $v$ to $T$ in $H$.

The remaining case is when exactly one path in $\mathcal{P}$, say $\mathcal{P}_r$, contains $e_i$ (it cannot be more than one due to vertex-disjointness). Since $\mathcal{P}$ avoids $t_i$, we can modify $\mathcal{P}_r$. Instead of traversing $e_i = (p_i, q_i)$ and continuing along the remainder of the path, we can directly redirect the path from $p_i$ to $t_i$ via the matching edge $(p_i, t_i)$ (\Cref{fig:terminal-path-rerouting}). Vertex-disjointness is preserved because no new vertices are introduced other than the previously avoided terminal $t_i$. Consequently, we have created a new path family $\mathcal{P}'$ in $G_j$ that avoids $e_i$, and the argument from the previous case applies.

Therefore, in both cases we maintain a path family of $\kappa - 1$ paths in $H$, which implies $\connvg{v}{H} \geq \kappa - 1$. Since $H$ is a subgraph of $G_j$ and we have already established that $\connvg{v}{G_j} = \kappa - 1$, the connectivity cannot exceed $\kappa - 1$. We thus conclude that $\connvg{v}{H} = \kappa - 1$, completing the proof of the claim.
\end{proof}
\begin{table}[htbp]
  \centering
  \begin{tabular}{l c c c c}
      \hline
      & $G$ & $G_i$ & $G_j$ & $H$ \\
      \hline
      Removed edges from $G$ & -- & $e_i$ & $e_j$ & $\{e_i, e_j\}$ \\ 
      $\connvg{v}{\cdot}$ & $\kappa$ & $\kappa - 1$ & $\kappa - 1$ & $\kappa - 1$ \\
      Essentiality of $t_i$ & Yes & Yes & No & -- \\
      Essentiality of $t$ & Yes & No & Yes & -- \\
      \hline
  \end{tabular}
  \caption{Summary of connectivity and terminal essentiality across modified graphs.}
  \label{tab:essentiality-overview}
\end{table}

The fact that the connectivity of $v$ is identical across $G_i$, $G_j$, and $H$ allows us to derive a contradiction by analyzing the structure of the union and intersection of the tightest minimum cuts, thereby completing the proof of \Cref{lem:critical-transfer}. Let $C_i$ be the tightest minimum cut separating $v$ from $T$ in $G_i$, and let $C_j$ be the tightest minimum cut separating $v$ from $T$ in $G_j$. Notice that $H$ is a subgraph of both $G_i$ and $G_j$. Due to \Cref{claim:all-conns-equal}, both $C_i$ and $C_j$ are also minimum cuts (though not necessarily the tightest ones) in $H$. 

By definition of cuts terminals must reside in either $\lcut$ or $\scut$. Furthermore, by \Cref{lem:essential-terminal-tightest-cut}, a terminal is in the $\scut$ of a tightest minimum cut $\cut$ separating vertex $v$ from $T$ if and only if it is essential for $v$. Therefore, based on the essentiality of $t$ and $t_i$ in $G_i$ and $G_j$ (summarized in \Cref{tab:essentiality-overview}), we immediately obtain:
\begin{align*}
        t&\in\lcutc{C_i}, & t_i&\in\scutc{C_i}, \\
        t&\in\scutc{C_j}, & t_i&\in\lcutc{C_j}.
\end{align*}

We next consider $C_i$ and $C_j$ in $G$. Since both cuts have size $\kappa-1$, whereas $\connvg{v}{G}=\kappa$, neither is a valid cut in $G$. In particular, $C_i$ is valid in $G_i$ but becomes invalid when $e_i$ is added back. Hence $e_i$ crosses $C_i$ from right to left, and thus $p_i\in\rcutc{C_i}$ and $q_i\in\lcutc{C_i}$. Symmetrically, $e_j$ crosses $C_j$ from right to left, so $p_j\in\rcutc{C_j}$ and $q_j\in\lcutc{C_j}$. Therefore,
\begin{align*}
        p_i&\in\rcutc{C_i}, & q_i&\in\lcutc{C_i}, \\
        p_j&\in\rcutc{C_j}, & q_j&\in\lcutc{C_j}.
\end{align*}

Now consider the cuts $C_i$ and $C_j$ on the graph $H$. By \Cref{claim:all-conns-equal} both are minimum cuts in $H$. By \Cref{lem:union-intersection-cut}, $C_i \cap C_j$ is a minimum cut on $H$ where both $t$ and $t_i$ are on the separator $\scutc{C_i \cap C_j}$. In contrast, both terminals belong to the left side $\lcutc{C_i \cup C_j}$ of the union cut $C_i \cup C_j$. Both cuts have the same size, $\kappa - 1$. \Cref{tab:cut-locations-terminals} summarizes these locations.
\begin{table}[h!]
    \centering
    \begin{tabular}{l c c c c}
        \hline
        & $C_i$ & $C_j$ & $C_i \cap C_j$ & $C_i \cup C_j$ \\
        \hline
        Location of $t$   & $L$ & $S$ & $S$ & $L$ \\
        Location of $t_i$ & $S$ & $L$ & $S$ & $L$ \\
        \hline
    \end{tabular}
    \caption{Locations of terminals in cuts $C_i$, $C_j$, $C_i \cap C_j$, and $C_i \cup C_j$.}
    \label{tab:cut-locations-terminals}
\end{table}

Consider the cut $C_i \cap C_j$. Both $t$ and $t_i$ belong to $\scutc{C_i \cap C_j}$. We claim that this cut is not valid in either $G_i$ or $G_j$. Indeed, if it were valid in $G_i$, then $t\in\scutc{C_i \cap C_j}$ would imply, by \Cref{obs:tightest-min-cut} and \Cref{lem:essential-terminal-tightest-cut}, that $t$ is essential for $v$ in $G_i$, contradicting the choice of $G_i$. Symmetrically, validity in $G_j$ would imply that $t_i$ is essential for $v$ in $G_j$, again a contradiction. Since $e_j$ is the only edge of $G_i$ absent from $H$, the edge $e_j$ must cross $C_i \cap C_j$ from $\rcutc{C_i \cap C_j}$ to $\lcutc{C_i \cap C_j}$, so $p_j\in\rcutc{C_i \cap C_j}$ and $q_j\in\lcutc{C_i \cap C_j}$. Similarly, $e_i$ is the only edge of $G_j$ absent from $H$, so $e_i$ crosses $C_i \cap C_j$ from $\rcutc{C_i \cap C_j}$ to $\lcutc{C_i \cap C_j}$, giving $p_i\in\rcutc{C_i \cap C_j}$ and $q_i\in\lcutc{C_i \cap C_j}$.
Consequently,
\begin{align*}
    q_i, q_j \in \lcutc{C_i \cap C_j} = \lcutc{C_i} \cap \lcutc{C_j}.
\end{align*}
In particular, $q_i\in\lcutc{C_j}$ and $q_j\in\lcutc{C_i}$. \Cref{tab:cut-locations-endpoints} summarizes these locations of endpoints.

\begin{table}[h!]
    \centering
    \begin{tabular}{l c c c}
        \hline
        & $C_i$ & $C_j$ & $C_i \cap C_j$ \\
        \hline
        Location of $p_i$ & $R$ & $?$ & $R$ \\
        Location of $q_i$ & $L$ & $L$ & $L$ \\
        Location of $p_j$ & $?$ & $R$ & $R$ \\
        Location of $q_j$ & $L$ & $L$ & $L$ \\
        \hline
    \end{tabular}
    \caption{Established locations of the endpoints in $C_i$, $C_j$, and $C_i\cap C_j$. A question mark denotes a location that is not yet determined.}
    \label{tab:cut-locations-endpoints}
\end{table}

Now consider $C_i$ in $G_i$. The edge $e_j=(p_j,q_j)$ is present in $G_i$, and $q_j\in\lcutc{C_i}$. Since $C_i$ is a valid cut in $G_i$, the edge $e_j$ cannot cross $C_i$ from right to left. Therefore, $p_j\in\lcutc{C_i}\cup\scutc{C_i}$. Symmetrically, considering $C_j$ in $G_j$ gives $p_i\in\lcutc{C_j}\cup\scutc{C_j}$. Together with $p_i\in\rcutc{C_i}$ and $p_j\in\rcutc{C_j}$, \Cref{lem:union-intersection-cut} yields $p_i,p_j\in\lcutc{C_i\cup C_j}\cup\scutc{C_i\cup C_j}$. We summarize all locations of the endpoints in \Cref{tab:endpoint-locations}.
\begin{table}[h!]
    \centering
    \begin{tabular}{l c c c c}
        \hline
        & $C_i$ & $C_j$ & $C_i \cap C_j$ & $C_i \cup C_j$ \\
        \hline
        Location of $p_i$ & $R$ & $L\cup S$ & $R$ & $L\cup S$ \\
        Location of $q_i$ & $L$ & $L$ & $L$ & $L$ \\
        Location of $p_j$ & $L\cup S$ & $R$ & $R$ & $L\cup S$ \\
        Location of $q_j$ & $L$ & $L$ & $L$ & $L$ \\
        \hline
    \end{tabular}
    \caption{Locations of the endpoints in the four cuts.}
    \label{tab:endpoint-locations}
\end{table}

Finally, consider the cut $C_i\cup C_j$ in $H$. It is a valid cut of size $\kappa-1$, but it cannot be valid in $G$, since $\connvg{v}{G}=\kappa$. Thus, when $e_i$ and $e_j$ are added back, at least one of them must cross $C_i\cup C_j$ from right to left. However, we have established that $p_i,p_j\in\lcutc{C_i\cup C_j}\cup\scutc{C_i\cup C_j}$ and $q_i,q_j\in\lcutc{C_i\cup C_j}$. Hence neither $e_i$ nor $e_j$ can cross the cut from right to left, a contradiction. This contradiction proves the lemma.
\end{proof}

Applied with $t=\assgn(v)$, \Cref{lem:critical-transfer} together with \Cref{lem:own-edge-not-critical} shows that if $e_i$ is critical for assigning $v$ to $\assgn(v)$, then $\crt_v(t_i)<\crt_v(\assgn(v))$.

\begin{proof}[Proof of \Cref{lem:cycle-shift-decreases-potential}]
For every $i\in[k]$, the edge $e_i$ is critical for assigning $v_i$ to $\assgn(v_i)$, while \Cref{lem:own-edge-not-critical} shows that it is not critical for assigning $v_i$ to $t_i$. Hence $\assgn(v_i)\neq t_i$, so $\caR$ has no self-loops. By construction, every terminal has in-degree one in $\caR$, so $\caR$ contains a directed cycle of length at least two. Let $\caC$ be such a cycle and let $r$ be its length. Without loss of generality, assume $\caC=(t_1,t_2,\dots,t_r,t_1)$. For every $i\in[r]$, the edge $t_{i-1}\to t_i$ is labeled by $v_i$, where indices are taken modulo $r$. Hence, $\assgn(v_i)=t_{i-1}$, and the shifted assignment assigns $v_i$ to $t_i$, i.e., $\assgn'(v_i) = t_i$. This also shows that the vertices $v_1,\dots,v_r$ are distinct, since each one is assigned to a different terminal in $\assgn$.

Since $e_i$ is critical for assigning $v_i$ to $\assgn(v_i)$, \Cref{lem:critical-implies-assignable} implies that $t_i$ is essential for $v_i$. Therefore, every reassigned vertex remains assigned to an essential terminal. Moreover, each terminal on $\caC$ loses one assigned vertex and gains one assigned vertex. Thus, the capacity constraints remain satisfied, and $(G,T,c,\assgn')$ satisfies the \essentialassignmentcondition.

For $i\in[r]$, among the edges $e_1,\dots,e_k$, let $E_i$ be those that are critical for assigning $v_i$ to $t_i$, and let $E'_i$ be those that are critical for assigning $v_i$ to $\assgn(v_i)=t_{i-1}$. \Cref{lem:critical-transfer} implies that $E_i \subseteq E'_i$. In particular, $e_i$ is not critical for assigning $v_i$ to $t_i$ by \Cref{lem:own-edge-not-critical}, but it is critical for assigning $v_i$ to $t_{i-1}$. In other words, $e_i \notin E_i$ but $e_i \in E'_i$. Hence, $E_i \subset E'_i$. This gives $\crt_{v_i}(t_i)=|E_i|<|E'_i|=\crt_{v_i}(\assgn(v_i))$.

The assignment of every vertex outside $v_1, \dots, v_r$ remains unchanged. Consequently, $\Phi(\assgn')<\Phi(\assgn)$, which completes the proof.
\end{proof}

\section{Generalized Polynomial-Time Algorithm for Weighted Graphs}\label{sec:weighted}

In this section, we extend our polynomial-time algorithm to weighted instances. The weighted algorithm (\Cref{alg:gl-weighted}) also applies to the unweighted instances handled by \Cref{alg:gl-partition}, where it improves the running time because instead of using cycle shifts to optimize a potential function, it finds a minimum-potential flow-essential split-assignment in one minimum-cost-flow computation.

Unlike the unweighted setting, the weighted setting allows an additive error of $w_{\max}-1$ in the capacity constraints. To see why this error is needed, consider a $k$-$T$-connected instance with $k$ terminals of capacity one and one non-terminal of weight $k$. Any partition places this vertex in one part, whose weight exceeds its capacity by $k-1=w_{\max}-1$. The weighted \gyorilovasz theorem guarantees a valid partition with this error~\citep{chen2007almost}. We give a polynomial-time algorithm that finds such a partition. The following theorem, restated from the introduction, states the result.

\thmweightedktconn*

To prove \Cref{thm:weighted-k-t-conn}, we relax $k$-$T$-connectivity to the \fractionalesscond. Under this condition, each vertex may split its weight among its essential terminals, subject to the terminal capacities. We prove the following stronger theorem.

\thmweighted*

Having the above theorem, we can directly show \Cref{thm:weighted-k-t-conn}.

\begin{proof}[Proof of \Cref{thm:weighted-k-t-conn}]
Suppose that $G$ is $k$-$T$-connected. For every non-terminal vertex $v$ and every terminal $t$, there are $k-1$ vertex-disjoint paths from $v$ to the terminals in $T\setminus\{t\}$. Hence deleting $t$ decreases $\connvg{v}{G}$ from $k$ to $k-1$, so every terminal is essential for every non-terminal vertex.

Since the capacities and weights are integral and $\sum_{v\in V\setminus T}\weight_v\le\sum_{t\in T}\cp_t$, we can distribute the weight of every vertex among the terminals without exceeding their capacities. As every terminal is essential for every vertex, this distribution witnesses the \fractionalesscond. The result now follows from \Cref{gyori-weighted-poly-theorem}.
\end{proof}

In the rest of this section, we prove \Cref{gyori-weighted-poly-theorem}. The weighted algorithm (\Cref{alg:gl-weighted}) follows the unweighted algorithm (\Cref{alg:gl-partition}) closely. Because the \fractionalesscond and the \essentialassignmentcondition use the same notion of essentiality, we can reuse several lemmas from \Cref{sec:algorithm}.

\subsection{\textsc{GLWeightedPartition} Algorithm}

\begin{algorithm}[!b]
        \caption{\textsc{GLWeightedPartition}$(G,T,\cp,\weight)$}
        \label{alg:gl-weighted}

        \KwIn{Weighted instance $(G=(V,E),T=\{t_1,\dots,t_k\},\cp,\weight)$ satisfying the \fractionalesscond}
        \KwOut{Partition $\langle V_t\rangle_{t\in T}$}

        \If{$G=\emptyset$}{ \label{alg:gl-weighted:line:empty-start}
                \Return $\emptyset$ \cmt{No vertices remain}
        }
        \label{alg:gl-weighted:line:empty-end}

        \If{$\exists t\in T$ such that $\cp_t = 0$}{ \label{alg:gl-weighted:line:remove-saturated-start}
                $G'\gets G\setminus\{t\}$; $T'\gets T\setminus\{t\}$ \cmt{Remove saturated terminal}
                $\cp'\gets(\cp_{t'})_{t'\in T'}$ \cmt{Restrict capacities}
                $\langle V_{t'}\rangle_{t'\in T'}\gets \textsc{GLWeightedPartition}(G',T',\cp',\weight)$ \cmt{Recursive call}
                $V_t\gets\{t\}$ \cmt{Add removed terminal}
                \Return $\langle V_{t'}\rangle_{t'\in T'}\cup\{V_t\}$
        }
        \label{alg:gl-weighted:line:remove-saturated-end}

        \If{$\exists p\in\PT(G,T)$ and $t\in T$ such that $d^+(p)=1$ and $(p,t)\in E$}{ \label{alg:gl-weighted:line:degree-1-start}
                Contract $(p,t)$ in $G$ to obtain $G'$ \cmt{Remove the pre-terminal}
                Let $\cp'$ be obtained from $\cp$ by setting $\cp'_t=\cp_t-\weight_p$ \cmt{Update residual capacity}
                $\langle V_{t'}\rangle_{t'\in T}\gets \textsc{GLWeightedPartition}(G',T,\cp',\weight)$ \cmt{Recursive call}
                $V_t\gets V_t\cup\{p\}$ \cmt{Add removed vertex}
                \Return $\langle V_{t'}\rangle_{t'\in T}$
        }
        \label{alg:gl-weighted:line:degree-1-end}
        \If{there is a matching $M=\{(p_i,t_i):i\in[k]\}$ from $\PT(G,T)$ to $T$ saturating $T$}{ \label{alg:gl-weighted:line:matching-start}
                For each $i\in[k]$, choose an edge $e_i=(p_i,q_i)$ distinct from $(p_i,t_i)$ \cmt{Possible since no pre-terminal has out-degree one}
                For every $v\in V\setminus T$ and $t\in T$, define $\crt_v(t)$ as the number of edges among $e_1,\dots,e_k$ that are critical for assigning $v$ to $t$ \;
                $\assgnw\gets \textsc{MinCostSplitAssignment}(G,T,\cp,\weight,\crt)$ \cmt{Minimum-potential split-assignment witness}
                Choose $e_{\mathrm{nc}}\in\{e_1,\dots,e_k\}$ that is not critical for any pair $(v,t)$ with $\assgnw(v,t)>0$

                $G'\gets G\setminus e_{\mathrm{nc}}$ \cmt{Delete one safe edge}
                \Return $\textsc{GLWeightedPartition}(G',T,\cp,\weight)$
        }
        \label{alg:gl-weighted:line:matching-end}
        \Else{ \label{alg:gl-weighted-line-roundcall}
                $(G',T',\cp',\langle V_t\rangle_{t\in T\setminus T'})\gets \textsc{RoundAndRemove}(G,T,\cp,\weight)$ \cmt{Round deficient terminals}
                $\langle V_t\rangle_{t\in T'}\gets \textsc{GLWeightedPartition}(G',T',\cp',\weight)$ \cmt{Recursive call}
                \Return $\langle V_t\rangle_{t\in T}$
        }
        \label{alg:gl-weighted:line:round-end}
\end{algorithm}

The weighted algorithm in \Cref{alg:gl-weighted} follows the same structure as the unweighted algorithm. It adds one rounding operation to the three operations of the unweighted algorithm. Operation (i) removes a terminal with zero capacity (Lines~\ref{alg:gl-weighted:line:remove-saturated-start}--\ref{alg:gl-weighted:line:remove-saturated-end}). Operation (ii) contracts a pre-terminal of out-degree one (Lines~\ref{alg:gl-weighted:line:degree-1-start}--\ref{alg:gl-weighted:line:degree-1-end}). Operation (iii) deletes a non-critical secondary edge after finding a matching from pre-terminals to terminals (Lines~\ref{alg:gl-weighted:line:matching-start}--\ref{alg:gl-weighted:line:matching-end}). Operation (iv) rounds and removes a Hall-deficient set of terminals together with its pre-terminal neighborhood when no such matching exists (Lines~\ref{alg:gl-weighted-line-roundcall}--\ref{alg:gl-weighted:line:round-end}).

Suppose that operations (i) and (ii) do not apply and that there is a matching $M=\{(p_i,t_i):i\in[k]\}$ from pre-terminals to terminals. For each $i\in[k]$, we choose an edge $e_i=(p_i,q_i)$ distinct from $(p_i,t_i)$. These edges define the criticality cost $\crt_v(t)$ of \Cref{def:critical}, and hence the split-assignment potential $\Phifrac$ of \Cref{def:fractional-assignment-potential}. As in the unweighted algorithm, the proof uses a cycle shift in the reassignment graph $\caR$ to decrease $\Phifrac$ by at least one. The weighted potential, however, is not polynomially bounded in the input size because the weights may be exponentially large. We instead find a minimum-potential flow-essential split-assignment in one minimum-cost-flow computation. The cycle-shift argument then shows that this split-assignment yields a removable edge among $e_1,\dots,e_k$.

The remaining challenge is the case in which no such matching $M$ exists. The unweighted analysis shows that this case cannot occur there. In the weighted setting, we use \textsc{RoundAndRemove} (\Cref{alg:gl-round-and-remove}). The subroutine chooses an inclusion-minimal Hall-deficient set $S\subseteq T$ and an arbitrary terminal $t_S\in S$. As in the unweighted argument, let $\PT(G,S)$ be the set of neighboring pre-terminals. By \Cref{lem:minimal-hall-set}, there is a matching between $S\setminus\{t_S\}$ and $\PT(G,S)$ that saturates both sides, as in \Cref{fig:round-and-remove-before}. For each $t\in S\setminus\{t_S\}$, we form $V_t$ from $t$ and its matched pre-terminal, and we set $V_{t_S}=\{t_S\}$. We then remove $S\cup \PT(G,S)$ and recurse on the residual instance shown in \Cref{fig:round-and-remove-after}. By \Cref{lem:terminal-leaf-hall-set-not-essential-outside}, no remaining vertex sends weight to a terminal in $S$. To justify the recursive call, it remains to show that every terminal receiving positive weight from a remaining vertex stays essential after the removal. This is part of \Cref{lem:weighted-round-and-remove}.

Operation (iv) is the only source of the additive $w_{\max}-1$ error. Each terminal in $S$ has positive capacity and receives at most one non-terminal of weight at most $w_{\max}$. Thus, the weight of its part exceeds its capacity by at most $w_{\max}-1$.

\begin{figure}[htbp]
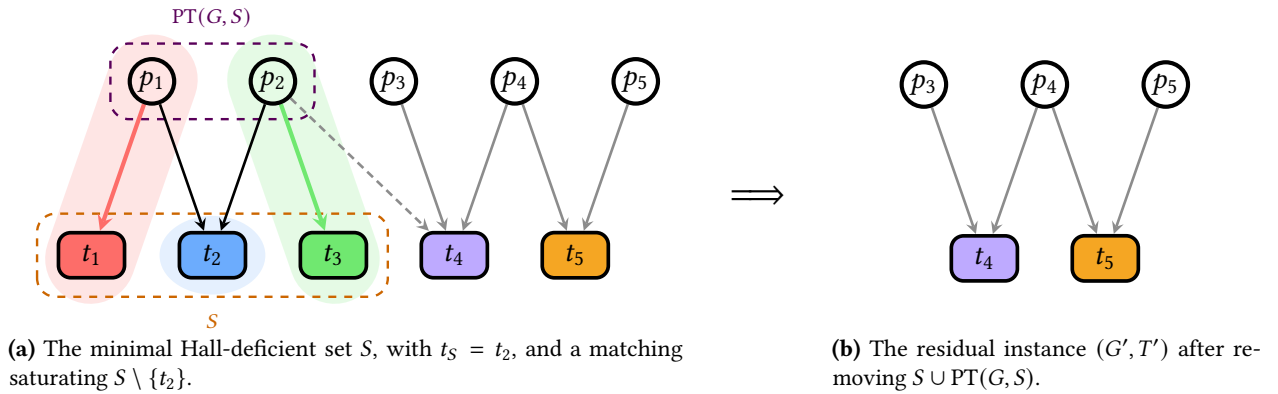

  \centering

  \begin{subfigure}[b]{0.54\linewidth}
    \centering
    \begin{essgraph}[
      scale=1.00,
      >=stealth,
      node scale=1.05,
      term scale=0.92,
      line cap=round,
      line join=round
    ]
      \essterminal{t1}{(0.00,0.25)}{$t_1$}
      \essterminal{ts}{(1.60,0.25)}{$t_2$}
      \essterminal{t3}{(3.20,0.25)}{$t_3$}
      \essterminal{t4}{(4.80,0.25)}{$t_4$}
      \essterminal{t5}{(6.40,0.25)}{$t_5$}

      \essvertex[extra={fill=essg@c1!18}]{p1}{(0.80,2.55)}{$p_1$}
      \essvertex[extra={fill=essg@c3!18}]{p2}{(2.40,2.55)}{$p_2$}
      \essvertex{p3}{(4.00,2.55)}{$p_3$}
      \essvertex{p4}{(5.60,2.55)}{$p_4$}
      \essvertex{p5}{(7.20,2.55)}{$p_5$}

      \begin{scope}[on background layer]
        \draw[
          draw=essg@c1,
          draw opacity=0.18,
          preaction={draw=white,draw opacity=1,line width=44pt},
          line width=34pt,
          line cap=round
        ] ($(p1)!-0.01!(t1)$) -- ($(t1)!-0.02!(p1)$);
        \fill[fill=essg@c2,fill opacity=0.18]
          (ts) ellipse (0.70 and 0.52);
        \draw[
          draw=essg@c3,
          draw opacity=0.18,
          preaction={draw=white,draw opacity=1,line width=44pt},
          line width=34pt,
          line cap=round
        ] ($(p2)!-0.01!(t3)$) -- ($(t3)!-0.02!(p2)$);

        \draw[
          draw=violet!70!black,
          fill=none,
          rounded corners=7pt,
          dashed,
          line width=0.9pt
        ] (0.25,2.05) rectangle (2.95,3.05);
        \draw[
          draw=orange!80!black,
          fill=none,
          rounded corners=7pt,
          dashed,
          line width=0.9pt
        ] (-0.72,-0.30) rectangle (3.92,0.80);
      \end{scope}

      \essedge[extra={->,draw=essg@c1,ultra thick}]{p1}{t1}
      \essedge[extra={->}]{p1}{ts}
      \essedge[extra={->}]{p2}{ts}
      \essedge[extra={->,draw=essg@c3,ultra thick}]{p2}{t3}

      \essedge[non-crit,extra={->,densely dashed}]{p2}{t4}
      \essedge[non-crit,extra={->}]{p3}{t4}
      \essedge[non-crit,extra={->}]{p4}{t4}
      \essedge[non-crit,extra={->}]{p4}{t5}
      \essedge[non-crit,extra={->}]{p5}{t5}

      \node[
        font=\scriptsize\bfseries,
        text=violet!70!black,
        anchor=south,
        inner sep=1pt
      ] at (1.60,3.23) {$\PT(G,S)$};
      \node[
        font=\scriptsize\bfseries,
        text=orange!80!black,
        anchor=north,
        inner sep=1pt
      ] at (1.60,-0.48) {$S$};
    \end{essgraph}
    \caption{The minimal Hall-deficient set $S$, with $t_S=t_2$, and a
      matching saturating $S\setminus\{t_2\}$.}
    \label{fig:round-and-remove-before}
  \end{subfigure}
  \hfill
  \raisebox{2.50cm}{\Large$\Longrightarrow$}
  \hfill
  \begin{subfigure}[b]{0.34\linewidth}
    \centering
    \begin{essgraph}[
      scale=1.00,
      >=stealth,
      node scale=1.05,
      term scale=0.92,
      line cap=round,
      line join=round
    ]
      \essterminal[extra={fill=essg@c4}]{t4r}{(2.40,0.25)}{$t_4$}
      \essterminal[extra={fill=essg@c5}]{t5r}{(4.00,0.25)}{$t_5$}

      \essvertex{p3r}{(1.60,2.55)}{$p_3$}
      \essvertex{p4r}{(3.20,2.55)}{$p_4$}
      \essvertex{p5r}{(4.80,2.55)}{$p_5$}

      \essedge[non-crit,extra={->}]{p3r}{t4r}
      \essedge[non-crit,extra={->}]{p4r}{t4r}
      \essedge[non-crit,extra={->}]{p4r}{t5r}
      \essedge[non-crit,extra={->}]{p5r}{t5r}
      \draw[white,line width=0.1pt] (3.20,0.00) -- (3.20,-0.70);
    \end{essgraph}
    \caption{The residual instance $(G',T')$ after removing
      $S\cup \PT(G,S)$.}
    \label{fig:round-and-remove-after}
  \end{subfigure}

  \caption{One step of
    \textsc{RoundAndRemove}. Here $S=\{t_1,t_2,t_3\}$ and $t_S=t_2$.
    The highlighted groups are the completed parts
    $V_{t_1}=\{t_1,p_1\}$, $V_{t_2}=\{t_2\}$, and
    $V_{t_3}=\{t_3,p_2\}$; removing their vertices leaves the residual
    instance $(G',T')$.}
  \label{fig:round-and-remove}
\end{figure}

\begin{algorithm}[!ht]
        \caption{\textsc{RoundAndRemove}$(G,T,\cp,\weight)$}
        \label{alg:gl-round-and-remove}
        \KwIn{Weighted instance $(G=(V,E),T=\{t_1,\dots,t_k\},\cp,\weight)$ satisfying the \fractionalesscond
        }
        \KwRequire{No matching from $\PT(G,T)$ to $T$ saturating $T$ exists, $\cp_t>0$ for every $t\in T$, and every pre-terminal has out-degree at least two.}
        \KwOut{A smaller instance and completed parts for the removed terminals}

        Find an inclusion-minimal nonempty set $S\subseteq T$ such that there is no matching from $S$ to $\PT(G,T)$ saturating $S$ \cmt{Hall-deficient set}
        Pick an arbitrary terminal $t_S\in S$

        Find a matching $M$ from $S\setminus\{t_S\}$ to $\PT(G,S)$ \cmt{Exists by \Cref{lem:minimal-hall-set}}

        $V_{t_S}\gets\{t_S\}$ \cmt{$t_S$ receives no pre-terminal}
        \ForEach{$t\in S\setminus\{t_S\}$}{
                $V_t\gets\{t,M(t)\}$ \cmt{Assign one matched pre-terminal}
        }

        $G'\gets G\setminus (S\cup \PT(G,S))$; $T'\gets T\setminus S$ \cmt{Remove rounded vertices}
        $\cp'\gets(\cp_t)_{t\in T'}$ \cmt{Remaining capacities}
        \Return $(G',T',\cp',\langle V_t\rangle_{t\in S})$
\end{algorithm}

\subsection{Correctness of the Weighted Algorithm}

We first show that the first two operations of \Cref{alg:gl-weighted} preserve the \fractionalesscond. \Cref{lem:frac-remove-saturated-terminal} shows that removing a terminal with zero capacity preserves the \fractionalesscond (used in Lines~\ref{alg:gl-weighted:line:remove-saturated-start}--\ref{alg:gl-weighted:line:remove-saturated-end}) and \Cref{lem:weighted-contract-degree-one-potential-leaf} shows that a pre-terminal of out-degree one can be contracted without violating the \fractionalesscond (used in Lines~\ref{alg:gl-weighted:line:degree-1-start}--\ref{alg:gl-weighted:line:degree-1-end}). These are the weighted analogues of the corresponding lemmas in the unweighted section.

\begin{lemma}[Removing a terminal with zero capacity] \label{lem:frac-remove-saturated-terminal}
Let $(G,T,\cp,\weight)$ satisfy the \fractionalesscond, and let $t_i\in T$ satisfy $\cp_{t_i}=0$. Let $G'=G\setminus\{t_i\}$, let $T'=T\setminus\{t_i\}$, and let $\cp'=(\cp_t)_{t\in T'}$. Then $(G',T',\cp',\weight)$ satisfies the \fractionalesscond. Moreover, any valid weighted partition of this instance extends to $(G,T,\cp,\weight)$ by setting $V_{t_i}=\{t_i\}$.
\end{lemma}

\begin{proof}
Let $\assgnw$ be a flow-essential split-assignment witnessing the \fractionalesscond in $G$. Since $\cp_{t_i}=0$, the capacity constraints give $\assgnw(v,t_i)=0$ for every non-terminal $v$. Thus, restricting $\assgnw$ to the pairs in $(V(G')\setminus T')\times T'$ preserves the amount assigned by every vertex and the total weight assigned to every terminal in $T'$.

If $\assgnw(v,t)>0$ for $t\in T'$, then $t$ is essential for $v$ in $G$. By \Cref{lem:essential-survives-terminal-removal}, $t$ remains essential for $v$ in $G'$. Hence the restricted split-assignment witnesses the \fractionalesscond in $G'$.

Finally, $V_{t_i}=\{t_i\}$ is connected and has non-terminal weight zero. Therefore, it extends any valid weighted partition of $(G',T',\cp',\weight)$ to one of $(G,T,\cp,\weight)$.
\end{proof}

\begin{lemma}[Contracting a pre-terminal of out-degree one] \label{lem:weighted-contract-degree-one-potential-leaf}
Let $(G,T,\cp,\weight)$ satisfy the \fractionalesscond. Assume that $p\in\PT(G,T)$, $d^+(p)=1$, and $(p,t_i)\in E$. Let $G'$ be obtained by contracting $(p,t_i)$ into $t_i$, and let $\cp'_{t_i}=\cp_{t_i}-\weight_p$, while all other capacities stay unchanged. Then $\cp'_{t_i}\ge0$ and $(G',T,\cp',\weight)$ satisfies the \fractionalesscond. Moreover, any valid weighted partition $\langle V'_t\rangle_{t\in T}$ of the contracted instance extends to the original instance by replacing $V'_{t_i}$ with $V'_{t_i}\cup\{p\}$.
\end{lemma}
\begin{proof}
Let $\assgnw$ be a flow-essential split-assignment witnessing the \fractionalesscond in $G$. Since $p$ has out-degree one, $t_i$ is its only essential terminal. Hence $\assgnw(p,t_i)=\weight_p$, so $\cp_{t_i}\ge\weight_p$ and $\cp'_{t_i}\ge0$.

Let $\assgnw'$ be the restriction of $\assgnw$ to $(V(G')\setminus T)\times T$. Removing $p$ decreases the total weight assigned to $t_i$ by exactly $\weight_p$, so this amount is at most $\cp_{t_i}-\weight_p=\cp'_{t_i}$. The total weight assigned to every other terminal does not increase. If $\assgnw'(v,t)>0$, then $t$ is essential for $v$ in $G$, and \Cref{lem:essential-survives-degree-one-leaf-contraction} implies that $t$ remains essential for $v$ in $G'$. Thus, $\assgnw'$ witnesses the \fractionalesscond in the contracted instance.

Let $\langle V'_t\rangle_{t\in T}$ be a valid weighted partition of the contracted instance. For each vertex in $V'_{t_i}$, take a path to $t_i$ in $G'[V'_{t_i}]$ and replace any edge created by the contraction with its two-edge path through $p$. Together with $(p,t_i)$, these paths show that $G[V'_{t_i}\cup\{p\}]$ is connected to $t_i$. All other parts are unchanged. The non-terminal weight of the new part is at most $\cp'_{t_i}+w_{\max}-1+\weight_p=\cp_{t_i}+w_{\max}-1$, so the extended partition is valid.
\end{proof}

We next prove the correctness of operation (iii) in Lines~\ref{alg:gl-weighted:line:matching-start}--\ref{alg:gl-weighted:line:matching-end} of \Cref{alg:gl-weighted}. Since operations (i) and (ii) do not apply, all terminal capacities are positive and all pre-terminals have out-degree at least two. Given a matching $M=\{(p_i,t_i):i\in[k]\}$, choose a secondary edge $e_i=(p_i,q_i)$ for every $i\in[k]$. We then compute a minimum-potential flow-essential split-assignment. If every $e_i$ were critical for a pair $(v_i,s_i)$ with positive assigned weight, the edges $s_i\to t_i$ would form a reassignment graph containing a directed cycle. Shifting one unit along this cycle would lower the potential, contradicting minimality. Thus, some $e_i$ is removable. This argument uses \Cref{lem:critical-implies-assignable,lem:own-edge-not-critical,lem:critical-transfer}, which do not depend on an unweighted assignment.

We first define the weighted potential, which extends \Cref{def:assignment-potential}.

\begin{definition}[Split-assignment potential] \label{def:fractional-assignment-potential}
Let $\crt$ be a criticality cost as in \Cref{def:critical}. For a \FractionalEssAssign $\assgnw$, define its potential by
\begin{equation*}
        \Phifrac(\assgnw)=\sum_{v\in V\setminus T}\sum_{t\in T}\assgnw(v,t)\crt_v(t).
\end{equation*}
\end{definition}

We are now ready to show that a cycle shift in the reassignment graph $\caR$ decreases the potential $\Phifrac$.

\begin{lemma}[Cycle shift decreases the weighted potential] \label{lem:weighted-cycle-shift-decreases-cost}
Let $(G,T,\cp,\weight)$ satisfy the \fractionalesscond. Write $T=\{t_1,\dots,t_k\}$, let $M=\{(p_i,t_i):i\in[k]\}$ be a matching from $\PT(G,T)$ to $T$ saturating $T$, and for each $i$ let $e_i=(p_i,q_i)$ be an edge distinct from $(p_i,t_i)$. Let $\assgnw$ be a \FractionalEssAssign witnessing the \fractionalesscond. Assume that for every $i\in[k]$ there are a vertex $v_i$ and a terminal $s_i$ such that $\assgnw(v_i,s_i)>0$ and $e_i$ is critical for assigning $v_i$ to $s_i$. Let $\caR$ be the reassignment graph on $T$ containing, for every $i\in[k]$, the edge $s_i\to t_i$ labeled by $v_i$. Then $\caR$ contains a directed cycle, and shifting a single unit of weight along a directed cycle yields another \FractionalEssAssign $\assgnw'$ satisfying the \fractionalesscond such that $\Phifrac(\assgnw')<\Phifrac(\assgnw)$.
\end{lemma}
\begin{proof}
Every terminal has in-degree one in $\caR$, so $\caR$ contains a directed cycle. This cycle has length at least two: if $s_i=t_i$, then $e_i$ would be critical for assigning $v_i$ to $t_i$, contrary to \Cref{lem:own-edge-not-critical}. Let $\caC$ be a simple directed cycle of length $r$. Without loss of generality, assume that $\caC=(t_1,t_2,\dots,t_r,t_1)$, with indices taken modulo $r$. For every $i\in[r]$, the edge $t_{i-1}\to t_i$ is labeled by $v_i$. Thus, $\assgnw(v_i,t_{i-1})>0$, and $e_i$ is critical for assigning $v_i$ to $t_{i-1}$.

For every $i\in[r]$, \Cref{lem:critical-implies-assignable} implies that $t_i$ is essential for $v_i$. By \Cref{lem:critical-transfer}, every edge that is critical for assigning $v_i$ to $t_i$ is also critical for assigning $v_i$ to $t_{i-1}$. Moreover, $e_i$ is critical for assigning $v_i$ to $t_{i-1}$ but not to $t_i$, by \Cref{lem:own-edge-not-critical}. Hence $\crt_{v_i}(t_i)<\crt_{v_i}(t_{i-1})$.

Define $\assgnw'$ by decreasing $\assgnw(v_i,t_{i-1})$ by one and increasing $\assgnw(v_i,t_i)$ by one for every $i\in[r]$. Each $\assgnw(v_i,t_{i-1})$ is positive. Since $\caC$ is simple, the terminals $t_{i-1}$ are distinct, so no value $\assgnw(v_i,t_{i-1})$ is decreased more than once. Hence $\assgnw'$ is nonnegative. The shift preserves the total weight assigned by every vertex. Each terminal on $\caC$ loses and gains one unit, so its total assigned weight is unchanged. Every increased value assigns weight from $v_i$ to $t_i$, which is essential for $v_i$. Thus, $\assgnw'$ satisfies the \fractionalesscond. Finally, $\Phifrac(\assgnw')-\Phifrac(\assgnw)=\sum_{i=1}^r\bigl(\crt_{v_i}(t_i)-\crt_{v_i}(t_{i-1})\bigr)<0$, since $\crt_{v_i}(t_i)<\crt_{v_i}(t_{i-1})$ for every $i\in[r]$.
\end{proof}

We can now prove that a minimizer of the potential function $\Phifrac$ yields a removable edge, because otherwise one can apply the above lemma to decrease the potential, contradicting minimality.

\begin{lemma}[A matching gives a removable edge] \label{lem:weighted-matching-removable-edge}
Assume that $(G,T,\cp,\weight)$ satisfies the \fractionalesscond, every pre-terminal has out-degree at least two, and there is a matching $M=\{(p_i,t_i):i\in[k]\}$ from $\PT(G,T)$ to $T$ saturating $T$. Then Lines~\ref{alg:gl-weighted:line:matching-start}--\ref{alg:gl-weighted:line:matching-end} of \Cref{alg:gl-weighted} find, in polynomial time, an edge $e_{\mathrm{nc}}$ such that $(G\setminus e_{\mathrm{nc}},T,\cp,\weight)$ satisfies the \fractionalesscond.
\end{lemma}
\begin{proof}
For each matched pair $(p_i,t_i)$, choose an edge $e_i=(p_i,q_i)$ distinct from $(p_i,t_i)$; this is possible because every pre-terminal has out-degree at least two. Define $\crt$ as in \Cref{def:critical}, and let $\assgnw=\textsc{MinCostSplitAssignment}(G,T,\cp,\weight,\crt)$ be the resulting minimum-potential flow-essential split-assignment, whose existence follows from the \fractionalesscond and \Cref{prop:best-assign-is-poly}.

Suppose every edge $e_i$ is critical for some pair $(v_i,s_i)$ with $\assgnw(v_i,s_i)>0$. Build the reassignment graph $\caR$ on terminals with edge $s_i\to t_i$ labeled by $v_i$. By \Cref{lem:weighted-cycle-shift-decreases-cost}, there is another flow-essential split-assignment satisfying the \fractionalesscond with smaller $\Phifrac$, contradicting the choice of $\assgnw$.

Thus some edge $e_i$ is not critical for any pair $(v,t)$ with $\assgnw(v,t)>0$. Deleting this edge preserves the essentiality of $t$ for $v$ for every such pair. The same flow-essential split-assignment therefore witnesses the \fractionalesscond in $G\setminus e_i$.
\end{proof}

The first three operations of \Cref{alg:gl-weighted} preserve the \fractionalesscond. It remains to analyze \textsc{RoundAndRemove}, which is used when no matching from $\PT(G,T)$ to $T$ saturates $T$. The subroutine chooses a nonempty set $S\subseteq T$ that no matching saturates and that is inclusion-minimal with this property. Hall's theorem gives a subset of $S$ with fewer neighboring pre-terminals than terminals. Minimality forces this subset to be $S$, since every proper subset of $S$ can be saturated. Thus, $S$ is inclusion-minimal subject to $|S|>|\PT(G,S)|$. The subroutine removes $S\cup \PT(G,S)$.

We first handle the removed vertices. \Cref{alg:gl-round-and-remove} chooses $t_S\in S$ and a matching from $S\setminus\{t_S\}$ to $\PT(G,S)$. Each matched pair forms one part, and $V_{t_S}=\{t_S\}$. The next lemma proves that these parts are valid.

\begin{lemma}[Rounded parts are valid] \label{lem:weighted-rounded-parts-valid}
Assume that all capacities are positive and that there is no matching from $\PT(G,T)$ to $T$ saturating $T$. Then the parts completed by \textsc{RoundAndRemove} are connected and satisfy the bound in \Cref{gyori-weighted-poly-theorem}.
\end{lemma}
\begin{proof}
Let $S$ be the set chosen by \textsc{RoundAndRemove}. By \Cref{lem:minimal-hall-set}, there is a matching from $S\setminus\{t_S\}$ to $\PT(G,S)$ that saturates both sides. Hence every pre-terminal removed with $S$ belongs to exactly one returned part.

The part $V_{t_S}=\{t_S\}$ is connected and has non-terminal weight zero. For $t\in S\setminus\{t_S\}$, the part $V_t=\{t,M(t)\}$ is connected because $M(t)$ is adjacent to $t$. Since all capacities are positive and integral, $\weight_{M(t)}\le w_{\max}\le \cp_t+w_{\max}-1$. Thus all rounded parts satisfy the required weighted bound.
\end{proof}

It remains to show that the \fractionalesscond is preserved on the remaining instance. Let $\assgnw$ be a flow-essential split-assignment witnessing this condition. By \Cref{lem:terminal-leaf-hall-set-not-essential-outside}, $\assgnw(v,t)>0$ with $t\in S$ implies that $v\in \PT(G,S)$. Thus, no remaining non-terminal assigns positive weight to a terminal in $S$, so restricting $\assgnw$ to the remaining instance preserves the total weight assigned by every remaining non-terminal. Moreover, removing vertices cannot increase the total weight assigned to any remaining terminal, and hence all capacity constraints continue to hold. It remains only to show that every terminal $t\notin S$ that is essential for a remaining non-terminal $v$ remains essential after removing $S\cup \PT(G,S)$; this is the content of the following lemma.

\begin{lemma}
\label{lem:essentiality-survives-rounding}
Let $G=(V,E)$ be a directed graph with terminal set $T$, and let
$S\subseteq T$ be inclusion-minimal subject to
$|S|>|\PT(G,S)|$. Suppose that a terminal $t \in T \setminus S$ is
essential for a non-terminal $v \in (V\setminus T)\setminus \PT(G,S)$ in $G$. Then $t$ remains essential for $v$ after removing
$S\cup \PT(G,S)$ from the graph and $S$ from the terminal set.
\end{lemma}

\begin{figure}[htbp]
\centering
\begingroup
\newcommand{\panelcaption}[1]{%
  \par\vspace{-8pt}\begin{minipage}{6.85cm}\centering\small #1\end{minipage}%
}

\tikzset{
  round path avoids/.style={
    ->,
    draw=blue!68!black,
    line width=1.55pt,
    line cap=round,
    shorten <=1.5pt,
    shorten >=2pt
  },
  round path meets/.style={
    ->,
    draw=orange!88!black,
    line width=1.8pt,
    line cap=round,
    shorten <=1.5pt,
    shorten >=2pt
  },
  round graph edge/.style={
    ->,
    draw=gray!55,
    line width=0.72pt,
    shorten <=1.5pt,
    shorten >=2pt
  },
  round omitted/.style={
    circle,
    fill=black,
    draw=black,
    inner sep=1.45pt
  }
}

\resizebox{\textwidth}{!}{%
\begin{minipage}{16.4cm}
\centering

\begin{tikzpicture}[baseline=-0.5ex]
  \draw[round path meets] (0,0) -- (0.72,0)
    node[right,font=\scriptsize,text=black] {meets $\PT(G,S)$};
  \draw[round path avoids] (3.75,0) -- (4.47,0)
    node[right,font=\scriptsize,text=black] {avoids $\PT(G,S)$};
\end{tikzpicture}

\vspace{0.65cm}

\begin{minipage}[t]{7.10cm}
\centering
\begin{essgraph}[
  scale=0.76,
  >=stealth,
  node font={\small}
]
  \essterminal{sone}{(-3.30,-0.45)}{$s_1$}
  \essterminal{stwo}{(-1.95,-0.45)}{$s_2$}
  \essterminal{sthree}{(-0.60,-0.45)}{$s_3$}

  \node[round omitted] (tdot) at (1.50,-0.45) {};
  \essterminal{t}{(2.65,-0.45)}{$t$}
  \essterminal{tprime}{(4.00,-0.45)}{$t'$}

  \essvertex{pone}{(-3.05,2.20)}{$p_1$}
  \essvertex{ptwo}{(-0.35,2.20)}{$p_2$}

  \essvertex{v}{(0,4.35)}{$v$}

  \begin{scope}[on background layer]
    \node[
      fit=(sone)(stwo)(sthree),
      fill=orange!12,
      draw=orange!80!black,
      dashed,
      rounded corners=7pt,
      line width=0.9pt,
      inner xsep=7pt,
      inner ysep=14pt,
      label={
        [font=\small\bfseries,text=orange!80!black,
         anchor=north west,yshift=-2pt]
        south west:$S$
      }
    ] {};

    \node[
      fit=(tdot)(t)(tprime),
      fill=blue!7,
      draw=blue!60!black,
      dashed,
      rounded corners=7pt,
      line width=0.9pt,
      inner xsep=7pt,
      inner ysep=14pt,
      label={
        [font=\small\bfseries,text=blue!60!black,
         anchor=north east,yshift=-2pt]
        south east:$T\setminus S$
      }
    ] {};

    \node[
      fit=(pone)(ptwo),
      fill=violet!9,
      draw=violet!70!black,
      dashed,
      rounded corners=7pt,
      line width=0.9pt,
      inner xsep=8pt,
      inner ysep=13pt,
    ] {};

    \draw[round graph edge,opacity=0.5] (pone) -- (sone);

    \draw[round path meets]
      (v) .. controls (-2.05,4.18) and (-3.00,3.25) .. (pone);
    \draw[round path meets,opacity=0.5] (pone) -- (sone);
    \draw[round path meets]
      (v) .. controls (-0.55,3.75) and (-0.48,2.85) .. (ptwo);
    \draw[round path meets,opacity=0.5]
      (ptwo) .. controls (0.00,1.00) and (0.82,-0.02) .. (tdot);

    \draw[round path avoids]
      (v) .. controls (1.05,3.45) and (2.15,1.35) .. (t);
    \draw[round path avoids]
      (v) .. controls (2.45,3.95) and (3.90,2.10) .. (tprime);
  \end{scope}

  \node[font=\small\bfseries,anchor=south] at (0,4.72)
    {$\mathcal P$ in $(G,T)$};
  \node[
    font=\small\bfseries,
    text=violet!70!black,
    fill=violet!9,
    inner sep=1pt
  ] at (-1.70,2.20) {$\PT(G,S)$};
\end{essgraph}
\phantomsubcaption\label{fig:round-essential-paths-in-g}
\panelcaption{\textbf{(\thesubfigure)} The path family $\mathcal P$ in the original
  instance $(G,T)$.}
\end{minipage}
\hspace{-0.28cm}
\raisebox{2.55cm}{\Large$\Longrightarrow$}
\hspace{-0.28cm}
\begin{minipage}[t]{7.10cm}
\centering
\begin{essgraph}[
  scale=0.76,
  >=stealth,
  node font={\small}
]
  \essterminal[extra={opacity=0}]{soneghost}{(0,4.35)}{}
  \essterminal[extra={opacity=0}]{stwoghost}{(0,4.35)}{}
  \essterminal[extra={opacity=0}]{sthreeghost}{(0,4.35)}{}

  \node[round omitted] (tdot) at (1.30,-0.45) {};
  \essterminal{t}{(2.45,-0.45)}{$t$}
  \essterminal{tprime}{(3.80,-0.45)}{$t'$}

  \essterminal[extra={fill=violet!24}]
    {pone}{(-2.45,-0.45)}{$p_1$}
  \essterminal[extra={fill=violet!24}]
    {ptwo}{(-0.95,-0.45)}{$p_2$}

  \essvertex{v}{(0,4.35)}{$v$}

  \begin{scope}[on background layer]
    \node[
      fit=(tdot)(t)(tprime),
      fill=blue!7,
      draw=blue!60!black,
      dashed,
      rounded corners=7pt,
      line width=0.9pt,
      inner xsep=7pt,
      inner ysep=14pt,
      label={
        [font=\small\bfseries,text=blue!60!black,
         anchor=north east,yshift=-2pt]
        south east:$T\setminus S\subseteq T_H$
      }
    ] {};

    \node[
      fit=(pone)(ptwo),
      fill=violet!12,
      draw=violet!70!black,
      dashed,
      rounded corners=7pt,
      line width=0.9pt,
      inner xsep=8pt,
      inner ysep=14pt,
      label={
        [font=\small\bfseries,text=violet!70!black,yshift=-2pt]
        below:$\PT(G,S)\subseteq T_H$
      }
    ] {};

    \draw[round path meets]
      (v) .. controls (-1.82,3.95) and (-2.42,1.35) .. (pone);
    \draw[round path meets]
      (v) .. controls (-0.52,3.28) and (-0.82,1.10) .. (ptwo);

    \draw[round path avoids]
      (v) .. controls (0.95,3.45) and (1.98,1.35) .. (t);
    \draw[round path avoids]
      (v) .. controls (2.28,3.95) and (3.68,2.10) .. (tprime);
  \end{scope}

  \node[font=\small\bfseries,anchor=south] at (0,4.72)
    {$\mathcal P_H$ in $(H,T_H)$};
\end{essgraph}
\vspace{0.42cm}
\phantomsubcaption\label{fig:round-essential-paths-in-h}
\panelcaption{\textbf{(\thesubfigure)} The path family $\mathcal P_H$ after replacing
  $S$ by $\PT(G,S)$ in the terminal set.}
\end{minipage}

\end{minipage}
}%

\caption{The path-family transformation used in
  \Cref{lem:essentiality-survives-rounding}. In
  \subref{fig:round-essential-paths-in-g}, the orange paths of
  $\mathcal P$ meet $\PT(G,S)$, whereas the blue paths avoid it. In
  \subref{fig:round-essential-paths-in-h},
  the vertices of $S$ have been
  removed and the vertices of $\PT(G,S)$ are terminals with no outgoing
  edges. Truncating each orange path at its first vertex in $\PT(G,S)$
  and leaving every blue path unchanged yields $\mathcal P_H$ in
  $(H,T_H)$, where $T_H=(T\setminus S)\cup \PT(G,S)$.}
\label{fig:round-essential-intermediate}

\endgroup
\end{figure}
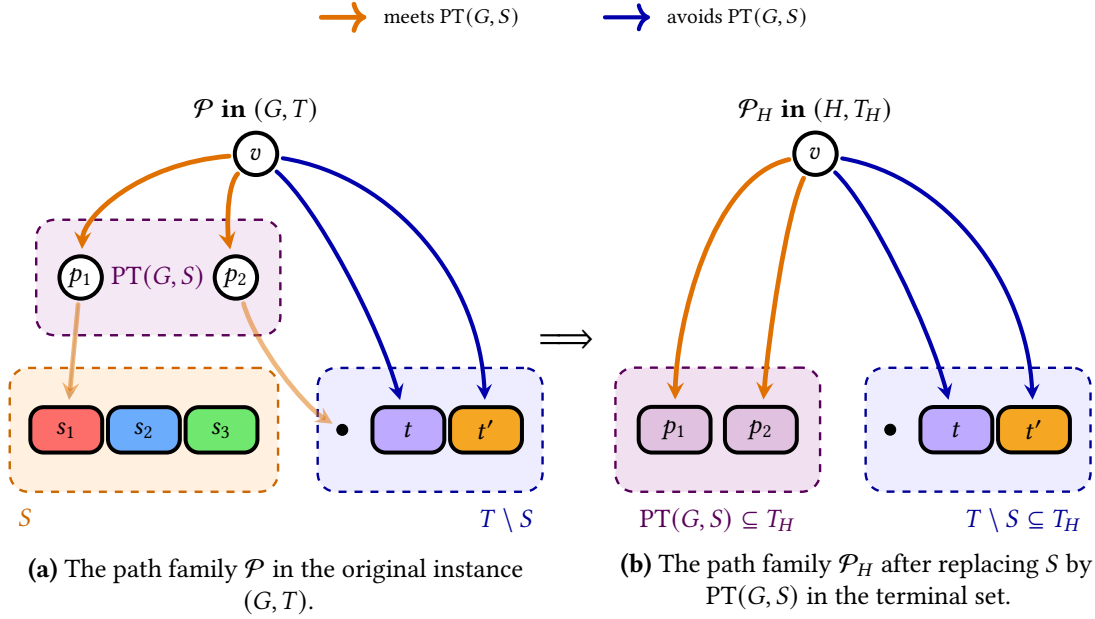

\begin{proof}
We prove the statement in two steps. First, we remove $S$, replace it in the
terminal set by its pre-terminal neighborhood, and show that this change
preserves the essentiality of $t$. We then remove the new temporary terminals
one at a time.

For clarity, throughout this proof a pair $(J,U)$ means that $J$ is the graph
under consideration and $U$ is its designated terminal set. Write
$T':=T\setminus S$ and $G':=G\setminus(S\cup \PT(G,S))$. Let $H$ be obtained
from $G\setminus S$ by deleting every outgoing edge of every vertex in
$\PT(G,S)$, and let $T_H:=(T \setminus S)\cup \PT(G,S)$. Thus $(G,T)$ is the original instance,
shown in \Cref{fig:round-essential-paths-in-g}; $(H,T_H)$ is the intermediate
instance in which the vertices of $\PT(G,S)$ replace the terminals of $S$, shown
in \Cref{fig:round-essential-paths-in-h}; and $(G',T')$ is the instance after
rounding.

We first show that $t$ is essential for $v$ in $(H,T_H)$. Write $\kappa_G$
and $\kappa_H$ for the terminal connectivity of $v$ in $(G,T)$ and
$(H,T_H)$, respectively. We begin by proving that
$\kappa_G=\kappa_H$.

Let $\mathcal P$ be a family of $\kappa_G$ vertex-disjoint paths from $v$ to
distinct terminals in $T$. Construct a family $\mathcal P_H$ by truncating
each path that meets $\PT(G,S)$ at its first vertex in $\PT(G,S)$. Since we only take
prefixes, the paths remain vertex-disjoint; this transformation is depicted in
\Cref{fig:round-essential-paths-in-g,fig:round-essential-paths-in-h}. Every path that originally ended
in $S$ meets $\PT(G,S)$ and, after truncation, ends at its first such vertex;
every path that does not meet $\PT(G,S)$ ends in $T'$. Thus the paths in
$\mathcal P_H$ lie in $H$, and end at
distinct terminals in $T_H$. Therefore $\kappa_H\ge\kappa_G$.

For the reverse inequality, let
$C=(\lcutc{C},\scutc{C},\rcutc{C})$ be the tightest minimum cut separating
$v$ from $T$ in $G$. By
\Cref{lem:terminal-leaf-hall-set-not-essential-outside}, no terminal in $S$
is essential for $v$. Hence
\Cref{lem:essential-terminal-tightest-cut} implies that
$S\cap\scutc{C}=\emptyset$. Since no terminal can lie on the right side of
a cut separating $v$ from $T$, it follows that $S\subseteq\lcutc{C}$.

Moreover, $\PT(G,S)\subseteq\lcutc{C}\cup\scutc{C}$. Indeed, every vertex
$p\in \PT(G,S)$ has an edge to some terminal in $S\subseteq\lcutc{C}$, so
$p\in\rcutc{C}$ would create an edge from the right side of $C$ to its left
side. The placement of $S$, $\PT(G,S)$, and $v$ with respect to $C$ is shown in
\Cref{fig:round-essential-tightest-cut}. Remove the vertices of $S$ from the left side of $C$ and define
$C_H:=(\lcutc{C}\setminus S,\scutc{C},\rcutc{C})$. This is a valid cut in
$H$ separating $v$ from $T_H$: the vertex $v$ remains on the right side,
every terminal in $T_H$ lies on the left side or in the separator (this was
already true for $T'$, and was just proved for $\PT(G,S)$), and no edge goes from
right to left. Its separator is unchanged and has size $\kappa_G$, so
$\kappa_H\le\kappa_G$. Together with
$\kappa_H\ge\kappa_G$, this gives $\kappa_H=\kappa_G$, and hence $C_H$ is a
minimum cut in $H$.
The resulting minimum cut $C_H$ is shown in
\Cref{fig:round-essential-induced-cut}.
\begin{figure}[htbp]
\centering
\begingroup
\renewcommand{\EssCutPadding}{1pt}
\newcommand{\panelcaption}[1]{%
  \par\vspace{1pt}\begin{minipage}{6.85cm}\centering\small #1\end{minipage}%
}

\tikzset{
  round cut witness edge/.style={
    ->,
    draw=gray!72,
    line width=1.0pt,
    shorten <=1.5pt,
    shorten >=2pt
  },
  round cut S dot/.style={
    circle,
    fill=orange!80!black,
    draw=orange!80!black,
    inner sep=1.8pt
  },
  round cut remaining dot/.style={
    circle,
    fill=blue!60!black,
    draw=blue!60!black,
    inner sep=1.8pt
  },
  round cut neighborhood dot/.style={
    circle,
    fill=violet!70!black,
    draw=violet!70!black,
    inner sep=1.8pt
  },
  round cut S region/.style={
    fill=orange!12,
    draw=orange!80!black,
    dashed,
    rounded corners=7pt,
    line width=0.9pt,
    inner xsep=9pt,
    inner ysep=12pt
  },
  round cut remaining region/.style={
    fill=blue!7,
    draw=blue!60!black,
    dashed,
    rounded corners=7pt,
    line width=0.9pt,
    inner xsep=9pt,
    inner ysep=12pt
  },
  round cut neighborhood region/.style={
    fill=violet!9,
    draw=violet!70!black,
    dashed,
    rounded corners=7pt,
    line width=0.9pt,
    inner xsep=12pt,
    inner ysep=12pt
  },
  round cut terminal universe/.style={
    fill=black!12,
    fill opacity=0.20,
    draw=black!58,
    dash pattern=on 1.2pt off 1.6pt,
    rounded corners=9pt,
    line width=0.9pt,
    inner sep=7pt
  }
}

\resizebox{\textwidth}{!}{%
\begin{minipage}{16.4cm}
\centering

\begin{minipage}[t]{7.10cm}
\centering
\begin{essgraph}[scale=0.74,>=stealth,node font={\small}]
  \essterminal[extra={opacity=0}]{soneghost}{(0,3.10)}{}
  \essterminal[extra={opacity=0}]{stwoghost}{(0,3.10)}{}
  \essterminal[extra={opacity=0}]{sthreeghost}{(0,3.10)}{}

  \essterminal{t}{(-2.70,0)}{$t$}

  \node[round cut remaining dot] (tleftone) at (-3.15,-2.25) {};
  \node[round cut remaining dot] (tlefttwo) at (-2.25,-2.25) {};
  \node[round cut S dot] (sone) at (-0.35,-2.25) {};
  \node[round cut S dot] (stwo) at (0.55,-2.25) {};
  \node[round cut S dot] (sthree) at (1.45,-2.25) {};

  \node[round cut neighborhood dot] (pone) at (3.40,-2.25) {};
  \node[round cut neighborhood dot] (ptwo) at (3.40,0) {};

  \begin{scope}[shift={(0,3.10)},scale=1.28,transform shape]
    \essvertex[colors={t/1}]{v}{(0,0)}{$v$}
  \end{scope}

  \essvertex[extra={opacity=0}]{la}{(-4.40,-3.00)}{}
  \essvertex[extra={opacity=0}]{lb}{(4.40,-3.00)}{}
  \essvertex[extra={opacity=0}]{lc}{(-4.40,-2.00)}{}
  \essvertex[extra={opacity=0}]{ld}{(4.40,-2.00)}{}
  \essvertex[extra={opacity=0}]{ca}{(-4.40,-0.65)}{}
  \essvertex[extra={opacity=0}]{cb}{(4.40,-0.65)}{}
  \essvertex[extra={opacity=0}]{cc}{(-4.40,1.05)}{}
  \essvertex[extra={opacity=0}]{cd}{(4.40,1.05)}{}
  \essvertex[extra={opacity=0}]{ra}{(-4.40,2.40)}{}
  \essvertex[extra={opacity=0}]{rb}{(4.40,2.40)}{}
  \essvertex[extra={opacity=0}]{rc}{(-4.40,3.80)}{}
  \essvertex[extra={opacity=0}]{rd}{(4.40,3.80)}{}

  \esscut[
    name=C,
    mode=full,
    left={la,lb,lc,ld},
    separator={t,ca,cb,cc,cd},
    right={v,ra,rb,rc,rd}
  ]

  \begin{scope}[on background layer]
    \node[round cut S region,fit=(sone)(stwo)(sthree)] (sregion) {};
    \node[round cut remaining region,fit=(t)(tleftone)(tlefttwo)]
      (remainingregion) {};
    \node[round cut neighborhood region,fit=(pone)(ptwo)]
      (nregion) {};
    \node[round cut terminal universe,fit=(sregion)(remainingregion)]
      (tregion) {};

  \end{scope}

  \draw[round cut witness edge] (pone) -- (sthree);
  \draw[round cut witness edge]
    (ptwo) .. controls (2.90,-0.88) and (1.30,-1.43) .. (stwo);

  \node[font=\small\bfseries,text=orange!80!black,
        fill=orange!12,inner sep=1pt]
    at (0.55,-1.60) {$S$};
  \node[font=\small\bfseries,text=blue!60!black,
        fill=blue!7,inner sep=1pt]
    at (-2.70,-1.60) {$T\setminus S$};
  \node[font=\small\bfseries,text=violet!70!black,
        fill=violet!9,inner sep=1pt]
    at (3.40,-1.60) {$\PT(G,S)$};
  \node[font=\small\bfseries,text=black!65,fill=black!8,inner sep=1pt]
    at (-0.90,0.88) {$T$};

  \node[anchor=east,font=\small\itshape] at (-4.70,-2.30)
    {$\lcutc{C}$};
  \node[anchor=east,font=\small\itshape] at (-4.70,0)
    {$\scutc{C}$};
  \node[anchor=east,font=\small\itshape] at (-4.70,3.10)
    {$\rcutc{C}$};
\end{essgraph}
\phantomsubcaption\label{fig:round-essential-tightest-cut}
\panelcaption{\textbf{(\thesubfigure)} The tightest minimum cut $C$ in
  $(G,T)$.}
\end{minipage}
\hspace{-0.22cm}
\raisebox{2.05cm}{\Large$\Longrightarrow$}
\hspace{-0.22cm}
\begin{minipage}[t]{7.10cm}
\centering
\begin{essgraph}[scale=0.74,>=stealth,node font={\small}]
  \essterminal[extra={opacity=0}]{soneghost}{(0,3.10)}{}
  \essterminal[extra={opacity=0}]{stwoghost}{(0,3.10)}{}
  \essterminal[extra={opacity=0}]{sthreeghost}{(0,3.10)}{}

  \essterminal{t}{(1.00,0)}{$t$}
  \node[round cut remaining dot] (tleftone) at (0.55,-2.25) {};
  \node[round cut remaining dot] (tlefttwo) at (1.45,-2.25) {};

  \node[round cut neighborhood dot] (pone) at (-1.25,-2.25) {};
  \node[round cut neighborhood dot] (ptwo) at (-1.25,0) {};

  \begin{scope}[shift={(0,3.10)},scale=1.28,transform shape]
    \essvertex[colors={t/1}]{v}{(0,0)}{$v$}
  \end{scope}

  \essvertex[extra={opacity=0}]{hla}{(-4.40,-3.00)}{}
  \essvertex[extra={opacity=0}]{hlb}{(4.40,-3.00)}{}
  \essvertex[extra={opacity=0}]{hlc}{(-4.40,-2.00)}{}
  \essvertex[extra={opacity=0}]{hld}{(4.40,-2.00)}{}
  \essvertex[extra={opacity=0}]{hca}{(-4.40,-0.65)}{}
  \essvertex[extra={opacity=0}]{hcb}{(4.40,-0.65)}{}
  \essvertex[extra={opacity=0}]{hcc}{(-4.40,1.05)}{}
  \essvertex[extra={opacity=0}]{hcd}{(4.40,1.05)}{}
  \essvertex[extra={opacity=0}]{hra}{(-4.40,2.40)}{}
  \essvertex[extra={opacity=0}]{hrb}{(4.40,2.40)}{}
  \essvertex[extra={opacity=0}]{hrc}{(-4.40,3.80)}{}
  \essvertex[extra={opacity=0}]{hrd}{(4.40,3.80)}{}

  \esscut[
    name=CH,
    mode=full,
    left={hla,hlb,hlc,hld},
    separator={t,hca,hcb,hcc,hcd},
    right={v,hra,hrb,hrc,hrd}
  ]

  \begin{scope}[on background layer]
    \node[round cut remaining region,fit=(t)(tleftone)(tlefttwo)]
      (hremainingregion) {};
    \node[round cut neighborhood region,fit=(pone)(ptwo)]
      (hnregion) {};
    \node[round cut terminal universe,fit=(hremainingregion)(hnregion)]
      (thregion) {};
  \end{scope}

  \node[font=\small\bfseries,text=blue!60!black,
        fill=blue!7,inner sep=1pt]
    at (1.00,-1.60) {$T\setminus S$};
  \node[font=\small\bfseries,text=violet!70!black,
        fill=violet!9,inner sep=1pt]
    at (-1.25,-1.60) {$\PT(G,S)$};
  \node[font=\small\bfseries,text=black!65,fill=black!8,inner sep=1pt]
    at (-0.12,0.88) {$T_H$};

  \node[anchor=east,font=\small\itshape] at (-4.70,-2.30)
    {$\lcutc{C_H}$};
  \node[anchor=east,font=\small\itshape] at (-4.70,0)
    {$\scutc{C_H}$};
  \node[anchor=east,font=\small\itshape] at (-4.70,3.10)
    {$\rcutc{C_H}$};
\end{essgraph}
\phantomsubcaption\label{fig:round-essential-induced-cut}
\panelcaption{\textbf{(\thesubfigure)} The induced minimum cut $C_H$ in
  $(H,T_H)$.}
\end{minipage}

\end{minipage}
}%

\caption{The cut transformation used in
  \Cref{lem:essentiality-survives-rounding}. In
  \subref{fig:round-essential-tightest-cut}, $S\subseteq\lcutc{C}$,
  $v\in\rcutc{C}$, and the edges from $\PT(G,S)$ to $S$ force $\PT(G,S)$ to lie
  in $\scutc{C}\cup\lcutc{C}$. The essential terminal
  $t\in T\setminus S$ lies in $\scutc{C}$. In
  \subref{fig:round-essential-induced-cut}, deleting $S$ and making $\PT(G,S)$
  terminal yields $T_H=(T\setminus S)\cup \PT(G,S)$ and the cut
  $C_H=(\lcutc{C}\setminus S,\scutc{C},\rcutc{C})$, with the same
  separator and with $t$ still in it.}
\label{fig:round-essential-cut-transfer}

\endgroup
\end{figure}

Since $t$ is essential for $v$ in $(G,T)$,
\Cref{lem:essential-terminal-tightest-cut} gives $t\in\scutc{C}$. Thus $t$
also belongs to the separator of the minimum cut $C_H$. Because
$t\in T_H$, \Cref{obs:tightest-min-cut} implies that $t$ belongs to the
separator of the tightest minimum cut in $(H,T_H)$. By
\Cref{lem:essential-terminal-tightest-cut}, $t$ is therefore essential for
$v$ in $(H,T_H)$.

It remains to remove the vertices of $\PT(G,S)$, which are now terminals in
$(H,T_H)$. Remove them one at a time. Since each of them is distinct from
$t$, \Cref{lem:essential-survives-terminal-removal} guarantees that $t$
remains essential after every removal. Once all vertices of $\PT(G,S)$ have been
removed, the resulting instance is $(G',T')$. Therefore $t$ remains essential
for $v$ after the rounding step.
\end{proof}

We can now use the preceding structural lemma to show that the split-assignment
witness remains valid on the smaller instance returned by rounding.
\begin{lemma}[Preservation after rounding] \label{lem:weighted-round-and-remove}
Assume that $(G,T,\cp,\weight)$ satisfies the \fractionalesscond, all capacities are positive, there is no pre-terminal of out-degree one, and there is no matching from $\PT(G,T)$ to $T$ saturating $T$. Then the smaller instance returned by \textsc{RoundAndRemove} satisfies the \fractionalesscond.
\end{lemma}
\begin{proof}
Let $S$ be the inclusion-minimal Hall-deficient set chosen by
\textsc{RoundAndRemove}, and let $G':=G\setminus(S\cup \PT(G,S))$ and
$T':=T\setminus S$. Let $\assgnw$ be a
flow-essential split-assignment witnessing the \fractionalesscond in $(G,T)$, and
restrict it to the remaining vertices and terminals.

First, every remaining vertex still sends all of its weight to terminals in
$T'$. Indeed, by
\Cref{lem:terminal-leaf-hall-set-not-essential-outside}, no terminal in $S$
is essential for a vertex outside $\PT(G,S)$. Since $\assgnw$ sends positive
weight only to essential terminals, a remaining vertex sends no weight to
$S$.

Next, suppose the restricted split-assignment sends positive weight from a
remaining vertex $v$ to a terminal $t\in T'$. Then $t$ was essential for $v$
in $G$, and \Cref{lem:essentiality-survives-rounding} guarantees that it
remains essential for $v$ after $S\cup \PT(G,S)$ is removed. Finally, the total weight assigned to every remaining terminal
can only decrease when the split-assignment is
restricted, so it remains within the corresponding capacity. Thus the
restricted split-assignment is a valid witness to the \fractionalesscond on the
smaller instance.
\end{proof}

We are now ready to prove the main theorem of this section.

\begin{proof}[Proof of \Cref{gyori-weighted-poly-theorem}]
We show that \textsc{GLWeightedPartition} in \Cref{alg:gl-weighted} is the promised polynomial-time algorithm. We argue by induction on $|V(G)|+|E(G)|$. If $G=\emptyset$, Lines~\ref{alg:gl-weighted:line:empty-start}--\ref{alg:gl-weighted:line:empty-end} return the empty partition, and the claim is immediate.

Suppose a terminal with zero capacity exists. Then the algorithm deletes it in Lines~\ref{alg:gl-weighted:line:remove-saturated-start}--\ref{alg:gl-weighted:line:remove-saturated-end} and recurses on the smaller instance. By \Cref{lem:frac-remove-saturated-terminal}, the output of the recursive call extends to a valid partition of the original instance. Suppose a pre-terminal of out-degree one exists. Then the algorithm contracts it in Lines~\ref{alg:gl-weighted:line:degree-1-start}--\ref{alg:gl-weighted:line:degree-1-end} and recurses on the smaller instance. By \Cref{lem:weighted-contract-degree-one-potential-leaf}, the output of the recursive call extends to a valid partition of the original instance.

If neither of the above cases applies, then all remaining terminals have positive capacity and every pre-terminal has out-degree at least two. If a matching from $\PT(G,T)$ to $T$ saturating $T$ exists, then the algorithm finds a removable edge in Lines~\ref{alg:gl-weighted:line:matching-start}--\ref{alg:gl-weighted:line:matching-end} and recurses on the smaller instance. By \Cref{lem:weighted-matching-removable-edge}, the removable edge exists.

If no matching from $\PT(G,T)$ to $T$ saturating $T$ exists, the algorithm chooses a nonempty set $S\subseteq T$ that no matching saturates and that is inclusion-minimal with this property, together with a terminal $t_S\in S$. Hall's theorem gives a subset of $S$ with fewer neighboring pre-terminals than terminals. Minimality forces this subset to be $S$, since every proper subset of $S$ can be saturated. Thus, $S$ is inclusion-minimal subject to $|S|>|\PT(G,S)|$. By \Cref{lem:minimal-hall-set}, $|\PT(G,S)|=|S|-1$ and there is a matching from $S\setminus\{t_S\}$ to $\PT(G,S)$ that saturates both sides. The subroutine uses this matching to form the parts for the terminals in $S$, and \Cref{lem:weighted-rounded-parts-valid} shows that these parts are connected and satisfy the weighted bound. The remaining instance $(G',T',\cp',\weight)$ satisfies the \fractionalesscond by \Cref{lem:weighted-round-and-remove}.

It remains to prove that the algorithm takes polynomial time. Every recursive call decreases $|V(G)|+|E(G)|$. The first operation removes a terminal, the second removes a pre-terminal by contraction, the third deletes an edge, and the fourth removes the nonempty set $S$ and its pre-terminal neighborhood. Thus, along any recursion path, there are at most $|V(G)|+|E(G)|$ recursive calls.

The set $S$ in \textsc{RoundAndRemove} can also be found in polynomial time. Start with $S=T$, which no matching saturates by assumption. Repeatedly test each terminal $t\in S$ and remove $t$ whenever no matching saturates $S\setminus\{t\}$, restarting the scan after each removal. The empty set is saturated by the empty matching, so this process leaves a nonempty set, and it performs at most $|T|^2$ matching tests. When it stops, every set $S\setminus\{t\}$ can be saturated. Every proper subset of $S$ lies inside such a set and can therefore be saturated, so $S$ has the required minimality.

All other work in the operations are also polynomial, including finding the removable edge (by \Cref{lem:weighted-matching-removable-edge}). Hence the algorithm runs in polynomial time.
\end{proof}

\subsection{Running Time Analysis}

We briefly analyze the running time of \Cref{alg:gl-weighted}. Recall that $n=|V(G)|$, $m=|E(G)|$, and $k=|T|$.

\begin{proposition}
\label{prop:weighted-algorithm-time}
\Cref{alg:gl-weighted} runs in
$\caO\bigl(nkm^{2+o(1)}+m(nk)^{1+o(1)}\log(nw_{\max})\log n\bigr)$ time.
\end{proposition}
\begin{proof}
We first record a bound used throughout the proof. Every non-terminal has positive weight, so a witness to the \fractionalesscond sends positive weight from it to an essential terminal. Therefore, every non-terminal has an outgoing edge. It follows that $n-k\leq m$, and hence $n\leq k+m$.

We will use the following matching bound. Given a set of terminals $S\subseteq T$, we can either find a matching from $\PT(G,S)$ to $S$ that saturates $S$ or report that none exists. This takes $\caO(n+m^{1+o(1)})$ time. We first scan the graph in $\caO(n+m)$ time to test whether some terminal in $S$ has no neighboring pre-terminal. If such a terminal exists, then no matching saturates $S$. Otherwise, every vertex on either side of the bipartite graph is incident to an edge, so this graph has $\caO(m)$ vertices and at most $m$ edges. The standard reduction to unit-capacity maximum flow produces a network with $\caO(m)$ edges. Applying the exact maximum-flow algorithm of \citet{brand2023deterministic} yields the desired matching in $\caO(m^{1+o(1)})$ time.

We split the running-time analysis into two parts: the cost of checking the conditions in Lines~\ref{alg:gl-weighted:line:remove-saturated-start}, \ref{alg:gl-weighted:line:degree-1-start}, and \ref{alg:gl-weighted:line:matching-start}, and the cost of carrying out the four operations.

The conditions in Lines~\ref{alg:gl-weighted:line:remove-saturated-start} and \ref{alg:gl-weighted:line:degree-1-start} test whether a terminal has zero capacity and whether a pre-terminal has out-degree one, respectively. Both can be checked in $\caO(n+m)$ time. The condition in Line~\ref{alg:gl-weighted:line:matching-start} tests whether there is a matching from $\PT(G,T)$ to $T$ that saturates $T$. By the matching bound above, this test takes $\caO(n+m^{1+o(1)})$ time. Since there are at most $n+m$ recursive calls, the total cost of checking these conditions is
$\caO\bigl((n+m)^2+(n+m)m^{1+o(1)}\bigr)$, which is within $\caO(nkm^{2+o(1)})$.

We next bound the total time spent carrying out the four operations when their conditions hold.

Removing a zero-capacity terminal in Lines~\ref{alg:gl-weighted:line:remove-saturated-start}--\ref{alg:gl-weighted:line:remove-saturated-end} takes $\caO(n+m)$ time.
Contracting a pre-terminal of out-degree one in Lines~\ref{alg:gl-weighted:line:degree-1-start}--\ref{alg:gl-weighted:line:degree-1-end} also takes $\caO(n+m)$ time. Each call to either operation removes a vertex, so there are at most $n$ such calls. Their total cost is $\caO(n(n+m))$. Since $k,m\geq 1$ and $n\leq k+m$, this cost is within $\caO(nkm^2)$.

We now consider the third operation, which finds a removable edge in Lines~\ref{alg:gl-weighted:line:matching-start}--\ref{alg:gl-weighted:line:matching-end}. Let $e_1,\dots,e_k$ be the chosen secondary edges. This operation runs only when a matching from $\PT(G,T)$ to $T$ saturating $T$ exists. That matching uses $k$ distinct edges of $G$, so $k\leq m$ and hence $n\leq k+m\leq 2m$. Therefore every bound of the form $\caO((n+m)^{1+o(1)})$ here is $\caO(m^{1+o(1)})$. By \Cref{prop:essential-terminals-single-flow}, the essential terminals of all vertices can be computed in $\caO(nm^{1+o(1)})$ time. We perform this computation in $G$ and in each graph $G\setminus e_i$. Comparing the resulting sets determines, for every $i$, $v$, and $t$, whether $e_i$ is critical for assigning $v$ to $t$. The total time is $\caO(nkm^{1+o(1)})$.
The criticality information gives all values of $\crt$. By \Cref{prop:best-assign-is-poly}, a minimum-potential flow-essential split-assignment can then be found in
$\caO\bigl(nm^{1+o(1)}+(nk)^{1+o(1)}\log(nw_{\max})\log n\bigr)$ time. Finally, checking the $k$ secondary edges against all pairs with positive assigned weight takes $\caO(nk^2)$ time. Since $k\leq m$, this last cost is within $\caO(nkm)$. Therefore, one such call takes
$\caO\bigl(nkm^{1+o(1)}+(nk)^{1+o(1)}\log(nw_{\max})\log n\bigr)$ time. Since each call removes an edge, there are at most $m$ such calls, so their total cost is $\caO\bigl(nkm^{2+o(1)}+m(nk)^{1+o(1)}\log(nw_{\max})\log n\bigr)$.

In the fourth operation, Lines~\ref{alg:gl-weighted-line-roundcall}--\ref{alg:gl-weighted:line:round-end} call \textsc{RoundAndRemove} (\Cref{alg:gl-round-and-remove}). If a terminal has no neighboring pre-terminal, its singleton is an inclusion-minimal Hall-deficient set and can be found in $\caO(n+m)$ time.

Otherwise, every terminal has a neighboring pre-terminal, so the current number of terminals is at most $\min\{k,m\}$. To find an inclusion-minimal Hall-deficient set, start with the current terminal set and remove a terminal whenever the remaining set still cannot be saturated. Restart the scan after each removal. This procedure uses $\caO(\min\{k,m\}^2)$ matching tests. By the matching bound above, these tests take $\caO\bigl(\min\{k,m\}^2(n+m^{1+o(1)})\bigr)$ time in total. Since \textsc{RoundAndRemove} removes at least one terminal in each call, there are at most $k$ such calls. Therefore, all calls to \textsc{RoundAndRemove} take
$\caO\bigl(k(n+m)+k\min\{k,m\}^2(n+m^{1+o(1)})\bigr)$ total time. Since $\min\{k,m\}^2\leq m^2$ and $\min\{k,m\}^2\leq nm$, this is within $\caO(nkm^{2+o(1)})$.

The bounds above cover both checking the conditions and carrying out the four operations. Each cost is within the claimed bound. This concludes the proof.
\end{proof}

\FloatBarrier

\section{Improved Algorithm for DAGs}
\label{sec:dag}

In this section, we consider the weighted \gyorilovasz partition problem on a directed acyclic graph. We show that in this case the weighted problem admits a near-linear-time algorithm.

\thmdag*

Since the unweighted setting is a special case of the weighted setting, we also obtain the following corollary.

\begin{corollary}
    \label{cor:gl-dag-unweighted}
    Let $D=(V,E)$ be a directed acyclic graph and let $T=\{t_1,\dots,t_k\}$ be a set of distinct terminals. Assume that $D$ is $k$-$T$-connected, and let $\cp_1,\dots,\cp_k$ be nonnegative integers such that $\sum_{i=1}^k \cp_i=|V\setminus T|$. Then there is an $\caO(m\log n)$-time algorithm that partitions $V$ into $V_1,\dots,V_k$ such that $t_i\in V_i$, $|V_i|=\cp_i+1$, and $D[V_i]$ is connected to $t_i$ for every $i\in [k]$.
\end{corollary}
\begin{proof}
Set $w_v=1$ for all non-terminal vertices $v\in V\setminus T$, so $w_{\max}=1$. Then $\sum_{v\in V\setminus T} w_v=|V\setminus T|=\sum_{i=1}^k \cp_i$. Applying \Cref{dag-algo-theorem} gives a partition with the guarantee that $|V_i| \leq \cp_i + 1$ for every $i\in[k]$. Since the total number of non-terminals is $\sum_{i=1}^k \cp_i$, we must have $|V_i| = \cp_i + 1$ for every $i\in[k]$.
\end{proof}

We begin with the ordering used by the algorithm and in the proofs below.

\begin{definition}[Canonical topological order] \label{def:canonical-topological-order}
Let $T = \{t_1, t_2, \dots, t_k\}$ be an indexed set of terminals. A topological order $\prec$ of a DAG $D=(V,E)$ is an order such that $(v,u) \in E$ implies $v \prec u$. It is called a \emph{canonical topological order} if every non-terminal vertex precedes every terminal, and the terminals appear at the end in the order. That is, $(V \setminus T) \prec t_1 \prec t_2 \prec \dots \prec t_k$. Such an order always exists, since terminals have no outgoing edges.
\end{definition}

Unlike \Cref{sec:algorithm,sec:weighted}, we do not relax the connectivity requirement and work with $k$-$T$-connected DAGs directly. The key structural fact is that a DAG is $k$-$T$-connected if and only if every non-terminal has out-degree at least $k$ (\Cref{lem:k-conn-dag}). The algorithm therefore repeatedly contracts pre-terminals into terminals while preserving this out-degree condition. For an unsaturated terminal $t$, it chooses the earliest pre-terminal in the canonical topological order among those that currently have an edge to the part of $t$, and contracts this pre-terminal into $t$. Intuitively, this choice is safe since any remaining vertex whose out-degree could drop would have to be an even earlier predecessor of the same part (\Cref{lem:dag-contract-theorem,lem:contract-exists-dag}). We defer the full argument to the proof outline below. The design of this algorithm is inspired by the algorithm of \citet{casel2023efficient} for chordal graphs.

\subsection{\textsc{GLDAGPartition} Algorithm} \label{subsec:dag-algo}

The algorithm maintains the following objects:
\begin{itemize}
    \item $\hat{\cp}_t$ for each terminal $t$, its residual capacity, initialized to the input capacity $\cp_t$.
    \item $T_{\text{active}}$, the set of terminals $t$ with $\hat{\cp}_t>0$. The algorithm only grows parts of terminals in this set.
    \item $V_1,\dots,V_k$, where $V_i$ is the current part of $t_i$. Initially $V_i=\{t_i\}$, and every contraction into $t_i$ adds one more vertex to $V_i$.
    \item $\chi_v$ for each non-terminal vertex $v$, where $\chi_v=\textsc{true}$ means that $v$ has already been contracted and $\chi_v=\textsc{false}$ otherwise. This lets us discard stale heap entries.
    \item $H_i$ for each $i\in[k]$, a heap ordered by $\prec$ that stores pre-terminals with an outgoing edge to the current part $V_i$. Its role is to return the earliest available pre-terminal of $V_i$.
\end{itemize}

\begin{algorithm}[!ht]
    \caption{\textsc{GLDAGPartition}: Near-Linear Time \gyori--\lovasz Partition for $k$-$T$-Connected DAGs}
    \label{alg:dag-partition}
    \LinesNumbered
    \KwIn{DAG $D=(V,E)$, terminals $T=\{t_1,\dots,t_k\}$, weights $\weight$, terminal capacities $\cp$}
    \KwOut{Partition $\langle V_1,\dots,V_k\rangle$}

    Compute canonical topological order $\prec$ \cmt{Non-terminals precede $t_1,\dots,t_k$}
    
    \ForEach{$i \in [k]$}{
        $\hat{\cp}_{t_i} \gets \cp_{t_i}$ \cmt{Initialize the residual capacity of $t_i$}
        $V_i \gets \{t_i\}$ \cmt{Initialize the output part of $t_i$}
        $H_i \gets \textsc{MinHeap}(\text{in-neighbors of } t_i, \prec)$ \cmt{Pre-terminals with an outgoing edge to a vertex in $V_i$} \label{alg:dag-partition:line:init-heap}
    }

    $T_{\text{active}} \gets \{t\in T:\hat{\cp}_t>0\}$ \cmt{Terminals with positive residual capacity}
    $r \gets |V \setminus T|$ \cmt{Number of non-terminals not yet contracted}

    \ForEach{$v \in V \setminus T$}{
        $\chi_v \gets \textsc{false}$ \cmt{$v$ has not been used yet}
    }

    \While{$r > 0$}{ \label{alg:dag-partition:line:while-start}
        Pick $t_i \in T_{\text{active}}$ \cmt{Choose any unsaturated terminal} \label{alg:dag-partition:line:pick-terminal}
        $p \gets \textsc{ExtractMin}(H_i)$ \cmt{Earliest pre-terminal with an edge to $V_i$} \label{alg:dag-partition:line:extract}
        
        \If{$\chi_p=\textsc{true}$}{ \label{alg:dag-partition:line:used-check-start}
            continue \cmt{Ignore stale heap entries}
        }
        \label{alg:dag-partition:line:used-check-end}

        $V_i \gets V_i \cup \{p\}$, $\hat{\cp}_{t_i} \gets \hat{\cp}_{t_i}-\weight_p$ \cmt{Contract $p$ to $t_i$, update residual capacity of $t_i$} \label{alg:dag-partition:line:contract}
        $\chi_p \gets \textsc{true}$, $r \gets r-1$ \cmt{Mark $p$ as used} \label{alg:dag-partition:line:mark}
        Insert every in-neighbor of $p$ into $H_i$ \cmt{New pre-terminals of the enlarged part} \label{alg:dag-partition:line:push-neighbors}

        \If{$\hat{\cp}_{t_i} \leq 0$}{ \label{alg:dag-partition:line:remove-terminal-start}
            $T_{\text{active}} \gets T_{\text{active}} \setminus \{t_i\}$ \cmt{$t_i$ is now saturated}
        }
        \label{alg:dag-partition:line:remove-terminal-end}
    }
    \label{alg:dag-partition:line:while-end}
    
    \Return $\langle V_1,\dots,V_k\rangle$
\end{algorithm}

The algorithm starts by computing a canonical topological order and initializing the five maintained objects. For each $i\in[k]$, the heap $H_i$ stores the pre-terminals that currently have an edge to the part $V_i$. In each iteration, the algorithm picks an active terminal $t_i$, extracts the earliest pre-terminal from $H_i$, and discards the entry if that pre-terminal was already used earlier (Lines~\ref{alg:dag-partition:line:pick-terminal}--\ref{alg:dag-partition:line:used-check-end}). Otherwise, it contracts this pre-terminal into $t_i$, adds it to $V_i$, decreases the residual capacity of $t_i$, and inserts all in-neighbors of the contracted pre-terminal into $H_i$ so that the heap continues to store pre-terminals with an edge to the enlarged part (Lines~\ref{alg:dag-partition:line:contract}--\ref{alg:dag-partition:line:push-neighbors}). Once $t_i$ becomes saturated, the algorithm removes it from $T_{\text{active}}$ and never grows its part again (Lines~\ref{alg:dag-partition:line:remove-terminal-start}--\ref{alg:dag-partition:line:remove-terminal-end}). \Cref{fig:dag-contract} illustrates one iteration of the algorithm.

\begin{figure}[htbp]
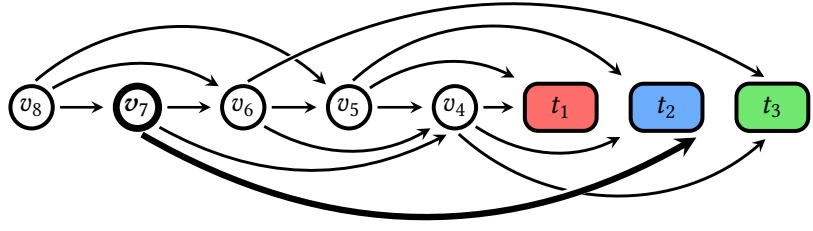
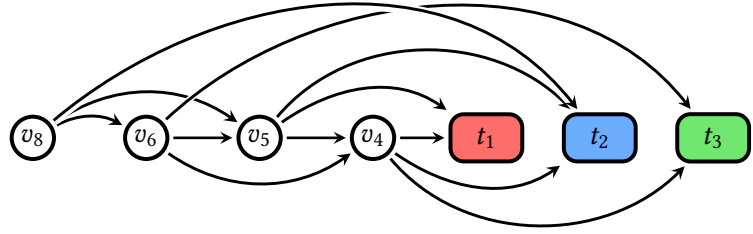

\centering

\tikzset{
    dag curved edge/.style={
        draw=black,
        line width=1.1pt,
        preaction={draw=white,line width=2.3pt,-}
    },
    dag highlighted edge/.style={
        draw=black,
        line width=2.4pt,
        preaction={draw=white,line width=3.6pt,-}
    }
}

\begin{subfigure}{\textwidth}
\centering

\begin{essgraph}[
    scale=1.0,
    every node/.style={transform shape},
    >=stealth,
    line cap=round,
    line join=round
]

\path[use as bounding box] (-5.60,-2.00) rectangle (5.60,2.00);

\essvertex{v8}{(-4.90,0)}{$v_8$}
\essvertex[extra={line width=2.6pt}]{v7}{(-3.50,0)}{$\boldsymbol{v_7}$}
\essvertex{v6}{(-2.10,0)}{$v_6$}
\essvertex{v5}{(-0.70,0)}{$v_5$}
\essvertex{v4}{( 0.70,0)}{$v_4$}

\essterminal{t1}{( 2.10,0)}{$t_1$}
\essterminal{t2}{( 3.50,0)}{$t_2$}
\essterminal{t3}{( 4.90,0)}{$t_3$}

\begin{scope}[
    on background layer,
    every path/.append style={
        ->,
        shorten >=3.5pt,
        shorten <=2pt
    }
]

\essedge{v8}{v7}
\essedge{v7}{v6}
\essedge{v6}{v5}
\essedge{v5}{v4}
\essedge{v4}{t1}

\draw[dag curved edge,bend left=30]
    (v8.north east) to (v6.north west);
\draw[dag curved edge,bend left=42]
    (v8.north) to (v5.north west);
\draw[dag curved edge,bend left=30]
    (v5.north east) to (t1.north west);
\draw[dag curved edge,bend left=42]
    (v5.north) to (t2.north west);

\draw[dag curved edge,bend left=30]
    (v6.north) to (t3.north);

\draw[dag curved edge,bend right=30]
    (v6.south east) to (v4.south west);
\draw[dag curved edge,bend right=28]
    (v7.south east) to (v4.south);
\draw[dag curved edge,bend right=30]
    (v4.south east) to (t2.south west);
\draw[dag curved edge,bend right=42]
    (v4.south) to (t3.south);

\draw[dag highlighted edge,bend right=30]
    (v7.south) to (t2.south east);

\end{scope}

\end{essgraph}

\caption{The original DAG.}
\label{fig:dag-contract-a}

\end{subfigure}

\vspace{0.45cm}

\begin{subfigure}{\textwidth}
\centering

\begin{essgraph}[
    scale=1.0,
    every node/.style={transform shape},
    >=stealth,
    line cap=round,
    line join=round
]

\path[use as bounding box] (-5.60,-2.00) rectangle (5.60,2.00);

\essvertex{v8c}{(-4.50,0)}{$v_8$}
\essvertex{v6c}{(-3.00,0)}{$v_6$}
\essvertex{v5c}{(-1.50,0)}{$v_5$}
\essvertex{v4c}{( 0.00,0)}{$v_4$}

\essterminal{t1c}{( 1.50,0)}{$t_1$}
\essterminal{t2c}{( 3.00,0)}{$t_2$}
\essterminal{t3c}{( 4.50,0)}{$t_3$}

\begin{scope}[
    on background layer,
    every path/.append style={
        ->,
        shorten >=1.5pt,
        shorten <=1.5pt
    }
]

\essedge{v6c}{v5c}
\essedge{v5c}{v4c}
\essedge{v4c}{t1c}

\draw[dag curved edge]
    (v8c) to[out=28,in=152] (v6c);
\draw[dag curved edge]
    (v8c) to[out=35,in=145] (v5c);
\draw[dag curved edge]
    (v5c) to[out=35,in=145] (t1c);
\draw[dag curved edge]
    (v5c) to[out=50,in=135] (t2c);

\draw[dag curved edge]
    (v6c) to[out=43,in=128,looseness=0.96] (t3c);

\draw[
    dag curved edge,
    preaction={draw=white,line width=4.0pt,-}
]
    (v8c) to[out=40,in=128,looseness=1.00] (t2c);

\draw[dag curved edge]
    (v6c) to[out=-35,in=-145] (v4c);
\draw[dag curved edge]
    (v4c) to[out=-35,in=-145] (t2c);
\draw[dag curved edge]
    (v4c) to[out=-50,in=-135] (t3c);

\end{scope}

\end{essgraph}

\caption{The DAG obtained by contracting $v_7$ into terminal $t_2$.}
\label{fig:dag-contract-b}

\end{subfigure}

\caption{One step of \textsc{GLDAGPartition}, assuming all terminal capacities are positive. The horizontal vertex order is canonical, and the DAG is $3$-$T$-connected by \Cref{lem:k-conn-dag}, since every non-terminal has out-degree $3$. Initially, $H_2=\langle v_7,v_5,v_4\rangle$ in increasing $\prec$ order; when $t_2$ is active, the algorithm extracts $v_7$, which can be contracted into $t_2$ by \Cref{lem:dag-contract-theorem} because they have no common predecessor. The contracted DAG remains $3$-$T$-connected, $\hat{\cp}_{t_2}$ decreases by $w_{v_7}$, and inserting the in-neighbor $v_8$ yields $H_2=\langle v_8,v_5,v_4\rangle$.}
\label{fig:dag-contract}

\end{figure}

The correctness hinges on the contraction step. Let $p$ be the pre-terminal contracted into $t_i$. We claim that no remaining non-terminal loses out-degree. Indeed, suppose some remaining vertex $u$ loses an outgoing neighbor when $p$ is contracted. Then $u$ must have edges both to $p$ and to the current part of $t_i$; otherwise replacing $p$ by $t_i$ would not reduce the number of distinct out-neighbors of $u$. Since $(u, p) \in E$, we have $u \prec p$ in the canonical topological order. Moreover, because $u$ already has an edge to the current part of $t_i$, the vertex $u$ is already present in $H_i$. Therefore $u$ should have been extracted before $p$, contradicting the choice of $p$ as the minimum unused element of $H_i$. Hence the out-degree condition is preserved, so each contraction preserves $k$-$T$-connectivity on the remaining active terminals.

\begin{lemma}[Time complexity] \label{algo-complexity}
\Cref{alg:dag-partition} executes in $\caO(m \log n)$ time.
\end{lemma}

\begin{proof}
Computing the canonical topological order and building the initial heaps takes $\caO(n+m)$ time. The running time is then dominated by heap operations.

Whenever a pre-terminal $p$ is contracted into $t_i$, the algorithm inserts every in-neighbor of $p$ into $H_i$. Over the entire execution, each edge contributes to such an insertion at most once, so there are $\caO(m)$ pushes in total. Likewise, every heap entry is extracted at most once, so there are $\caO(m)$ pops in total. Using ordinary binary heaps, each push and pop costs $\caO(\log n)$.

Therefore the total time spent on heap operations is $\caO(m \log n)$, and this dominates the $\caO(n+m)$ initialization cost. Hence the overall running time of \Cref{alg:dag-partition} is $\caO(m \log n)$.
\end{proof}

\subsection{Correctness of the Algorithm}

To prove the correctness of this algorithm, we first establish several foundational structural properties of $k$-$T$-connected DAGs.

\begin{lemma} [$k$-$T$-connected DAG] \label{lem:k-conn-dag}
Let $D = (V, E)$ be a simple DAG with terminals $T$. Then $D$ is $k$-$T$-connected if and only if every non-terminal vertex has an out-degree of at least $k$. That is, for all $v \in V \setminus T$, we have $d^+(v) \geq k$.
\end{lemma}
\begin{proof}
The necessity is immediate: $k$ vertex-disjoint paths from a non-terminal vertex require $k$ distinct outgoing edges.

For sufficiency, suppose for a contradiction that $D$ is not $k$-$T$-connected. By \Cref{lem:mengers-theorem}, there is a cut $C=(\lcutc{C},\scutc{C},\rcutc{C})$ of size less than $k$ that separates a vertex $u$ from $T$. Let $v$ be the maximum vertex of $\rcutc{C}$ under $\prec$. By maximality, $v$ has no outgoing edge to $\rcutc{C}$; by the cut property, it has no outgoing edge to $\lcutc{C}$. Thus every out-neighbor of $v$ lies in $\scutc{C}$. As $d^+(v)\geq k$ and $D$ is simple, $|\scutc{C}|\geq k$, a contradiction.
\end{proof}

\begin{definition}[Common predecessor] \label{def:common-predecessor}
        For a terminal $t \in T$ and a pre-terminal $p$, a vertex $u \in V$ is called a \emph{common predecessor} of $p$ and $t$ if $(u, p)$ $\in E$ and $(u, t) \in E$.
\end{definition}

\begin{lemma}[Contractibility of a pre-terminal and a terminal without a common predecessor] \label{lem:dag-contract-theorem}
        Let $D = (V, E)$ with terminals $T$ be a $k$-$T$-connected DAG. If there exist a pre-terminal $p$ $\in V \setminus T$ and a terminal $t \in T$ such that $(p, t)$ $\in E$ and $p$ and $t$ have no common predecessor, then contracting $p$ into $t$ results in a $k$-$T$-connected DAG.
\end{lemma}

\begin{proof}
        Let $D'$ be obtained by removing $p$ and redirecting every edge $(u,p)$ to $(u,t)$. We first show that this operation preserves the out-degree of every remaining vertex. Fix $u\neq p$. If $(u,p)\notin E$, then the outgoing edges of $u$ are unchanged. If $(u,p)\in E$, then $(u,t)\notin E$, because otherwise $u$ would be a common predecessor of $p$ and $t$. Thus the edge to $p$ is replaced by a new edge to $t$, and $u$ loses and gains exactly one distinct out-neighbor. Consequently,
        $d^+_{D'}(u)=d^+_D(u)$ for every remaining vertex $u$. In particular, every non-terminal vertex of $D'$ has out-degree at least $k$ by \Cref{lem:k-conn-dag}. The same observation shows that no parallel edges are created, so $D'$ remains simple.
        
        Finally, the previous topological order with $p$ deleted is still a valid topological order for $D'$. Indeed, every redirected edge $(u,t)$ comes from edges $(u,p)$ and $(p,t)$ in $D$, and therefore $u\prec p\prec t$. Hence all redirected edges point forward in the inherited order, and $D'$ is a DAG. Applying \Cref{lem:k-conn-dag} to $D'$ proves that it is $k$-$T$-connected.
\end{proof}

Although \Cref{fig:k-conn-counter-example} shows that a general directed graph may be $k$-$T$-connected without having a contractible pre-terminal, we show that $k$-$T$-connectivity in a DAG does guarantee such an edge.

\begin{lemma}[Existence of a contractible edge for each terminal] \label{lem:contract-exists-dag}
Let $D=(V,E)$ be a $k$-$T$-connected DAG with $V\setminus T\neq\emptyset$, and let $\prec$ be a canonical topological order. For every terminal $t\in T$, let $p$ be the earliest pre-terminal satisfying $(p,t)\in E$. Then, contracting $p$ into $t$ preserves $k$-$T$-connectivity.
\end{lemma}

\begin{proof}
Fix $t\in T$. Since $V\setminus T\neq\emptyset$, choose a non-terminal vertex $x$. The $k$-$T$-connectivity of $D$ gives $k$ vertex-disjoint paths from $x$ to distinct terminals, so every terminal, including $t$, has an in-neighbor. Thus $p$ is well-defined.

If $p$ and $t$ had a common predecessor $u$, then $(u,p)\in E$ and $(u,t)\in E$. The first edge implies $u\prec p$, and the second makes $u$ another pre-terminal of $t$. This contradicts the choice of $p$ as the earliest such pre-terminal. Hence $p$ and $t$ have no common predecessor, and the conclusion follows from \Cref{lem:dag-contract-theorem}.
\end{proof}

With these structural properties established, we can now prove the correctness of our proposed algorithm.

\begin{lemma}[Algorithm correctness] \label{lem:dag-algo-correctness}
Let $D=(V,E)$ be a $k$-$T$-connected DAG with terminals $T=\{t_1,\dots,t_k\}$. Suppose that every non-terminal vertex $v$ has a positive integer weight $\weight_v$, every terminal $t$ has a nonnegative integer capacity $\cp_t$, and $\sum_{v\in V\setminus T}\weight_v\leq\sum_{t\in T}\cp_t$. Then \Cref{alg:dag-partition} returns a partition $V_1,\dots,V_k$ of $V$ such that, for every $i\in[k]$, $t_i\in V_i$, $D[V_i]$ is connected to $t_i$, and
\begin{equation*}
    \sum_{v\in V_i\setminus T}\weight_v\leq \cp_{t_i}+w_{\max}-1.
\end{equation*}
\end{lemma}

\begin{proof}
Consider the graph obtained after any number of successful contractions. Its non-terminals are precisely the vertices $v$ with $\chi_v=\textsc{false}$, and the restriction of $\prec$ to these vertices, followed by $t_1,\dots,t_k$, is a canonical topological order. The algorithm maintains the residual capacity $\hat{\cp}_t$, initially equal to the input capacity $\cp_t$, for every terminal $t$. Moreover, the unused pre-terminals represented in $H_i$ are exactly the remaining pre-terminals with an edge to the contracted copy of $t_i$. Indeed, this is true initially by the heap initialization in Line~\ref{alg:dag-partition:line:init-heap}, and whenever a pre-terminal is contracted into $t_i$, the algorithm inserts all of its in-neighbors into $H_i$ (Line~\ref{alg:dag-partition:line:push-neighbors}). Thus, after any stale entries preceding it are discarded (Lines~\ref{alg:dag-partition:line:used-check-start}--\ref{alg:dag-partition:line:used-check-end}), the minimum unused pre-terminal in $H_i$ is the earliest pre-terminal of $t_i$ in the current contracted graph.

We prove inductively that the current contracted graph is $k$-$T$-connected. This holds initially by assumption. Suppose that non-terminals remain and the algorithm selects an active terminal $t_i$ (Line~\ref{alg:dag-partition:line:pick-terminal}). By \Cref{lem:contract-exists-dag}, the earliest pre-terminal of $t_i$ exists and can be contracted into $t_i$ while preserving $k$-$T$-connectivity. By the heap property above, this pre-terminal is exactly the first unused pre-terminal extracted from $H_i$ (Line~\ref{alg:dag-partition:line:extract}, after stale entries are handled in Lines~\ref{alg:dag-partition:line:used-check-start}--\ref{alg:dag-partition:line:used-check-end}). Hence every successful iteration preserves $k$-$T$-connectivity. This conclusion ultimately follows from \Cref{lem:dag-contract-theorem}: the earliest pre-terminal has no common predecessor with $t_i$, so its contraction is safe.

It remains to show that an active terminal is always available while a non-terminal remains (Line~\ref{alg:dag-partition:line:while-start}). Suppose otherwise. Every terminal $t_i$ is then inactive, so $\hat{\cp}_{t_i}\leq 0$. Since $\hat{\cp}_{t_i}$ was initialized to $\cp_{t_i}$ and decreases by the total weight assigned to $t_i$, the vertices assigned to $t_i$ have total weight at least $\cp_{t_i}$. Hence the total weight already assigned is at least $\sum_{t\in T}\cp_t$, and therefore at least $\sum_{v\in V\setminus T}\weight_v$. This is impossible while an unassigned vertex of positive weight remains. Furthermore, whenever a non-terminal remains, the preceding invariant and \Cref{lem:contract-exists-dag} ensure that the heap of each active terminal contains an unused entry. Since each stale entry is removed when extracted (Lines~\ref{alg:dag-partition:line:used-check-start}--\ref{alg:dag-partition:line:used-check-end}) and every successful iteration decreases $r$ (Line~\ref{alg:dag-partition:line:mark}), the algorithm terminates with $r=0$.

The returned sets therefore form a partition of $V$: each non-terminal is added to exactly one set, and $V_i$ contains $t_i$. Whenever a pre-terminal $p$ is added to $V_i$ (Line~\ref{alg:dag-partition:line:contract}), it has an edge to a vertex already in $V_i$. Fix one such edge for each added vertex. Starting from any $v\in V_i$ and repeatedly following these edges leads to vertices added strictly earlier, and therefore eventually reaches $t_i$. All these edges belong to the original graph $D$, so they form a directed path from $v$ to $t_i$ in $D[V_i]$. Thus undoing the contractions preserves the required connectivity.

Finally, fix $i\in[k]$. If the input capacity $\cp_{t_i}=0$, then the initial residual capacity $\hat{\cp}_{t_i}$ is zero, so $t_i$ is never active and $V_i=\{t_i\}$. Otherwise, the algorithm adds vertices to $V_i$ only while its residual capacity $\hat{\cp}_{t_i}$ is positive (Lines~\ref{alg:dag-partition:line:pick-terminal} and~\ref{alg:dag-partition:line:contract}). If $p$ is the last pre-terminal added before $t_i$ becomes inactive (Lines~\ref{alg:dag-partition:line:remove-terminal-start}--\ref{alg:dag-partition:line:remove-terminal-end}), then immediately before adding $p$, $\hat{\cp}_{t_i}$ is a positive integer and hence at least one. Therefore, the total weight assigned to $t_i$ before adding $p$ is at most $\cp_{t_i}-1$. Since $\weight_p \leq w_{\max}$, the total weight assigned to $V_i$ is at most $\cp_{t_i}+w_{\max}-1$. If $t_i$ is still active when the algorithm terminates, its assigned weight is strictly less than the input capacity $\cp_{t_i}$, so the same bound holds. This proves the claim.
\end{proof}

\begin{proof}[Proof of \Cref{dag-algo-theorem}]
\Cref{lem:dag-algo-correctness} shows that \Cref{alg:dag-partition} returns a valid weighted partition and that every part satisfies the required capacity bound. The running time is $\caO(m\log n)$ by \Cref{algo-complexity}.
\end{proof}

\subsection*{Acknowledgments}
The first author thanks Christos Papadimitriou for introducing this problem during a 2017 program at the Simons Institute. He also thanks Alireza Farhadi, Christos Papadimitriou, and Saeed Seddighin for early discussions on approximation algorithms for the problem and Robert Kleinberg for recent  discussions on its current state. The authors further thank Shayan Chashmjahan, Samira Goudarzi, and Suho Shin for helpful early discussions related to this project.

\clearpage
\printbibliography
\newcommand{\flowcond}{local connectivity\xspace}
\newcommand{\focuscond}{compact connectivity\xspace}
\newcommand{\compactset}[1]{\mathcal{C}(#1)}
\newcommand{\localset}[1]{\mathcal{L}(#1)}
\appendix

\section{Relaxed Conditions for the \gyorilovasz Theorem}
\label{app:relaxed-conditions}

The classical \gyorilovasz theorem is stated under $k$-$T$-connectivity. In our setting, however, this hypothesis is stronger than necessary: a valid partition exists whenever the \essentialassignmentcondition holds, even if the graph is not $k$-$T$-connected. This observation is central to our polynomial-time algorithm, since it permits contractions without destroying the connectivity structure needed for the inductive argument.

In this appendix, we restrict attention to the unweighted setting, so we count non-terminal vertices rather than total weight. We introduce two further relaxations of $k$-$T$-connectivity: \flowcond\ (\Cref{def:local-connectivity-condition}) and \focuscond\ (\Cref{def:compact-connectivity-condition}). These conditions form the hierarchy shown in \Cref{fig:connectivity-condition-hierarchy}: $k$-$T$-connectivity implies \focuscond, which in turn implies \flowcond\ (\Cref{lem:local-compact-kcon}). We show that the \gyorilovasz conclusion holds under either condition by \Cref{thm:glvs-generalized}, whose proof appears in \Cref{sec:gyori-generalized}.
Compact connectivity cannot always be preserved under contraction: some instances satisfying it have no edge whose contraction preserves the condition. Nevertheless, every such instance satisfies the \essentialassignmentcondition (\Cref{def:essential-assignment-condition} and \Cref{lem:compact-connectivity-implies-essential-assignment}) and therefore admits a polynomial-time partitioning algorithm by \Cref{thm:ess-assign-cond}. This contrast demonstrates the greater flexibility of the \essentialassignmentcondition.
\begin{definition}[Internally vertex-disjoint paths to terminals]
    \label{def:internally-vertex-disjoint-to-terminals}
    Let $G = (V, E)$ be a directed graph with terminal set $T \subseteq V$, and let $v \in V$. A family of directed paths from $v$ to $T$ is \emph{internally vertex-disjoint} if no two paths share a vertex other than the common starting vertex $v$ and their terminal endpoints. In particular, unlike the standard notion of vertex-disjoint paths, several paths are allowed to end at the same terminal.
\end{definition}

\subsection{Local Connectivity Condition}

\begin{definition}[Local connectivity and local connected set]
\label{def:local-connectivity}
Let $G = (V, E)$ be a directed graph with terminal set $T \subseteq V$, and let $T' \subseteq T$. A vertex $v \in V \setminus T$ is \emph{locally connected} to $T'$ if there exist $|T'|$ internally vertex-disjoint paths from $v$ to vertices of $T'$.

By convention, a pre-terminal of some $t \in T'$ is considered to have arbitrarily many paths to $t$, and hence is locally connected to $T'$.

The \emph{local connected set} of $T'$, denoted by $\localset{T'}$, is the set of all non-terminal vertices $v \in V \setminus T$ that are locally connected to $T'$.
\end{definition}
\begin{definition}[Local connectivity condition]

    \label{def:local-connectivity-condition}

    A directed graph $G = (V, E)$ with $|V| = n$ and terminal set
    $T \subseteq V$ with $|T| = k$ satisfies the \emph{local connectivity condition}
    with respect to capacities $\cp$ satisfying
    $\sum_{t \in T} \cp_t = n - k$ if, for every subset $T' \subseteq T$,
    the local connected set of $T'$ has size at least the total capacity
    of the terminals in $T'$. Equivalently, for every subset $T' \subseteq T$:
    \begin{equation*}
        |\localset{T'}| \geq \sum_{t \in T'} \cp_t.
    \end{equation*}
\end{definition}

\begin{lemma}[$k$-$T$-connectivity implies local connectivity]
    \label{lem:kcon-implies-localcon}
    Let $G = (V, E)$ be a directed graph with terminal set $T \subseteq V$ of size $k$. If $G$ is $k$-$T$-connected, then $G$ satisfies the local connectivity condition for any nonnegative capacities $\cp$ satisfying $\sum_{\term \in T} \cp_\term = |V \setminus T|$.
\end{lemma}

\begin{proof}
    Fix a non-terminal vertex $v \in V \setminus T$. Since $G$ is $k$-$T$-connected and $|T| = k$, there exists a family of $k$ vertex-disjoint paths from $v$ to $T$, with one path ending at each terminal.

    Now fix any subset $T' \subseteq T$. Restricting this family to the terminals in $T'$ yields $|T'|$ internally vertex-disjoint paths from $v$ to $T'$. Hence every non-terminal vertex is locally connected to $T'$, and therefore $\localset{T'} = V \setminus T$.

    It follows that $|\localset{T'}| = |V \setminus T|$. Since $\sum_{\term \in T} \cp_\term = |V \setminus T|$, we obtain
    $$|\localset{T'}| = |V \setminus T| = \sum_{\term \in T} \cp_\term \geq \sum_{\term \in T'} \cp_\term.$$
Thus the local connectivity condition holds for every $T' \subseteq T$.
\end{proof}

Consequently, local connectivity provides a relaxation of $k$-$T$-connectivity.
\begin{theorem}[The \gyorilovasz theorem under local connectivity]
    \label{thm:glvs-generalized}
    Let $G = (V, E)$ be a directed graph with terminal set $T = \{t_1, \dots, t_k\} \subseteq V$, and let $\cp_t \in \mathbb{N}$ for all $t \in T$. Assume that $\sum_{t \in T} \cp_t = |V \setminus T|$. If $G$ satisfies the local connectivity condition with respect to $T$ and $\cp$, then there exists a partition $\langle V_i \rangle_{i \in [k]}$ of $V$ such that, for every $i \in [k]$, we have $t_i \in V_i$, $|V_i \setminus \{t_i\}| = \cp_{t_i}$, and $G[V_i]$ is connected to $t_i$.
\end{theorem}

We prove \Cref{thm:glvs-generalized} in \Cref{sec:gyori-generalized} by adapting the framework of \citet{gyori1976division, hoyer2016gyori}. This proof is constructive, but it does not lead to a polynomial-time algorithm; taken at face value, it yields only an exponential-time procedure, even for small values of $k$.

\begin{conjecture}
    \label{conj:local-connectivity-contraction}
    Let $G = (V, E)$ be a directed graph with terminal set $T \subseteq V$ and capacities $\cp$ that satisfies the local connectivity condition. Then there always exists either:
    \begin{enumerate}
        \item a terminal $\term \in T$ where $\cp_\term = 0$ that can be removed from the graph while preserving the condition, or
        \item a pre-terminal $p \in V \setminus T$ and a terminal $\term \in T$ with $(p,\term) \in E$ such that contracting $p$ into $\term$ and decrementing $\cp_\term$ by $1$ preserves the local connectivity condition.
    \end{enumerate}
\end{conjecture}

We do not know a counterexample to this conjecture. A proof would immediately imply fixed-parameter tractability with parameter $k$: when $k$ is fixed, the local connectivity condition can be checked in polynomial time, and one could simply test whether contracting a given pre-terminal into a terminal joined to it by an outgoing edge preserves the condition. On the other hand, we conjecture that checking the local connectivity condition is hard to do in polynomial time when $k$ is not fixed. If so, any algorithm that repeatedly verifies the condition during the execution would necessarily incur super-polynomial overhead.

\subsection{Compact Connectivity Condition}
\begin{definition}[Compact connectivity and compact-connected set]
\label{def:compact-connectivity}
Let $G = (V, E)$ be a directed graph with terminal set $T \subseteq V$ and $|T| = k$. A non-terminal vertex $v \in V \setminus T$ is \emph{compact-connected} to a terminal $t \in T$ if there exists a family $\mathcal{P}$ of $k$ internally vertex-disjoint paths from $v$ to the terminals in $T$ such that, for every terminal $t' \neq t$, at most one path in $\mathcal{P}$ ends at $t'$. Equivalently, several paths of $\mathcal{P}$ may end at $t$, but every other terminal receives at most one path.

By convention, a pre-terminal of $t$ is considered to have arbitrarily many paths to $t$, and hence is compact-connected to $t$.

The \emph{compact-connected set} of a terminal $t$, denoted by $\compactset{t}$, is defined as the set of all non-terminal vertices $v \in V \setminus T$ that are compact-connected to $t$.
\end{definition}

This notion gives a second relaxation of $k$-$T$-connectivity that is still strong enough to imply the \gyorilovasz conclusion.

\begin{definition}[Compact connectivity condition]
\label{def:compact-connectivity-condition}
A directed graph $G = (V, E)$ with $|V| = n$ and terminal set $T \subseteq V$ with $|T| = k$ satisfies the \emph{compact connectivity condition} with respect to capacities $\cp$ satisfying $\sum_{t \in T} \cp_t = n - k$ if, for every subset $T' \subseteq T$, the number of vertices that are compact-connected to at least one terminal of $T'$ is at least the total capacity of the terminals in $T'$. Equivalently,
\begin{equation*}
\left| \bigcup_{t \in T'} \compactset{t} \right| \geq \sum_{t \in T'} \cp_t.
\end{equation*}
\end{definition}

By Hall's Marriage Theorem, the compact connectivity condition is equivalent to the existence of a capacity-respecting assignment in which each non-terminal vertex is assigned to a terminal to which it is compact-connected.

\begin{lemma}[Compact connectivity and witness assignment]
\label{lem:compact-connectivity-implies-assignment}

Let $G=(V,E)$ with $|V|=n$ and $|T|=k$, and let $\cp$ be capacities satisfying $\sum_{t\in T}\cp_t=n-k$. Then $G$ satisfies the compact connectivity condition with respect to $T$ and $\cp$ if and only if there exists an assignment
$\assgncompact:V\setminus T\to T$
such that every $v\in V\setminus T$ is compact-connected to $\assgncompact(v)$ and
$|\assgncompact^{-1}(t)|=\cp_t$
for every $t\in T$. We call such an assignment $\assgncompact$ a \emph{witness} for the compact connectivity condition.
\end{lemma}

\begin{proof}

First, suppose that $G$ satisfies the compact connectivity condition. Replace each terminal $t\in T$ by $\cp_t$ copies, and form a bipartite graph whose left side is $V\setminus T$ and whose right side consists of these copies. A vertex $v\in V\setminus T$ is adjacent to every copy of $t$ exactly when $v$ is compact-connected to $t$.

Consider any set $S$ of terminal copies, and let $T'\subseteq T$ be the set of terminals having at least one copy in $S$. The copies in $S$ are adjacent to exactly the vertices that are compact-connected to some terminal of $T'$, so the compact connectivity condition gives
\begin{equation*}
    \left|\bigcup_{t\in T'}\compactset{t}\right|
    \geq \sum_{t\in T'}\cp_t
    \geq |S|.
\end{equation*}
Thus Hall's condition holds. Since both sides of the bipartite graph have size $n-k$, there is a perfect matching. Assigning each $v\in V\setminus T$ to the terminal whose copy it is matched to gives the desired witness $\assgncompact$.

Conversely, suppose that such a witness $\assgncompact$ exists. For any $T'\subseteq T$, every vertex assigned to a terminal in $T'$ is compact-connected to some terminal in $T'$. Hence
\begin{equation*}
    \bigcup_{t\in T'}\assgncompact^{-1}(t)
    \subseteq
    \bigcup_{t\in T'}\compactset{t}.
\end{equation*}
Therefore,
\begin{equation*}
    \left|\bigcup_{t\in T'}\compactset{t}\right|
    \geq
    \sum_{t\in T'}|\assgncompact^{-1}(t)|
    =
    \sum_{t\in T'}\cp_t.
\end{equation*}
Thus $G$ satisfies the compact connectivity condition.
\end{proof}

The compact connectivity condition sits between local connectivity and standard $k$-$T$-connectivity in the sense made precise below.

\begin{lemma}[Compact connectivity implies local connectivity to subsets] \label{lem:compact-implies-local}
    Let $G = (V, E)$ be a directed graph with terminal set $T \subseteq V$ and $|T| = k$. If a non-terminal vertex $v \in V \setminus T$ is compact-connected to a terminal $t \in T$, then $v$ is locally connected to any subset of terminals $T' \subseteq T$ that contains $t$.
\end{lemma}

\begin{proof}
    Since $v$ is compact-connected to $t$, there exists a family $\mathcal{P}$ of $k$ internally vertex-disjoint paths from $v$ to $T$ such that every terminal $t' \neq t$ is the endpoint of at most one path of $\mathcal{P}$.

    Fix a subset $T' \subseteq T$ containing $t$, and let $f$ be the number of paths in $\mathcal{P}$ whose endpoints lie in $T'$. Suppose, for contradiction, that $f < |T'|$. Then more than $k-|T'|$ paths of $\mathcal{P}$ end in $T \setminus T'$. Since $|T \setminus T'| = k-|T'|$, the pigeonhole principle implies that some terminal in $T \setminus T'$ is the endpoint of at least two paths of $\mathcal{P}$. This terminal cannot be $t$, because $t \in T'$, and so we contradict the definition of compact connectivity.

    Therefore $f \geq |T'|$. Choose any $|T'|$ of the paths in $\mathcal{P}$ whose endpoints lie in $T'$. These paths are internally vertex-disjoint paths from $v$ to $T'$, so $v$ is locally connected to $T'$.
\end{proof}

\begin{lemma}[Comparing local, compact, and $k$-$T$-connectivity]
\label{lem:local-compact-kcon}

    Let $G = (V, E)$ be a directed graph with terminal set $T \subseteq V$, 
    where $|V| = n$ and $|T| = k$, and capacities $\cp$ satisfying
    $\sum_{t \in T} \cp_t = n - k$. Then:

    \begin{enumerate}

        \item If $G$ is $k$-$T$-connected, then $G$ satisfies the compact connectivity condition with respect to terminal set $T$ and capacities $\cp$.

        \item If $G$ satisfies the compact connectivity condition with respect to $T$ and $\cp$, then it also satisfies the local connectivity condition with respect to $T$ and $\cp$.

    \end{enumerate}

    Thus, $k$-$T$-connectivity implies compact connectivity, and compact connectivity implies local connectivity.
\end{lemma}

\begin{proof}
    For the first implication, assume that $G$ is $k$-$T$-connected. Fix a non-terminal vertex $v \in V \setminus T$ and a terminal $t \in T$. Since $|T|=k$, there is a family $\mathcal{P}$ of $k$ vertex-disjoint paths from $v$ to $T$ with distinct endpoints, and hence exactly one path ends at each terminal. In particular, every terminal $t' \neq t$ is the endpoint of at most one path of $\mathcal{P}$. Therefore $v$ is compact-connected to $t$.
    More precisely, every non-terminal vertex belongs to every set $\compactset{t}$. For $T'=\emptyset$, both sides of the required inequality are zero and the condition holds. For a nonempty subset $T' \subseteq T$, the union below is over at least one terminal and therefore equals all of $V \setminus T$, so
    $$\left|\bigcup_{t \in T'} \compactset{t}\right| = |V \setminus T| = \sum_{t \in T} \cp_t \geq \sum_{t \in T'} \cp_t,$$
    and the compact connectivity condition holds.

    For the second implication, assume that $G$ satisfies the compact connectivity condition. Fix any subset $T' \subseteq T$. By \Cref{lem:compact-implies-local}, every vertex that is compact-connected to some terminal $t \in T'$ is locally connected to $T'$. Hence
    $$\bigcup_{t \in T'} \compactset{t} \subseteq \localset{T'}.$$
    Taking cardinalities and applying the compact connectivity condition gives
    $$|\localset{T'}| \geq \left|\bigcup_{t \in T'} \compactset{t}\right| \geq \sum_{t \in T'} \cp_t.$$
    Thus $G$ satisfies the local connectivity condition as well.
\end{proof}

\begin{corollary} \label{col:glvs-exists-compact}
    Let $G = (V, E)$ be a directed graph with terminal set $T = \{t_1, \dots, t_k\} \subseteq V$, and let $\cp_t \in \mathbb{N}$ for all $t \in T$. Assume that $\sum_{t \in T} \cp_t = |V \setminus T|$. If $G$ satisfies the compact connectivity condition with respect to $T$ and $\cp$, then there exists a partition $\langle V_i \rangle_{i \in [k]}$ of $V$ such that, for every $i \in [k]$, we have $t_i \in V_i$, $|V_i \setminus \{t_i\}| = \cp_{t_i}$, and $G[V_i]$ is connected to $t_i$.
\end{corollary}

\begin{proof}
    This is immediate from \Cref{lem:local-compact-kcon} and \Cref{thm:glvs-generalized}: compact connectivity implies local connectivity, and \Cref{thm:glvs-generalized} then yields the required partition.
\end{proof}

\begin{lemma}[Recognizing compact connectivity]
    \label{lem:compact-connectivity-recognition}
    Let $G=(V,E)$ be a directed graph with terminal set $T$ and capacities $\cp$, where $|V|=n$, $|E|=m$, and $|T|=k$. One can determine whether $G$ satisfies the compact connectivity condition in polynomial time using one maximum-flow computation for each pair $(v,t)\in(V\setminus T)\times T$, followed by one maximum-flow computation in a bipartite network with $\caO(nk)$ edges.
\end{lemma}

\begin{proof}
First verify that $\sum_{t\in T}\cp_t=n-k$, as required by \Cref{def:compact-connectivity-condition}. We first compute the sets $\compactset{t}$ for all $t\in T$.

Fix a non-terminal $v\in V\setminus T$ and a terminal $t\in T$. We use a vertex-splitting network to test whether $v\in\compactset{t}$. For every vertex $x\in V$, create two copies $x_{\mathrm{in}}$ and $x_{\mathrm{out}}$. Add the split arc $(x_{\mathrm{in}},x_{\mathrm{out}})$ with capacity $k$ when $x\in\{v,t\}$ and with capacity $1$ otherwise. Add a source $s_{\mathrm{in}}$, a sink $s_{\mathrm{out}}$, and the arc $(s_{\mathrm{in}},v_{\mathrm{in}})$ of capacity $k$. For every edge $(x,y)\in E$, add the arc $(x_{\mathrm{out}},y_{\mathrm{in}})$ of capacity $k$. Finally, for every terminal $t'\in T$, add the arc $(t'_{\mathrm{out}},s_{\mathrm{out}})$ of capacity $k$. Call the resulting network $H_{v,t}$.

For a family $\mathcal{P}=\{\mathcal{P}_1,\ldots,\mathcal{P}_k\}$ witnessing that $v$ is compact-connected to $t$, send one unit of flow along the copy of each path $\mathcal{P}_i$ in $H_{v,t}$. This flow route starts with the arc $(s_{\mathrm{in}},v_{\mathrm{in}})$, traverses the split arc $(v_{\mathrm{in}},v_{\mathrm{out}})$, and then alternates between arcs corresponding to edges of $G$ and the split arcs of the subsequent vertices of $\mathcal{P}_i$. After traversing the split arc of the terminal at which $\mathcal{P}_i$ ends, the route follows the arc to $s_{\mathrm{out}}$. These $k$ units respect all capacities. Indeed, every vertex other than $v$ and $t$ occurs in at most one path of $\mathcal{P}$, while the split arcs of $v$ and $t$ have capacity $k$ and may therefore be used by multiple flow routes. Moreover, every arc corresponding to an edge of $G$ has capacity $k$. Thus they form an integral $s_{\mathrm{in}}$--$s_{\mathrm{out}}$ flow of value $k$.

Conversely, let $f$ be an integral $s_{\mathrm{in}}$--$s_{\mathrm{out}}$ flow of value $k$. After discarding flow cycles, decompose $f$ into $k$ unit $s_{\mathrm{in}}$--$s_{\mathrm{out}}$ flow paths. Each such flow path begins with $(s_{\mathrm{in}},v_{\mathrm{in}})$ and, after projecting its split arcs and the arcs corresponding to edges of $G$ back to $G$, yields a directed path from $v$ to a terminal. In this way we obtain $k$ paths from $v$ to $T$. The unit capacities of the split arcs ensure that these paths are internally vertex-disjoint and that at most one ends at each terminal $t'\neq t$; several may end at $t$, whose split arc has capacity $k$. Hence these are precisely the $k$ paths required to show that $v\in\compactset{t}$. Therefore $v\in\compactset{t}$ if and only if the maximum-flow value in $H_{v,t}$ is $k$. Note that this test also includes every pre-terminal of $t$ in $\compactset{t}$, in accordance with the convention in \Cref{def:compact-connectivity}. Repeating the test for every pair $(v,t)\in(V\setminus T)\times T$ computes all compact-connected pairs. Each network has $\caO(n)$ vertices and $\caO(n+m)$ arcs.

We now build a bipartite flow network that tests for a witness to the compact connectivity condition. Create a source $s_{\mathrm{in}}$ and a sink $s_{\mathrm{out}}$. For every non-terminal $v\in V\setminus T$, add an arc $(s_{\mathrm{in}},v)$ of capacity $1$. For every terminal $t\in T$, add an arc $(t,s_{\mathrm{out}})$ of capacity $\cp_t$. Finally, add an arc $(v,t)$ of capacity $1$ exactly when $v\in\compactset{t}$.

An integral $s_{\mathrm{in}}$--$s_{\mathrm{out}}$ flow of value $n-k$ is exactly a witness $\assgncompact$ from \Cref{lem:compact-connectivity-implies-assignment}: the arc $(s_{\mathrm{in}},v)$ forces every non-terminal $v$ to be assigned to one terminal to which it is compact-connected, and the arcs $(t,s_{\mathrm{out}})$ force exactly $\cp_t$ vertices to be assigned to every terminal $t$. Conversely, every witness $\assgncompact$ defines such a flow. By \Cref{lem:compact-connectivity-implies-assignment}, $G$ satisfies the compact connectivity condition if and only if this network has maximum-flow value $n-k$.

The bipartite network has $\caO(n)$ vertices and $\caO(nk)$ arcs. Thus, after the polynomially many preprocessing flows, one additional maximum-flow computation recognizes the compact connectivity condition in polynomial time.
\end{proof}

The next two lemmas show that every instance satisfying the compact connectivity condition also satisfies the \essentialassignmentcondition.

\begin{lemma}[Compact connectivity implies essentiality] \label{lem:compact-connectivity-implies-essentiality}
    Let $G=(V,E)$ be a graph with a terminal set $T$ of size $k$, and let $v \in V \setminus T$. If $v$ is compact-connected to a terminal $t \in T$, then $t$ is essential for $v$.
\end{lemma}

\begin{proof}
    Since $v$ is compact-connected to $t$, there exists a collection $\mathcal{P}=\{\mathcal{P}_1,\ldots,\mathcal{P}_k\}$ of $k$ paths from $v$ to $T$ such that the paths are internally vertex-disjoint, the only terminal at which more than one path may end is $t$, and all other terminal endpoints are distinct. In particular, apart from the common starting vertex $v$, the only vertex that may belong to more than one path in $\mathcal{P}$ is $t$.

    Suppose, for contradiction, that $t$ is not essential for $v$. Let $\kappa=\connvg{v}{G}$, and let $\cut=(\lcut,\scut,\rcut)$ be the tightest minimum cut separating $v$ from $T$. Thus $|\scut|=\kappa$. Since $T$ itself is a separator of size $k$, we have $\kappa\leq k$.

    By the assumption that $t$ is not essential and \Cref{lem:essential-terminal-tightest-cut}, we have $t\in\lcut$, and in particular $t\notin\scut$.

    Every path $\mathcal{P}_j$ must intersect $\scut$. Otherwise, $\mathcal{P}_j$ would give a path from $v$ to a terminal that avoids $\scut$, contradicting the fact that $\cut$ separates $v$ from $T$. For each $j\in[k]$, let $x_j$ be the first vertex of $\mathcal{P}_j$ that is in $\scut$.

    We claim that the vertices $x_1,\ldots,x_k$ are all distinct. Indeed, suppose that $x_j=x_{j'}$ for two distinct indices $j,j'\in[k]$. Then $\mathcal{P}_j$ and $\mathcal{P}_{j'}$ share a vertex other than $v$. By compact connectivity, the only such shared vertex can be $t$. However, $x_j\in\scut$ while $t\in\lcut$, so $x_j\neq t$, a contradiction. Therefore the $k$ paths in $\mathcal{P}$ intersect $\scut$ in $k$ distinct vertices, and hence $|\scut|\geq k$. Since $|\scut|=\kappa\leq k$, it follows that $\kappa=k$.

    Thus the terminal connectivity of $v$ is $k$. After removing $t$, only $k-1$ terminals remain, so the terminal connectivity of $v$ in $G\setminus\{t\}$ is exactly $k-1$. Hence removing $t$ decreases the terminal connectivity of $v$, which means that $t$ is essential for $v$, contradicting our assumption. Therefore $t$ is essential for $v$.
\end{proof}

\begin{lemma}[Compact connectivity condition implies the \essentialassignmentcondition]
\label{lem:compact-connectivity-implies-essential-assignment}
If a graph $G = (V, E)$ satisfies the compact connectivity condition with respect to a terminal set $T$ and capacities $\cp$, then it also satisfies the \essentialassignmentcondition with respect to $T$ and $\cp$.
\end{lemma}

\begin{proof}
By \Cref{lem:compact-connectivity-implies-assignment}, the compact connectivity condition yields an assignment $\assgncompact: V \setminus T \to T$ mapping each non-terminal $v$ to a compact-connected terminal $\assgncompact(v)$ such that $|\assgncompact^{-1}(t)| = \cp_t$ for all $t \in T$. By \Cref{lem:compact-connectivity-implies-essentiality}, any terminal compact-connected to $v$ is also essential for $v$. Thus, $\assgn = \assgncompact$ is a valid flow-essential assignment satisfying the \essentialassignmentcondition.
\end{proof}

Hence, \Cref{alg:gl-partition} partitions any graph satisfying the compact connectivity condition in polynomial time.

\begin{corollary}
\label{cor:compact-connectivity-partition}
Let $G = (V, E)$ be a directed graph satisfying the compact connectivity condition with respect to terminal set $T = \{t_1, \dots, t_k\}$ and capacities $\cp$. Then $G$ satisfies the \essentialassignmentcondition, and $V$ can be partitioned in polynomial time into $k$ disjoint subsets $V_1, \dots, V_k$ such that for each $i \in \{1, \dots, k\}$, $t_i \in V_i$, $|V_i \setminus \{t_i\}| = \cp_{t_i}$, and $G[V_i]$ is connected to $t_i$.
\end{corollary}

\providecommand{\essentialassignmentcond}{Flow-Essential Assignment Condition}

\begin{figure}[htbp]
  \centering
  \begin{tikzpicture}[
      x=1cm,
      y=1cm,
      font=\small,
      hierarchy boundary/.style={line width=0.95pt, draw=black!78},
      condition label/.style={
        font=\small\bfseries,
        align=center,
        text=black!82,
        inner sep=0pt,
        execute at begin node={\hyphenpenalty=10000\exhyphenpenalty=10000}
      }
    ]
    \filldraw[
      fill=green!18,
      draw=green!32!black,
      line width=1.25pt
    ] (0,0) circle[radius=4.05];

    \filldraw[
      fill=cyan!18,
      fill opacity=0.58,
      draw=cyan!48!black,
      draw opacity=1,
      line width=1.05pt
    ] (1.40,-0.35) ellipse [x radius=4.85, y radius=2.95];

    \draw[hierarchy boundary] (0,-0.35) circle[radius=2.40];
    \draw[hierarchy boundary] (0,-0.35) circle[radius=1.10];

    \node[condition label, text=green!24!black, text width=4.60cm] at (0,3.1)
      {\essentialassignmentcond};
    \node[condition label, text width=2.60cm] at (0,1.35)
      {Compact Connectivity Condition};
    \node[condition label, font=\footnotesize\bfseries] at (0,-0.35)
      {$k$-connectivity};

    \node[
      condition label,
      font=\footnotesize\bfseries,
      text=cyan!42!black,
      text width=2.10cm
    ] at (5.15,-0.35) {Local Connectivity Condition};
  \end{tikzpicture}
  \par\vspace{0.4em}
  \caption{Hierarchy of the connectivity conditions considered in this paper. We give a polynomial-time algorithm, and consequently a constructive proof of existence, for instances satisfying the \essentialassignmentcond.}
  \label{fig:connectivity-condition-hierarchy}
\end{figure}
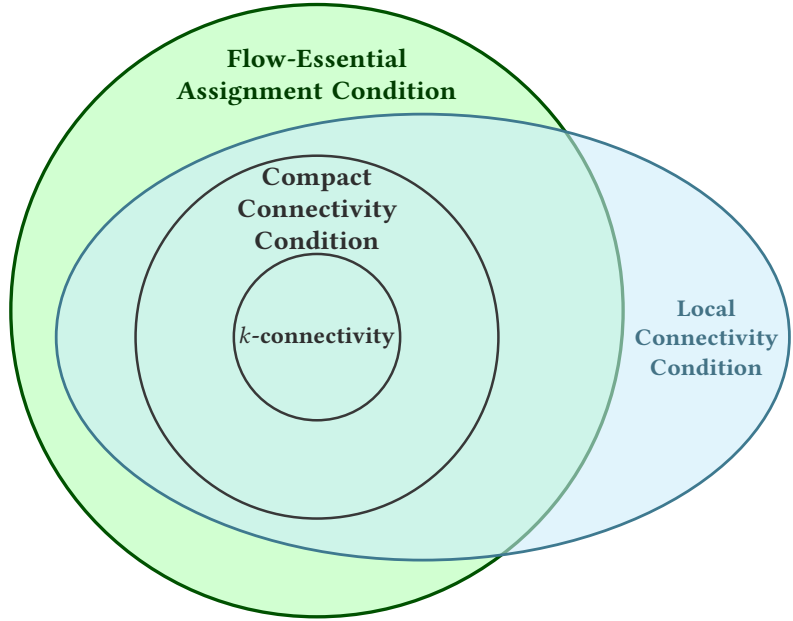

Although we showed that every instance satisfying compact connectivity also satisfies the \essentialassignmentcondition and can therefore be partitioned by our algorithm, preserving compact connectivity within the contract/remove framework is not always possible. In particular, some instances contain neither a non-critical edge that can be removed nor a pre-terminal that can be contracted while preserving compact connectivity. In contrast, the \essentialassignmentcondition is flexible enough to support contractions and edge removals while still capturing the structure needed to construct the desired partition.

We omit the counterexample and the proof that it does not contain any contractible pre-terminal or removable edge, as the proof is long and does not contribute to the main result of the paper. For the interested reader, we provide a code\footnote{\url{https://github.com/mahdi-jfri/Gyori-Lovasz-Codes}} that constructs and checks the validity of this stronger counterexample.

Instead, we present an example where all pre-terminals have out-degree two and there are no non-critical edges. Because no pre-terminal has out-degree one, there is no obvious contraction to perform; nevertheless, valid contractions are still possible despite this restriction. The stronger counterexample we omitted ensures that even these non-trivial pre-terminals of out-degree two cannot be contracted either. To construct that main, stronger counterexample, we utilize more copies of the same gadget to increase the load on the critical edges, ensuring that even contracting them fails to preserve the compact connectivity condition.

\begin{lemma}
    \label{lem:all-edges-compact-critical}

    There exist a directed graph $G=(V,E)$, a terminal set $T\subseteq V$
    with $|T|=9$, and capacities $\cp$ such that $G$ satisfies the compact
    connectivity condition with respect to $\cp$, every pre-terminal has
    out-degree exactly two, but every edge of $G$ is critical. That is, for every
    $e\in E$, the graph $G\setminus e$ fails to satisfy the compact connectivity
    condition with respect to $\cp$.
\end{lemma}
\begin{proof}

    Let $T=\{t_1,\ldots,t_9\}$. For every unordered pair
    $\{i,j\}\subseteq[9]$, create three pre-terminals
    \begin{equation*}
        P_{i,j}=\{p_{i,j}^1,p_{i,j}^2,p_{i,j}^3\},
    \end{equation*}
    and give each of them exactly two outgoing edges, one to $t_i$ and one
    to $t_j$.

    We add one forcing vertex for each of these pre-terminal-to-terminal
    edges. Fix such an edge $e=(p_{x,y}^r,t_y)$, where $t_x$ is the other
    out-neighbor of $p_{x,y}^r$. Choose four distinct indices
    $a,b,c,d\in[9]\setminus\{x,y\}$, and add a vertex $v_e$ with
    \begin{equation*}
        N^+(v_e)
        =P_{a,b}\cup P_{a,c}
        \cup\{p_{d,x}^1,p_{d,x}^2,p_{x,y}^r\}.
    \end{equation*}
    Such a choice is always possible because seven terminal indices remain.
    The graph has no other vertices or edges. In particular, every
    pre-terminal has out-degree exactly two, while every forcing vertex has
    out-degree nine. \Cref{fig:compact-connectivity-forcing-gadget} shows the forcing vertex $v=v_e$ targeting
    $(p_{5,6}^1,t_6)$ for the choice
    $(a,b,c,d,x,y)=(1,2,3,4,5,6)$.

    \begin{figure}[htbp]
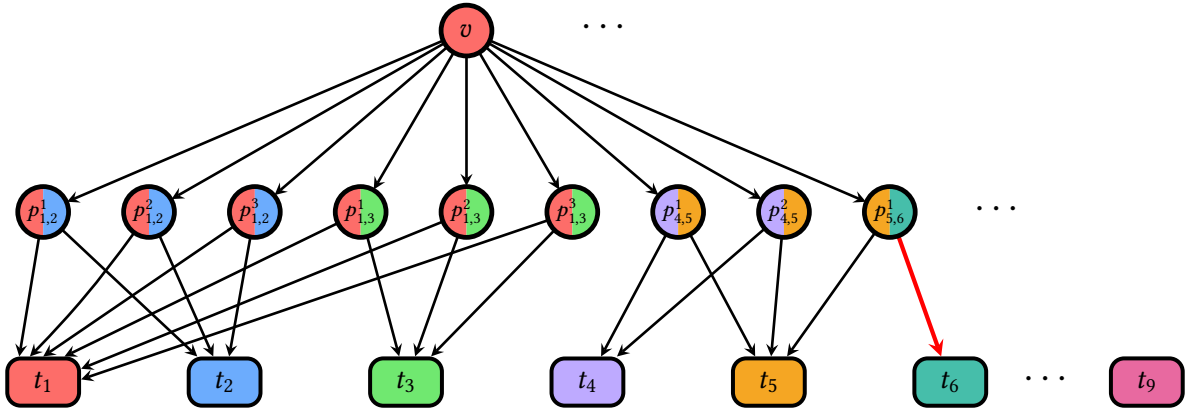

  \centering
  \begin{essgraph}[scale=1, >=stealth, node scale=1.2]

    \essterminal{t1}{(-7.2,0)}{$t_1$}
    \essterminal{t2}{(-4.8,0)}{$t_2$}
    \essterminal{t3}{(-2.4,0)}{$t_3$}
    \essterminal{t4}{( 0.0,0)}{$t_4$}
    \essterminal{t5}{( 2.4,0)}{$t_5$}
    \essterminal{t6}{( 4.8,0)}{$t_6$}

    \essterminal{t9}{(7.4,0)}{$t_9$}
    \essterminal[extra={opacity=0}]{t7}{(6.1,0)}{}
    \essterminal[extra={opacity=0}]{t8}{(6.1,0)}{}

    \essvertex[colors={t1/0.5,t2/0.5},
      extra={font=\scriptsize}
    ]{p121}{(-7.2,2.25)}{$p_{1,2}^{1}$}

    \essvertex[colors={t1/0.5,t2/0.5},
      extra={font=\scriptsize}
    ]{p122}{(-5.8,2.25)}{$p_{1,2}^{2}$}

    \essvertex[colors={t1/0.5,t2/0.5},
      extra={font=\scriptsize}
    ]{p123}{(-4.4,2.25)}{$p_{1,2}^{3}$}

    \essvertex[colors={t1/0.5,t3/0.5},
      extra={font=\scriptsize}
    ]{p131}{(-3.0,2.25)}{$p_{1,3}^{1}$}

    \essvertex[colors={t1/0.5,t3/0.5},
      extra={font=\scriptsize}
    ]{p132}{(-1.6,2.25)}{$p_{1,3}^{2}$}

    \essvertex[colors={t1/0.5,t3/0.5},
      extra={font=\scriptsize}
    ]{p133}{(-0.2,2.25)}{$p_{1,3}^{3}$}

    \essvertex[colors={t4/0.5,t5/0.5},
      extra={font=\scriptsize}
    ]{p451}{(1.2,2.25)}{$p_{4,5}^{1}$}

    \essvertex[colors={t4/0.5,t5/0.5},
      extra={font=\scriptsize}
    ]{p452}{(2.6,2.25)}{$p_{4,5}^{2}$}

    \essvertex[colors={t5/0.5,t6/0.5},
      extra={font=\scriptsize}
    ]{p561}{(4.0,2.25)}{$p_{5,6}^{1}$}

    \essvertex[colors={t1/1}]{v}{(-1.6,4.65)}{$v$}

    \begin{scope}[on background layer]

      \essedge[extra={->}]{v}{p121}
      \essedge[extra={->}]{v}{p122}
      \essedge[extra={->}]{v}{p123}

      \essedge[extra={->}]{v}{p131}
      \essedge[extra={->}]{v}{p132}
      \essedge[extra={->}]{v}{p133}

      \essedge[extra={->}]{v}{p451}
      \essedge[extra={->}]{v}{p452}
      \essedge[extra={->}]{v}{p561}

      \draw[->,line width=\EssEdgeLineWidth]
  (p121) -- ([xshift=-9pt]t1.north);
      \essedge[extra={->}]{p121}{t2}

\draw[->,line width=\EssEdgeLineWidth]
  (p122) to[bend left=4] ([xshift=-5pt]t1.north);
      \essedge[extra={->}]{p122}{t2}

      \draw[->,line width=\EssEdgeLineWidth]
  (p123) to[bend left=0] (t1.north);
      \essedge[extra={->}]{p123}{t2}

      \draw[->,line width=\EssEdgeLineWidth]
  (p131) -- ([xshift=8pt]t1.north);
      \essedge[extra={->}]{p131}{t3}

\draw[->,line width=\EssEdgeLineWidth]
  (p132) to[bend left=0]
  ([yshift=5pt]t1.east);
      \essedge[extra={->}]{p132}{t3}

      \draw[->,line width=\EssEdgeLineWidth]
  (p133) to[bend left=0]
  ([yshift=1pt]t1.east);
      \essedge[extra={->}]{p133}{t3}

      \essedge[extra={->}]{p451}{t4}
      \essedge[extra={->}]{p451}{t5}

      \essedge[extra={->}]{p452}{t4}
      \essedge[extra={->}]{p452}{t5}

      \essedge[extra={->}]{p561}{t5}
      \essedge[extra={->, ultra thick, red}]{p561}{t6}

    \end{scope}

    \node[font=\Large] at (6.1,0)     {$\cdots$};
    \node[font=\Large] at (5.45,2.25) {$\cdots$};
    \node[font=\Large] at (0.25,4.65) {$\cdots$};

  \end{essgraph}

\caption{The forcing gadget for the
edge $e=(p_{5,6}^1,t_6)$. In this figure, the inner colors indicate the terminals to which a vertex is compact-connected. Any compact-connection witness for $v$ consists of
nine internally vertex-disjoint paths and must therefore use all nine
out-neighbors of $v$. The triples $P_{1,2}$ and $P_{1,3}$ force $t_1$ to be
the only terminal permitted to receive multiple paths. The two paths through
$p_{4,5}^1$ and $p_{4,5}^2$ must then split between $t_4$ and $t_5$, forcing
the path through $p_{5,6}^1$ to use the highlighted edge $e$. Consequently,
$v$ is compact-connected only to $t_1$, and every witness for this connection
uses $e$. Relabeling the six displayed terminals gives the forcing vertex
$v_e$ for any pre-terminal-to-terminal edge $e$ in the construction.}
\label{fig:compact-connectivity-forcing-gadget}
\end{figure}

    We first establish the key forcing
    property by analyzing the displayed vertex $v$. Suppose that a family of nine
    internally vertex-disjoint paths witnesses that $v$ is compact-connected
    to some terminal $t$. Since $v$ has exactly nine out-neighbors, the nine
    paths must leave $v$ through all of them.

    Three of the paths pass through $P_{1,2}$ and can end only at $t_1$ or
    $t_2$. Thus, the terminal $t$ that is allowed to receive multiple paths
    must belong to $\{t_1,t_2\}$; otherwise these three paths would have only
    two possible endpoints, each permitted to receive at most one path. The
    three paths through $P_{1,3}$ similarly force
    $t\in\{t_1,t_3\}$. Hence $t=t_1$.

    It now remains to determine the last three paths. Because $t_1$ is the
    only terminal allowed to receive multiple paths, the two paths through
    $p_{4,5}^1$ and $p_{4,5}^2$ must split between $t_4$ and $t_5$. The path
    through $p_{5,6}^1$ therefore cannot also end at $t_5$ and must use
    $(p_{5,6}^1,t_6)$. Conversely, routing the six paths through
    $P_{1,2}\cup P_{1,3}$ to $t_1$, the two paths through $P_{4,5}$ to
    $t_4$ and $t_5$, and the final path to $t_6$ gives the required compact
    connection. We have proved both that $v$ is compact-connected only to
    $t_1$ and that every witness for this connection uses
    $(p_{5,6}^1,t_6)$.

    The same argument applies to every forcing vertex $v_e$: relabeling
    $(1,2,3,4,5,6)$ as $(a,b,c,d,x,y)$ shows that $v_e$ is
    compact-connected only to $t_a$, and every witness for this connection
    uses the edge $e$ that $v_e$ was created to target.

    We now
    define the capacities together with a witness. Let
    $\assgncompact:V\setminus T\to T$ assign every pre-terminal to one of its
    two terminal out-neighbors, choosing these assignments so that every
    terminal receives at least one pre-terminal, and assign each forcing
    vertex $v_e$ to its unique compact-connected terminal $t_a$. Set
    \begin{equation*}
        \cp_t=|\assgncompact^{-1}(t)|
        \qquad\text{for every }t\in T.
    \end{equation*}
    By the pre-terminal convention, every pre-terminal is compact-connected
    to its assigned terminal, and the forcing property proves the same for
    every $v_e$. Thus $\assgncompact$ is a witness for the compact
    connectivity condition. Moreover, it assigns every non-terminal exactly
    once, so $\sum_{t\in T}\cp_t=|V\setminus T|$. By
    \Cref{lem:compact-connectivity-implies-assignment}, $G$ satisfies the
    compact connectivity condition with respect to these capacities.
    
    It remains to prove that every
    edge is critical. Let $f$ be any edge of $G$. If $f$ leaves a forcing
    vertex $v_e$, then $v_e$ has only eight out-neighbors in $G\setminus f$. Since
    internally vertex-disjoint paths from $v_e$ must use distinct first
    out-neighbors, $v_e$ cannot have the nine paths required for a compact
    connection and is therefore compact-connected to no terminal.

    Now suppose that $f$ joins a pre-terminal to a terminal, and consider the
    forcing vertex $v_f$ created for this particular edge. The forcing
    construction gives $v_f$ exactly one compact-connected terminal, and
    every witness for that compact connection uses $f$. After $f$ is deleted,
    $v_f$ can no longer be compact-connected to that terminal; nor can it
    become compact-connected to another terminal, since deleting an edge
    cannot create a new path family. Thus $v_f$ is compact-connected to no
    terminal in $G\setminus f$.

    In either case, some non-terminal of $G\setminus f$ has no terminal to which it can
    be assigned in a witness for the compact connectivity condition. Hence no
    such witness assignment exists, and
    \Cref{lem:compact-connectivity-implies-assignment} implies that $G\setminus f$
    fails the compact connectivity condition. Since $f$ was arbitrary, every
    edge of $G$ is critical.
\end{proof}

\newcommand{\unaloc}{U}

\section{The \gyorilovasz theorem under local connectivity}
\label{sec:gyori-generalized}

In this section, we present a self-contained proof of \Cref{thm:glvs-generalized} by adapting and generalizing the configuration and cascade framework introduced by \citet{gyori1976division, hoyer2016gyori} to directed graphs under the local connectivity condition. 

To establish the theorem inductively, it suffices to prove the following incremental version, which demonstrates that any valid partial allocation of vertices can be extended to expand the component of a chosen terminal.

\begin{theorem}
	\label{thm:glvs-incremental}
	Let $G = (V, E)$ be a directed graph with terminal set $T = \{t_1, \dots, t_k\} \subseteq V$ that satisfies the local connectivity condition with respect to capacities $\cp$. Let $n_1, \dots, n_k$ be integers such that $n_i \le \cp_{t_i}$ for all $i \in [k]$ and $n_1 < \cp_{t_1}$. Let $V_1, \dots, V_k$ be pairwise disjoint vertex sets such that $t_i \in V_i$, $G[V_i]$ is connected to $t_i$ for all $i \in [k]$, and $|V_i \setminus \{t_i\}| = n_i$ for all $i \in [k]$.
	
	Then, there exist pairwise disjoint vertex sets $V_1', \dots, V_k'$ such that $t_i \in V_i'$, $G[V_i']$ is connected to $t_i$ for all $i \in [k]$, $|V_1' \setminus \{t_1\}| = n_1 + 1$, and $|V_i' \setminus \{t_i\}| = n_i$ for all $i \in [k] \setminus \{1\}$.
\end{theorem}

To prove \Cref{thm:glvs-incremental}, we define the unallocated set of non-terminal vertices as $\unaloc := V \setminus \bigcup_{i=1}^k V_i$. By hypothesis, $\unaloc$ is non-empty. We adapt the terminology from \citet{hoyer2016gyori} to track structural dependencies within the induced subgraphs $G[V_i]$.

\begin{definition} [Reservoir and cascade]
	For any $i \in [k] \setminus \{1\}$ and any vertex $v \in V_i \setminus \{t_i\}$, the reservoir of $v$, denoted by $R(v)$, is the set of all vertices in $V_i$ that can reach $t_i$ via a directed path in $G[V_i \setminus \{v\}]$. By definition, $v \notin R(v)$ and $t_i \in R(v)$.
	
	A cascade in $G[V_i]$ is a (possibly empty) sequence $w_1, w_2, \dots, w_m$ of distinct vertices in $V_i \setminus \{t_i\}$ such that $w_{j+1} \notin R(w_j)$ for all $1 \le j < m$.
\end{definition}

\begin{definition} [Configuration, rank, and bridges]
	A configuration consists of a valid choice of vertex sets $V_1, \dots, V_k$ such that $|V_i \setminus \{t_i\}| = n_i$ for all $i \in [k]$, along with exactly one cascade in each $G[V_i]$ for $i \in [k] \setminus \{1\}$. A vertex belonging to any of these cascades is a cascade vertex.
	
	The rank of a cascade vertex $w \in V_i$ is defined recursively:
	\begin{itemize}
		\item $\text{rank}(w) = 1$ if $w$ has a directed edge to a vertex in $V_1$.
		\item $\text{rank}(w) = r$ (for $r \ge 2$) if the minimum rank of any cascade vertex $w' \in V_j$ ($j \neq i$) such that $w$ has a directed edge to $R(w')$ is $r-1$.
	\end{itemize}
	If no such outgoing directed edge exists, the rank of $w$ is undefined.
	
	A configuration is valid if every cascade vertex has a well-defined rank that strictly increases along its cascade. For an integer $r \ge 1$, let $\rho_r$ denote the total number of vertices belonging to $R(w)$ for some cascade vertex $w$ of rank $r$.
	
	Finally, a bridge is a directed edge originating from a vertex in $\unaloc$ and terminating at a vertex within the reservoir of some cascade vertex. In a valid configuration, the rank of a bridge is the minimum rank among all cascade vertices whose reservoirs contain the target of the bridge.
\end{definition}

\begin{algorithm}[!ht]
	\caption{\textsc{GLIncremental}$(G,T,\cp,\langle V_1,\dots,V_k\rangle)$}
	\label{alg:gl-incremental}
	
	\KwIn{Directed graph $G=(V,E)$, terminals $T=\{t_1,\dots,t_k\}$, capacities $\cp$, and pairwise disjoint sets $\langle V_1,\dots,V_k\rangle$ satisfying the assumptions of \Cref{thm:glvs-incremental}}
	\KwOut{Sets $\langle V_1',\dots,V_k'\rangle$ satisfying the conclusion of \Cref{thm:glvs-incremental}}
	
	Initialize $\conf$ with parts $V_1,\dots,V_k$ and empty cascades \cmt{The initial valid configuration}
	
	\While{\Cref{lem:gl-cascade-extension} or \Cref{lem:gl-bridge-shift} is applicable}{
		\If{$\exists$ edge $(v, u)$ satisfying the premise of \Cref{lem:gl-cascade-extension}}{
			$\conf \gets \textsc{CascadeExtension}(\conf,v,u)$
		}
		\Else{
			$Z \gets \textsc{BridgeShift}(\conf)$ \cmt{$Z$ is a configuration or a final partition}
			\If{$Z$ is a final partition}{
				\Return $Z$
			}
			$\conf \gets Z$
		}
	}
	
	$\unaloc \gets V \setminus \bigcup_{i=1}^k V_i$ \cmt{Unallocated vertices in the terminal state of $\conf$}

	Choose an edge $(v, u)$ with $v \in \unaloc$ and $u \in V_1$ \cmt{Such an edge exists by the proof}
	$V_1 \gets V_1 \cup \{v\}$

	\Return $\langle V_1,\dots,V_k\rangle$
\end{algorithm}

We start from the given sets $V_1,\dots,V_k$ with all cascades empty. This is a valid configuration. For a valid configuration, let $\rho = (\rho_1,\dots,\rho_n)$. We repeatedly apply local modifications that preserve validity and strictly increase $\rho$ lexicographically. Since every coordinate $\rho_r$ is at most $|V|$, this process must terminate.

The constructive content of the proof is summarized in the following three procedures. The first repeatedly applies the two local update rules from \Cref{lem:gl-cascade-extension,lem:gl-bridge-shift}. The other two procedures implement the corresponding local changes to the current configuration.

The following monotonicity property of reservoirs will be used throughout.

\begin{lemma}
        \label{lem:gl-reservoir-monotonicity}
        Let $i \in [k] \setminus \{1\}$ and let $x,y \in V_i \setminus \{t_i\}$ be distinct. If $y \notin R(x)$, then
        $R(x) \cup \{x\} \subseteq R(y)$.

        In other words, if deleting $x$ prevents $y$ from reaching $t_i$, then deleting $y$ still leaves every vertex in $R(x)$, as well as $x$ itself, connected to $t_i$.
\end{lemma}
\begin{proof}
        Fix $z \in R(x)$. By definition, there is a directed path $P$ from $z$ to $t_i$ in $G[V_i \setminus \{x\}]$. We claim that $P$ does not contain $y$. Indeed, if it did, then the suffix of $P$ from $y$ to $t_i$ would avoid $x$, implying $y \in R(x)$, a contradiction. Hence $P$ avoids $y$ as well, so $z \in R(y)$. Therefore $R(x) \subseteq R(y)$.

        It remains to show that $x \in R(y)$. Since $G[V_i]$ is connected to $t_i$, there is a directed path $Q$ from $x$ to $t_i$ in $G[V_i]$. This path cannot contain $y$, since otherwise its suffix from $y$ to $t_i$ would avoid $x$, again implying $y \in R(x)$. Thus $Q$ avoids $y$, and hence $x \in R(y)$. Therefore $R(x) \cup \{x\} \subseteq R(y)$.
\end{proof}

\begin{algorithm}[!ht]
	\caption{\textsc{CascadeExtension}$(\conf,v,u)$}
	\label{alg:gl-cascade-extension}
	
	\KwIn{A valid configuration $\conf$ and an edge $(v, u)$ satisfying the premise of \Cref{lem:gl-cascade-extension}}
	\KwOut{A valid configuration $\conf'$ with lexicographically larger $\rho$}
	
	Let $i$ be the unique index with $v \in V_i \setminus \{t_i\}$

	Let $r \geq 0$ be minimum such that either $r=0$ and $u \in V_1$, or $r \geq 1$ and $u \in R(w)$ for some cascade vertex $w$ of rank $r$
	
	In the cascade of $V_i$, keep precisely the vertices of rank at most $r$

	Append $v$ to the cascade of $V_i$ \cmt{$v$ has rank $r+1$ by the minimality of $r$}

	Delete every cascade vertex of rank larger than $r+1$ from every cascade
	
	\Return the resulting configuration
\end{algorithm}

\begin{lemma}
	\label{lem:gl-cascade-extension}
	Let $\conf$ be a valid configuration. Suppose there exists an edge $(v, u)$ such that $v \in V_i \setminus \{t_i\}$ for some $i\in[k]\setminus\{1\}$, the vertex $v$ is neither a cascade vertex nor contained in any reservoir, and either $u \in V_1$ or $u \in R(w)$ for some cascade vertex $w$. Let $r \ge 0$ be minimum such that either $r=0$ and $v$ has a directed edge to a vertex of $V_1$, or $r \ge 1$ and $v$ has a directed edge to a vertex in $R(w)$ for some cascade vertex $w$ of rank $r$.
	
	Then there exists a valid configuration $\conf'$ with lexicographically larger $\rho$.
\end{lemma}

\begin{proof}
	In the cascade of $V_i$, keep precisely the vertices of rank at most $r$, then append $v$. If $V_i$ already contains a cascade vertex $x$ of rank $r+1$, delete $x$ as well. In every $V_j$, delete every cascade vertex of rank larger than $r+1$. The resulting sequences are still cascades: in $V_i$ this holds because $v$ was contained in no reservoir, so in particular it lies outside the reservoir of the last retained cascade vertex, and in every other $V_j$ we only truncate a suffix. By the minimality of $r$, the new vertex $v$ has rank exactly $r+1$.

	All reservoirs of ranks at most $r$ remain unchanged, so $\rho_1,\dots,\rho_r$ are unchanged. If there was no rank-$(r+1)$ cascade vertex in $V_i$ before the modification, then the new reservoir $R(v)$ is nonempty, so $\rho_{r+1}$ strictly increases. Otherwise, let $x$ be the former rank-$(r+1)$ cascade vertex of $V_i$. Since $v$ was contained in no reservoir, we have $v \notin R(x)$, and \Cref{lem:gl-reservoir-monotonicity} gives $R(x) \cup \{x\} \subseteq R(v)$. Hence the rank-$(r+1)$ reservoir in $V_i$ strictly grows. The global deletion of cascade vertices of rank larger than $r+1$ implies that $\rho_{r+2}=\dots=\rho_n=0$. Therefore $\rho_{r+1}$ strictly increases, while $\rho_1,\dots,\rho_r$ stay fixed. The new configuration is valid and has lexicographically larger $\rho$.
\end{proof}

\begin{algorithm}[!ht]
	\caption{\textsc{BridgeShift}$(\conf)$}
	\label{alg:gl-bridge-shift}
	
	\KwIn{A valid configuration $\conf$ containing a bridge}
	\KwOut{Either a valid configuration $\conf'$ with lexicographically larger $\rho$, or a final partition}
	
	\While{\textbf{true}}{
		Choose a bridge $(v, u)$ of minimum rank
		
		Choose a cascade vertex $w \in V_i$ of rank $r$ such that $u \in R(w)$ and the bridge rank is $r$
		
		Construct an in-arborescence of $G[V_i]$ rooted at $t_i$ whose subtree rooted at $w$ is exactly $V_i \setminus R(w)$
		
		\If{$V_i \setminus (R(w) \cup \{w\}) \neq \emptyset$}{
			Choose a leaf $x \neq w$ in the subtree rooted at $w$

			$V_i \gets (V_i \setminus \{x\}) \cup \{v\}$ \cmt{Replace $x$ by $v$}

			Delete every cascade vertex of rank larger than $r$ from every cascade
			
			\Return the resulting configuration
		}
		
		$V_i \gets (V_i \setminus \{w\}) \cup \{v\}$ \cmt{Replace $w$ by $v$}
		Delete every cascade vertex of rank at least $r$ from every cascade
		
		\If{$r=1$}{
			$V_1 \gets V_1 \cup \{w\}$ \cmt{$w$ has an edge to $V_1$}
			\Return $\langle V_1,\dots,V_k\rangle$
		}
		
		\cmt{$w$ is now the tail of a bridge of smaller rank, so the loop continues}
	}
\end{algorithm}

\begin{lemma}
	\label{lem:gl-bridge-shift}
	Let $\conf$ be a valid configuration. If $\conf$ contains a bridge, then either the conclusion of \Cref{thm:glvs-incremental} already holds, or there exists a valid configuration $\conf'$ with lexicographically larger $\rho$.
\end{lemma}

\begin{proof}
	Choose a bridge $(v, u)$ of minimum rank. Thus $v \in \unaloc$ and $u \in R(w)$ for some cascade vertex $w \in V_i$ of minimum possible rank $r$.

	We first describe a convenient spanning arborescence of $G[V_i]$. Since every vertex of $R(w)$ reaches $t_i$ in $G[V_i \setminus \{w\}]$, we may begin with an in-arborescence of $G[R(w)]$ rooted at $t_i$. Next, every vertex of $V_i \setminus R(w)$ reaches $w$ in $G[V_i \setminus R(w)]$, because every directed path from such a vertex to $t_i$ must pass through $w$ before visiting a vertex in $R(w)$. Thus we can add $w$ as a child of a suitable vertex of $R(w)$ on a directed path from $w$ to $t_i$, and then attach every vertex separated from $t_i$ by $w$ below $w$ along directed paths to $w$. In this way we obtain an in-arborescence of $G[V_i]$ rooted at $t_i$ whose subtree rooted at $w$ is exactly $V_i \setminus R(w)$.

	For every cascade vertex $y\in V_i$ of rank less than $r$, repeated applications of \Cref{lem:gl-reservoir-monotonicity} along the cascade give $R(y)\subseteq R(w)$ and $w\notin R(y)$, so neither possible removed vertex lies in $R(y)$. Adding $v$ cannot enlarge $R(y)$, since a new path through $v$ would give a bridge of rank less than $r$, contrary to the choice of $(v,u)$. Thus, both replacements preserve all lower-rank reservoirs.

	If $V_i \setminus (R(w) \cup \{w\})$ is nonempty, let $x$ be a leaf of the subtree of $w$ distinct from $w$. Replace $x$ by $v$, that is, remove $x$ from $V_i$ and add $v$ to $V_i$. Since $x$ is a leaf of the chosen arborescence, deleting $x$ preserves connectivity of all remaining vertices in $V_i$ to $t_i$. Since $u \in R(w)$ and $(v, u)$ is an edge, the new vertex $v$ also reaches $t_i$. Moreover, $x \notin R(w)$, while after the replacement the new vertex $v$ belongs to $R(w)$. Thus $\rho_r$ strictly increases. Finally, delete every cascade vertex of rank larger than $r$ from every cascade. The resulting configuration is valid and has lexicographically larger $\rho$.

	We are left with the case $V_i \setminus \{w\} = R(w)$. Then deleting $w$ from $V_i$ preserves connectivity of all remaining vertices in $V_i$ to $t_i$. Replace $w$ by $v$: remove $w$ from $V_i$, add $v$ to $V_i$, and delete every cascade vertex of rank at least $r$ from every cascade. The new vertex $v$ reaches $t_i$ through $(v, u)$ and a path from $u$ to $t_i$ inside $R(w)$. Now $w$ becomes a vertex of $\unaloc$.

	Because $w$ has rank $r$, either $r=1$ and $w$ has an edge to a vertex of $V_1$, or $r>1$ and $w$ has an edge to the reservoir of a cascade vertex of rank $r-1$. In the first case, adding $w$ to $V_1$ preserves connectivity to $t_1$, so the conclusion of \Cref{thm:glvs-incremental} holds. In the second case, $w$ is now the tail of a bridge of smaller rank. Repeating the same replacement step strictly decreases the rank each time, so after at most $r$ repetitions we either reach rank $1$ and finish, or encounter the first case above and obtain a lexicographically larger valid configuration.
\end{proof}

\begin{proof}[Proof of \Cref{thm:glvs-incremental}]
	Start from the initial valid configuration with empty cascades and repeatedly apply \Cref{lem:gl-cascade-extension,lem:gl-bridge-shift} as long as possible, unless the conclusion of the theorem has already been obtained. By the lexicographic monotonicity of $\rho$, this process cannot continue indefinitely.

	If there is an edge from some vertex of $\unaloc$ to a vertex of $V_1$, then adding that vertex to $V_1$ immediately proves the theorem. Hence, if the theorem has not yet been proved, we may assume that no such edge exists and that we have reached a valid configuration in which neither lemma is applicable.

	Let $I \subseteq [k] \setminus \{1\}$ be the set of indices whose cascades are nonempty. For each $i \in I$, let $z_i$ be the last vertex of the cascade in $V_i$, and define
	$$T^\star := \{t_1\} \cup \{t_i : i \in I\}, \qquad
	Y := V_1 \cup \bigcup_{i \in I} \bigl(R(z_i) \cup \{z_i\}\bigr).$$
	By repeated application of \Cref{lem:gl-reservoir-monotonicity} along each cascade, every cascade vertex of $V_i$ belongs to $R(z_i) \cup \{z_i\}$. Hence every cascade vertex lies in $Y$.

	We claim that no vertex of $V \setminus Y$ is locally connected to $T^\star$. Fix $x \in V \setminus Y$, and consider any directed path $P$ from $x$ to a terminal in $T^\star$. Since the target lies in $Y$, the path must enter $Y$ at some point. Let $(a, b)$ be the first edge of $P$ with $b \in Y$.

	If $b \in V_1$, then the iterative process would already have terminated if $a \in \unaloc$. Thus $a \notin \unaloc$. Since every cascade vertex lies in $Y$ and $a \notin Y$, the vertex $a$ is neither a cascade vertex nor contained in a reservoir. Therefore \Cref{lem:gl-cascade-extension} would apply with rank $0$, contradicting the choice of the configuration.

	If $b \in R(z_i)$ for some $i \in I$, then the edge $(a, b)$ cannot start in $\unaloc$, since otherwise it would be a bridge and \Cref{lem:gl-bridge-shift} would apply. Again, every cascade vertex lies in $Y$, so $a \notin Y$ implies that $a$ is neither a cascade vertex nor contained in a reservoir. Thus \Cref{lem:gl-cascade-extension} applies, another contradiction.

	Therefore the only possible first entry into $Y$ is through some vertex $z_i$. In particular, every directed path from $x$ to a terminal in $T^\star$ contains one of the vertices $z_i$ as an internal vertex. Thus, $x$ is not a pre-terminal of any terminal in $T^\star$.

	The set $\{z_i : i \in I\}$ has size $|I| = |T^\star|-1$. Therefore every family of internally vertex-disjoint paths from $x$ to terminals in $T^\star$ has size at most $|I|$, so $x \notin \localset{T^\star}$. We conclude that
	$$\localset{T^\star} \subseteq Y \setminus T.$$

	Now we count. Since $R(z_i) \cup \{z_i\} \subseteq V_i$, we have $\bigl|(R(z_i) \cup \{z_i\}) \setminus T\bigr| \le n_i$. Hence
	$$|\localset{T^\star}| \le |Y \setminus T|
	= |V_1 \setminus \{t_1\}| + \sum_{i \in I} \bigl|(R(z_i) \cup \{z_i\}) \setminus T\bigr|
	\le n_1 + \sum_{i \in I} n_i.$$
	Using $n_1 < \cp_{t_1}$ and $n_i \le \cp_{t_i}$ for every $i$, we obtain
	$$|\localset{T^\star}| < \cp_{t_1} + \sum_{i \in I} \cp_{t_i}
	= \sum_{t \in T^\star} \cp_t,$$
	which contradicts the local connectivity condition for the terminal set $T^\star$.

	This contradiction shows that the iterative process cannot terminate in a configuration in which neither lemma is applicable and no edge from $\unaloc$ enters $V_1$. Therefore, when the iterative process terminates before either lemma can be applied again, there is an edge from $\unaloc$ to $V_1$, and then we are simply done.
\end{proof}

\begin{proof}[Proof of \Cref{thm:glvs-generalized}]
	Start with the trivial partition $V_i = \{t_i\}$ for every $i \in [k]$. Repeatedly apply \Cref{thm:glvs-incremental}: at each step, choose a terminal whose current part still has fewer than its prescribed number of non-terminal vertices, relabel that terminal as $t_1$, and enlarge its part by one while keeping all other part sizes unchanged. Each application increases the total number of assigned non-terminals by exactly one and never exceeds the prescribed capacities. After exactly $\sum_{t \in T} \cp_t = |V \setminus T|$ applications, every terminal $t$ receives precisely $\cp_t$ non-terminals, so the resulting sets form a partition of $V$ with the sizes required by \Cref{thm:glvs-generalized}.
\end{proof}

\end{document}